\documentclass[%
superscriptaddress,
nofootinbib,
 amsmath,amssymb,
 aps,
prd,
11pt,a4paper]{revtex4-2}

\usepackage[british]{babel}
\usepackage[lmargin=2.25cm,rmargin=2.25cm,tmargin=2.5cm,bmargin=3.6cm]{geometry}

\usepackage{graphicx}
\usepackage{dcolumn}
\usepackage{bm}

\usepackage{amsfonts}
\usepackage{amsthm}
\usepackage[dvipsnames]{xcolor}
\usepackage{physics}
\usepackage{url}
\usepackage{tensor}
\usepackage{float}
\usepackage{mathrsfs}
\usepackage{mathtools}
\usepackage{bm}
\usepackage{cancel}
\usepackage{tikz}
\usepackage[shortlabels]{enumitem}
\usepackage[pdftex,hidelinks,colorlinks=true, allcolors=blue]{hyperref}
\usepackage{tikz-cd}

\usepackage[normalem]{ulem}

\usepackage{cleveref}

\theoremstyle{definition}
\newtheorem{defi}{Definition}[section]
\newtheorem*{defi*}{Definition}

\theoremstyle{plain}
\newtheorem{thm}[defi]{Theorem}
\newtheorem*{thm*}{Theorem}
\newtheorem{lem}[defi]{Lemma}
\newtheorem*{lem*}{Lemma}

\newtheorem*{clm*}{Claim}
\newtheorem{cor}[defi]{Corollary}
\newtheorem*{cor*}{Corollary}

\DeclareSymbolFont{matha}{OML}{txmi}{m}{it}
\DeclareMathSymbol{v}{\mathord}{matha}{118}
\numberwithin{equation}{section}

\def\fheq{\overset{\mathcal{H}^+}{=}}
\def\pheq{\overset{\mathcal{H}^-}{=}}

\def\ceq{\overset{\mathcal{C}(v)}{=}}

\newcommand{\tceq}{\overset{\tilde{\mathcal{C}}}{=}}
\def\hinteq{\overset{\mathcal{H}_{v_1}^{v_2}}{=}}

\newcommand{\numberthis}{\addtocounter{equation}{1}\tag{\theequation}}

\allowdisplaybreaks

  \makeatletter
\AtBeginDocument{%
  \let\l@subsubsection\@gobbletwo 
}
\makeatother

\begin{document}
\title{\Large Dynamical entropy of charged black objects} 
 
\author{Manus~R.~Visser}
\email{manus.visser@ru.nl}
\affiliation{Institute for Mathematics, Astrophysics and Particle Physics, and Radboud Center for Natural Philosophy, Radboud University, 6525 AJ Nijmegen, The Netherlands}
\author{Zihan~Yan}
\email{zy286@cam.ac.uk}
\affiliation{DAMTP, Centre for Mathematical Sciences, University of Cambridge, Wilberforce Road, Cambridge, U.K. CB3 0WA}

\date{\today}

\begin{abstract}

We develop a general framework for electromagnetic potential-charge contributions to the first law of black hole mechanics, applicable to dynamical first-order perturbations of stationary black objects with possibly non-compact bifurcate Killing horizons. Working in the covariant phase space formalism, we derive both comparison and physical process versions of the first law. We consider generic diffeomorphism-invariant  theories of gravity in $D$ spacetime dimensions,  containing non-minimally coupled abelian 
 $p$-form gauge fields. The pullback of the gauge field to the horizon is allowed to diverge while its field strength remains smooth, yielding   gauge-invariant electric potential-charge pairs in the first law. We further extend the construction to include magnetic charges by developing a bundle-covariant, gauge-invariant prescription that fixes the Jacobson-Kang-Myers ambiguity in the improved Noether charge.  Electric and magnetic charges are, respectively,  associated with non-trivial 
\((D - p - 1)\)- and \((p + 1)\)-cycles of the horizon cross-section,   whose homology classes determine the number of independent potential-charge pairs through the Betti numbers \(b_{D - p - 1}\) and \(b_{p + 1}\).
  Further, the dynamical gravitational entropy entering the first law is identified with the gauge-invariant part of the improved Noether charge,  giving a   gauge-invariant extension of the recent proposal by Hollands, Wald and Zhang \cite{Hollands:2024vbe}. We illustrate our framework with dyonic AdS black holes, dipole black rings, and charged black branes.  
  \end{abstract}

\maketitle

\newpage
\tableofcontents

\newpage

\section{Introduction}

Since the pioneering work of Bekenstein \cite{Bekenstein:1973ur} and Hawking \cite{Hawking:1975vcx} it has been understood that black holes are thermodynamic objects, with   associated energy, temperature and entropy.  Bekenstein proposed that  the entropy $S$ of 
  black hole solutions to general relativity  is proportional to the horizon area $A$, and Hawking discovered that the black hole temperature $T$ is proportional to the surface gravity  $\kappa $ of the Killing horizon, i.e.,   
\begin{equation}\label{bekensteinhawking}
    S_{\text{BH}} = \frac{A}{4G}\,  \qquad  \text{and} \qquad  T=\frac{\kappa}{2\pi}\,,
\end{equation} 
where $G$ is Newton's gravitational constant. 
The Bekenstein-Hawking entropy $S_{\text{BH}}$ satisfies a second law, due to the area increase theorem \cite{Hawking:1971tu}, and it obeys the so-called first law of black hole thermodynamics~\cite{Bekenstein:1973ur,Bardeen:1973gs,HaywardPRD}. 

Yet, there is an ambiguity in the Bekenstein-Hawking formula \eqref{bekensteinhawking} for black hole entropy, since there are multiple possible definitions of the horizon \cite{Hawking:1973uf,HaywardPRD,Ashtekar:2002ag,Ashtekar:2004cn,Bousso:2015mqa}. Two possible contenders that both satisfy an area theorem are the event horizon and the apparent horizon \cite{Hawking:1973uf}. For stationary black holes these two horizons coincide, but for   dynamical black holes their locations differ. That is, the apparent horizon of dynamical black holes must lie inside the event horizon if the null energy condition  holds and the weak cosmic censorship is valid. Hence, based on the second law of black hole thermodynamics, there seems to be multiple definitions of   black hole entropy, which possibly correspond to different coarse grainings in the microscopic theory (see, e.g., \cite{Kelly:2013aja,Engelhardt:2017aux,Engelhardt:2018kcs}). Importantly, however, the second law is not the only law of black hole thermodynamics  that defines black hole entropy; the  entropy also features in the first law of black hole thermodynamics. Therefore, the first law should   be taken into account in the search for the right definition(s) of black hole entropy. 
 
 The first law of black hole mechanics was originally   derived for stationary perturbations   of a stationary, axisymmetric black hole background geometry with a bifurcate Killing horizon~\cite{Bekenstein:1973ur,Bardeen:1973gs}. The stationary comparison version of the first law --- that  compares two different,  infinitesimally close,
stationary, axisymmetric, charged black hole geometries --- reads
 \begin{equation} \label{originalfirstlaw}
     \delta M  = T \delta S + \Phi \delta Q  + \Omega \delta J \,,
 \end{equation} 
where $M$ is the mass, $T$ is the Hawking temperature of the background Killing horizon, $\Phi$ is the electric potential, $Q$ is the electric charge, $\Omega$ is the angular velocity of the Killing horizon and  $J$ is the angular momentum.  
In general relativity, the entropy entering the first law for black holes is the Bekenstein-Hawking entropy \eqref{bekensteinhawking}. Although originally established for stationary perturbations~\cite{Bardeen:1973gs}, the validity of the first law with $S=S_{\text{BH}}$ extends to first-order non-stationary perturbations, provided the entropy is evaluated on the bifurcation surface (as demonstrated in \cite{Sudarsky:1992ty,Wald:1993ki} using the Hamiltonian formalism for Einstein-Yang-Mills theory).

More generally, the first law  \eqref{originalfirstlaw} applies   to any   classical diffeomorphism-invariant theory of gravity. For those general higher-curvature theories of gravity, Wald  \cite{Wald:1993nt} showed  that the entropy in the first law is given by  the Noether charge associated to the horizon generating
Killing field, integrated over a horizon cross-section.  
The Noether charge formula for black hole entropy is valid for arbitrary cross-sections of bifurcate Killing horizons in the case of stationary  perturbations, and it holds for first-order non-stationary perturbations as well, if the Noether charge is evaluated at the bifurcation surface \cite{Iyer:1994ys}. By contrast, for arbitrary horizon cross-sections of non-stationary black holes, that are first-order perturbations of stationary black holes,  the Noether charge no longer gives the correct  entropy in the first law of black hole thermodynamics.

Recently, Hollands, Wald, and Zhang (HWZ) \cite{Hollands:2024vbe} demonstrated that when the first law is generalised to non-stationary perturbations of a bifurcate Killing horizon and to arbitrary horizon cross-sections, the black hole entropy acquires a dynamical correction term, while the form of the first law \eqref{originalfirstlaw} remains unchanged. For general relativity, they proposed that the  \emph{entropy of dynamical black holes} is 
\begin{equation} \label{dynamicalentropygr}
    S_{\text{dyn}} =\left  (1- v \frac{\dd}{\dd v} \right)S_{\text{BH}}  \,.
\end{equation}
This dynamical black hole entropy was also studied in other recent work \cite{Rignon-Bret:2023fjq,Kar:2024dqk,Visser:2024pwz,Ciambelli:2023mir}. Here, $v$ is the (future-directed) affine parameter along the outgoing null geodesics of the future event
horizon, and $S_{\text{BH}}$ is the Bekenstein-Hawking entropy of the future event horizon. On the bifurcation surface (where $v=0$) and on Killing horizons (where $\dd A/ \dd v=0$)  $S_{\text{dyn}}$   reduces to $S_{\text{BH}}$. However, for arbitrary affine null time $v$ the entropy of dynamical black holes is smaller than $S_{\text{BH}}$ of the event horizon, if the derivative of the horizon area $\dd A/\dd v$  is positive (which follows from
the area theorem). 
Moreover, HWZ proved, to first order in perturbation theory, the  dynamical black hole entropy \eqref{dynamicalentropygr} is equal to the Bekenstein-Hawking entropy of the apparent horizon (see also \cite{Visser:2024pwz} for a pedagogical proof). Therefore, at least in this restricted setting, it seems that the area of the apparent horizon yields the correct thermodynamic entropy of black holes, and not the area of the event horizon.   

Furthermore, HWZ generalised  the non-stationary comparison version of the first law to arbitrary diffeomorphism-invariant theories   of gravity for which the metric
is the only dynamical field. In this case the dynamical black hole entropy is given by  \cite{Hollands:2024vbe}
\begin{equation}\label{dynentropywall}
    S_{\text{dyn}} =\left  (1- v \frac{\dd}{\dd v} \right)S_{\text{Wall}}  \,.
\end{equation}
The defining relation for $S_{\text{Wall}}$, introduced by Wall in \cite{Wall:2015raa}, is that its second $v$-derivative is related to the $vv$-component of the gravitational field equations
\begin{equation} \label{definingrelationwall}
    \partial_v^2 \delta S_{\text{Wall}} = - 2\pi \int_{\mathcal C(v)} \dd{A} \delta E_{vv}\,,
\end{equation}
where $\mathcal C(v)$ is the horizon cross-section at fixed affine parameter $v $ with area element $\dd A$. The entropy  $S_{\text{Wall}} $ has been computed    explicitly by Wall \cite{Wall:2015raa} for Lagrangians that are contractions of the metric and Riemann tensor, so-called $f$(Riemann) theories of gravity. Notably, this expression differs from the Noether horizon entropy by Wald \cite{Wald:1993nt,Iyer:1994ys} due to extrinsic curvature corrections, which, mysteriously, match the extrinsic curvature corrections in holographic entanglement entropy for $f$(Riemann) theories derived by Dong \cite{Dong:2013qoa}.  Subsequent work established that the Wall entropy is gauge invariant to linear order in perturbations, and extended the construction in effective field theory to second order in perturbation theory \cite{Hollands:2022fkn} (see also \cite{Davies:2022xdq}) as well as to non-perturbative settings \cite{Davies:2023qaa}. Moreover, 
 $S_{\text{Wall}}$
  has been generalised to arbitrary diffeomorphism-covariant theories of gravity with non-minimal couplings to scalar fields, gauge fields \cite{Biswas:2022grc}, (non-gauged) vector fields~\cite{Wall:2024lbd},  and higher-spin fields \cite{Yan:2024gbz}. Finally, a covariant entropy current associated with the Wall entropy for diffeomorphism-invariant theories of gravity was constructed in \cite{Hollands:2024vbe} (see also \cite{Hollands:2022fkn,Bhattacharyya:2021jhr}).

From the defining relation \eqref{dynentropywall}, it follows   that the Wall entropy obeys a second law in higher-curvature gravity, valid for linear non-stationary perturbations of a Killing horizon \cite{Wall:2015raa}. This result holds under the assumptions of the null energy condition and the teleological boundary condition that the perturbation settles into a stationary state as $v\to +\infty.$ Moreover, the dynamical black hole entropy also obeys a linearised second law, since combining \eqref{dynentropywall} and \eqref{definingrelationwall}, and imposing the linearised gravitational field equations $\delta E_{vv} = \delta T_{vv}$, yields 
\begin{equation} \label{linearisedsecondlaw1}
    \partial_v \delta S_{\text{dyn}} = 2\pi \int_{\mathcal C(v)} \dd{A} v  \delta T_{vv} \,.
\end{equation}
Now, if the perturbation is sourced by a stress-energy tensor that satisfies the null energy condition,
$\delta T_{vv}\ge 0$,  then $S_{\text{dyn}}$ obeys a   linearised second law
$ \partial_v \delta S_{\text{dyn}} \ge 0$. This is a local version of the second law, because the entropy changes after matter crosses the horizon. On the other hand, $S_{\text{Wall}}$ and $S_{\text{BH}}$ already change 
before matter crosses the horizon due to the teleological condition at $ v=+\infty$.    Consequently, even though $S_{\text{Wall}}$ and $S_{\text{dyn}}$ both satisfy a (linearised) second law, they exhibit different physical responses to energy influx. Importantly, the entropy relevant to the first law for non-stationary perturbations of stationary black hole backgrounds   is $S_{\text{dyn}}$, when evaluated on an arbitrary horizon cross-section. 

The dynamical black hole entropy is defined in its most general form in the covariant phase space formalism, a.k.a., the Noether charge formalism \cite{Wald:1993nt,Iyer:1994ys}. Compared to Wald's original definition of stationary black hole entropy as the horizon Noether charge, HWZ defined   $S_{\text{dyn}}$ with an additional dynamical correction term  \cite{Hollands:2024vbe}
\begin{equation} \label{improvednoetherdynentr}
    S_{\text{dyn}} = \int_{\mathcal C(v)} \frac{2\pi}{\kappa_3} (\mathbf Q_\xi - \xi \cdot \mathbf B_{\mathcal H^+})\,.
\end{equation}
Here, $\mathbf Q_\xi$ is the Noether charge associated to the horizon generating Killing vector field $\xi$, and $\mathbf B_{\mathcal H^+}$ is a codimension-1 differential form  defined on the future event horizon $\mathcal H^+$, that vanishes on the background Killing horizon and whose variation equals the symplectic potential, cf.~\eqref{entropyconditiona}. Further,    $\kappa_3$ is the surface gravity of the   event horizon, defined with respect to the background horizon Killing field $\xi$, cf.~\eqref{kappa3}. The combination of terms between brackets in \eqref{improvednoetherdynentr} is also known as the improved Noether charge $\tilde{\mathbf Q}_\xi$, hence the entropy of dynamical black holes is the horizon integral of the improved Noether charge (times the inverse temperature $2\pi / \kappa_3$). 

HWZ established that formula  \eqref{improvednoetherdynentr} for $S_{\text{dyn}}$ obeys the non-stationary first law for arbitrary diffeomorphism-invariant theories of gravity for which the metric is the only dynamical field. In a follow-up paper \cite{Visser:2024pwz}, we generalised this result  to  include non-minimal bosonic matter fields that are stationary in the background,   whose pullback to the horizon is smooth and that satisfy    an asymptotically flat fall-off condition (although the last condition was not made explicit). Moreover, our gauge conditions for the perturbations were less restrictive, since we did not set the variations of the horizon Killing field and the surface gravity to zero, thus allowing for variations of the black hole temperature.  Further, a substantial consistency check that we did was evaluating \eqref{improvednoetherdynentr}
explicitly for 
  $f(\text{Riemann})$ theories of gravity and showing that it matches the definition in \eqref{dynentropywall}. 

An important caveat in the literature on dynamical black hole entropy is the absence of  
an electromagnetic potential-charge term  of the form $\Phi \delta Q$ in the non-stationary comparison version of the first law. We intend to fill this gap in the present paper. We restrict to charges associated to abelian  $p$-form gauge fields and mostly focus on electric charges in this paper, but we also have a proposal for the inclusion of magnetic charges in the first law.  For the non-stationary physical process version   of the first law in general relativity, Rignon-Bret \cite{Rignon-Bret:2023fjq} already extended the derivation to include  an electric potential-charge  term in the case of minimally coupled Maxwell fields. In his formulation, the $\Phi_{\mathcal H}   \delta Q$ term arises directly from the matter Killing energy flux across the horizon. We   take a different perspective on the physical process first law where the electric charge   does not only arise from the matter stress-energy tensor,  but is   an intrinsic property of the horizon itself, since in the   non-minimally coupled case there is no clear division between the gravitational and electromagnetic sector. 
We derive a form of the physical process first law in which both the horizon charges and matter charges --- given by an integral of the Hodge dual of the $p$-form electric current sourced by external matter fields  --- appear explicitly on each side of the relation. Since these two charge terms are equal, our formulation reduces to that of \cite{Rignon-Bret:2023fjq,Hollands:2024vbe},  which locally is given by \eqref{linearisedsecondlaw1}.  

The reason the electric charge term was missing    in our previous   derivation \cite{Visser:2024pwz} of   the comparison first law is the following.  If the pullback of an abelian one-form (Maxwell) gauge potential $\mathbf A$ to a bifurcate Killing  horizon~$\mathcal H$ is smooth, and   the stationarity gauge condition $\mathcal L_\xi \mathbf A = 0$ is imposed,   then the electric potential at the horizon, $\Phi_{\mathcal H} = - \xi \cdot \mathbf A |_{\mathcal H}$,  vanishes.  That is because   $\xi$ vanishes on the bifurcation surface, while $\Phi_{\mathcal H}$ is constant along the horizon whenever $\mathcal L_\xi \mathbf A = 0$ \cite{Carter:2009nex,Gauntlett:1998fz,Gao:2001ut,Smolic:2014swa} (see also section \ref{ssec:zeroth-law} for a proof of the constancy of $\Phi_{\mathcal H}$). Thus, smoothness of the pullback of $\mathbf A$ forces $\Phi_{\mathcal H}=0$   everywhere on~$\mathcal H$, eliminating any    $\Phi_{\mathcal H} \delta Q$ contribution in the  first law. Moreover, the asymptotically flat fall-off condition for the gauge field, $A_a \to \mathcal O(1/r^{D-3})$ as $r \to \infty$, ensures that the electric potential at   infinity, $\Phi_\infty = -\xi \cdot \mathbf A |_{r\to \infty}$, also vanishes. Hence, under the assumptions in \cite{Visser:2024pwz},\footnote{Admittedly, the assumptions for the gauge field --- $\mathcal L_\xi \mathbf A =0$, $\mathbf A |_{\mathcal H}$ is smooth and $\mathbf A |_{r\to \infty} = \mathcal O(1/r^{D-3})$ ---  were so strong that  the Maxwell field strength $\mathbf F = \dd \mathbf A$ vanishes identically and the gauge field   is globally pure gauge.} both horizon and asymptotic potentials vanish, so that the first law contains no electric   potential-charge term of the form $\bar \Phi \delta Q$, with $\bar \Phi \equiv \Phi_{\mathcal H}- \Phi_\infty$. 
   
There are three ways to relax these assumptions and thereby obtain a non-vanishing electric potential-charge term in the first law:
\begin{itemize}
    \item[a)] \emph{Non-smoothness at the horizon}: one may allow the pullback of $\mathbf A$ to the horizon to diverge at the bifurcation surface, while keeping $\xi \cdot \mathbf A$  finite and nonzero at the bifurcation surface. In this case  the horizon potential does not vanish, $\Phi_{\mathcal H} \neq 0 $. Moreover, the field strength can remain regular at the entire Killing horizon (see section \ref{sec:nonsmoothness}).
    \item[b)] \emph{Stationary up to gauge}: Instead of imposing $\mathcal L_\xi \mathbf A =0$, one can allow the gauge field to be stationary up to a pure gauge, $\mathcal L_\xi \mathbf A = \dd \lambda_\xi$. This is allowed since the   physical field strength  is still stationary, $\mathcal L_\xi \mathbf F = \dd (\mathcal L_\xi \mathbf A)=0$. In this case,     the gauge-invariant horizon potential $\Phi_{\mathcal H}\equiv (- \xi \cdot \mathbf  A + \lambda_\xi)|_{\mathcal H}$ can be nonzero, because $\lambda_\xi$ might not vanish at $\mathcal H$, even if the pullback of $\mathbf A$ to the horizon is smooth (so that $\xi \cdot \mathbf  A |_{\mathcal H}=0$).
    \item[c)] \emph{Non-vanishing asymptotic potential}: If the gauge field does not vanish at spatial infinity, then the asymptotic electric potential is nonzero, $\Phi_\infty \neq 0 $.
\end{itemize}

In earlier work \cite{Gao:2003ys} on the stationary comparison version of the  first law, Gao considered a non-smooth $\mathbf A$ at the horizon, while keeping the standard stationarity and asymptotic fall-off conditions. This leads to an extension of the Noether charge derivation \cite{Wald:1993nt,Iyer:1994ys} of the first law for stationary perturbations  to both abelian and non-abelian gauge fields, thereby recovering an electric potential-charge term   (for the abelian case). In the present paper we adopt a more general perspective: we allow for all three types of relaxed gauge conditions simultaneously, which together lead to a black hole first law  that includes a gauge-invariant electric potential-charge term. 

The main goal of this work is therefore to establish a general framework for including electric potential-charge contributions in the non-stationary first law of black holes. Our analysis applies to arbitrary  diffeomorphism-invariant theories of gravity in which the metric is non-minimally coupled to bosonic matter fields, and in particular to abelian gauge fields. As in~\cite{Hollands:2024vbe,Visser:2024pwz}, we consider non-stationary perturbations of a stationary black hole background with a bifurcate Killing horizon. We allow the pullback of the gauge fields to be divergent at the horizon, while the gauge-invariant field strength is regular everywhere and stationary in the background. We derive two complementary versions of the first law --- the non-stationary comparison version and the physical process version --- both using the covariant phase space formalism. A key feature is that we  do not restrict the  behaviour of gauge field perturbations at the horizon,  similar to Gao's \cite{Gao:2003ys} derivation  of the stationary first law. On the other hand, we do impose standard fall-off conditions at spatial infinity, both for asymptotically flat and anti-de Sitter black holes (see section \ref{sec:falloff}).  These conditions as well as other gauge conditions for the metric perturbations (see section \ref{sec:setup} and section 2.2 in \cite{Visser:2024pwz})  are necessary to have a well-defined first law.

In addition to this central aim, we achieve three further generalisations compared to \cite{Hollands:2024vbe,Visser:2024pwz}. First, we extend the non-stationary first law to arbitrary 
 $p$-form abelian gauge fields, thereby capturing higher-form charges and potentials within the same framework. Second, while previous work only considered compact horizon cross-sections, we allow for non-compact cross-sections. This extends the applicability of our first law to a broad class of (higher-dimensional) \emph{black objects}, including planar and hyperbolic black holes, black strings, black branes, blackfolds, and black rings (see~\cite{Emparan:2008eg} for a review).  In such cases, the first law naturally  may incorporate a spectrum of charges associated with the non-trivial topology of the horizon cross-section  (see section \ref{sssec:non-compact}).

Third, apart from deriving the electric potential-charge contributions to the first law, we also incorporate magnetic monopole charges associated to abelian gauge fields in a fully covariant manner. 
Electric charges arise from flux integrals of the Hodge dual of the field strength,
$
    Q \propto \int \!\star \mathbf{F},
$ 
whereas magnetic   charges are obtained from integrals of the field strength itself,
$ 
    P \propto \int \!\mathbf{F}.
$
The latter are nonzero only when the field strength is closed but not globally exact, 
\(\dd \mathbf{F} = 0\) with \(\mathbf{F} \neq \dd \mathbf{A}\). This means there is an associated 
  non-trivial principal fiber bundle  for one-form gauge fields and     a non-trivial ``gerbe'' (higher principal bundle) for higher-form abelian gauge fields. 
In this bundle  formulation, the field strength --- the curvature of the  bundle/gerbe   --- remains smooth everywhere, 
while the gauge potential --- the connection  ---  can be defined only locally on overlapping patches related by gauge transformations. 
In our previous work~\cite{Visser:2024pwz}, we assumed instead that the gauge potential was a single globally defined smooth field on spacetime. 
This assumption restricts the gauge bundle to be topologically trivial, so that the field strength is globally exact, 
\(\mathbf{F} = \dd \mathbf{A}\), and no magnetic flux can exist. 
Consequently, magnetic monopoles --- and their associated potential--charge terms in the first law --- were   absent in that framework. 

In the present work, we relax this globally exactness condition and develop a bundle-covariant and gauge-invariant prescription for magnetic potential-charge terms that extends the covariant phase space formalism to non-trivial gauge bundles  (a similar framework was developed for non-abelian gauge fields by Prabhu \cite{Prabhu:2015vua}).  In our approach, the magnetic contribution to the first law is obtained by fixing the so-called Jacobson-Kang-Myers ambiguity of covariant phase space forms~\cite{Jacobson:1993vj}, which   affects the definition of the improved Noether charge. This ambiguity is an exact term, which, when integrated over a horizon cross-section, localises on a codimension-three boundary of overlapping patches, if the gauge potential is not smooth on the entire   cross-section. We show the ambiguity can be   fixed by the requirement of bundle covariance. The resulting expression for the magnetic term is patch- and gauge-independent, and correctly reduces to zero when the   bundle is trivial (see section \ref{ssec:magnetic}).

The rest of the paper is organised as follows. We first summarise the  non-stationary comparison first law with an electric potential-charge term for the simpler case of compact horizon slices,   illustrating our main results. In section \ref{sec:geometric} we describe the geometric setup in more detail, collect some facts about $p$-form gauge fields and the   electric potential $(p-1)$-form on Killing horizons, and state the asymptotic fall-off conditions. Further, in section \ref{sec:gaugedepen} we work out the covariant phase space formalism for Lagrangians that depend  on abelian gauge field strengths, where we distinguish between gauge-dependent and gauge-invariant parts of covariant phase space quantities. The   comparison first law is derived  in section \ref{sec:cfl}, and the physical process first law     in section \ref{sec:ppfl}. Next, in section \ref{sec:ambi} we describe the differential form ambiguities in the covariant phase space quantities, and we show for non-compact horizons that the dynamical black hole entropy is unambiguous, assuming appropriate  boundary conditions at the asymptotic ends of the extended horizon directions. In section \ref{sec:exam} we illustrate the potential-charge term for three different black objects, and we end with an outlook for future research. We have three appendices where we collect some technical details of our derivations. 

\newpage
\subsection{Summary of results}
For non-stationary first-order perturbations of   stationary black holes in $D$ spacetime dimensions with  compact horizon cross-section~$\tilde{\mathcal C}$,   the  comparison version of the    first law reads
\begin{equation} \label{mainresultfirstlaw}
    T \delta S_{\text{dyn}} = \delta M - \sum_I \Omega_I \delta J_I - \sum_\sigma \bar \Phi_\sigma \delta Q_\sigma\,,
\end{equation}
where $\Omega_I $ is the angular velocity conjugate to the angular momentum $J_I$,  and  $\bar \Phi_\sigma \equiv \Phi_\sigma^{\mathcal H^+} - \Phi_\sigma^\infty$ is the electric potential difference between the value at the horizon and at infinity,  respectively defined as
\begin{equation}
\Phi_\sigma^{\mathcal{H}^+} = \int_{\tilde \sigma} \mathbf{\Phi} |_{\tilde{\mathcal C}}   \qquad \text{and}\qquad \Phi_\sigma^{\infty} = \int_{\tilde \sigma_\infty} \mathbf{\Phi} |_{\mathcal S_\infty} \,,
\end{equation}
with $\mathbf{\Phi} = - \xi\cdot \mathbf{A} + \bm{\lambda}_\xi$  being the electric potential $(p-1)$-form, associated to the $p$-form gauge field~$\mathbf A$.
The conjugate charges $Q_\sigma$ are the higher-form electric charges associated to  $\mathbf A$:
\begin{equation}
    Q_\sigma = - \int_\sigma \mathbf \Upsilon \,.
\end{equation}
Here,   $\mathbf{\Upsilon}$ is an electric flux density ($D-p-1$)-form, obtained from the variation of the action with respect to the field strength (see  \eqref{upsilon} for a definition). Since $\mathbf \Upsilon$ only depends on gauge-invariant fields, the electric charge is gauge invariant. Further, the potential form  $\bm \Phi$ is gauge invariant up to an exact form, which vanishes when integrated over the compact horizon cycle $\tilde \sigma$, hence the potentials are also gauge invariant. 

The definitions of the different cycles are as follows. The cycle $\sigma $ denotes a non-trivial ($D-p-1$)-cycle of   the horizon cross-section $\tilde{\mathcal C}$, that is dual (by Poincar\'e duality) to the harmonic basis form   $\hat{\bm{\eta}}_\sigma$ of   $\mathbf \Phi $. Further, we assume that $\sigma$ is homologous,  inside a spacelike hypersurface with horizon boundary $\tilde{\mathcal C}$ and asymptotic boundary $\mathcal S_\infty$, to a cycle at spatial infinity $\mathcal S_\infty$. This implies the charge variations at the horizon and at infinity are the same (denoted as $\delta Q_\sigma$)  due to Gauss's law. Moreover,  
  $\tilde \sigma$ is a     $(p-1)$-cycle  of $\tilde{\mathcal C}$ that is (Poincar\'e) dual to $\star_{\tilde{\mathcal C}} \hat{\bm \eta}_\sigma$, the Hodge dual of $\hat{\bm \eta}_\sigma$ within~$\tilde{\mathcal C}$,  and $\tilde \sigma_\infty$ is a ($p-1$)-cycle of $\mathcal S_\infty$ that is dual to $\star_{\mathcal S_\infty}\hat{\bm \eta}_{\sigma_\infty}$. The number of potential-charge pairs is equal to the number of non-trivial ($p-1$)-cycles, and is   counted by the $(p-1)$-th Betti number $b_{p-1}$ of the homology group $H_{p-1} (\tilde{\mathcal{C}})$ of the horizon cross-section. By Poincar\'{e} duality this is equal to the Betti number $b_{D-p-1}$ (see section \ref{sssec:compact}). 

As a trivial illustration, for an electrically charged, spherically symmetric black hole in Einstein-Maxwell theory, the potential reduces   to the standard scalar electrostatic potential (since $p=1$). 
The electric charge   is given by the Gauss-law flux of the $(D-2)$-form \( \mathbf \Upsilon \propto -\star \mathbf  F \) through a horizon cross-section \( \sigma = S^{D-2}_{\mathcal{H}} \). 
Because the Maxwell field satisfies the sourceless equation \( \dd(\star \mathbf F) = 0 \), the charge is conserved and can equivalently be evaluated over the \((D-2)\)-sphere at infinity, 
thereby reproducing the standard form of the first law for charged black holes (see section \ref{sec:examplA}).

A more interesting example is   a five-dimensional charged black ring solution to Einstein-Maxwell-Dilaton theory with   a two-form gauge field~\cite{Emparan:2004wy,Emparan:2001wn}. In this case, the electric potential $\mathbf \Phi$ is a one-form and the electric flux density $\mathbf \Upsilon$ is a two-form. The black ring has horizon cross-section $\tilde{\mathcal C} =  S^1\times S^2$, so the non-trivial cycles of $\tilde{ \mathcal C}$ are $\tilde \sigma =S^1$ and $\sigma = S^{2}$. The non-trivial $S^2$ cycle then supports a local dipole charge on the horizon.  Since spatial infinity has topology of $S^3$, which does not contain non-trivial  two-cycles (or one-cycles), the charge has no asymptotic counterpart, yet the horizon charge contributes to the thermodynamics of the black ring (see section \ref{sec:blackrings}). 

Another,  more involved example is a six-dimensional doubly-spinning toroidal blackfold with horizon topology $T^2 \times S^2 = S^1_a \times S^1_b \times S^2$ that is charged under a two-form gauge field \cite{Emparan:2009at,Emparan:2009vd}. The electric flux density $\mathbf \Upsilon$ is now a three-form that can be integrated over two different non-trivial three-cycles. The two non-trivial cycles dual to the harmonic part of $\mathbf \Phi$ are $\sigma_a = S^1_a \times S^2$ and $\sigma_b = S_b^1 \times S^2$, which give rise to two   species of charges, $Q_a$ and $Q_b$, and two conjugate electric potentials that are integrated over the remaining circle, $S_b^1$ and $S_a^1$, respectively. Indeed, the first Betti number of $T^2 \times S^2$ is two, so there are two potential-charge pairs (see section \ref{compacthorizontopologies}).

A noteworthy feature of our derivation of the first law is that terms of the form  $Q \delta \Phi$ --- actually, the integral of $\mathbf{\Upsilon} \wedge \delta (\xi \cdot \mathbf{A})$ --- cancel from the final expression, in agreement with earlier work for stationary perturbations \cite{Bardeen:1973gs,Gao:2003ys,Compere:2006my}. This cancellation closely parallels the cancellation of the 
$\delta \kappa$ term in the non-stationary first law that we identified in our previous work \cite{Visser:2024pwz} and that was observed  for stationary perturbations in the original work \cite{Bardeen:1973gs}. As a consequence, the first law can be written solely in terms of well-defined conjugate pairs, with variations acting on only one member of each pair. This structure highlights the thermodynamic consistency of the first law.

Crucially, one of our main results is that the dynamical black hole entropy in \eqref{mainresultfirstlaw} is defined as the gauge-invariant part of the improved Noether charge 
\begin{equation}   \label{finaldynengaugeinvariant}
    S_{\text{dyn}} = \int_{\tilde{\mathcal C}} \frac{2\pi}{\kappa_3} (\mathbf Q_\xi^{\text{GI}} - \xi \cdot \mathbf B^{\text{GI}}_{\mathcal H^+})\,.
\end{equation}
The Noether charge $ \mathbf Q_\xi$ and the horizon form $\mathbf B_{\mathcal H^+}$ can in general have gauge-dependent terms, so the proposal    is that only the gauge-invariant terms --- which depend, among others, on $\mathbf F$ and not on~$\mathbf A$ ---  contribute to $S_{\text{dyn}}$.  Hence, our definition of dynamical black hole entropy coincides with the HWZ proposal \eqref{improvednoetherdynentr} in the purely gravitational case, while for non-minimally coupled gauge fields it provides a natural gauge-invariant extension, consistent with the physical requirement that thermodynamic entropy itself must be gauge invariant. Moreover, for Killing horizons   formula \eqref{finaldynengaugeinvariant} reduces to   Wald's Noether charge horizon entropy, which is manifestly gauge invariant since it only depends on functional derivatives of the gauge-invariant Lagrangian  \cite{Wald:1993nt,Iyer:1994ys}. Finally, since \eqref{finaldynengaugeinvariant}  depends only on gauge-invariant quantities (which are smooth on the horizon), the results from  our previous paper  \cite{Visser:2024pwz} imply that the dynamical black hole entropy is universally related to the Wall entropy through \eqref{dynamicalentropygr}. 

\subsection{Conventions}

 Our conventions mainly follow those in Wald's textbook \cite{Wald:1984rg}. We assume a mostly positive signature metric, and work in $D$ spacetime dimensions. We use Latin letters $a,b,c,\dots$ to denote abstract spacetime indices, and $i,j,k, \dots$ to denote codimension-2 spatial indices. For non-compact horizon cross-sections, we use Greek letters $\alpha_1, \alpha_2,\cdots, \beta_1,\beta_2,\cdots$ to denote the indices of compact directions, while we employ Hebrew letter $\aleph_1,\aleph_2,\cdots$ to denote the indices of the non-compact directions.  On the horizon,  $k^a$ and  $l^a$ are the tangent vectors to  the outgoing and ingoing affinely parameterised null geodesics orthogonal to the cross-sections. We also introduce Gaussian null coordinates $(v,u,x^i)$ near the future event horizon. Then, on the horizon, $k^a = (\partial_v)^a$, $l^a = (\partial_u)^a$. For differential forms, we use boldface notation (e.g.~$\mathbf X$), but we write them in normal font when their indices are explicitly shown (e.g.~$X_{a_1 \cdots a_p}$). We use $V \cdot \mathbf X$ to denote the contraction of a vector $V^a$ with the first index of $\mathbf X$, and we similarly write $\mathbf X \cdot V$ to denote the contraction of the last index. On the horizon, the orientation of the spacetime volume form is chosen to be $\bm \epsilon = \bm k \wedge \bm l \wedge \bm \epsilon_{\mathcal C}$, where $\bm \epsilon_{\mathcal C}$ is the spatial codimension-2 volume (area) form of a cross-section. Our convention for the Hodge dual of a $p$-form $\mathbf X$ is 
\begin{equation}
    (\star X)_{a_1 \cdots a_{D-p}} = \frac{1}{p!} \epsilon_{a_1 \cdots a_{D-p} b_1 \cdots b_p} X^{b_1 \cdots b_p}.
\end{equation}
Note that the Hodge dual carries the volume factor $\sqrt{-g}$. Under this convention, the contraction $A_{a_1 \cdots a_p} B^{a_1 \cdots a_p}$ can be expressed as $A_{a_1 \cdots a_p} B^{a_1 \cdots a_p} \bm \epsilon = p! {\star}\mathbf A \wedge \mathbf B$. Finally, regarding our choice of   units, we keep Newton's constant $G$ explicitly, while setting $\hbar = c = k_B = 1$ throughout this paper.

\section{Geometric setup}
\label{sec:geometric}

\subsection{Bifurcate Killing horizon, Gaussian null coordinates and non-stationary perturbations}
\label{sec:setup}

We consider an asymptotically flat or anti-de Sitter (AdS),  stationary background geometry $(\mathcal{M}, g)$ in $D$ spacetime dimensions with a bifurcate Killing horizon $\mathcal H.$ The Killing vector field that becomes null on the Killing horizon is denoted by~$\xi^a.$ The Killing horizon is comprised of two null surfaces, $\mathcal H^+$ (the future horizon) and $\mathcal H^-$ (the past horizon), that intersect at the bifurcation surface $\mathcal B$ (see figure \ref{fig:gnc}). We assume the future Killing horizon is    orientable and connected but not necessarily simply connected, and we do not consider  topology-changing processes (such as mergers), since we are studying  first-order  perturbations around a stationary background. This means the future horizon $\mathcal{H}^+$ has topology   $\mathcal{C} \times \mathbb{R}$, where $\mathcal{C}$ is a codimension-2 spatial cross-section of the horizon and $\mathbb{R}$ is the range of  the  null parameter.   Furthermore, we allow the horizon cross-sections to be   non-compact, namely $\mathcal{C} = \tilde{\mathcal{C}} \times \Sigma_k$, where $\tilde{\mathcal{C}}$ is a compact subspace and $\Sigma_k$ is a non-compact subspace with   trivial topology $\mathbb R^k$, such as flat space or hyperbolic space. The compact subspace is assumed to possess no boundary, i.e., there is no boundary at a finite distance. Moreover, we allow the topology of spatial infinity to be different from the horizon topology, as long as   the AdS or flat asymptotic structure is kept intact. Thus, our work  applies to stationary black objects with a   Killing horizon, including asymptotically flat or AdS black holes with spherical, planar, or hyperbolic horizons, as well as extended and higher-dimensional objects such as black branes, black strings, and black rings. This generalises previous studies of dynamical gravitational entropy \cite{Hollands:2024vbe,Visser:2024pwz}, which were restricted to black hole geometries with  compact horizon topology.   

The field content under consideration in this paper is: the metric field $g$, a single $p$-form abelian gauge field~$\mathbf{A}$ --- although our framework also applies  to multiple gauge fields ---  and other gauge-invariant bosonic matter  fields $\{\psi\}$. These fields are (non-minimally) coupled in a gauge-invariant manner, i.e., only the field strength $\mathbf{F}=\dd{\mathbf{A}}$ of the gauge field appears in the Lagrangian. We assume that the background dynamical fields $g, \mathbf{F}, \{\psi\}$ are   smooth and   are preserved by the Killing flow generated by $\xi^a$, i.e., the background fields   obey the stationarity condition:
\begin{equation} \label{statcondition}
    \mathcal{L}_\xi g = 0\,, \qquad \mathcal{L}_\xi \mathbf{F} = 0\,, \qquad \mathcal{L}_\xi \psi = 0\,.
\end{equation}

As a powerful geometric tool, we introduce  \emph{Gaussian null coordinates} (GNC) which can be constructed around any null hypersurface.  We emphasise that GNC hugely simplifies the calculation of the horizon entropy in covariant phase space formalism due to the boost weight analysis (see below), but it is   not necessary to introduce these coordinate to derive the first law. In previous work, HWZ derived the dynamical black hole first law covariantly using a Killing field argument,   but it is not understood, as far as we know, how to generalise this in the presence of matter fields.

Specifically, we choose $v$ to be the affine   parameter of the null geodesic generators of $\mathcal{H}^+$, and $\{x^i\}$ with $ i = 1,\cdots, D-2$ to be the coordinates on the codimension-2 cross-section $\mathcal{C}$. At each point $(v,x^i)$ on $\mathcal{H}^+$, we  shoot an ingoing null geodesic off the horizon to extend it to a point $(v,u,x^i)$ whose affine distance from the horizon is $u$. Thus, near $\mathcal{H}^+$, we have   the GNC $(v,u,x^i)$ in affine parametrisation.  In GNC,  $\mathcal{H}^+$ is located at $u=0$ and $\mathcal H^-$ at $v=0$, and the bifurcation surface $\mathcal{B}$ is labelled by $u=v=0$. Now, the near-horizon line element reads
\begin{equation} 
\label{gnc}
    \dd{s^2} = - 2 \dd{u} \dd{v} - u^2 \alpha(v,u,x^i)\dd{v^2} - 2 u \omega_i(u,v,x^i) \dd{v} \dd{x^i} + \gamma_{ij}(v,u,x^i) \dd{x^i} \dd{x^j}\,,
\end{equation}
where $\alpha, \omega_i, \gamma_{ij}$ are functions of the coordinates $(v,u,x^i)$, which encode the dynamics of the metric.

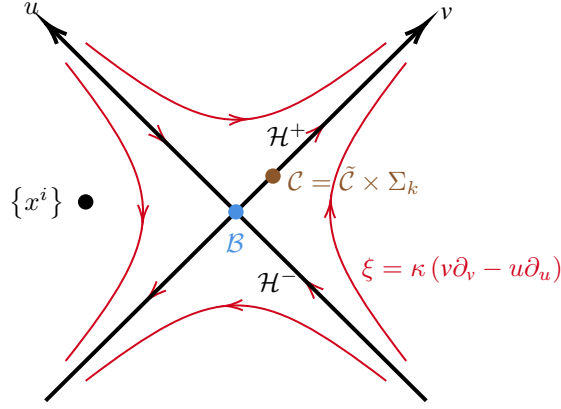
\begin{figure}[t!]
    \centering

\tikzset{every picture/.style={line width=0.75pt}} 

\begin{tikzpicture}[x=0.75pt,y=0.75pt,yscale=-1,xscale=1]

\draw [color={rgb, 255:red, 208; green, 2; blue, 27 }  ,draw opacity=1 ]   (285.5,71) .. controls (234.35,136.77) and (235.32,155.23) .. (285.5,221) ;
\draw [shift={(247.78,140.31)}, rotate = 92.42] [color={rgb, 255:red, 208; green, 2; blue, 27 }  ,draw opacity=1 ][line width=0.75]    (8.74,-2.63) .. controls (5.56,-1.12) and (2.65,-0.24) .. (0,0) .. controls (2.65,0.24) and (5.56,1.12) .. (8.74,2.63)   ;
\draw [color={rgb, 255:red, 208; green, 2; blue, 27 }  ,draw opacity=1 ]   (156.91,189.59) -- (244.09,102.41) ;
\draw [shift={(245.5,101)}, rotate = 135] [color={rgb, 255:red, 208; green, 2; blue, 27 }  ,draw opacity=1 ][line width=0.75]    (10.93,-3.29) .. controls (6.95,-1.4) and (3.31,-0.3) .. (0,0) .. controls (3.31,0.3) and (6.95,1.4) .. (10.93,3.29)   ;
\draw [shift={(155.5,191)}, rotate = 315] [color={rgb, 255:red, 208; green, 2; blue, 27 }  ,draw opacity=1 ][line width=0.75]    (10.93,-3.29) .. controls (6.95,-1.4) and (3.31,-0.3) .. (0,0) .. controls (3.31,0.3) and (6.95,1.4) .. (10.93,3.29)   ;
\draw [color={rgb, 255:red, 208; green, 2; blue, 27 }  ,draw opacity=1 ]   (166.91,112.41) -- (234.09,179.59) ;
\draw [shift={(234.09,179.59)}, rotate = 45] [color={rgb, 255:red, 208; green, 2; blue, 27 }  ,draw opacity=1 ][line width=0.75]    (10.93,-3.29) .. controls (6.95,-1.4) and (3.31,-0.3) .. (0,0) .. controls (3.31,0.3) and (6.95,1.4) .. (10.93,3.29)   ;
\draw [shift={(166.91,112.41)}, rotate = 225] [color={rgb, 255:red, 208; green, 2; blue, 27 }  ,draw opacity=1 ][line width=0.75]    (10.93,-3.29) .. controls (6.95,-1.4) and (3.31,-0.3) .. (0,0) .. controls (3.31,0.3) and (6.95,1.4) .. (10.93,3.29)   ;
\draw [line width=1.5]    (105.5,241) -- (293.38,53.12) ;
\draw [shift={(295.5,51)}, rotate = 135] [color={rgb, 255:red, 0; green, 0; blue, 0 }  ][line width=1.5]    (11.37,-3.42) .. controls (7.23,-1.45) and (3.44,-0.31) .. (0,0) .. controls (3.44,0.31) and (7.23,1.45) .. (11.37,3.42)   ;
\draw [line width=1.5]    (107.62,53.12) -- (295.5,241) ;
\draw [shift={(105.5,51)}, rotate = 45] [color={rgb, 255:red, 0; green, 0; blue, 0 }  ][line width=1.5]    (11.37,-3.42) .. controls (7.23,-1.45) and (3.44,-0.31) .. (0,0) .. controls (3.44,0.31) and (7.23,1.45) .. (11.37,3.42)   ;
\draw [color={rgb, 255:red, 74; green, 144; blue, 226 }  ,draw opacity=1 ]   (200.5,146) ;
\draw [shift={(200.5,146)}, rotate = 0] [color={rgb, 255:red, 74; green, 144; blue, 226 }  ,draw opacity=1 ][fill={rgb, 255:red, 74; green, 144; blue, 226 }  ,fill opacity=1 ][line width=0.75]      (0, 0) circle [x radius= 3.35, y radius= 3.35]   ;
\draw [color={rgb, 255:red, 208; green, 2; blue, 27 }  ,draw opacity=1 ]   (115.5,71) .. controls (167.02,137.13) and (167.02,155.6) .. (115.5,221) ;
\draw [shift={(153.89,151.33)}, rotate = 272.1] [color={rgb, 255:red, 208; green, 2; blue, 27 }  ,draw opacity=1 ][line width=0.75]    (8.74,-2.63) .. controls (5.56,-1.12) and (2.65,-0.24) .. (0,0) .. controls (2.65,0.24) and (5.56,1.12) .. (8.74,2.63)   ;
\draw [color={rgb, 255:red, 208; green, 2; blue, 27 }  ,draw opacity=1 ]   (125.5,231) .. controls (191.63,180.21) and (210.1,180.58) .. (275.5,231) ;
\draw [shift={(195.08,193.34)}, rotate = 357.03] [color={rgb, 255:red, 208; green, 2; blue, 27 }  ,draw opacity=1 ][line width=0.75]    (8.74,-2.63) .. controls (5.56,-1.12) and (2.65,-0.24) .. (0,0) .. controls (2.65,0.24) and (5.56,1.12) .. (8.74,2.63)   ;
\draw [color={rgb, 255:red, 208; green, 2; blue, 27 }  ,draw opacity=1 ]   (125.5,61) .. controls (191.63,112.52) and (210.1,112.52) .. (275.5,61) ;
\draw [shift={(205.83,99.39)}, rotate = 177.9] [color={rgb, 255:red, 208; green, 2; blue, 27 }  ,draw opacity=1 ][line width=0.75]    (8.74,-2.63) .. controls (5.56,-1.12) and (2.65,-0.24) .. (0,0) .. controls (2.65,0.24) and (5.56,1.12) .. (8.74,2.63)   ;
\draw    (125.5,141) ;
\draw [shift={(125.5,141)}, rotate = 0] [color={rgb, 255:red, 0; green, 0; blue, 0 }  ][fill={rgb, 255:red, 0; green, 0; blue, 0 }  ][line width=0.75]      (0, 0) circle [x radius= 3.35, y radius= 3.35]   ;
\draw [color={rgb, 255:red, 139; green, 87; blue, 42 }  ,draw opacity=1 ]   (219,128) ;
\draw [shift={(219,128)}, rotate = 0] [color={rgb, 255:red, 139; green, 87; blue, 42 }  ,draw opacity=1 ][fill={rgb, 255:red, 139; green, 87; blue, 42 }  ,fill opacity=1 ][line width=0.75]      (0, 0) circle [x radius= 3.35, y radius= 3.35]   ;
\draw  [draw opacity=0] (23.5,162.5) -- (93.5,162.5) -- (93.5,202.5) -- (23.5,202.5) -- cycle ;

\draw (305.5,46) node    {$v$};
\draw (98.5,43.5) node    {$u$};
\draw (226,107) node    {$\mathcal{H}^{+}$};
\draw (200,161.5) node  [color={rgb, 255:red, 74; green, 144; blue, 226 }  ,opacity=1 ]  {$\mathcal{B}$};
\draw (313.5,175) node  [color={rgb, 255:red, 208; green, 2; blue, 27 }  ,opacity=1 ]  {$\xi =\kappa \left( v\partial_v -u \partial_u\right)$};
\draw (101,141.5) node    {$\left\{x^{i}\right\}$};
\draw (259.5,130) node  [color={rgb, 255:red, 139; green, 87; blue, 42 }  ,opacity=1 ]  {$\mathcal{C} =\tilde{\mathcal{C}} \times \Sigma _{k}$};
\draw (221.5,180.5) node    {$\mathcal{H}^{-}$};

\end{tikzpicture}

    \caption{Bifurcate Killing horizon $\mathcal H$, consisting of a future horizon $\mathcal H^+$ at $u=0$ and a past horizon $\mathcal H^-$ at $v=0$ that intersect at the bifurcation surface $\mathcal B$. The Gaussian null coordinate system around $\mathcal H^+$  is $(v,u,x^i)$, where $v$ is the affine parameter on the null geodesics of $\mathcal H^+$, with $\mathcal B$ at $v=0$, $u$ is the affine null geodesic distance  away from $\mathcal H^+$, and $x^i $ are the $(D-2)$ spatial coordinates on the cross-section of $\mathcal H^+$. This cross-section $\mathcal C$  is the product of a compact subspace $\tilde{\mathcal C}$ and a non-compact manifold $\Sigma_k$.}
    \label{fig:gnc}
\end{figure}

In GNC the horizon generating  Killing field takes the form  
\begin{equation}
    \xi = \kappa (v \partial_v - u \partial_u)\,,
\end{equation}
where $\kappa$ is the surface gravity of the bifurcate Killing horizon. Hence, $\xi^a$ acts as a local Lorentz boost generator near the horizon.  At spatial infinity, for   stationary, axisymmetric black objects,   the horizon Killing vector field may be normalised  as   $\xi = \partial_t + \Omega_I \partial_{\vartheta^I},$ where $t$ and $\vartheta^I$ are the time and angular coordinates, respectively. In the absence of axisymmetry, the angular velocities $\Omega_I$ vanish.

For any quantity $X$, we   define its \emph{boost weight} (or Killing weight) in affine GNC as follows: under a rigid boost $u \to au, v \to v/a$ with constant $a$, a weight-$w$ quantity $X$ transforms as 
\begin{equation}
    X \to a^w X\,.
\end{equation}
For a weight-$w$ GNC component $T_{(w)}$ of a tensor field $T$, we can see that its weight $w$ is equal to the number of $v$-indices minus the number of $u$-indices, with all indices lowered. For example, $E_{vv}$ has weight 2, $X_{uiuj}$ has weight $-2$, $T_{uvj}$ has weight 0, et cetera. For more  details on the boost weight analysis we refer the reader to \cite{Bhattacharyya:2021jhr,Hollands:2022fkn,Wall:2024lbd}. 

The boost weight analysis of tensor components leads to the following powerful result that underlies many key proofs in   our previous work \cite{Visser:2024pwz} and in  this paper:
\begin{thm}
    For a tensor field $T$ that is preserved by the   flow generated by the horizon Killing field $\xi$, i.e., $\mathcal{L}_\xi T = 0$, and that is regular on the horizon $\mathcal{H}^+$, its GNC components $T_{(w)}$ with positive boost weight $w > 0$ vanish  on $\mathcal{H}^+$. \label{thm:boost-weight}
\end{thm}
\noindent A proof can be found in various works, for instance, in \cite{Wall:2024lbd}. \\

Furthermore, to study the dynamical entropy of black objects and the associated thermodynamic first law   with $p$-form charges, we consider first-order non-stationary perturbations 
\begin{equation}
    g \to g + \delta g\,, \qquad \mathbf{A} \to \mathbf{A} + \delta \mathbf{A}\,, \qquad \psi \to \psi + \delta \psi\,.
\end{equation}
These perturbations break the background Killing symmetry in the sense that   $\delta (\mathcal L_\xi g) \neq 0$, $\delta (\mathcal L_\xi \mathbf{A}) \neq 0$ and $\delta (\mathcal L_\xi \psi) \neq 0$. Additionally, we allow for variations of the Killing field $\xi$, i.e. $\delta \xi \neq 0 $, but it turns out they do not contribute to the first law, as shown in detail in \cite{Visser:2024pwz}. There are no gauge conditions imposed on $\delta \mathbf{A}$ and $\delta \psi$ at the horizon, but we do impose asymptotic fall-off conditions (see section \ref{sec:falloff}). Further,  we require the perturbations of the metric $g$ to preserve the form of the line element in Gaussian null coordinates, that is, we fix the GNC system $ (v,u, x^i)$ and set 
\begin{equation}
    \delta g_{ua} = 0, \qquad \delta g_{vv} \sim \mathcal{O}(u^2), \qquad \delta g_{vi} \sim \mathcal{O}(u).
\end{equation}
Note this implies $\delta g_{vv} = \delta g_{vi} \fheq 0$. In other words,  all the perturbations of the metric only vary the functions 
\begin{equation}
    \alpha \to \alpha + \delta \alpha\,, \qquad \omega_i \to \omega_i + \delta \omega_i\,, \qquad \gamma_{ij} \to \gamma_{ij} + \delta \gamma_{ij}\,.
\end{equation}

\subsection{Non-smooth $p$-form gauge field on a Killing horizon}
\label{sec:nonsmoothness}
    
Below we explain what type of divergence  the pullback of  a $p$-form abelian gauge field $\mathbf{A}$ can have on a Killing horizon, and why such a  singular gauge field has   an associated non-vanishing, regular electrostatic  potential form  on the horizon. 

As stated above, we assume the   curvature   $(p+1)$-form $\mathbf{F} = \dd{\mathbf{A}}$ obeys the  stationarity condition
    \begin{equation}
        \mathcal L_\xi \mathbf F = 0\,.
    \end{equation}
  In terms of the gauge field this reads 
    \begin{equation}
        \mathcal L_\xi \mathbf F = \mathcal L_\xi (\dd{\mathbf A}) = \dd{(\mathcal L_\xi \mathbf A)} = 0\,,
    \end{equation}
    where we   used that the Lie derivative  and exterior derivative commute. This means the gauge field is preserved by the Killing flow generated by $\xi$ up to a pure gauge factor
    \begin{equation} 
        \mathcal L_\xi \mathbf A = \dd{\bm \lambda_\xi} \,,\label{eq:A-Killing}
    \end{equation}
    where $\bm \lambda_\xi$ is a $(p-1)$-form gauge parameter \cite{Biswas:2022grc}. Notice that this gauge parameter is only defined up to an exact term, so any $\tilde{\bm \lambda}_\xi = \bm \lambda_\xi + \dd{\bm \alpha}$ for some $\bm \alpha$ satisfies the same condition \eqref{eq:A-Killing}. 

    To study the potentially divergent profile of the gauge potential, we simplify  equation \eqref{eq:A-Killing} by performing the following gauge transformation:
    \begin{equation}
        \mathbf{A} \to \tilde{\mathbf{A}} = \mathbf{A} + \dd{\bm \Lambda} \qquad \text{with} \qquad \dd{(\bm \lambda_\xi + \mathcal{L}_\xi \bm \Lambda)} = 0\,.
    \end{equation}
    In terms of $\tilde{\mathbf{A}}$, the stationarity condition \eqref{eq:A-Killing} reads
    \begin{equation} \label{killingauge}
        \mathcal L_\xi \tilde{\mathbf A }= 0\,.
    \end{equation}
    We call this the \emph{Killing gauge} for the gauge field. We will employ this gauge to show that the gauge field can be non-smooth on the Killing horizon, however in the rest of the paper we do not work with this particular gauge (or any other gauge choice for that matter). 

    In Gaussian null coordinates we can write the Killing gauge \eqref{killingauge} for a weight-$w$ component of $\tilde{\mathbf{A}}$ as \cite{Bhattacharyya:2021jhr} 
   \begin{equation}
        (v \partial_v - u \partial_u + w) \tilde A_{(w)} = 0\,.
    \end{equation}
    On the future horizon ($u=0$) the solution takes the form \cite{Wall:2024lbd}
    \begin{equation}
        \tilde A_{(w)}(v,0,x^i) \fheq f_{w}(x^i) v^{-w}\,,
    \end{equation}
    for some function $f_w(x^i)$ of the spatial directions only. Note that the gauge field components with $w > 0$ blow up at the bifurcation surface ($u=v=0$) for $f_w(x^i) \neq 0$. Due to the antisymmetry of the differential form~$\tilde{\mathbf{A}}$, there can be at most one $v$-index on any non-vanishing component of~$\tilde{\mathbf{A}}$. Therefore, the highest-weight components  have $w=1$ and they exhibit a $1/v$ divergence
   at the bifurcation surface (where $v=0$): \begin{equation} \label{divgauge}
        \tilde A_{v i_1 \cdots i_{p-1}} \fheq f_{i_1 \cdots i_{p-1}}(x^i) \frac{1}{v}\,,
    \end{equation}
    with $i_1, \cdots, i_{p-1}$ labeling different spatial directions. This is the only possible type of divergence of the pullback of the gauge field in the Killing gauge to a future Killing horizon.   Similarly, on the past Killing horizon in the Killing gauge, the pullback of the gauge field scales as $1/u$ and thus exhibits a divergence at the bifurcation surface, where $u=0$. 

    In previous works \cite{Hollands:2024vbe,Visser:2024pwz}  on dynamical black hole entropy, all the dynamical fields,  including the gauge fields, were assumed to be smooth on the horizon. Such a requirement would impose $f_{w=1}(x^i)= 0$ on the horizon. In this paper  we weaken this assumption by requiring  only the gauge-invariant fields  to be regular on the horizon. Hence, we relax the regularity condition for gauge fields, as long as the physical field strength    derived from it is smooth on the Killing horizon. 
    Requiring $\mathbf{F} = \dd\mathbf{A}$ to be smooth on the horizon implies, by theorem \ref{thm:boost-weight}, that all components of $\mathbf{F}$ with positive boost weight must vanish at $v = 0$. In particular, the weight-1 components of~$\mathbf{F}$, which could receive divergent contributions from $1/v$ terms in $\mathbf{A}$, must vanish on the horizon. This imposes a curl-free condition on the spatial functions $f_{i_1 \cdots i_{p-1}}(x^i)$ appearing in the divergent part of the gauge potential:
    \begin{equation}
        \partial_{[i_1} f_{i_2 \cdots i_{p}]} \fheq 0\,.
    \end{equation}
    This condition does not eliminate $f_{i_1 \cdots i_{p-1}}$ altogether, hence the $1/v$ divergence in the gauge field components could still be present at the bifurcation surface, while the physical field strength is smooth on the horizon. 
    
    In general, the weight-1 components of a gauge field that satisfy the stationarity  condition \eqref{eq:A-Killing} have the following irregular profile:
    \begin{equation}
        A_{v i_1 \cdots i_{p-1}} \fheq f_{i_1 \cdots i_{p-1}}(x^i) \frac{1}{v} - (\dd{\Lambda})_{v i_1 \cdots i_{p-1}}.
    \end{equation}
    In principle, one could choose the gauge parameter $\bm \Lambda$ such that $\mathbf{A}$ is regular at $v=0$, however, the price to pay is that we need to work in a particular gauge with irregular $\bm \Lambda$.  We can also choose a gauge parameter $\bm \Lambda$ such that $\mathbf A$ has a stronger $1/v^n$ divergence at $v=0$, while the field strength is still smooth.

    Physically, the divergence  of the pullback of the gauge field at the bifurcation surface in the Killing gauge is important because it leads to a non-vanishing, finite electric potential on the   horizon. For instance, for charged black hole solutions to Einstein-Maxwell theory, the electric potential  associated to the gauge field $\tilde{\mathbf{A}}$  is  
    \begin{equation} \label{potential}
        \Phi \equiv - \xi \cdot \tilde{\mathbf{A}} \Big |_{\mathcal H^+} = - \kappa v \tilde A_v\,.
    \end{equation}
    A regular background gauge potential (satisfying $\tilde A_v \fheq 0$) results in a vanishing electrostatic potential. To investigate the more interesting cases where horizons have non-zero electrostatic potentials (and their higher-form analogues), we work with a singular gauge potential profile   on the horizon. Interestingly, the electrostatic potential and its higher-form generalisation, $- \xi \cdot \tilde{\mathbf{A}}$, are always non-singular on the horizon. This is because at the future horizon the Killing vector field is linear in $v$, which cancels the $1/v$ divergence in $\tilde{\mathbf{A}}$. This   guarantees that   the charge term  $\Phi \delta Q$ appearing in the  first law of black hole thermodynamics is regular. 

 We should mention that the  divergence of the pullback of the gauge field at the bifurcation surface is not an artefact of working with affine GNC or the Killing gauge. This is true independent of the choice of coordinates and gauge for the following reason. We require that $\xi \cdot \mathbf A$ is nonzero on the horizon --- otherwise the  electric potential would vanish at the horizon --- which   implies that $\mathbf A$ diverges at the bifurcation surface, because $\xi$ vanishes there. 

 To summarise, in the Killing gauge,   the boost-weight-1 component of the gauge field either vanishes on a Killing horizon, in which case the horizon electric potential   also vanishes, or it diverges at the bifurcation surface, in which case the electric potential form is non-zero and regular on a Killing horizon.   In both cases, the field strength can remain regular at the Killing horizon. 

    \subsection{Electrostatic potential  $(p-1)$-form is closed on a   Killing horizon} \label{ssec:zeroth-law}

    In the case of electromagnetism, the electrostatic potential $\Phi = - \xi \cdot \tilde{\mathbf{A}}$ is constant on a    Killing horizon \cite{Carter:2009nex,Gauntlett:1998fz,Gao:2001ut,Smolic:2014swa}, which is sometimes called the ``zeroth law'' for the electrostatic potential.  Below, we derive an analogous    statement for $p$-form abelian gauge fields  without choosing a particular   gauge (such as the Killing gauge). We consider a $p$-form gauge field $\mathbf{A}$ that is stationary up to a pure gauge term, $\mathcal{L}_\xi \mathbf{A} = \dd{\bm \lambda_\xi}$. We define the     electrostatic potential $(p-1)$-form associated to $\mathbf{A} $  as 
    \begin{equation} \label{potentialform}
        \bm \Phi \equiv - \xi \cdot \mathbf{A} + \bm \lambda_\xi\,,
    \end{equation} 
    which is exactly the ``momentum map''  associated with the background Killing vector $\xi$ that appears in the gauge-covariant construction of the Lie derivative of gauge (Maxwell) fields, discussed in \cite{Elgood:2020svt}. This electric potential form satisfies the following theorem (see also \cite{Copsey:2005se,Compere:2007vx}): 
    \begin{thm} 
    \label{thm:zeroth-law}
        The pullback of the electrostatic potential form $\bm \Phi$  to the future Killing horizon $\mathcal H^+$  is closed. For a 1-form gauge field $\mathbf A$, this means $\Phi$ is constant on $\mathcal H^+$.
    \end{thm}
    \begin{proof}
        The exterior derivative of the  electrostatic potential form $\bm \Phi$ is
        \begin{equation} \label{extpotentialform}
            \dd{\bm \Phi} = - \dd{(\xi \cdot \mathbf{A})} + \dd{\bm \lambda_\xi} = - \mathcal{L}_\xi \mathbf{A} + \xi \cdot \dd{\mathbf{A}} + \dd{\bm \lambda_\xi} = \xi \cdot \mathbf{F} = \frac{1}{p!}\xi^a F_{a a_1 \cdots a_p} \dd{x^{a_1}} \wedge \cdots \wedge \dd{x^{a_p}}\,.
        \end{equation}
        When pulled back to the future Killing horizon $\mathcal H^+$, the $p$-form basis consists of wedge products of the null covector $\dd{v}$ and the spatial covectors $\dd{x^i}$. However, components involving $\dd{v}$ are identically zero, because $F_{v \cdots v \cdots} = 0$ by the antisymmetry of the tensor field $F_{a a_1 \cdots a_p}$. Hence, in affine GNC, we have  
        \begin{equation}
            \dd{\bm \Phi} \fheq \frac{1}{p!} \kappa v F_{v i_1 \cdots i_p} \dd{x^{i_1}} \wedge \cdots \wedge \dd{x^{i_p}}.
        \end{equation}
        But by \cref{thm:boost-weight}, the  $F_{v i_1 \cdots i_p}$ components are identically zero on $\mathcal H^+$, since they have weight one and $\mathbf{F}$ is assumed to be regular everywhere on the horizon. Therefore,
        \begin{equation}
            \dd{\bm \Phi} \fheq 0\,.
        \end{equation}
    \end{proof}
A similar theorem holds for the past Killing horizon $\mathcal H^-$ (where $v=0$), since tensor components with negative boost weight vanish on $\mathcal H^-$, in particular, $F_{ui_1\cdots i_p}$ vanishes. Hence, the electrostatic potential form is also closed on the past Killing horizon, $\dd{\bm \Phi} \pheq 0$, thus it is closed on the entire bifurcate Killing horizon $\mathcal H$. 

An important property of $\bm \Phi$ is that it is only defined up to  two types of ambiguities. Under a gauge transformation $\mathbf A \to \mathbf A + \dd \mathbf \Lambda$, the gauge parameter $\bm{\lambda}_\xi$ transforms as $\bm{\lambda}_\xi \to \bm{\lambda}_\xi + \mathcal L_\xi \mathbf \Lambda$, since $\mathcal L_\xi \mathbf A \to \mathcal L_\xi (\mathbf A + \dd \mathbf \Lambda) = \dd (\bm \lambda_\xi + \mathcal L_\xi \mathbf \Lambda)$, so the  electrostatic  potential form transforms as
\begin{equation}\label{gaugetranselectricpot}
    \mathbf \Phi \to - \xi \cdot (\mathbf A + \dd \mathbf \Lambda) + \bm \lambda_\xi + \mathcal L_\xi \mathbf \Lambda = \mathbf \Phi + \dd (\xi\cdot \mathbf \Lambda)\,.
\end{equation}
Hence, $\mathbf \Phi $ is gauge invariant modulo   an exact term. A second ambiguity   arises because   $\bm \lambda_\xi$ itself is defined by   \eqref{eq:A-Killing} only up to an exact term, $\bm \lambda_\xi \to \bm \lambda_\xi + \dd \bm \alpha$, which induces the shift $\bm \Phi \to \bm \Phi + \dd \bm \alpha$. These two ambiguities can be arranged to cancel by choosing $\bm \alpha = - \xi \cdot \bm \Lambda$, in which case $\bm \Phi$ is  strictly gauge invariant.

If $\bm \Phi$ is strictly gauge invariant, then it is smooth on the horizon, since it is already smooth in the Killing gauge~\eqref{killingauge}.  Smoothness of $\bm \Phi$ is possible even if  the gauge field has a strong $1/v^n$ divergence, since this  can be cancelled by a divergence in $\bm \lambda_\xi$. 
   
   In order to obtain the higher-form analogue of a constant electrostatic potential on the horizon, we   use these ambiguities to set\footnote{Sufficient conditions for the constant electric potential gauge are: $\dd (\xi \cdot \bm \lambda_\xi )=0$ and $ \mathcal{L}_{m_\aleph} (\xi \cdot \mathbf A) = \mathcal{L}_{m_\aleph}\bm \lambda_\xi$. }
    \begin{equation} \label{constpotgaugeoriginal}
        \mathcal{L}_\xi \bm \Phi \fheq 0\,, \qquad \mathcal{L}_{m_\aleph} \bm \Phi \fheq 0\,,
    \end{equation}
    for any extended dimensions with coordinates $y^\aleph$, $\aleph = 1,...,k$ and generators $m_\aleph = \pdv*{y^\aleph}$ that commute with the horizon Killing field: $[\xi, m_{\aleph}]=0$. Note that $k$ denotes the dimension of the non-compact part $\Sigma_k$ of a horizon cross-section with the same topology as $\mathbb R^k$.  These conditions   guarantee that the potential form is constant along the null and the extended directions of the future horizon. We call this the \emph{constant electric potential gauge}. We do not have to restrict $\bm \Phi$ to be constant along the compact subspace $\tilde{\mathcal{C}}$ of the horizon, because $\bm \Phi$ can be resolved  into a number of constants $\Phi_{\sigma}^{\mathcal{H^+}}$ on   non-trivial cycles  of $\tilde{\mathcal{C}}$ using algebraic topology techniques, as we will see   in \cref{sec:cfl}.

\subsection{Asymptotic fall-off conditions}
\label{sec:falloff}

In order to have a well-defined $\Phi \delta Q$ term in the black hole first law,  we require the following four asymptotic conditions:
\begin{itemize}
\item[1.]  The  interior product of the Killing field $\xi$ with the field strength vanishes at spatial infinity: \begin{equation}\label{conditiononFatinfinity}\xi \cdot \mathbf F \to 0 \qquad \text{ as} \qquad  r \to \infty\,.\end{equation} This ensures that the electric potential form is closed at spatial infinity, due to \eqref{extpotentialform}. If we assume in addition that $\bm \lambda_\xi \to 0$ at spatial infinity, so that the diffeomorphism generated by 
 $\xi$ does not induce a non-trivial (large) gauge transformation at infinity,   this means that $\xi \cdot \mathbf A $ is a closed form at spatial infinity (and it is constant if $\mathbf A$ is a 1-form). 

\item[2.]  For electric charges, the gauge field satisfies a Dirichlet asymptotic boundary condition: 
\begin{equation} \label{dirichletcond}
    \delta \mathbf A \to 0  \qquad \text{as}\qquad r \to \infty \,.
\end{equation}
  This condition implies that a term involving $\delta \mathbf A$     in the fundamental variational identity   in the covariant phase space formalism vanishes when integrated at spatial infinity, see  below  equation \eqref{gdimpnoether}. 
  
  For magnetic charges, $\delta \mathbf A$ does not vanish at the asymptotic boundary, which is crucial for obtaining a non-zero magnetic potential-charge contribution from spatial infinity in the first law, see the discussion around   \eqref{magneticatinfinity}.   In that case, we impose the asymptotic boundary condition
  \begin{equation}
      \delta \mathbf F \to \mathcal O(1) \qquad \text{as } \qquad r \to \infty\,,
  \end{equation}
  ensuring that the magnetic flux does not blow up at infinity.
  
  \item[3.]  The perturbation of the electric flux density $\mathbf{\Upsilon}$ ---  which for electromagnetism is proportional to $\star \mathbf F$ and more generally it is defined below  in \eqref{upsilon} ---   satisfies
\begin{equation}
\delta \mathbf{\Upsilon} \to \mathcal{O}(1) \quad \text{as} \quad r \to \infty \,,
\end{equation}
so that the flux integral at spatial infinity is finite or it vanishes. By the linearised gauge field equation
$
\dd(\delta\mathbf{\Upsilon}) = 0 ,
$ 
 and Stokes' theorem, this flux equals the integral of $\delta\mathbf{\Upsilon}$ over  an appropriate cycle within any horizon cross-section  and   defines the variation of the total electric charge, see \eqref{pcharge}. 

\item[4.] For non-compact horizon cross-sections, we impose an asymptotic boundary condition 
on  all  dynamical fields (and their perturbations) along each flat  extended 
direction $y^\aleph$:
\begin{equation}  
\varphi \to \varphi_\infty \,, \qquad \delta\varphi \;\to\; \delta\varphi_\infty \qquad \text{and}
\qquad 
\mathcal{L}_{m_\aleph}\delta\varphi \to 0 \,, \qquad \text{as} \qquad y^\aleph \to \pm\infty\,,
\label{fourthfalloff}
\end{equation}
where $m_\aleph = \pdv*{y^\aleph}$ with 
 $\aleph = 1,...,k$ and $\varphi_\infty, \delta \varphi_\infty$ are constants (independent of   $y^\aleph$). 
This ensures that the opposite boundary contributions in the JKM ambiguity 
cancel pairwise, and hence the dynamical entropy is JKM-invariant also for 
non-compact horizons (see section~\ref{sec:jkm}). Moreover, we used this condition to prove the closedness of  the downgraded electric flux density on the compact subspace of the horizon in appendix \ref{appc}.
 \end{itemize}

Specifically, focusing on electric charges,  for asymptotically flat and asymptotically AdS black holes the first two conditions are met in the following way. On the one hand, in the asymptotically flat case, components with one leg along the time direction  of the gauge field fall off as  $\mathcal O(1/r^{D-3})$, and   the field strength falls off as $\mathcal O(1/r^{D-2})$. This means both asymptotic conditions are automatically satisfied, if the radial coordinate is fixed under the perturbation. Moreover, this fall-off behaviour implies that the electric potential form vanishes at spatial infinity, and hence only its integral at the horizon features into the black hole first law.  

On the other hand, for asymptotically AdS black holes, components of the 
gauge potential with one index along the asymptotic time direction can 
approach a nonzero constant at spatial infinity, so that the electric 
potential form does not vanish there. Concretely,  
$ 
\mathbf{A} = \mathbf{A}_\infty + \mathcal{O}\!\left(1/r^{D-3}\right) $ with 
 $ \mathbf{A}_\infty \neq 0 ,
 $
which implies that the electric potential form has a finite, nonzero 
boundary value at infinity. In order to satisfy the second condition \eqref{dirichletcond}, we now need $\delta \mathbf A_\infty =0$. Moreover, as we will see below, in this case the  difference between the integral of the potential form at spatial infinity and at the horizon   will enter into the black hole first law. 
  
Regarding the third asymptotic condition, this is met in the  case of electromagnetism, for which $\delta \mathbf \Upsilon \propto \star \delta  \mathbf F,$ for the following reason. For asymptotically flat black holes, the perturbed radial electric component $\delta F_{tr}$ decays as $\mathcal O(1/r^{D-2})$ at large $r$. Since the area element of the codimension-2 surface at large radius grows as $r^{D-2}$ (for spherical topology; for planar topology it grows with the box size), the components of $\delta\mathbf{\Upsilon}$ tangent to the codimension-2 surface are $\mathcal O(1)$, ensuring that the flux integral at infinity is finite or zero. In the asymptotically AdS case we additionally fix $\delta \mathbf A_\infty=0$, which leads to the same decay for $\delta F_{tr}$ and thus for $\delta\mathbf{\Upsilon}$ as in the asymptotically flat case.

The fourth asymptotic condition is also automatically 
satisfied in several physically relevant situations. For instance, perturbations of uniform black branes that change only the mass or charge density to a new constant value obey~\eqref{fourthfalloff}, since the background fields are translationally invariant along the brane directions and both $\phi$ and $\delta\phi$ therefore approach the same constant values at $y^\aleph\to\pm\infty$. In this case the boundary contributions from opposite faces cancel pairwise. More generally, any translationally invariant perturbation along the brane directions --- such as adding a constant momentum density (e.g., an infinitesimal boost of the brane) --- or working in a compactified geometry with periodic boundary conditions along the flat directions also satisfies the equal-asymptotics condition.

The additional asymptotic condition 4 plays an important role in ensuring that the bifurcate Killing horizon structure assumed in the derivation of the first law is not spoiled by the presence of non-compact horizon directions. Without such a condition, fluxes of energy or charge along the flat directions could invalidate the assumption that the null geodesic generators of the horizon coincide with the orbits of a Killing vector field.\footnote{We thank Harvey Reall for this observation.} Indeed, explicit stationary AdS solutions are   known with non-compact horizons that evade the rigidity theorem and do not admit a bifurcate Killing horizon. Canonical examples are provided by the flowing black funnels in \cite{Fischetti:2012vt} and the stationary   plasma flows in \cite{Figueras:2012rb}. In both cases the background metrics approach different asymptotics at the two ends of the horizon: for flowing funnels, they interpolate between boundary black holes of unequal temperature, while for plasma flows they interpolate between distinct boosted black branes of different surface gravity and horizon area density. As a result, condition 4 is violated and the horizon is non-Killing with no bifurcation surface.

By contrast, condition~4 rules out precisely the 
pathological behaviour realised in the black funnel example. Physically, it 
ensures that the state of the brane is uniform at both ends, so that no net 
flux of conserved charges can flow along the horizon. Consequently the 
contributions from opposite boundary faces cancel, as in \eqref{VI.6}, and 
the horizon of the background geometry remains bifurcate and Killing. In this way, the equal-asymptotics condition provides a minimal restriction under which the first law extends consistently to non-compact horizons.

    \section{Gauge dependence in the covariant phase space formalism}
\label{sec:gaugedepen}
In this section we analyse the covariant phase space formalism  \cite{Wald:1993nt,Iyer:1994ys} for diffeomorphism-invariant gravity theories for which the metric is non-minimally coupled to a $p$-form gauge field and to other bosonic gauge-invariant tensor fields (see also, e.g.,  \cite{Gao:2003ys,Papadimitriou:2005ii,Compere:2007vx,Rogatko:2009th,Keir:2013jga,Harlow:2019yfa,Elgood:2020svt,Hajian:2022lgy}).    
To generalise   previous work \cite{Hollands:2024vbe,Visser:2024pwz} on dynamical horizon entropy to charged black holes, we need to single out the gauge-dependent terms in any covariant phase space quantity and treat them separately. In this paper, we will only highlight the contributions from gauge-dependent terms --- for the gauge-invariant terms  we use   results obtained previously in     \cite{Hollands:2024vbe,Visser:2024pwz}.  

    \subsection{Diffeomorphism-invariant Lagrangian}

    We consider a general diffeomorphism-invariant Lagrangian  $D$-form  that is locally constructed out of   bosonic fields and their symmetrised derivatives:
    \begin{equation}
       \mathbf L = \mathbf L_\text{grav-EM}(g^{ab}, R_{abcd}, F_{a_1 \cdots a_{p+1}}, \{\psi\}, \nabla_a) + \mathbf L_\text{source}(g^{ab}, A_{a_1 \cdots a_p}, \{\psi\}, \left\{ \Psi \right\},\nabla_a)\,.
    \end{equation}
    The first term, $\mathbf L_\text{grav-EM}$, is  the gravity-electromagnetism part of the Lagrangian. It is constructed from    contractions of the (inverse) metric $g^{ab}$,   the Riemann tensor $R_{abcd}$, the field strength tensor $\mathbf{F} = \dd{\mathbf{A}}$ of a $p$-form gauge field $\mathbf{A}$, a collection of gauge-independent bosonic matter tensor fields  $\{\psi_{a_1 \cdots a_s}\}$ with spin $s$, and their symmetric covariant derivatives.  Our framework readily applies  to multiple gauge fields, however, to simplify the   discussion, we   only consider one $p$-form field throughout the paper. In $\mathbf L_\text{grav-EM}$, the couplings between the metric and matter fields are in general \emph{non-minimal}, and all terms are   gauge-invariant, so there is no explicit $\mathbf{A}$ dependence. 
    
    The second term in the Lagrangian, $ \mathbf L_\text{source}$,  depends on a collection of external matter fields $\{\Psi_{a_1 \cdots a_r}\}$ that  source the dynamical fields in the gravity-electromagnetism sector via a stress-energy tensor $T_{ab}$, an electric current $j^{a_1 \cdots a_p}$, and other source tensors $\{\zeta^{a_1 \cdots a_s}\}$. These external source tensors are  determined by interactions in the   action $I_\text{source} = \int \mathbf L_\text{source}$:
    \begin{equation}
        T_{ab} \equiv - \frac{2}{\sqrt{-g}} \frac{\delta I_\text{source}}{\delta g^{ab}}, \qquad j^{a_1 \cdots a_p} \equiv (-1)^{D-p}\frac{p!}{\sqrt{-g}} \frac{\delta I_\text{source}}{\delta A_{a_1 \cdots a_p}}, \qquad \zeta^{a_1 \cdots a_s} \equiv \frac{1}{\sqrt{-g}} \frac{\delta I_\text{source}}{\delta \psi_{a_1 \cdots a_s}}.
    \end{equation}
    We assume that the $\mathbf{A}$-dependent term in $L_\text{source}$ is linear in $\mathbf{A}$,\footnote{The Chern-Simons term for abelian gauge fields is allowed in our framework as part of the source Lagrangian.} so it takes the form of an electric source $\sqrt{-g}\, j^{a_1 \cdots a_p} A_{a_1 \cdots a_p}$, and that the external matter fields $\{\Psi\}$ are minimally coupled to gravity. We see that the electric source term is the only gauge-dependent part of the total Lagrangian. However, the action can be made gauge invariant (up to a boundary term) by demanding $\{\Psi\}$ to have certain symmetries, such as a $\text{U}(1)$ gauge symmetry for a complex scalar field coupled to electromagnetism. 

    For now, for the purpose of deriving the comparison version of the first law, we only consider source-free Lagrangians with $\mathbf L_\text{source} = 0$, which implies the Lagrangians are gauge invariant; hence they depend only on  $\mathbf A$ implicitly through $\mathbf{F}$. We will turn on source terms in \cref{sec:ppfl}, where we study the physical process version of the first law.

    A final comment before we proceed: in the analysis that follows, we will often treat $\mathbf{F}$ as a fundamental variable, temporarily disregarding the relation $\mathbf{F} = \dd{\mathbf{A}}$. This approach is technically useful as an intermediate step in deriving certain covariant phase space quantities. Moreover, we   keep explicit $\mathbf{F}$-dependence wherever possible in order to preserve gauge invariance in the formalism.

    \subsection{Presymplectic potential}

    The presymplectic potential arises from the boundary terms in the variation of the Lagrangian. Here, we show that the boundary terms can be separated into a gauge-invariant and gauge-dependent part.
    
    For convenience, we denote $\tilde{\phi} \equiv (g, \{\psi\})$ as the collection of the metric field $g$ and all other gauge-invariant (non-gauge) matter fields $\{\psi\}$, $\phi \equiv (\tilde{\phi}, \mathbf F)$ as the  collection of all gauge-invariant fields, and $\varphi \equiv (\phi, \mathbf{A})$ as the collection of gauge-invariant and gauge-dependent fields.

  The variation of the Lagrangian form with respect to the dynamical fields $\tilde \phi$ and  $ \mathbf{F}$ is given by the following gauge-invariant equation 
    \begin{equation}
        \delta \mathbf L = \mathbf E(\phi) \cdot \delta \tilde{\phi} + \mathbf{\Upsilon} (\phi)\wedge \delta \mathbf F + \dd{\mathbf \Theta^{\text{GI}}(\phi, \delta \phi)}\,, \label{eq:dL-gi}
    \end{equation}
    where  
    \begin{equation}
        \mathbf{E}(\phi) \cdot \delta \tilde \phi = \frac{1}{2} \mathbf{E}_{(g)ab}(\phi) \delta g^{ab} + \mathbf{E}_{(\psi)}(\phi) \cdot \delta \psi
    \end{equation}
    contains the equations of motions for the metric and other gauge-invariant matter fields, and $ \mathbf{\Upsilon}$ is a $(D-p-1)$-form that is defined as 
    \begin{equation} \label{upsilon}
        \mathbf{\Upsilon} \equiv (p+1)! \star \left( \frac{1}{\sqrt{-g}}\frac{\delta I}{\delta \mathbf F} \right)\,,
    \end{equation}
    with $I = \int \mathbf L$ the action, $\sqrt{-g}$ the volume factor, $\delta I/\delta \mathbf F$ the functional derivative with respect to the form $\mathbf F$, and $\star$ the Hodge dual (see \cref{lem:contraction} for a derivation of this term). Here, $\bm \Upsilon$ can be thought of as an effective ``equation of motion'' as if we were treating $\mathbf{F}$ as fundamental,   forgetting that $\mathbf{F}$ is derived from $\mathbf{A}$. Further, in the last term of \eqref{eq:dL-gi}, $\mathbf \Theta^{\text{GI}}$ is the gauge invariant part of the presymplectic potential, which is locally constructed out of $\phi$, $\delta \phi$ and their derivatives, and it is linear in $\delta \phi$.
    
    Gauge dependence hides in the $\bm \Upsilon \wedge \delta \mathbf{F}$ term  in \eqref{eq:dL-gi}, which can be unpacked by using the definition of the field strength $\mathbf{F} = \dd{\mathbf{A}}$:
    \begin{equation}
        \mathbf{\Upsilon} \wedge \delta \mathbf F = \mathbf{\Upsilon} \wedge \dd{(\delta \mathbf A)} = (-1)^{D-p-1} \left[\dd{\left(\mathbf{\Upsilon} \wedge \delta \mathbf A\right)} - \dd{\mathbf{\Upsilon}} \wedge \delta \mathbf A\right].
    \end{equation}
    Then, we can rewrite $\delta \mathbf L$ as 
    \begin{equation}
        \delta \mathbf L = \mathbf E(\phi) \cdot \delta \tilde{\phi} + \bm{\mathcal E} (\phi) \wedge \delta \mathbf A + \dd{\left(\mathbf \Theta^{\text{GI}}(\phi, \delta \phi) + \mathbf \Theta^{\text{GD}}(\phi, \delta \mathbf A)\right)} \,,\label{eq:dL-gd}
    \end{equation}
    where
    \begin{equation}
        \bm{\mathcal E} = (-1)^{D-p} \dd{\mathbf{\Upsilon}} \label{eq:E-dUpsilon}
    \end{equation}
    is the true equation of motion for the gauge field, and 
    \begin{equation}
        \mathbf \Theta^{\text{GD}}(\phi, \delta \mathbf A) = (-1)^{D-p-1}\mathbf{\Upsilon} (\phi) \wedge \delta \mathbf A \label{eq:Theta-GD}
    \end{equation}
    is the gauge dependent part of the presymplectic potential.

    In total, the   presymplectic potential is given by 
    \begin{equation}
        \mathbf{\Theta}(\phi, \delta \varphi) = \mathbf{\Theta}^{\text{GI}}(\phi, \delta \phi) + \mathbf{\Theta}^{\text{GD}}(\phi, \delta \mathbf{A})\,.
    \end{equation}

    \subsection{Noether charge}
    
    Next, we want to find the Noether charge associated to a diffeomorphism generated by an arbitrary  smooth vector field $\chi.$ We consider a variation of the dynamical fields $\delta \phi = \mathcal L_\chi \phi$ induced by the flow of $\chi$. Diffeomorphism covariance of $\mathbf L$ implies that the variation $\delta_\chi 
    \mathbf L$ produced by $\delta \phi = \mathcal L_\chi \phi$  equals  $ \mathcal L_\chi \mathbf L$, i.e., 
    \begin{equation}
        \mathcal{L}_\chi \mathbf{L} = \mathbf E(\phi) \cdot \mathcal{L}_\chi \tilde{\phi} + \bm{\mathcal E} (\phi)  \wedge \mathcal{L}_\chi \mathbf A + \dd{\left(\mathbf \Theta^{\text{GI}}(\phi, \mathcal{L}_\chi \phi) + \mathbf \Theta^{\text{GD}}(\phi, \mathcal{L}_\chi \mathbf A)\right)}. \label{eq:L-diffeo}
    \end{equation}
Furthermore, it follows from the       Cartan magic formula and the fact that $\mathbf L$ is a top form, $\dd \mathbf L =0$, that $\mathcal{L}_\chi \mathbf{L} = \dd{(\chi \cdot \mathbf{L})}$. By \cite{Seifert:2006kv}, we can write the terms involving the equations of motion as the exterior derivative of the constraint form   
    \begin{equation}
        \mathbf{E}(\phi) \cdot \mathcal{L}_\chi \tilde \phi + \bm{\mathcal{E}} (\phi)  \wedge \mathcal{L}_\chi \mathbf{A} = \dd{\mathbf{C}_\chi} \,,\label{eq:phi-A-constraint}
    \end{equation}
    where the constraint $(D-1)$-form is defined as
    \begin{equation}
        \mathbf{C}_\chi = \left(-E_{(g)}^{ab} + C_{(\psi)}^{ab} + C_{(A)}^{ab}\right) \chi_b \bm \epsilon_a\,. \label{eq:constraint-form}
    \end{equation}
    Here,  
    \begin{equation}
        E_{(g)ab} \equiv \frac{2}{\sqrt{-g}} \frac{\delta I}{\delta g^{ab}}
    \end{equation}
    is the gravitational equation of motion tensor, and
    \begin{align}
        C^{ab}_{(\psi)} & \equiv E_{(\psi)}^{a a_2 \cdots a_s} \psi^{b}{}_{a_2 \cdots a_s} + E_{(\psi)}^{a_1 a a_3\cdots a_s} \psi_{a_1}{}^{b}{}_{a_3 \cdots a_s} + \cdots + E_{(\psi)}^{a_1 \cdots a_{s-1}a} \psi_{a_1 \cdots a_{s-1}}{}^{b},\\
        C^{ab}_{(A)} & \equiv E_{(A)}^{a a_2 \cdots a_p} A^{b}{}_{a_2 \cdots a_p} + E_{(A)}^{a_1 a a_3\cdots a_p} A_{a_1}{}^{b}{}_{a_3 \cdots a_p} + \cdots + E_{(A)}^{a_1 \cdots a_{p-1}a} A_{a_1 \cdots a_{p-1}}{}^{b}
    \end{align}
    are the constraints for a $\psi$ tensor field with rank $s$ and a gauge field $A$ with rank $p$, and
    \begin{align}
        E_{(\psi)}^{a_1 \cdots a_s} & \equiv \frac{1}{\sqrt{-g}} \frac{\delta I}{\delta \psi_{a_1 \cdots a_s}}\,,\\
        E_{(A)}^{a_1 \cdots a_p} & \equiv \frac{1}{\sqrt{-g}} \frac{\delta I}{\delta A_{a_1 \cdots a_p}}
    \end{align}
    are the equations of motion for $\psi$ and $A$, respectively. We review the detailed derivation of the constraint form in appendix \ref{app:constraint-form}. We also note that the equation of motion form $\mathbf{E}_{(A)}$ is related to the form $\bm{\mathcal{E}}$ in \eqref{eq:E-dUpsilon} by a Hodge dual:
    \begin{equation}
        \bm{\mathcal{E}} = p!\,{\star} \mathbf{E}_{(A)} \qquad \Longleftrightarrow \qquad \mathbf{E}_{(A)} = \frac{(-1)^{1 + p(D-p)}}{p!} {\star}\bm{\mathcal{E}}\,.
    \end{equation}
     
We can now rewrite \cref{eq:L-diffeo} as a conservation equation
    \begin{equation}
        \dd{\left( \mathbf{J}_\chi + \mathbf{C}_\chi \right)} = 0\,, \label{eq:os-conservation}
    \end{equation}
    where 
    \begin{equation}
        \mathbf{J}_\chi \equiv \mathbf \Theta^{\text{GI}}(\phi, \mathcal{L}_\chi \phi) + \mathbf \Theta^{\text{GD}}(\phi, \mathcal{L}_\chi \mathbf A) - \chi \cdot \mathbf{L}
    \end{equation}
    is the Noether current $(D-1)$-form associated to    $\chi$.

    Further, from    Poincar\'e's lemma  it follows that~\cite{Wald:1990mme}
    \begin{equation}
        \mathbf{J}_\chi + \mathbf{C}_\chi = \dd{\mathbf{Q}_\chi}\,,
    \end{equation}
    where $\mathbf{Q}_\chi$ is the Noether charge $(D-2)$-form associated to $\chi$. Below we split the Noether charge into a gauge-independent and gauge-dependent part. 

    \subsubsection{Gauge-independent part}
    Notice the off-shell conservation equation \eqref{eq:os-conservation} is a mixture of gauge-independent and gauge-dependent terms. To separate out the gauge-independent part of the Noether charge, we   once again treat $\mathbf{F}$ as fundamental and derive a gauge-invariant conservation equation. Replacing $\delta \phi $ by $\mathcal L_\chi \phi$ in \eqref{eq:dL-gi} and using $\delta_\chi \mathbf L = \mathcal L_\chi \mathbf L$ yields  
    \begin{equation}
        \mathcal{L}_\chi \mathbf{L} = \mathbf E(\phi) \cdot \mathcal{L}_\chi \tilde{\phi} + \bm \Upsilon(\phi) \wedge \mathcal{L}_\chi \mathbf F + \dd{\mathbf \Theta^{\text{GI}}(\phi, \mathcal{L}_\chi \phi)}\,. \label{GILagrdiffeo}
    \end{equation}
    Here, $\mathbf{\Upsilon}$ is the ``equation of motion'' for $\mathbf{F}$, although the correct on-shell condition is not $\mathbf{\Upsilon} = 0$. Using the proof that yields \cref{eq:phi-A-constraint}, where   $\mathbf{A}$ is treated as fundamental, we can find the off-shell relation   
    \begin{equation}
        \mathbf{E}(\phi) \cdot \mathcal{L}_\chi \tilde \phi + \mathbf{\Upsilon}(\phi) \wedge \mathcal{L}_\chi \mathbf{F} = \dd{\mathbf{c}^{(F)}_\chi}\,,\label{eq:phi-F-constraint}
    \end{equation}
    where 
    \begin{equation}
        \mathbf{c}^{(F)}_\chi = \left(-E_{(g)}^{ab} + C_{(\psi)}^{ab} + c_{(F)}^{ab}\right) \chi_b \bm \epsilon_a\,, \label{eq:F-constraint-form}
    \end{equation}
    with 
    \begin{equation}
        c_{(F)}^{ab} \equiv \frac{1}{\sqrt{-g}}\left(\frac{\delta I}{\delta F_{a a_2 \cdots a_{p+1}}} F\indices{^b_{a_2 \cdots a_{p+1}}} + \frac{\delta I}{\delta F_{a_1 a a_3 \cdots a_{p+1}}} F\indices{_{a_1}^b_{a_3 \cdots a_{p+1}}}  + \cdots + \frac{\delta I}{\delta F_{a_1  \cdots a_{p} a}} F\indices{_{a_1 \cdots a_{p}}^b}\right)
    \end{equation}
      the ``would-be constraint form''   if $\mathbf{F}$ were a fundamental field. Importantly, this relation holds off shell and it is manifestly gauge-invariant. See appendix \ref{app:constraint-form} for a detailed derivation.

   Combining \eqref{GILagrdiffeo} and \eqref{eq:phi-F-constraint} and inserting $\mathcal L_\chi \mathbf L = \dd{(\chi \cdot \mathbf L)}$ gives a gauge-invariant conservation equation:
    \begin{equation}
        \dd{\left( \mathbf \Theta^{\text{GI}}(\phi, \mathcal{L}_\chi \phi) - \chi \cdot \mathbf{L} + \mathbf{c}^{(F)}_\chi \right)} = 0\,.
    \end{equation}
 Thus, we  define the gauge-invariant part of the Noether charge via
    \begin{equation}
        \dd{\mathbf{Q}^\text{GI}_\chi} = \mathbf \Theta^{\text{GI}}(\phi, \mathcal{L}_\chi \phi) - \chi \cdot \mathbf{L} + \mathbf{c}^{(F)}_\chi. \label{GINoether}
    \end{equation}
    
    \subsubsection{Gauge-dependent part}
    Next, we calculate the gauge dependent part of   $\mathbf{Q}_\chi$ by studying the gauge dependencies in the presymplectic potential $\mathbf{\Theta}^{\text{GD}}$ and in the constraint form $\mathbf{C}_\chi$. Specifically, we want to rewrite the off-shell conservation equation \eqref{eq:os-conservation} so that the gauge dependence become apparent.
    
    To proceed, we first compute
    \begin{align} \label{eq:E-Lie-A}
            \bm{\mathcal E} \wedge \mathcal{L}_\chi \mathbf A & = \bm{\mathcal E} \wedge \left( \dd{(\chi \cdot \mathbf{A})} + \chi \cdot \mathbf{F}\right) \\
            & = (-1)^{D-p} \dd{(\bm{\mathcal E} \wedge (\chi \cdot \mathbf{A}))} +(-1)^{D-p}\dd{\mathbf{\Upsilon}} \wedge (\chi \cdot \mathbf{F}) \nonumber \\
            & = (-1)^{D-p} \dd{(\bm{\mathcal E} \wedge (\chi \cdot \mathbf{A}))} +(-1)^{D-p} \left[\dd{(\bm \Upsilon \wedge (\chi \cdot \mathbf F))} - (-1)^{D-p-1} \bm \Upsilon \wedge \dd{(\chi \cdot \mathbf F)}\right]  \nonumber\\
            & = (-1)^{D-p}\dd{\left( \bm{\mathcal E} \wedge (\chi \cdot \mathbf{A}) + \mathbf{\Upsilon} \wedge (\chi \cdot \mathbf{F})\right)} + \mathbf{\Upsilon} \wedge \mathcal{L}_\chi \mathbf{F}  \,, \nonumber 
        \end{align}
       where we  used $\bm{\mathcal{E}} = (-1)^{D-p} \dd{\bm \Upsilon}$ and $\dd{\bm{\mathcal{E}}} = 0$ in the second line, and $\dd{(\chi \cdot \mathbf{F})} = \mathcal{L}_\chi \mathbf{F}$ in the final line, which follows from  the Bianchi identity $\dd\mathbf F=0$.
    
    Combining equations \cref{eq:phi-A-constraint}, \cref{eq:phi-F-constraint} and \cref{eq:E-Lie-A} yields    
    \begin{equation}
        \begin{split}
            \dd{\mathbf{C}_\chi} & = \dd{\left[\mathbf{c}^{(F)}_\chi + (-1)^{D-p} \left(  \mathbf{\Upsilon} \wedge (\chi \cdot \mathbf{F}) + \bm{\mathcal E} \wedge (\chi \cdot \mathbf{A}) \right)  \right]}\,,
        \end{split}
    \end{equation}
    so the off-shell conservation equation becomes
    \begin{equation}
        \dd{\left(\mathbf{J}_\chi + \mathbf{c}^{(F)}_\chi + (-1)^{D-p} \left(  \mathbf{\Upsilon} \wedge (\chi \cdot \mathbf{F}) + \bm{\mathcal E} \wedge (\chi \cdot \mathbf{A}) \right)\right)} = 0\,.
    \end{equation}
   We now have all the necessary ingredients to determine the gauge dependence of the Noether charge. The full Noether charge is defined by
    \begin{align}
            \dd{\mathbf{Q}_\chi} & = \mathbf{J}_\chi + \mathbf{c}^{(F)}_\chi + (-1)^{D-p} \left(  \mathbf{\Upsilon} \wedge (\chi \cdot \mathbf{F}) + \bm{\mathcal E} \wedge (\chi \cdot \mathbf{A}) \right)\\
            & = \mathbf \Theta^{\text{GI}}(\phi, \mathcal{L}_\chi \phi) + \mathbf \Theta^{\text{GD}}(\phi, \mathcal{L}_\chi \mathbf A) - \chi \cdot \mathbf{L} + \mathbf{c}^{(F)}_\chi + (-1)^{D-p} \left(  \mathbf{\Upsilon} \wedge (\chi \cdot \mathbf{F}) + \bm{\mathcal E} \wedge (\chi \cdot \mathbf{A}) \right)\nonumber \\
            & = \underbrace{\left[\mathbf \Theta^{\text{GI}}(\phi, \mathcal{L}_\chi \phi) - \chi \cdot \mathbf{L} + \mathbf{c}^{(F)}_\chi \right]}_{\dd{ \mathbf{Q}_\chi^{\text{GI}}}} + (-1)^{D-p-1} \mathbf{\Upsilon} \wedge \dd{(\chi \cdot \mathbf{A})} + (-1)^{D-p} \bm{\mathcal{E}} \wedge (\chi \cdot \mathbf{A})\,,\nonumber
    \end{align} 
    where, in the third line, we identified the gauge-independent part, and inserted the gauge-dependent presymplectic potential
    \begin{equation}
        \mathbf \Theta^{\text{GD}}(\phi, \mathcal{L}_\chi \mathbf A) = (-1)^{D-p-1}\mathbf{\Upsilon} \wedge \left( \dd{(\chi \cdot \mathbf{A})} + \chi \cdot \mathbf{F}\right).
    \end{equation}
    Using \cref{eq:E-dUpsilon}, the gauge dependent terms evaluate to
    \begin{equation}
        \begin{split}
            & (-1)^{D-p-1} \mathbf{\Upsilon} \wedge \dd{(\chi \cdot \mathbf{A})} + (-1)^{D-p} \bm{\mathcal{E}} \wedge (\chi \cdot \mathbf{A})\\
            =~& (-1)^{D-p-1} \mathbf{\Upsilon} \wedge \dd{(\chi \cdot \mathbf{A})} + \dd{\mathbf{\Upsilon}} \wedge (\chi \cdot \mathbf{A})\\
            =~& \dd{(\mathbf{\Upsilon} \wedge (\chi \cdot \mathbf{A}))}.
        \end{split}
    \end{equation}
        Therefore, the Noether charge can be split as
    \begin{equation}
        \mathbf{Q}_\chi = \mathbf{Q}^\text{GI}_\chi + \mathbf{Q}^\text{GD}_\chi \,, 
    \end{equation}
    where the gauge-invariant part satisfies \eqref{GINoether} and the gauge-dependent part reads
    \begin{equation}
        \mathbf{Q}^\text{GD}_\chi = \mathbf{\Upsilon} \wedge (\chi \cdot \mathbf{A})\,.
    \end{equation}

      \subsection{Presymplectic current}
    The presymplectic current is derived from the presymplectic potential via
    \begin{equation}
        \bm \omega(\phi, \delta_1 \varphi, \delta_2 \varphi) \equiv \delta_1 \mathbf{\Theta}(\phi, \delta_2 \varphi) - \delta_2 \mathbf{\Theta}(\phi, \delta_1 \varphi)\,.
    \end{equation}
    The gauge-dependent part follows immediately from the expression for $\mathbf{\Theta}^{\text{GD}}$ in \cref{eq:Theta-GD}:
    \begin{equation}
        \begin{split}
            \bm \omega^{\text{GD}}(\phi, \delta_1 \varphi, \delta_2 \varphi) & = \delta_1 \mathbf{\Theta}^{\text{GD}}(\phi, \delta_2 \varphi) - \delta_2 \mathbf{\Theta}^{\text{GD}}(\phi, \delta_1 \varphi)\\
            & = (-1)^{D-p-1} \left( \delta_1 \mathbf{\Upsilon} \wedge \delta_2 \mathbf{A} - \delta_2 \mathbf{\Upsilon} \wedge \delta_1 \mathbf{A} \right).
        \end{split}
    \end{equation}
    If the second variation is induced  by a diffeomorphism generated by an arbitrary vector field $\chi$, the gauge-dependent part reads
    \begin{equation}
        \bm \omega^{\text{GD}}(\phi, \delta \varphi, \mathcal{L}_\chi \varphi) = (-1)^{D-p-1} \left( \delta \mathbf{\Upsilon} \wedge \mathcal{L}_\chi \mathbf{A} - \mathcal{L}_\chi \mathbf{\Upsilon} \wedge \delta \mathbf{A} \right).
    \end{equation}
    For a diffeomorphism generated by a Killing vector field $\xi$, the gauge-invariant part vanishes 
    \begin{equation} \label{vanishingcurrent}
 \bm \omega^\text{GI}(\phi, \delta \varphi, \mathcal{L}_\xi \varphi) =0\,,
    \end{equation}
   as it depends linearly on $\mathcal L_\xi \varphi$, which is zero due to the stationarity conditions \eqref{statcondition}. Further, the gauge-dependent part becomes
    \begin{equation}
        \bm \omega^\text{GD}(\phi, \delta \phi, \mathcal{L}_\xi \varphi) = (-1)^{D-p-1} \left( \delta \bm \Upsilon \wedge \dd{\bm \lambda_\xi} \right) = \dd{\left( \delta \bm \Upsilon \wedge \bm \lambda_\xi \right)} - (-1)^{D-p} \delta \bm{\mathcal{E}} \wedge \bm \lambda_\xi \,.\label{eq:omega-xi-GD}
    \end{equation}
    In the first equality, we   used $\mathcal{L}_\xi \mathbf A = \dd{\bm \lambda_\xi}$ and $\mathcal L_\xi \mathbf \Upsilon =0$, which follows from the fact that $\mathbf \Upsilon $ is locally constructed from the gauge-invariant fields $(g, \psi, \mathbf F)$ and their derivatives. Moreover, in the second equality we inserted    the perturbed equation of motion $\delta \bm{\mathcal{E}} = (-1)^{D-p}\dd{(\delta \bm \Upsilon)}$.

    \subsection{Fundamental variational identity}

    Finally, by varying the Noether current $\mathbf{J}_\chi$ with respect to both the dynamical fields $\varphi$ and an arbitrary vector field $\chi$, one can derive the fundamental variational identity of the covariant phase space formalism  \cite{Hollands:2012sf}
    \begin{equation}
        \bm \omega(\phi, \delta \varphi, \mathcal{L}_\chi \varphi) = \dd{(\delta_\varphi \mathbf{Q}_\chi - \chi \cdot \mathbf{\Theta}(\phi, \delta \varphi))} - \delta_\varphi \mathbf{C}_\chi + \chi \cdot \left( \mathbf{E}(\phi) \cdot \delta \tilde \phi + \bm{\mathcal{E}}(\phi) \wedge \delta \mathbf{A} \right)\,, \label{fundidentityorigin}
    \end{equation}
    where  we used the definitions \cite{Compere:2006my,Visser:2024pwz}
    \begin{equation}
        \delta_\varphi \mathbf{Q}_\chi \equiv \delta \mathbf{Q}_\chi - \mathbf{Q}_{\delta \chi} \qquad \text{and} \qquad \delta_\varphi \mathbf{C}_\chi \equiv \delta \mathbf{C}_\chi - \mathbf{C}_{\delta \chi}\,.
    \end{equation}
    The fundamental identity  is the key equation which gives rise to both the comparison version and the physical process version of the non-stationary first law of black hole thermodynamics. In the presence of a $p$-form gauge field, it can be separated into a gauge-independent and a gauge-dependent part 
    \begin{equation}
        \begin{split} \label{fundamentalidentity}
            & \quad~\bm \omega^{\text{GI}}(\phi, \delta \varphi, \mathcal{L}_\chi \varphi) +  \bm \omega^{\text{GD}}(\phi, \delta \varphi, \mathcal{L}_\chi \varphi) + \delta_\varphi \mathbf{C}_\chi - \chi \cdot \left( \mathbf{E}(\phi) \cdot \delta \tilde \phi + \bm{\mathcal{E}} (\phi) \wedge \delta \mathbf{A}\right)\\
            & = \dd{\left(\delta_\phi \mathbf{Q}^{\text{GI}}_\chi - \chi \cdot \mathbf{\Theta}^{\text{GI}}(\phi, \delta \phi)+ \delta_\varphi \mathbf{Q}^{\text{GD}}_\chi - \chi \cdot \mathbf{\Theta}^{\text{GD}}(\phi, \delta \mathbf{A})\right)}\,.
        \end{split}
    \end{equation}
    As we shall see shortly, the gauge-dependent contribution precisely yields the $\Phi \delta Q$ term in the first law associated with the electric charge.  
    
    \section{Comparison version of the first law with $p$-form charge} \label{sec:cfl}
 In this section we derive a non-stationary comparison version of the   first law for arbitrary cross-sections of the black hole horizon and for  a general diffeomorphism-invariant Lagrangian that depends on a $p$-form gauge field (together with the metric and other matter fields). We consider dynamical perturbations of a stationary black object geometry  and compare the macroscopic physical quantities before and after at first order in perturbation theory. Here, we do not consider external sources, i.e., $\mathbf L_\text{source} = 0$.  The derivation below extends earlier work \cite{Copsey:2005se,Compere:2007vx,Rogatko:2009th,Keir:2013jga} on the black hole first law for $p$-form   charges to non-stationary perturbations.   Moreover, it extends an earlier proof \cite{Hollands:2024vbe} (see also~\cite{Visser:2024pwz}) of the comparison first law for dynamical black hole entropy to include  electromagnetic potential-charge   terms. As we will see, the formal definition of the dynamical entropy of a black object is  almost the same as  the one proposed by Hollands-Wald-Zhang~\cite{Hollands:2024vbe}, but there is one subtle difference, namely    we define the dynamical entropy solely in terms of   the \emph{gauge-invariant} part of the improved Noether charge, see equation~\eqref{defdynentropyfinal}. 
 
 We specialise to a stationary black hole background, where the vector field $\chi$ is now replaced by the horizon generating Killing field $\xi.$ In that case, the gauge-invariant part of the presymplectic current (evaluated
on the Lie derivative of the fields along $
\xi$) vanishes due to the Killing symmetry, as we saw in \eqref{vanishingcurrent}. Further, we assume the constraint form, the equation-of-motion forms, and their variations are zero, because the setup is source-free and the perturbations are on shell. Then the gauge-dependent part of the presymplectic current, in equation \eqref{eq:omega-xi-GD},     simplifies to
    \begin{equation}
        \bm \omega^{\text{GD}}(\phi, \delta \phi, \mathcal{L}_\xi \varphi) = \dd{(\delta \bm \Upsilon \wedge \bm \lambda_\xi)}\,. \label{eq:omega-xi-GD-horizon}
    \end{equation}
    Therefore, in this setup,  the fundamental   identity  \eqref{fundamentalidentity} becomes 
    \begin{equation} \label{varid2}
        \dd{\left(\delta_\phi \mathbf{Q}^{\text{GI}}_\xi - \xi \cdot \mathbf{\Theta}^{\text{GI}}(\phi, \delta \phi) + \delta_\varphi \mathbf{Q}^{\text{GD}}_\xi - \xi \cdot \mathbf{\Theta}^{\text{GD}}(\phi, \delta \mathbf{A}) - \delta \bm \Upsilon \wedge \bm \lambda_\xi\right)} = 0\,.
    \end{equation}
    The non-stationary black hole first law follows by integrating this identity over a codimension-1 spacelike hypersurface $\Sigma$ with boundary $\partial \Sigma = \mathcal C \cup \mathcal S_\infty$, where  $\mathcal{C}$ is an arbitrary spatial    cross-section of the event horizon, and $\mathcal S_\infty$ is a codimension-2 surface at spatial infinity. In general,  $\mathcal{C}$ and $\mathcal S_\infty$ need not share the same topology --- for example,   in five dimensions a black ring has $\mathcal C \cong S^1 \times S^2$ while $\mathcal S_{\infty} \cong S^3$. Moreover,  $\mathcal{C}$ and $\mathcal S_\infty$ may be non-compact,   as  for black  $(p-1)$-branes which extend along  $p-1$ flat spatial directions. 

Now, it follows from Stokes' theorem that the boundary integral at $\mathcal C$ of the expression between brackets in \eqref{varid2} is equal to the boundary integral at $\mathcal S_\infty$:
  \begin{equation}
        \begin{split}
            &\int_{\mathcal{C}} \left( \delta_\phi \mathbf{Q}^{\text{GI}}_\xi - \xi \cdot \mathbf{\Theta}^{\text{GI}}(\phi, \delta \phi) \right) + \int_{\mathcal{C}} \left( \delta_\varphi \mathbf{Q}^{\text{GD}}_\xi - \xi \cdot \mathbf{\Theta}^{\text{GD}}(\phi, \delta \mathbf{A}) - \delta \bm \Upsilon \wedge \bm \lambda_\xi \right)\\
            &= \int_{\mathcal S_\infty} \left( \delta_\phi \mathbf{Q}^{\text{GI}}_\xi - \xi \cdot \mathbf{\Theta}^{\text{GI}}(\phi, \delta \phi) \right)+\int_{\mathcal S_\infty} \left( \delta_\varphi \mathbf{Q}^{\text{GD}}_\xi - \xi \cdot \mathbf{\Theta}^{\text{GD}}(\phi, \delta \mathbf A)  - \delta \bm \Upsilon \wedge \bm \lambda_\xi  \right)\,.
     \label{fundintid}   \end{split}
    \end{equation}
 We could in principle get rid of the integral of $\delta \bm \Upsilon \wedge \bm \lambda_\xi $ at infinity, by assuming   $\bm \lambda_\xi \to 0$ at $\mathcal S_\infty$, so that the Killing symmetry does not induce a large gauge transformations at spatial infinity. However, below we do keep this term, as $\bm \lambda_\xi$ is formally part of the definition of the electric potential form.

Next, we   evaluate each integral separately, starting with the two integrals of the gauge-dependent terms, and then we compute the integrals of the gauge-invariant terms.

\subsection{Gauge-dependent part: electric charge variation}
\label{subsec:gd}

   We first focus on the gauge-dependent terms, and show   they constitute the $p$-form  electric charge term in the first law for black objects.
    The variation of the gauge-dependent part of the Noether charge is:
    \begin{equation} \label{gdnoether}
        \delta_\varphi \mathbf{Q}^{\text{GD}}_\xi = \delta \mathbf{Q}^\text{GD}_\xi - \mathbf{Q}^\text{GD}_{\delta \xi} = \delta \left( \mathbf{\Upsilon} \wedge (\xi \cdot \mathbf{A}) \right) - \mathbf{\Upsilon} \wedge (\delta \xi \cdot \mathbf{A}) = \delta \mathbf{\Upsilon} \wedge (\xi \cdot \mathbf{A}) + \mathbf{\Upsilon} \wedge (\xi \cdot \delta \mathbf{A})\,.
    \end{equation}
    Then, we calculate the interior product of the presymplectic potential with $\xi$:
    \begin{equation}
        \begin{split} \label{gdsympt}
            \xi \cdot \mathbf \Theta^{\text{GD}}(\phi, \delta \mathbf A) & = (-1)^{D-p-1}\xi \cdot (\mathbf{\Upsilon} \wedge \delta \mathbf A)\\
            & = \mathbf{\Upsilon} \wedge (\xi \cdot \delta \mathbf{A}) + (-1)^{D-p-1}(\xi \cdot \mathbf{\Upsilon}) \wedge \delta \mathbf{A}\,.
        \end{split}
    \end{equation}
    Combining these results yields   
    \begin{equation}
        \delta_\varphi \mathbf{Q}^\text{GD}_\xi - \xi \cdot \mathbf \Theta^{\text{GD}}  = \delta \mathbf{\Upsilon} \wedge (\xi \cdot \mathbf{A}) + (-1)^{D-p} (\xi \cdot \bm \Upsilon) \wedge \delta \mathbf{A}\,. \label{gdimpnoether}
    \end{equation}
Note that $ \mathbf{\Upsilon} \wedge (\xi \cdot \delta \mathbf{A})$ cancels in this combination of differential forms. The second term on the right side vanishes both at spatial infinity and at the future Killing horizon. This is because, for electric charges,  we impose a Dirichlet boundary condition on the gauge field at spatial infinity: $\delta \mathbf A \to 0 $ as $r \to \infty$ (see fall-off condition 2 in section \ref{sec:falloff}). Further,   when pulled back to the future Killing horizon, the second term can be written in GNC as
    \begin{equation} \label{secondtermvanishesupsilon}
        (\xi \cdot \bm \Upsilon) \wedge \delta \mathbf A \fheq \frac{1}{(D-p-2)! p!} \kappa v \Upsilon_{v i_1 \cdots i_{D-p-2}} \delta A_{i_{D-p-1} \cdots i_{D-2}} \dd{x^{i_1}} \wedge \cdots \wedge \dd{x^{i_{D-2}}} \fheq 0\,.
    \end{equation}
    It vanishes on the horizon as a result of \cref{thm:boost-weight}, because $\Upsilon_{v i_1 \cdots i_{D-p-2}}$ is a background weight-1 tensor component that is gauge invariant (and hence regular).

    Thus, the pullback of this combination of gauge-dependent forms to   the future horizon and spatial infinity is simply
    \begin{equation}
        \delta_\varphi \mathbf{Q}^\text{GD}_\xi - \xi \cdot \mathbf \Theta^{\text{GD}}(\phi, \delta \mathbf A)  - \delta \bm \Upsilon \wedge \bm \lambda_\xi  \overset{\mathcal H^+,\, \mathcal S_\infty}{=}  \delta \mathbf{\Upsilon} \wedge (\xi \cdot \mathbf{A} - \bm \lambda_\xi) = -  \delta \mathbf{\Upsilon} \wedge \bm \Phi \,, \label{eq:improv-Q-GD}
    \end{equation}
    where we identified the $p$-form electric potential form $\bm \Phi \equiv - \xi \cdot \mathbf{A} + \bm \lambda_\xi$. 
Note that $\delta \mathbf{\Upsilon} $ does not vanish at spatial infinity, but rather asymptotes to a finite value as $r 
\to \infty$, due to our fall-off condition 3 in section \ref{sec:falloff}.

    We now argue the difference in integrals of the differential form in \eqref{eq:improv-Q-GD} over $\mathcal C$ and $\mathcal S_\infty$ gives rise to the $p$-form analogue of the electric charge term $ \bar \Phi \delta Q$   in the comparison version of the first law, where $\bar \Phi$ is the difference between the electric potential at the horizon and infinity. For convenience, we only describe how to evaluate the   integral at a horizon cross-section $\mathcal C$ in detail, since the integral at spatial infinity $\mathcal S_\infty$ gives a similar result. This follows from the fact that   the electric   potential form $\bm \Phi$ is both closed on a Killing horizon and at spatial infinity (see \cref{thm:zeroth-law} and equation~\eqref{conditiononFatinfinity}). 
    
    For a 1-form gauge field $A_a$, the electric potential is just a  scalar $\Phi$ that is constant on a Killing horizon, so 
    \begin{equation} \label{maxwellchargepotential}
        - \int_\mathcal{C} \delta \bm \Upsilon \wedge \bm \Phi = \Phi \delta Q\,,
    \end{equation}
    where we defined the electric charge as
    \begin{equation} \label{chargevariation}
        Q = - \int_\mathcal{C}   \bm \Upsilon\,.
    \end{equation}
However, for higher-form gauge fields of degree   $p > 1$, the electric charge term has a more complicated structure, as $\bm \Phi$ is no longer a scalar. We follow the methods in \cite{Copsey:2005se} (see also \cite{Compere:2007vx}) 
  to uncover this in detail. We first explain  how to generalise \eqref{maxwellchargepotential} to   charges associated to $p$-form gauge fields for compact horizon slices $\mathcal{C} =\tilde{\mathcal{C}} $, and then we treat the more general case of   non-compact horizon cross-sections $\mathcal C=\tilde{\mathcal C}\times \Sigma_k$.
  
  \subsubsection{Compact horizon cross-sections}\label{sssec:compact}
  We follow   standard   procedures for differential forms in algebraic topology, discussed in \cite{Bott:1982xhp} and \cite{PetersenManifolds}. By Hodge decomposition, the closed $\bm \Phi$ form can be decomposed as a sum of a harmonic form and an exact form on a compact horizon  slice $\tilde{\mathcal{C}}$, i.e., 
    \begin{equation}
        \bm \Phi \overset{\tilde{\mathcal{C}}}{=} \bm \eta + \dd{\bm \alpha}\,,
    \end{equation}
    where $\bm \eta$ is a harmonic  $(p-1)$-form, which satisfies $\Delta \bm \eta = (\dd \dd^{\dagger} + \dd^{\dagger}\dd) \bm \eta = 0$, with $\dd^{\dagger}$ the adjoint of the exterior derivative with respect to the inner product $\expval{\mathbf{X}, \mathbf{Y}} = \int {\star} \mathbf{X} \wedge \mathbf{Y}$. The exact term $\dd{\bm \alpha}$ is an ambiguity, as $\bm \lambda_\xi$ is defined up to an exact term (see \cref{ssec:zeroth-law}). For compact horizons,  $\delta \bm \Upsilon \wedge \dd{\bm \alpha}$ will integrate out as $\bm \Upsilon$ is closed on-shell, so only the harmonic part of $\bm \Phi$ matters. The harmonic form    $\bm{\eta
   }$ is an   element of the $(p-1)$-th de Rham cohomology $H^{p-1}_{\text{dR}}(\tilde{\mathcal{C}})$. We write it in a harmonic, orthonormal basis $\{\hat{\bm \eta}_\sigma\}$
    \begin{equation} \label{harmonicbasis}
        \bm \eta = \sum_\sigma \Phi_\sigma^{\mathcal H^+} \hat{\bm \eta}_\sigma\,.
    \end{equation}
 By the Poincar\'e duality between cohomology and homology, each $\hat{\bm \eta}_\sigma$ is dual to a non-trivial $(D-p-1)$-cycle $\sigma$ of $\tilde{\mathcal{C}}$. Furthermore, the coefficients $\Phi_\sigma^{\mathcal H^+}$ are constant on $\tilde{\mathcal C}$, since they are integrals of the potential form over the non-trivial $(p-1)$-cycles $\tilde{\sigma}\subset \tilde{\mathcal C}$ dual to $\star_{\tilde{\mathcal C}} \hat{\bm \eta}_\sigma$ (the Hodge dual of $\hat{\bm \eta}_\sigma$ within $\tilde{\mathcal C}$),
    \begin{equation} \label{electricpotentialfundeq}
      \Phi_\sigma^{\mathcal H^+} = \int_{\tilde \sigma} \bm \Phi   \big |_{\tilde{\mathcal{C}}} = \int_{\tilde \sigma} \bm \eta\,.
  \end{equation}
Here, we normalised the harmonic basis such that $\int_{\tilde \sigma_1}\hat{\bm \eta}_{\sigma_2}=\delta_{\sigma_{1} \sigma_{2}}. $
      
    Next, we use   Poincar\'e duality to further convert the integrals over $\tilde{\mathcal{C}}$ into integrals along $\sigma$,
 \begin{equation}
        \int_{\tilde{\mathcal{C}}} \mathbf{X} \wedge \hat{\bm \eta}_\sigma = \int_{\sigma} \mathbf{X}
    \end{equation}
    for any closed $(D-p-1)$-form $\mathbf{X}$. 
    Applying this to our setup (recall that $\delta \bm \Upsilon$ is closed by the perturbed equations of motion), we find 
    \begin{equation}
        - \int_{\tilde{\mathcal{C}}} \delta \bm \Upsilon \wedge \bm \Phi = - \int_{\tilde{\mathcal{C}}} \delta \bm \Upsilon \wedge \bm \eta = - \sum_\sigma \Phi_\sigma^{\mathcal H^+} \int_{\sigma} \delta \bm \Upsilon = \sum_\sigma \Phi_\sigma^{\mathcal H^+} \delta Q_\sigma \,,\label{eq:Phi-sigma-Q-sigma}
    \end{equation}
    where 
    \begin{equation} \label{pcharge}
        Q_\sigma = -\int_{\sigma} \bm \Upsilon
    \end{equation}
    is the electric charge on the cycle $\sigma$. We note that $\Phi_\sigma^{\mathcal H^+}$ and $Q_\sigma$ are defined up to a proportionality factor, since the final expression in   \eqref{eq:Phi-sigma-Q-sigma} is invariant under the rescalings $\Phi_\sigma^{\mathcal H^+} \to \Phi_\sigma^{\mathcal H^+} \lambda$ and $Q_\sigma \to Q_\sigma / \lambda$, if the variation of $\lambda$ vanishes. This $\lambda$ parameter relates to the different normalisation schemes in the Poincar\'e duality.
    
    The number of different species of electric charges equals the number of  non-trivial $(D-p-1)$-cycles, given by the rank of the homology group $H_{D-p-1} (\tilde{\mathcal{C}})$ of the horizon cross-section, i.e., by the $(D-p-1)$-th Betti number $b_{D-p-1}$. By Poincar\'{e} duality for the compact, oriented manifold $\tilde{\mathcal{C}}$, the homology group $H_{D-p-1} (\tilde{\mathcal{C}})$ is dual to the  de Rham cohomology group $H^{p-1}_{\text{dR}}(\tilde{\mathcal{C}})$, and their dimensions coincide,  $b_{D-p-1}=b_{p-1}$. Therefore, the number of independent potential-charge pairs is $b_{p-1}$, corresponding to the number of distinct topological sectors supporting electric flux on $\tilde{\mathcal C}$. Below we illustrate this with some examples of compact horizon topologies. 

   Further, assuming the horizon cycles  $\sigma$'s  are homologous inside the spacelike hypersurface~$\Sigma$ to   corresponding  non-trivial cycles of spatial infinity $\mathcal S_\infty$, the boundary integral  at $\mathcal S_\infty$ becomes
      \begin{equation} \label{infintegralcharge}
        - \int_\mathcal{\mathcal S_\infty} \delta \bm \Upsilon \wedge \bm \Phi    = \sum_\sigma \Phi_\sigma^{\infty} \delta Q_\sigma \,,  
    \end{equation}
   where  the coefficients $\Phi_\sigma^{\infty}$ are constants on $\mathcal S_\infty$. Here, we used that  the electric potential form is closed at infinity because of   fall-off condition 1 in section \ref{sec:falloff}. This implies that $\bm \Phi$ admits a decomposition into a harmonic and exact form at infinity, hence   the same arguments as above lead to \eqref{infintegralcharge}, with the coefficients 
$\Phi_\sigma^{\infty}$
 fixed by cycle integrals on $\mathcal S_\infty$, i.e., $\Phi_\sigma^{\infty} = \int_{\tilde \sigma_\infty} \mathbf{\Phi} |_{\mathcal S_\infty}$, where $\tilde \sigma_\infty \subset \mathcal S_\infty$ is a ($p-1$)-cycle that is Poincar\'{e} dual to $\star_{\mathcal S_\infty}\hat{\bm \eta}_{\sigma_\infty}$. Note that the charge variation associated to a non-trivial cycle of a horizon cross-section is the same as the charge variation at a homologous cycle of   spatial infinity --- both denoted as $ \delta Q_\sigma$ --- as follows from the linearised gauge field equation $\dd{(\delta \bm \Upsilon)}=0$. Further, we want to emphasise that the electric charge $Q_\sigma$ is gauge invariant, since $\mathbf \Upsilon$ only depends on gauge-invariant fields $\phi.$

We should point out, however, some black objects have   no non-trivial  $(D-p-1)$-cycles at spatial infinity. An example   is a 5-dimensional black ring charged under a 2-form gauge field, with the topology of $\mathcal S_\infty$ being   $S^3$ and   the horizon topology being  $S^1 \times S^2$. In this case, the horizon cycles are $\sigma = S^2$ and $\tilde \sigma = S^1$. However, the three-sphere at infinity  has no non-trivial  $S^2$ cycles (nor $S^1$ cycles). Therefore, the boundary integral \eqref{infintegralcharge} at $\mathcal S_\infty$ should vanish for these configurations. Indeed, if $\mathcal S_\infty$ has no non-trivial  $(D-p-1)$-cycles, then by Poincar\'{e} duality we have: $\text{dim}\,H_{D-p-1}(\mathcal S_\infty)=\text{dim}\, H^{p-1}_{\text{dR}}(\mathcal S_\infty)=0$,  hence every closed $(p-1)$-form on $\mathcal S_\infty$ is exact. In particular, $\dd{\mathbf \Phi}=0$ implies $ \mathbf \Phi = \dd \bm{\alpha}$. Since  $\dd {(\delta \mathbf \Upsilon)} =0$, the boundary integral \eqref{infintegralcharge} becomes the integral of an exact form $\dd (\delta   \mathbf{ \Upsilon } \wedge \bm{\alpha}  )$, which vanishes on the compact surface $\mathcal S_\infty$ due to Stokes' theorem. In the non-compact case this boundary integral can be set to zero by  fall-off condition~4 in section \ref{sec:falloff}. Thus, in this case, there are no electric potential-charge terms at infinity that contribute  to the first law.  
   
Now, assuming all $(D-p-1)$ horizon cycles $\sigma$ have a homologous counterpart at
$\mathcal S_\infty$,\footnote{It would be interesting to identify black objects for which $\mathcal S_\infty$ contains non-trivial cycles  (unlike the black ring) that are not homologous inside $\Sigma$ to $\sigma$, thus   violating the homology condition. In such cases, the electric charges defined at the horizon and at infinity differ and contribute separately to the first law.}   the difference between the boundary integrals of the gauge-dependent terms at $\mathcal C$ and $\mathcal S_\infty$ is  
 \begin{equation} \label{finalchargetermdifference}
   \left (  \int_{\mathcal{C}}- \int_{\mathcal{S}_\infty} \right) \left[ \delta_\varphi \mathbf{Q}^{\text{GD}}_\xi - \xi \cdot \mathbf{\Theta}^{\text{GD}}(\phi, \delta \mathbf{A}) - \delta \bm \Upsilon \wedge \bm \lambda_\xi \right] =   - \left ( \int_\mathcal{C} - \int_\mathcal{\mathcal S_\infty} \right) \delta \bm \Upsilon \wedge \bm \Phi    = \sum_\sigma \bar \Phi_\sigma \delta Q_\sigma \,,  
    \end{equation}
where $\bar \Phi_\sigma \equiv \Phi_\sigma^{\mathcal H^+} -  \Phi_\sigma^{\infty}$\,.  On the one hand, for asymptotically flat black holes the electric potential vanishes at spatial infinity, $\Phi_\sigma^{\infty} =0$, since the gauge field falls off as $\mathcal O (1/r^{D-3})$, so that  $\Phi_\sigma^{\mathcal H^+} \delta Q_\sigma$ appears in the first law. On the other hand, for asymptotically AdS black object, the electric potential at spatial infinity may be non-zero, if we impose the fall-off condition $\mathbf A = \mathbf A_\infty + \mathcal  O (1/r^{D-3})$.  In that case the charge term in the first law involves the difference in electric potential between the horizon and infinity. 

    We end with an important observation that the term $\mathbf{\Upsilon} \wedge \delta (\xi \cdot \mathbf{A})$ cancels in the combination $\delta_\varphi \mathbf{Q}^{\text{GD}}_\xi - \xi \cdot \mathbf{\Theta}^{\text{GD}}(\phi, \delta \mathbf{A}),$ as follows from \eqref{gdnoether} and \eqref{gdsympt}. In fact, $\mathbf{\Upsilon} \wedge (\delta  \xi \cdot \mathbf{A})$ already cancels within $\delta_\varphi \mathbf{Q}^{\text{GD}}_\xi$, since $\delta_\varphi$ only acts on the dynamical fields. This implies that a term like $Q_\sigma \delta \Phi_\sigma$ does not appear in the first law, as is required for a (contact) one-form on     thermodynamic phase space, since the electric potential and charge are conjugate variables, and only the variation of one of them  should feature in the thermodynamic first law. 

\subsubsection{Examples of compact horizon topologies}
\label{compacthorizontopologies}

  We give three examples of compact horizon cross-sections with various topologies to demonstrate the non-trivial cycles dual to the potential form and the different species of associated electric charges. 
\begin{itemize}
\item[A)] We first consider the case of a 5-dimensional  black ring with horizon topology $S^1 \times S^2$ charged under a 2-form gauge field $\mathbf B$ \cite{Emparan:2004wy,Emparan:2001wn}. In this case, the electric potential $\bm \Phi$ is a  1-form and the electric flux density $\bm \Upsilon$ is a  2-form. The electric dipole charge is associated with non-trivial 2-cycles $\sigma$ of the horizon cross-section, where, in this case, there is only one: the sphere $\sigma = S^2$. On the other hand, the cycle $\tilde \sigma$, dual to $\star_{\tilde{\mathcal C}} \hat{\bm \eta}_\sigma$, is the horizon cycle over which   the potential form is integrated, in this case:  the circle $\tilde \sigma = S^1$. The horizon charge term   that contributes to the first law is 
\begin{equation}
    - \int_{S^1 \times S^2} \delta \bm \Upsilon \wedge \bm \Phi = -\Phi_{S^2}\int_{S^2} \delta \bm \Upsilon =  \Phi_{S^2} \delta Q_{S^2}\,,
\end{equation}
where $\Phi_{S^2} = \int_{S^1} \bm \Phi$ is the electrostatic potential. We refer to section \ref{sec:blackrings} for more details.

\begin{figure}[h!]
    \centering

\tikzset{every picture/.style={line width=0.75pt}} 

\begin{tikzpicture}[x=0.75pt,y=0.75pt,yscale=-1,xscale=1]

\draw  [color={rgb, 255:red, 208; green, 2; blue, 27 }  ,draw opacity=1 ] (60,110) .. controls (60,93.43) and (73.43,80) .. (90,80) .. controls (106.57,80) and (120,93.43) .. (120,110) .. controls (120,126.57) and (106.57,140) .. (90,140) .. controls (73.43,140) and (60,126.57) .. (60,110) -- cycle ;
\draw  [color={rgb, 255:red, 74; green, 144; blue, 226 }  ,draw opacity=1 ] (156.67,108.33) .. controls (156.67,85.32) and (175.32,66.67) .. (198.33,66.67) .. controls (221.35,66.67) and (240,85.32) .. (240,108.33) .. controls (240,131.35) and (221.35,150) .. (198.33,150) .. controls (175.32,150) and (156.67,131.35) .. (156.67,108.33) -- cycle ;
\draw  [draw opacity=0][dash pattern={on 4.5pt off 4.5pt}] (240,108.33) .. controls (240,108.33) and (240,108.33) .. (240,108.33) .. controls (240,108.33) and (240,108.33) .. (240,108.33) .. controls (240,117.08) and (221.35,124.17) .. (198.33,124.17) .. controls (175.32,124.17) and (156.67,117.08) .. (156.67,108.33) -- (198.33,108.33) -- cycle ; \draw  [color={rgb, 255:red, 74; green, 144; blue, 226 }  ,draw opacity=1 ][dash pattern={on 4.5pt off 4.5pt}] (240,108.33) .. controls (240,108.33) and (240,108.33) .. (240,108.33) .. controls (240,108.33) and (240,108.33) .. (240,108.33) .. controls (240,117.08) and (221.35,124.17) .. (198.33,124.17) .. controls (175.32,124.17) and (156.67,117.08) .. (156.67,108.33) ;  
\draw  [draw opacity=0][dash pattern={on 4.5pt off 4.5pt}] (156.67,108.33) .. controls (156.67,99.59) and (175.32,92.5) .. (198.33,92.5) .. controls (221.35,92.5) and (240,99.59) .. (240,108.33) -- (198.33,108.33) -- cycle ; \draw  [color={rgb, 255:red, 74; green, 144; blue, 226 }  ,draw opacity=0.5 ][dash pattern={on 4.5pt off 4.5pt}] (156.67,108.33) .. controls (156.67,99.59) and (175.32,92.5) .. (198.33,92.5) .. controls (221.35,92.5) and (240,99.59) .. (240,108.33) ;  

\draw (91.17,68.17) node  [color={rgb, 255:red, 208; green, 2; blue, 27 }  ,opacity=1 ]  {$S^{1}$};
\draw (200.83,54.17) node  [color={rgb, 255:red, 74; green, 144; blue, 226 }  ,opacity=1 ]  {$S^{2}$};
\draw (85.67,168.5) node  [color={rgb, 255:red, 208; green, 2; blue, 27 }  ,opacity=1 ]  {$\displaystyle \int _{S^{1}}\mathbf{\Phi } =\Phi _{S^{2}}$};
\draw (205.67,168.83) node  [color={rgb, 255:red, 74; green, 144; blue, 226 }  ,opacity=1 ]  {$-\displaystyle \int _{S^{2}} \delta \mathbf{\Upsilon } =\delta Q_{S^{2}}$};
\draw (139.5,110) node    {$\times $};

\end{tikzpicture}

    \caption{A five-dimensional black ring charged under a 2-form gauge field has horizon topology $S^1 \times S^2$. The   potential form $\bm \Phi$ integrates on the $S^1$, while the electric flux density form $\delta \bm \Upsilon$ integrates on the~$S^2$.}
    \label{fig:black-ring}
\end{figure}
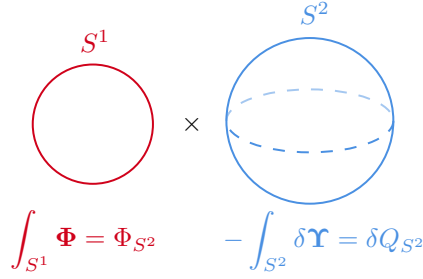

\item[B)] For a more involved example, consider a 6-dimensional toroidal blackfold with horizon topology $T^2 \times S^2= S^1_a \times S^1_b \times S^2$ \cite{Emparan:2009at,Emparan:2009vd}. To distinguish the two circles, we have labelled them by $a$ and $b$, respectively. 

Case (i): when a 6-dimensional blackfold is charged under a 2-form gauge field $\mathbf B_{(2)}$, with a 1-form electric potential $\bm \Phi_{(1)}$ and a 3-form electric flux density $\bm \Upsilon_{(3)}$, the non-trivial 3-cycles dual to the harmonic part of the potential form are: $\sigma_a = S^1_a \times S^2$ and $\sigma_b = S^1_b \times S^2$. Therefore, there exist two different species of electric charge that contribute to the first law:
\begin{equation}
    - \int_{S^1_a \times S^1_b \times S^2} \delta \bm \Upsilon_{(3)} \wedge \bm \Phi_{(1)} = - \Phi_{a} \int_{S^1_a \times S^2} \delta \bm \Upsilon_{(3)} - \Phi_{b} \int_{S^1_b \times S^2} \delta \bm \Upsilon_{(3)} = \Phi_{a} \delta Q_{a} + \Phi_{b} \delta Q_{b} 
\end{equation}
with
\begin{equation}
    \Phi_a = \int_{S^1_b} \bm \Phi_{(1)} \quad \text{and} \quad \Phi_b = \int_{S^1_a} \bm \Phi_{(1)}.
\end{equation}

\begin{figure}[t!]
    \centering

\tikzset{every picture/.style={line width=0.75pt}} 

\begin{tikzpicture}[x=0.75pt,y=0.75pt,yscale=-1,xscale=1]

\draw  [color={rgb, 255:red, 74; green, 144; blue, 226 }  ,draw opacity=1 ] (151,179) .. controls (151,162.43) and (164.43,149) .. (181,149) .. controls (197.57,149) and (211,162.43) .. (211,179) .. controls (211,195.57) and (197.57,209) .. (181,209) .. controls (164.43,209) and (151,195.57) .. (151,179) -- cycle ;
\draw  [color={rgb, 255:red, 65; green, 117; blue, 5 }  ,draw opacity=1 ] (247.67,177.33) .. controls (247.67,154.32) and (266.32,135.67) .. (289.33,135.67) .. controls (312.35,135.67) and (331,154.32) .. (331,177.33) .. controls (331,200.35) and (312.35,219) .. (289.33,219) .. controls (266.32,219) and (247.67,200.35) .. (247.67,177.33) -- cycle ;
\draw  [draw opacity=0][dash pattern={on 4.5pt off 4.5pt}] (331,177.33) .. controls (331,177.33) and (331,177.33) .. (331,177.33) .. controls (331,186.08) and (312.35,193.17) .. (289.33,193.17) .. controls (266.32,193.17) and (247.67,186.08) .. (247.67,177.33) -- (289.33,177.33) -- cycle ; \draw  [color={rgb, 255:red, 65; green, 117; blue, 5 }  ,draw opacity=1 ][dash pattern={on 4.5pt off 4.5pt}] (331,177.33) .. controls (331,177.33) and (331,177.33) .. (331,177.33) .. controls (331,186.08) and (312.35,193.17) .. (289.33,193.17) .. controls (266.32,193.17) and (247.67,186.08) .. (247.67,177.33) ;  
\draw  [draw opacity=0][dash pattern={on 4.5pt off 4.5pt}] (247.67,177.33) .. controls (247.67,177.33) and (247.67,177.33) .. (247.67,177.33) .. controls (247.67,168.59) and (266.32,161.5) .. (289.33,161.5) .. controls (312.35,161.5) and (331,168.59) .. (331,177.33) -- (289.33,177.33) -- cycle ; \draw  [color={rgb, 255:red, 65; green, 117; blue, 5 }  ,draw opacity=0.5 ][dash pattern={on 4.5pt off 4.5pt}] (247.67,177.33) .. controls (247.67,177.33) and (247.67,177.33) .. (247.67,177.33) .. controls (247.67,168.59) and (266.32,161.5) .. (289.33,161.5) .. controls (312.35,161.5) and (331,168.59) .. (331,177.33) ;  
\draw  [color={rgb, 255:red, 208; green, 2; blue, 27 }  ,draw opacity=1 ] (59,178.5) .. controls (59,161.93) and (72.43,148.5) .. (89,148.5) .. controls (105.57,148.5) and (119,161.93) .. (119,178.5) .. controls (119,195.07) and (105.57,208.5) .. (89,208.5) .. controls (72.43,208.5) and (59,195.07) .. (59,178.5) -- cycle ;
\draw [color={rgb, 255:red, 155; green, 155; blue, 155 }  ,draw opacity=1 ]   (90,250) -- (90,215.5) ;
\draw [shift={(90,213.5)}, rotate = 90] [color={rgb, 255:red, 155; green, 155; blue, 155 }  ,draw opacity=1 ][line width=0.75]    (10.93,-3.29) .. controls (6.95,-1.4) and (3.31,-0.3) .. (0,0) .. controls (3.31,0.3) and (6.95,1.4) .. (10.93,3.29)   ;
\draw [color={rgb, 255:red, 155; green, 155; blue, 155 }  ,draw opacity=1 ]   (230,252.5) -- (181.1,215.7) ;
\draw [shift={(179.5,214.5)}, rotate = 36.96] [color={rgb, 255:red, 155; green, 155; blue, 155 }  ,draw opacity=1 ][line width=0.75]    (10.93,-3.29) .. controls (6.95,-1.4) and (3.31,-0.3) .. (0,0) .. controls (3.31,0.3) and (6.95,1.4) .. (10.93,3.29)   ;
\draw [color={rgb, 255:red, 155; green, 155; blue, 155 }  ,draw opacity=1 ]   (230,252.5) -- (288.7,223.88) ;
\draw [shift={(290.5,223)}, rotate = 154.01] [color={rgb, 255:red, 155; green, 155; blue, 155 }  ,draw opacity=1 ][line width=0.75]    (10.93,-3.29) .. controls (6.95,-1.4) and (3.31,-0.3) .. (0,0) .. controls (3.31,0.3) and (6.95,1.4) .. (10.93,3.29)   ;
\draw [color={rgb, 255:red, 155; green, 155; blue, 155 }  ,draw opacity=1 ]   (179.5,142.5) -- (179.5,126.5) ;
\draw [shift={(179.5,144.5)}, rotate = 270] [color={rgb, 255:red, 155; green, 155; blue, 155 }  ,draw opacity=1 ][line width=0.75]    (10.93,-3.29) .. controls (6.95,-1.4) and (3.31,-0.3) .. (0,0) .. controls (3.31,0.3) and (6.95,1.4) .. (10.93,3.29)   ;
\draw [color={rgb, 255:red, 155; green, 155; blue, 155 }  ,draw opacity=1 ]   (133,96) -- (90.88,140.05) ;
\draw [shift={(89.5,141.5)}, rotate = 313.71] [color={rgb, 255:red, 155; green, 155; blue, 155 }  ,draw opacity=1 ][line width=0.75]    (10.93,-3.29) .. controls (6.95,-1.4) and (3.31,-0.3) .. (0,0) .. controls (3.31,0.3) and (6.95,1.4) .. (10.93,3.29)   ;
\draw [color={rgb, 255:red, 155; green, 155; blue, 155 }  ,draw opacity=1 ]   (217.5,90) -- (288.74,128.06) ;
\draw [shift={(290.5,129)}, rotate = 208.11] [color={rgb, 255:red, 155; green, 155; blue, 155 }  ,draw opacity=1 ][line width=0.75]    (10.93,-3.29) .. controls (6.95,-1.4) and (3.31,-0.3) .. (0,0) .. controls (3.31,0.3) and (6.95,1.4) .. (10.93,3.29)   ;

\draw (182.17,178.17) node  [color={rgb, 255:red, 74; green, 144; blue, 226 }  ,opacity=1 ]  {$S_{b}^{1}$};
\draw (291.83,177.17) node  [color={rgb, 255:red, 65; green, 117; blue, 5 }  ,opacity=1 ]  {$S^{2}$};
\draw (89.67,263) node  [color={rgb, 255:red, 0; green, 0; blue, 0 }  ,opacity=1 ]  {$\displaystyle \int _{\textcolor[rgb]{0.82,0.01,0.11}{S_{a}^{1}}}\mathbf{\Phi }_{( 1)} =\Phi _{b}$};
\draw (231.67,263.83) node  [color={rgb, 255:red, 0; green, 0; blue, 0 }  ,opacity=1 ]  {$-\displaystyle \int _{\textcolor[rgb]{0.29,0.56,0.89}{S_{b}^{1}} \times \textcolor[rgb]{0.25,0.46,0.02}{S^{2}}} \delta \mathbf{\Upsilon }_{( 3)} =\delta Q_{b}$};
\draw (229.5,178) node    {$\times $};
\draw (90.17,177.67) node  [color={rgb, 255:red, 208; green, 2; blue, 27 }  ,opacity=1 ]  {$S_{a}^{1}$};
\draw (134.5,178) node    {$\times $};
\draw (179.67,117) node  [color={rgb, 255:red, 0; green, 0; blue, 0 }  ,opacity=1 ]  {$\displaystyle \int _{\textcolor[rgb]{0.29,0.56,0.89}{S_{b}^{1}}}\mathbf{\Phi }_{( 1)} =\Phi _{a}$};
\draw (186.67,76.5) node  [color={rgb, 255:red, 0; green, 0; blue, 0 }  ,opacity=1 ]  {$-\displaystyle \int _{\textcolor[rgb]{0.82,0.01,0.11}{S_{a}^{1}} \times \textcolor[rgb]{0.25,0.46,0.02}{S}\textcolor[rgb]{0.25,0.46,0.02}{^{2}}} \delta \mathbf{\Upsilon }_{( 3)} =\delta Q_{a}$};

\end{tikzpicture}

    \caption{Blackfold with horizon topology $T^2 \times S^2$ charged under a 2-form gauge field. There are two species of electric charge: $Q_a$ associated to $\sigma_a = S^1_a \times S^2$ and $Q_b$ associated to $\sigma_b = S^1_b \times S^2$.}
    \label{fig:blackfold1}
\end{figure}
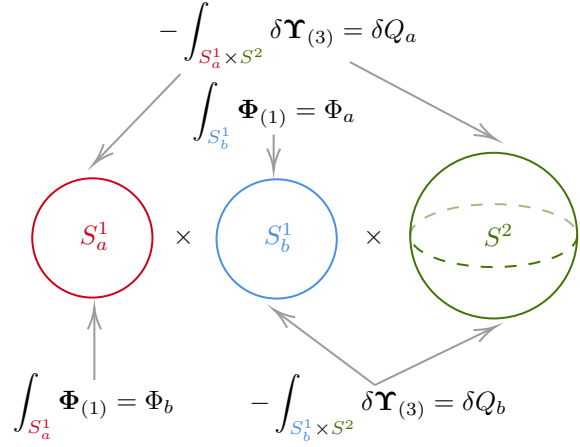

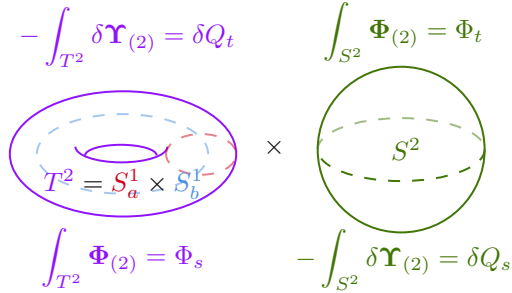
\begin{figure}[h!]
    \centering

\tikzset{every picture/.style={line width=0.75pt}} 

\begin{tikzpicture}[x=0.75pt,y=0.75pt,yscale=-1,xscale=1]

\draw  [color={rgb, 255:red, 144; green, 19; blue, 254 }  ,draw opacity=1 ] (13.5,113.5) .. controls (13.5,96.1) and (38.8,82) .. (70,82) .. controls (101.2,82) and (126.5,96.1) .. (126.5,113.5) .. controls (126.5,130.9) and (101.2,145) .. (70,145) .. controls (38.8,145) and (13.5,130.9) .. (13.5,113.5) -- cycle ;
\draw  [draw opacity=0] (92.11,109.89) .. controls (92.11,114.24) and (81.84,117.76) .. (69.18,117.76) .. controls (56.52,117.76) and (46.25,114.24) .. (46.25,109.89) -- (69.18,109.89) -- cycle ; \draw  [color={rgb, 255:red, 144; green, 19; blue, 254 }  ,draw opacity=1 ] (92.11,109.89) .. controls (92.11,114.24) and (81.84,117.76) .. (69.18,117.76) .. controls (56.52,117.76) and (46.25,114.24) .. (46.25,109.89) ;  
\draw  [draw opacity=0] (51.07,114.31) .. controls (51.07,111.33) and (59.07,108.92) .. (68.94,108.92) .. controls (78.8,108.92) and (86.8,111.33) .. (86.8,114.31) -- (68.94,114.31) -- cycle ; \draw  [color={rgb, 255:red, 144; green, 19; blue, 254 }  ,draw opacity=1 ] (51.07,114.31) .. controls (51.07,111.33) and (59.07,108.92) .. (68.94,108.92) .. controls (78.8,108.92) and (86.8,111.33) .. (86.8,114.31) ;  
\draw  [color={rgb, 255:red, 65; green, 117; blue, 5 }  ,draw opacity=1 ] (167.17,111.33) .. controls (167.17,88.32) and (185.82,69.67) .. (208.83,69.67) .. controls (231.85,69.67) and (250.5,88.32) .. (250.5,111.33) .. controls (250.5,134.35) and (231.85,153) .. (208.83,153) .. controls (185.82,153) and (167.17,134.35) .. (167.17,111.33) -- cycle ;
\draw  [draw opacity=0][dash pattern={on 4.5pt off 4.5pt}] (250.5,111.33) .. controls (250.5,111.33) and (250.5,111.33) .. (250.5,111.33) .. controls (250.5,111.33) and (250.5,111.33) .. (250.5,111.33) .. controls (250.5,120.08) and (231.85,127.17) .. (208.83,127.17) .. controls (185.82,127.17) and (167.17,120.08) .. (167.17,111.33) -- (208.83,111.33) -- cycle ; \draw  [color={rgb, 255:red, 65; green, 117; blue, 5 }  ,draw opacity=1 ][dash pattern={on 4.5pt off 4.5pt}] (250.5,111.33) .. controls (250.5,111.33) and (250.5,111.33) .. (250.5,111.33) .. controls (250.5,111.33) and (250.5,111.33) .. (250.5,111.33) .. controls (250.5,120.08) and (231.85,127.17) .. (208.83,127.17) .. controls (185.82,127.17) and (167.17,120.08) .. (167.17,111.33) ;  
\draw  [draw opacity=0][dash pattern={on 4.5pt off 4.5pt}] (167.17,111.33) .. controls (167.17,111.33) and (167.17,111.33) .. (167.17,111.33) .. controls (167.17,102.59) and (185.82,95.5) .. (208.83,95.5) .. controls (231.85,95.5) and (250.5,102.59) .. (250.5,111.33) -- (208.83,111.33) -- cycle ; \draw  [color={rgb, 255:red, 65; green, 117; blue, 5 }  ,draw opacity=0.5 ][dash pattern={on 4.5pt off 4.5pt}] (167.17,111.33) .. controls (167.17,111.33) and (167.17,111.33) .. (167.17,111.33) .. controls (167.17,102.59) and (185.82,95.5) .. (208.83,95.5) .. controls (231.85,95.5) and (250.5,102.59) .. (250.5,111.33) ;  
\draw  [color={rgb, 255:red, 74; green, 144; blue, 226 }  ,draw opacity=0.5 ][dash pattern={on 4.5pt off 4.5pt}] (27,113.5) .. controls (27,102.45) and (46.14,93.5) .. (69.75,93.5) .. controls (93.36,93.5) and (112.5,102.45) .. (112.5,113.5) .. controls (112.5,124.55) and (93.36,133.5) .. (69.75,133.5) .. controls (46.14,133.5) and (27,124.55) .. (27,113.5) -- cycle ;
\draw  [color={rgb, 255:red, 208; green, 2; blue, 27 }  ,draw opacity=0.5 ][dash pattern={on 4.5pt off 4.5pt}] (91.5,113.75) .. controls (91.5,108.09) and (99.34,103.5) .. (109,103.5) .. controls (118.66,103.5) and (126.5,108.09) .. (126.5,113.75) .. controls (126.5,119.41) and (118.66,124) .. (109,124) .. controls (99.34,124) and (91.5,119.41) .. (91.5,113.75) -- cycle ;

\draw (70,165.83) node  [color={rgb, 255:red, 144; green, 19; blue, 254 }  ,opacity=1 ]  {$\displaystyle \int _{T^{2}}\mathbf{\Phi }_{( 2)} =\Phi _{s}$};
\draw (210.17,166.67) node  [color={rgb, 255:red, 65; green, 117; blue, 5 }  ,opacity=1 ]  {$-\displaystyle \int _{S^{2}} \delta \mathbf{\Upsilon }_{( 2)} =\delta Q_{s}$};
\draw (145.5,110) node    {$\times $};
\draw (211.33,111.17) node  [color={rgb, 255:red, 65; green, 117; blue, 5 }  ,opacity=1 ]  {$S^{2}$};
\draw (70.33,128.67) node  [color={rgb, 255:red, 0; green, 0; blue, 0 }  ,opacity=1 ]  {$\textcolor[rgb]{0.56,0.07,1}{T^{2}} =\textcolor[rgb]{0.82,0.01,0.11}{S_{a}^{1}} \times \textcolor[rgb]{0.29,0.56,0.89}{S_{b}^{1}}$};
\draw (210,52.33) node  [color={rgb, 255:red, 65; green, 117; blue, 5 }  ,opacity=1 ]  {$\displaystyle \int _{S^{2}}\mathbf{\Phi }_{( 2)} =\Phi _{t}$};
\draw (72.17,54.67) node  [color={rgb, 255:red, 144; green, 19; blue, 254 }  ,opacity=1 ]  {$-\displaystyle \int _{T^{2}} \delta \mathbf{\Upsilon }_{( 2)} =\delta Q_{t}$};

\end{tikzpicture}

    \caption{Blackfold with horizon topology $T^2 \times S^2$ charged under 3-form gauge field $\mathbf B_{(3)}$. There are two species of electric charge: $Q_t$ associated to $\sigma_t = S^1_a\times S^1_b$ and $Q_s$ associated to $\sigma_s = S^2$. }
    \label{fig:blackfold2}
\end{figure}

Case (ii): when the blackfold is charged under a 3-form gauge field $\mathbf B_{(3)}$ with electric potential $\bm \Phi_{(2)}$ and electric flux density $\bm \Upsilon_{(2)}$, the non-trivial horizon 2-cycles are: $\sigma_t = S^1_a \times S^1_b$ and $\sigma_s = S^2$. Accordingly, in the first law we have the potential-charge terms
\begin{equation}
    - \int_{S^1_a \times S^1_b \times S^2} \delta \bm \Upsilon_{(2)} \wedge \bm \Phi_{(2)} = - \Phi_t \int_{S^1_a \times S^1_b} \delta \bm \Upsilon_{(2)} - \Phi_s \int_{S^2} \delta \bm \Upsilon_{(2)} = \Phi_t \delta Q_t + \Phi_s \delta Q_s 
\end{equation}
with 
\begin{equation}
    \Phi_t = \int_{S^2} \bm \Phi_{(2)} \quad \text{and} \quad \Phi_s = \int_{S^1_a \times S^1_b} \bm \Phi_{(2)}.
\end{equation}

\item[C)] Finally, we generalise the first two examples to a blackfold with horizon topology $T^m \times S^n$ where $T^m = S^1 \times \cdots \times S^1$ ($m$ copies of $S^1$) is an $m$-torus, and $S^n$ is a topological $n$-sphere. When it is charged under a $p$-form gauge field $\mathbf A_{(p)}$, the electric potential form is a $(p-1)$-form~$\bm \Phi_{(p-1)}$. Then, the non-trivial cycles $\{\sigma\}$ dual to the harmonic part of $\bm \Phi_{(p-1)}$ are 
\begin{equation}
    \sigma = \begin{cases}
        T^{m-p+1} \times S^n, & p-1 < n\\
        T^{m-p+1} \times S^n \quad \text{or} \quad  T^{m+n-p+1}, & p - 1 \geq n.
    \end{cases}
\end{equation}
where $T^{m-p+1}$ are $(m-p+1)$-subtori of $T^m$, and there are $\binom{m}{p-1}$ species of them. Similarly for $T^{m+n-p+1}$. The out-of-range cases are zero automatically. In the end, each of these non-trivial cycles would have a different type of electric charge associated to it. The number of potential-charge pairs is given by K\"{u}nneth's formula for the ($p-1$)-th Betti number of~$T^m \times S^n:$
\begin{equation}
    b_{p-1} (T^m \times S^n) =  \begin{dcases}
        \binom{m}{p-1}, & p - 1 <n\\
        \binom{m}{p-1} +\binom{m}{p-n-1}, & p-1 \geq n\,.
    \end{dcases}
\end{equation}
For instance, $b_1(S^1 \times S^2)=1$, showing the black ring example  only has one type of electric charge. Further, $b_1 (T^2 \times S^2 ) = 2$ and $b_2 (T^2 \times S^2)=2$, in agreement with the two potential-charge pairs of case (i) and (ii) in the previous example.
\end{itemize}

  \subsubsection{Non-compact horizon cross-sections} \label{sssec:non-compact}

For non-compact horizon slices $\mathcal C$ of the form $ \tilde{\mathcal{C}} \times \Sigma_k$, where $\tilde{\mathcal{C}}$ is a $(D-2-k)$-dimensional compact submanifold and $\Sigma_k$ is a non-compact submanifold with    the topology of $\mathbb R^k$,  the generalisation of \eqref{finalchargetermdifference} is more involved, since the Hodge decomposition theorem only holds on compact manifolds.  We want to decompose    $\int_{\mathcal C} \delta \bm \Upsilon \wedge \bm \Phi$ into integrations along $\tilde{\mathcal{C}}$ and $\Sigma_k$, respectively, and use the ``fibre integration'' technique to ``downgrade'' $\delta \bm \Upsilon$ and $\bm \Phi$ to lower differential forms living on $\tilde{\mathcal{C}}$. 

Recall that $\delta \bm \Upsilon$ is a $(D-p-1)$-form, and $\bm \Phi$ is a $(p-1)$-form. We impose the constant electric potential gauge \eqref{constpotgaugeoriginal}, such that   $\bm \Phi$ is constant along the null and the extended $k $ directions of the future horizon.  Further, we assign   $x^{i_1}, x^{i_2}, \cdots $ as coordinates on~$\mathcal{C}$;
        $\theta^{\alpha_1}, \theta^{\beta_2}, \cdots$ as coordinates of the compact subspace $\tilde{\mathcal{C}}$; and 
         $y^{\aleph_1}, y^{\aleph_2}, \cdots$ as coordinates along the extended (topologically flat) directions of $\Sigma_k$. 
   
    We first calculate the pullback of $ \delta \bm \Upsilon \wedge \bm \Phi$ to $\mathcal C$ 
    \begin{align}
            & \delta \bm \Upsilon \wedge \bm \Phi \ceq \mathcal N \delta \Upsilon_{[i_1 \cdots i_{D-p-1}} \Phi_{i_{D-p} \cdots i_{D-2}]} \dd{x^{i_1}} \wedge \cdots \wedge \dd{x^{i_{D-2}}} \nonumber\\
            & \ceq \mathcal N \Biggl[C^{D-p-1}_{0} C^{p-1}_k \delta \Upsilon_{[\alpha_1 \cdots \alpha_{D-p-1}} \Phi_{\aleph_1 \cdots \aleph_k \beta_1 \cdots \beta_{p-k-1}]} \dd[D-p-1]{\theta^{\alpha}} \wedge \dd[k]{y^\aleph} \wedge \dd[p-k-1]{\theta^\beta}\nonumber\\
            & \qquad \qquad + C^{D-p-1}_{1} C^{p-1}_{k-1} \delta \Upsilon_{[\alpha_1 \cdots \alpha_{D-p-2}\aleph_1} \Phi_{\aleph_2 \cdots \aleph_k \beta_1 \cdots \beta_{p-k}]} \dd[D-p-2]{\theta^{\alpha}} \wedge \dd[k]{y^\aleph} \wedge \dd[p-k]{\theta^\beta}\nonumber \\
            & \qquad \qquad + \cdots\\
            & \qquad \qquad + C^{D-p-1}_{k} C^{p-1}_{0} \delta \Upsilon_{[\alpha_1 \cdots \alpha_{D-p-k-1} \aleph_1 \cdots \aleph_k } \Phi_{\beta_1 \cdots \beta_{p-1}]} \dd[D-p-k-1]{\theta^{\alpha}} \wedge \dd[k]{y^\aleph} \wedge \dd[p-1]{\theta^\beta} \Biggr]\nonumber\\
            & \ceq \mathcal N \sum_{r=0}^{k} C^{D-p-1}_{r} C^{p-1}_{k-r} \delta \Upsilon_{[\alpha_1 \cdots \alpha_{D-p-r-1}\aleph_1 \cdots \aleph_r} \Phi_{\aleph_{r+1} \cdots \aleph_k \beta_{1} \cdots \beta_{p+r-k-1}]} \dd[D-p-r-1]{\theta^\alpha} \wedge \dd[k]{y^\aleph} \wedge \dd[p+r-k-1]{\theta^\beta}\,, \nonumber
        \end{align}
    where, for convenience, we have denoted
    \begin{equation}
        \mathcal N = \frac{1}{(p-1)!(D-p-1)!}
    \end{equation}
    and $C^m_n$ is the binomial factor 
    \begin{equation}
        C^m_{n} = {m \choose n} = \frac{m!}{n! (m-n)!}
    \end{equation}
    We also introduced shorthand notations
    \begin{equation}
        \dd[n]{\theta^{\alpha}} = \dd{\theta^{\alpha_1}}\wedge \cdots \wedge \dd{\theta^{\alpha_n}}
    \end{equation}
    and, similarly, for $\dd[k]{y^\aleph}$.

    Then, the integral over $\mathcal{C}$ becomes 
    \begin{equation} \label{noncomapctintermediate}
        \int_{\mathcal{C}} \delta \bm \Upsilon \wedge \bm \Phi = \sum_{r=0}^{k} \int_{\tilde{\mathcal{C}}} \delta \tilde{\bm \Upsilon}_{(D-p-r-1)}^{\aleph_{r+1} \cdots \aleph_{k}} \wedge \tilde{\bm \Phi}^{(p+r-k-1)}_{\aleph_{r+1} \cdots \aleph_{k}}\,,
    \end{equation}
    where summation over $\aleph_{r+1} \cdots \aleph_{k}$ is implied. The ($p+r-k-1$)-form is defined as
    \begin{equation} \label{downgradedphiform}
        \tilde{\bm \Phi}^{(p+r-k-1)}_{\aleph_{r+1} \cdots \aleph_{k}} \equiv \frac{1}{(k-r)!}\,\iota^*[m_{\aleph_k}  \cdot (m_{\aleph_{k-1}} \cdot ({\dots} \cdot (m_{\aleph_{r+1}} \cdot \bm \Phi)))]\,,
    \end{equation}
    with $\iota: \tilde{\mathcal{C}} \to \mathcal{C}$ denoting the inclusion map and $\iota^*$ the pullback. Moreover, the other $(D-p-r-1)$-form is defined as
    \begin{equation}
        \delta \tilde{\bm \Upsilon}^{\aleph_{r+1} \cdots \aleph_k}_{(D-p-r-1)} \equiv \frac{1}{r!}\, \iota^* \int_{\Sigma_k} [((\delta \bm \Upsilon \cdot m_{\aleph_r}) \cdot \dots) \cdot m_{\aleph_1}] \dd{y^{\aleph_1}} \cdots \dd{y^{\aleph_k}}\,,
    \end{equation}
 where $(\cdots)\cdot m_\aleph$ is the contraction with the last index. Such procedure is known as ``fibre integration'', that is,  we integrated out the legs in the extended directions to obtain a lower-degree differential form on $\tilde{\mathcal{C}}$. Note that the integration in the extended directions is included in the definition of $\delta \tilde{\bm{\Upsilon}}$ as the perturbation could depend on the $y^\aleph$ coordinates, but it is not included in $\tilde{\bm{\Phi}}$ because the potential form is independent of the $y^\aleph$ coordinates.

Now that we have ``downgraded''  $\delta \bm \Upsilon$ and $\bm \Phi$ to lower-degree  forms   on (submanifolds of) the compact manifold $\tilde{\mathcal{C}}$, we can use similar algebraic topological arguments as we did before in the case of compact horizons. 
In order to do so, we need one more fact, namely that   both of these downgraded forms are closed on $\tilde{\mathcal{C}}$
 \begin{equation} \label{downgradedoncompactsubspace}
        \dd{(\delta \tilde{\bm \Upsilon}^{\aleph_{r+1} \cdots \aleph_k}_{(D-p-r-1)})} \tceq 0\,, \qquad \dd{\tilde{\bm \Phi}^{(p+r-k-1)}_{\aleph_{r+1}\cdots \aleph_k}} \tceq 0\,.
    \end{equation}
    We leave the proofs of these equations to    appendix \ref{appc}.
    Since $\tilde{\bm \Phi}^{(p+r-k-1)}_{\aleph_{r+1}\cdots \aleph_k}$ is a closed form on a compact manifold, we can  apply Hodge decomposition   
    \begin{equation}
        \tilde{\bm \Phi}^{(p+r-k-1)}_{\aleph_{r+1}\cdots \aleph_k} \tceq \bm \eta^{(p+r-k-1)}_{\aleph_{r+1}\cdots \aleph_k} + \dd{\bm \alpha^{(p+r-k-2)}_{\aleph_{r+1}\cdots \aleph_k}}\,,
    \end{equation}
    where $\bm \eta^{(p+r-k-1)}_{\aleph_{r+1}\cdots \aleph_k}$ is a harmonic   form, which hence satisfies the Laplace equation $\Delta \bm \eta^{(p+r-k-1)}_{\aleph_{r+1}\cdots \aleph_k} = 0$.
    We can ignore the exact terms, because they would integrate to zero in \eqref{noncomapctintermediate}, since $\delta \tilde{\bm \Upsilon}^{\aleph_{r+1} \cdots \aleph_k}_{(D-p-r-1)}$ is closed and $\tilde{\mathcal{C}}$ is compact. Hence, only the   harmonic parts of the downgraded potential forms contribute to the integral in \eqref{noncomapctintermediate}. As before, they can be resolved in an orthonormal harmonic basis $\{\hat {\bm \eta}^{(p+r-k-1)}_{\sigma_r}\}$ 
    \begin{equation}
        \bm \eta^{(p+r-k-1)}_{\aleph_{r+1}\cdots \aleph_k} = \sum_{\sigma_r} \Phi^{\sigma_r, \mathcal{H}^+}_{\aleph_{r+1}\cdots \aleph_k} \hat{\bm \eta}^{(p+r-k-1)}_{\sigma_r}\,.
    \end{equation}
Here, $\sigma_r$'s are the non-trivial $(D-p-r-1)$-cycles dual to $\hat {\bm \eta}^{(p+r-k-1)}_{\sigma_r}$'s, and $\Phi^{\sigma_r, \mathcal{H}^+}_{\aleph_{r+1}\cdots \aleph_k}$ are constants on $\tilde{\mathcal{C}}$ that are fixed by integrals of the downgraded potential form~\eqref{downgradedphiform} over non-trivial $(p+r-k-1)$-dimensional cycles $\tilde\sigma_r \subset \tilde{\mathcal{C}}$ dual to $\star_{\tilde{\mathcal C}}\hat{\bm \eta}^{(p+r-k-1)}_{\sigma_r}$.
    Further, we proceed as before using the cohomology/homology duality:
    \begin{equation}
        \int_{\mathcal{C}} \mathbf{X} \wedge \hat{\bm \eta}^{(p+r-k-1)}_{\sigma_r} = \int_{\sigma_r} \mathbf{X} \,,
    \end{equation}
    for any closed $(D-p-r-1)$-form $\mathbf{X}$. This allows us to compute the boundary integral under consideration 
    \begin{equation}
        \begin{split}
            - \int_\mathcal{C} \delta \bm \Upsilon \wedge \bm \Phi & = -\sum_{r=0}^{k} \sum_{\sigma_r} \Phi^{\sigma_r, \mathcal{H}^+}_{\aleph_{r+1}\cdots \aleph_k} \int_{\tilde{\mathcal{C}}} \delta \tilde{\bm \Upsilon}_{(D-p-r-1)}^{\aleph_{r+1} \cdots \aleph_{k}} \wedge  \hat{\bm \eta}^{(p+r-k-1)}_{\sigma_r}\\
            & = - \sum_{r=0}^{k} \sum_{\sigma_r} \Phi^{\sigma_r, \mathcal{H}^+}_{\aleph_{r+1}\cdots \aleph_k} \int_{\sigma_r} \delta \tilde{\bm \Upsilon}_{(D-p-r-1)}^{\aleph_{r+1} \cdots \aleph_{k}}\\
            & = \sum_{r=0}^{k} \sum_{\sigma_r} \Phi^{\sigma_r, \mathcal{H}^+}_{\aleph_{r+1}\cdots \aleph_k} \delta Q^{\aleph_{r+1} \cdots \aleph_{k}}_{\sigma_r}\,, \label{finalnoncompactchargeterm}
        \end{split}
    \end{equation}
    where 
    \begin{equation}
        Q^{\aleph_{r+1} \cdots \aleph_{k}}_{\sigma_r} = - \int_{\sigma_r} \tilde{\bm \Upsilon}_{(D-p-r-1)}^{\aleph_{r+1} \cdots \aleph_{k}}
    \end{equation}
    is the electric charge on cycle $\sigma_r$ whose species is labelled by $\aleph_{r+1} \cdots \aleph_{k}$.

We can compute the boundary integral over $\mathcal S_\infty$ in a similar fashion, so that the generalisation of \eqref{finalchargetermdifference} to non-compact horizons becomes 
 \begin{equation} \label{finalchargetermdifferencenoncompact}
   \left (  \int_{\mathcal{C}}- \int_{\mathcal{S}_\infty} \right) \left[ \delta_\varphi \mathbf{Q}^{\text{GD}}_\xi - \xi \cdot \mathbf{\Theta}^{\text{GD}}(\phi, \delta \mathbf{A}) - \delta \bm \Upsilon \wedge \bm \lambda_\xi \right] =    \sum_{r=0}^{k} \sum_{\sigma_r} \bar \Phi^{\sigma_r}_{\aleph_{r+1}\cdots \aleph_k} \delta Q^{\aleph_{r+1} \cdots \aleph_{k}}_{\sigma_r} \,,  
    \end{equation}
where the bar on the electric  potential constants denotes the difference between the values at the horizon and at infinity. Here, we assumed again that $\mathcal S_\infty$ has non-trivial cycles that are homologous inside the spatial hypersurface $\Sigma$ to the horizon cycles $\sigma_r$. Compared to \eqref{finalchargetermdifference} we now have an additional sum over labels associated to the   extended directions.

For non-compact horizon cross-sections, the integrals of the electric flux over the extended dimensions diverge, because the integration along these extended directions produces an infinite volume. To obtain finite and physically meaningful quantities, one therefore considers densities rather than total charges. Concretely,  one could quotient out the divergent factor  
$
\omega_k \;=\; \int_{\Sigma_k}\bm{\epsilon}_k ,
$
where $\bm{\epsilon}_k$ is the volume form on $\Sigma_k$, and define the finite charge density as  
$
\rho_{\sigma_r} \;=\; Q_{\sigma_r}/\omega_k 
$
for each cycle. The coefficients $\bar \Phi_{\sigma_r}$   are already   finite and well defined. This prescription matches the usual treatment of planar black holes and black branes.  We will see an explicit realisation of this procedure in section \ref{sec:exam}, but in the following  we will formulate the first law in terms of  the total charge.
      
\subsection{Gauge-invariant part:  dynamical entropy variation}
Next, we move to the gauge-invariant part of the fundamental identity~\eqref{fundintid}. That is, we evaluate the  two integrals of the gauge-independent terms  over an arbitrary  horizon cross-section and  over spatial infinity.  For stationary, axisymmetric black holes and branes, one can express the horizon Killing   vector field at infinity generically in terms   of asymptotic time translations and rotations
    \begin{equation}
        \xi = \pdv{t} + \sum_I \Omega_{I} \pdv{\vartheta^I}\,,
    \end{equation} 
    where $\Omega_I$ is the angular velocity of the horizon associated with rotation in the $\vartheta^I$ direction. If the geometry has no compact angular isometries, like   planar black holes, then the angular velocities vanish. 
   The boundary integral of the gauge-invariant forms at infinity    gives rise to the variation of the canonical mass $ M$ and angular momenta $ J_I$:
    \begin{equation}
        \int_{\mathcal S_\infty} \left( \delta_\phi \mathbf{Q}_\xi^{\text{GI}} - \xi \cdot \mathbf{\Theta}^{\text{GI}}(\phi, \delta \phi) \right) = \delta M - \sum_I \Omega_I \delta J_I\,. \label{eq:int-S-infty}
    \end{equation}
    Plugging \eqref{eq:int-S-infty} and \eqref{finalchargetermdifferencenoncompact} into  \eqref{fundintid} yields
    \begin{equation} \label{prefirstlaws}
        \int_{\mathcal{C}} \left( \delta_\phi \mathbf{Q}^{\text{GI}}_\xi - \xi \cdot \mathbf{\Theta}^{\text{GI}}(\phi, \delta \phi) \right) = \delta M - \sum_I \Omega_I \delta J_I - \sum_{r=0}^{k} \sum_{\sigma_r} \bar \Phi^{\sigma_r, \mathcal{H}^+}_{\aleph_{r+1}\cdots \aleph_k} \delta Q^{\aleph_{r+1} \cdots \aleph_{k}}_{\sigma_r}.
    \end{equation}
    We immediately see that the left-hand side~should be identified with the product of the Hawking temperature $T=\kappa/2 \pi$ and the variation of the \emph{dynamical entropy} of the black object:
    \begin{equation} \label{dynvarent}
        \int_{\mathcal{C}}\left(\delta_\phi \mathbf{Q}^\text{GI}_\xi - \xi \cdot \mathbf{\Theta}^\text{GI}(\phi, \delta \phi)\right) = T \delta S_\text{dyn}\,.
    \end{equation}
In previous work  \cite{Hollands:2024vbe,Visser:2024pwz} it was shown that the dynamical entropy is   well defined if the following two consistency conditions hold:
\begin{itemize}
    \item[a)]  There exists a codimension-1 form $\mathbf B_{\mathcal H^+} (\phi)$ such that
    \begin{equation} \label{entropyconditiona}
    \mathbf \Theta^\text{GI} (\phi, \delta \phi) \fheq\delta \mathbf B^\text{GI}_{\mathcal H^+} (\phi)\,,
\,\qquad \mathbf B_{\mathcal H^+}  (\phi)\fheq 0.
\end{equation}
\item[b)] The field-only variation  of the Noether charge $\delta_\phi \mathbf{Q}_\xi^{\text{GI}}\equiv \delta \mathbf{Q}_\xi^{\text{GI}} -  \mathbf{Q}_{\delta \xi}^{\text{GI}}$ must satisfy
\begin{equation} \label{entropyconditionb}
    \delta_\phi \mathbf{Q}_\xi^{\text{GI}} = \kappa \delta \!\left ( \mathbf{Q}^{\text{GI}}_\xi / \kappa_3 \right)\,.
\end{equation}
Here, the surface gravity $\kappa_3$ is defined on the background Killing horizon  as 
\begin{equation} \label{kappa3}
    \kappa_3^2 \overset{\mathcal H^+}{\equiv} -\frac{1}{2} (\nabla^a \xi^b) (\nabla_{[a} \xi_{b]})\,, 
\end{equation}
and its variation on the perturbed event horizon can be expressed in terms of the variation of the background horizon Killing field, see equation (2.31) in \cite{Visser:2024pwz}.
\end{itemize}
Condition a) was proposed by Hollands, Wald and Zhang \cite{Hollands:2024vbe}, and they   proved   that such a form $\mathbf{B}_{\mathcal H^+}$ exists  for general diffeomorphism-invariant Lagrangians for which the metric is the only dynamical field. In our follow-up work \cite{Visser:2024pwz} we showed   that \eqref{entropyconditiona} holds for any diffeomorphism-invariant Lagrangians whose dynamical fields are the metric and arbitrary non-minimally coupled bosonic matter
fields that are smooth on the horizon.   Further, condition b) appeared for the first time in our work \cite{Visser:2024pwz}, and was proven for the same  general class of diffeomorphism-invariant Lagrangian. This condition is necessary  if one allows for variations of the Killing field and the surface gravity --- which were both kept fixed in \cite{Hollands:2024vbe} --- since the Noether charge depends on the Killing field. Crucially, it was shown that the $\delta \kappa_3$ term (i.e, the temperature variation) drops out in the dynamical first law of black objects, as it should, so that only $T \delta S_{\text{dyn}}$ appears in the first law and not $S_{\text{dyn}} \delta T$. This is similar to our earlier observation that $Q_\sigma \delta \Phi_\sigma$ does not appear in the first law (see the end of section~\ref{subsec:gd}).

Now, what is relevant for the present paper, is that the two conditions continue to hold for   diffeomorphism-invariant Lagrangians that depend on non-minimally coupled abelian $p$-form gauge fields that are non-smooth on the Killing horizon. That is because the differential forms in \eqref{entropyconditiona} and \eqref{entropyconditionb} are gauge invariant, hence they do not depend explicitly on the gauge field $\mathbf A$, and thus they are smooth on the horizon. Therefore, the proofs of conditions a) and b) from our previous paper  \cite{Visser:2024pwz} remain valid here, since  they apply to non-minimally coupled bosonic matter fields (such as the field strength $\mathbf{F}$)  that are smooth on the horizon. 

Given that these two conditions are satisfied in our setup,  the dynamical gravitational entropy  is defined up to linear order in perturbation theory away from a stationary background  as
\begin{equation} \label{defdynentropyfinal}
    S_{\text{dyn}} = \int_{\mathcal C} \frac{2\pi}{\kappa_3}  \tilde{\mathbf{Q}}_\xi^{\text{GI}} = \int_{\mathcal C} \frac{2\pi }{\kappa_3} (\mathbf{Q}_\xi^{\text{GI}} - \xi \cdot \mathbf{B}^{\text{GI}}_{\mathcal H^+})\,,
\end{equation}  
where $\tilde{\mathbf{Q}}_\xi$ is called the improved Noether charge. 
This definition  of the dynamical gravitational entropy is the same as the one proposed by Hollands, Wald and Zhang \cite{Hollands:2024vbe}, except that here we define the entropy in terms of the \emph{gauge-invariant} part of the improved Noether charge. In their original paper \cite{Hollands:2024vbe},   HWZ did not consider gauge fields, hence this extra subtlety about gauge dependence did not play a role in their work. It is satisfying that the gravitational entropy is a gauge-invariant quantity (see, e.g.,   \cite{Hajian:2022lgy}), since otherwise the entropy would depend on a gauge choice. 
    
    Finally, inserting \eqref{dynvarent} into \eqref{prefirstlaws}, we obtain the first law for non-stationary perturbations of a  black object charged under $p$-form gauge  fields
    \begin{equation}
       T \delta S_\text{dyn} = \delta M - \sum_I \Omega_I \delta J_I - \sum_{r=0}^{k} \sum_{\sigma_r}\bar  \Phi^{\sigma_r, \mathcal{H}^+}_{\aleph_{r+1}\cdots \aleph_k} \delta Q^{\aleph_{r+1} \cdots \aleph_{k}}_{\sigma_r}\,.
    \end{equation}

    \subsection{Comments on magnetic charges}
    \label{ssec:magnetic}

In \cref{subsec:gd} we only addressed the inclusion of electric charges in the black hole first law. Here, we   comment on the more subtle case of magnetically charged black objects. Again, for simplicity, we assume the horizon cross-section to be compact $\mathcal C = \tilde{\mathcal C}$. 

The  charge of a magnetic monopole is the integral of the field strength over a compact surface, i.e., $P \propto \oint \mathbf F$. This is only nonzero if the field strength is not globally exact  $\mathbf F \neq \dd \mathbf A.$ However, locally, within a particular coordinate chart, we can still have $\mathbf F = \dd \mathbf A.$
Magnetic monopole charges thus arise from non-trivial gauge field configurations that cannot be globally described within a single coordinate chart. For ordinary 1-form abelian gauge fields, this obstruction reflects the non-trivial topology of the underlying principal $U(1)$ bundle. For higher-form abelian gauge fields, the analogous obstruction is the non-triviality of the corresponding ``gerbe'' (higher principal bundle). Because the gauge potential $\mathbf{A}$ cannot be   globally defined as a connection on a non-trivial principal bundle, gauge-dependent integrals must be evaluated patchwise. One covers the integration domain by local charts $\mathcal{U}_\mathcal{I}$ with potentials $\mathbf{A}_{\mathcal I}$ and patches them on overlaps (or ``seams'')  $\mathcal{U}_\mathcal{IJ}=\mathcal{U}_\mathcal{I} \cap \mathcal U_\mathcal{J} $ via gauge transformations $\mathbf{A}_{\mathcal I} - \mathbf{A}_{\mathcal J} = \dd{\mathbf{\Lambda}_{\mathcal{IJ}}}$. When assembling the global result, these overlap regions contribute additional boundary terms along the seams $\mathcal{U}_{\mathcal{IJ}}$, which must be included to obtain the correct, globally consistent value.

    \begin{figure}[t]
        \centering

\tikzset{every picture/.style={line width=0.75pt}} 

\begin{tikzpicture}[x=0.75pt,y=0.75pt,yscale=-1,xscale=1]

\draw  [color={rgb, 255:red, 155; green, 155; blue, 155 }  ,draw opacity=0.5 ][fill={rgb, 255:red, 208; green, 2; blue, 27 }  ,fill opacity=0.25 ] (103.4,44.44) .. controls (133.5,30.25) and (147.76,33.31) .. (169.2,33.31) .. controls (190.64,33.31) and (223.9,43.79) .. (208.38,64.75) .. controls (192.85,85.7) and (105.62,78.5) .. (90.1,86.36) .. controls (74.57,94.22) and (43.52,145.3) .. (47.22,119.1) .. controls (50.91,92.91) and (74.57,60.82) .. (103.4,44.44) -- cycle ;
\draw  [color={rgb, 255:red, 155; green, 155; blue, 155 }  ,draw opacity=1 ][fill={rgb, 255:red, 208; green, 2; blue, 27 }  ,fill opacity=0.5 ] (104,44.75) .. controls (134.25,30.75) and (151.25,31.5) .. (169.5,32.5) .. controls (187.75,33.5) and (161.91,51.07) .. (172,74.75) .. controls (182.09,98.43) and (215.78,123.54) .. (210.5,142.25) .. controls (205.22,160.96) and (142.26,175.19) .. (102,169.25) .. controls (61.74,163.31) and (43.69,145.96) .. (47.22,119.35) .. controls (50.75,92.75) and (73.75,58.75) .. (104,44.75) -- cycle ;
\draw  [color={rgb, 255:red, 0; green, 0; blue, 0 }  ,draw opacity=1 ] (45,138.75) .. controls (45,79.79) and (95.14,32) .. (157,32) .. controls (218.86,32) and (269,79.79) .. (269,138.75) .. controls (269,197.71) and (218.86,245.5) .. (157,245.5) .. controls (95.14,245.5) and (45,197.71) .. (45,138.75) -- cycle ;
\draw  [draw opacity=0][dash pattern={on 4.5pt off 4.5pt}] (269,138.75) .. controls (269,161.15) and (218.86,179.32) .. (157,179.32) .. controls (95.14,179.32) and (45,161.15) .. (45,138.75) -- (157,138.75) -- cycle ; \draw  [color={rgb, 255:red, 0; green, 0; blue, 0 }  ,draw opacity=1 ][dash pattern={on 4.5pt off 4.5pt}] (269,138.75) .. controls (269,161.15) and (218.86,179.32) .. (157,179.32) .. controls (95.14,179.32) and (45,161.15) .. (45,138.75) ;  
\draw  [draw opacity=0][dash pattern={on 4.5pt off 4.5pt}] (45,138.75) .. controls (45,138.75) and (45,138.75) .. (45,138.75) .. controls (45,116.35) and (95.14,98.18) .. (157,98.18) .. controls (218.86,98.18) and (269,116.35) .. (269,138.75) -- (157,138.75) -- cycle ; \draw  [color={rgb, 255:red, 0; green, 0; blue, 0 }  ,draw opacity=0.5 ][dash pattern={on 4.5pt off 4.5pt}] (45,138.75) .. controls (45,138.75) and (45,138.75) .. (45,138.75) .. controls (45,116.35) and (95.14,98.18) .. (157,98.18) .. controls (218.86,98.18) and (269,116.35) .. (269,138.75) ;  
\draw  [color={rgb, 255:red, 155; green, 155; blue, 155 }  ,draw opacity=1 ][fill={rgb, 255:red, 74; green, 144; blue, 226 }  ,fill opacity=0.5 ] (185,93) .. controls (240.45,88.42) and (205.3,117.14) .. (232,134.5) .. controls (258.7,151.86) and (245,171.56) .. (243.5,189.5) .. controls (242,207.44) and (221.29,184.35) .. (189.5,204) .. controls (157.71,223.65) and (141.37,263.95) .. (115.5,224) .. controls (89.63,184.05) and (67.65,166.52) .. (80.5,139) .. controls (93.35,111.48) and (129.55,97.58) .. (185,93) -- cycle ;

\draw (197.92,217.7) node  [color={rgb, 255:red, 74; green, 144; blue, 226 }  ,opacity=1 ]  {$\mathcal{U}_{\mathcal{I}}$};
\draw (79.92,39.7) node  [color={rgb, 255:red, 208; green, 2; blue, 27 }  ,opacity=1 ]  {$\mathcal{U}_{\mathcal{J}}$};
\draw (181.49,187.69) node    {$\mathbf{A}_{\mathcal{\textcolor[rgb]{0.29,0.56,0.89}{I}}}$};
\draw (124.06,85.19) node    {$\mathbf{A}_{\mathcal{\textcolor[rgb]{0.82,0.01,0.11}{J}}}$};
\draw (149.04,136.19) node  [font=\normalsize]  {$\mathbf{A}_{\mathcal{\textcolor[rgb]{0.29,0.56,0.89}{I}}} -\mathbf{A}_{\mathcal{\textcolor[rgb]{0.82,0.01,0.11}{J}}} =\mathrm{d}\mathbf{\Lambda }_{\mathcal{\textcolor[rgb]{0.29,0.56,0.89}{I}\textcolor[rgb]{0.82,0.01,0.11}{J}}}$};
\draw (167.42,108.2) node  [color={rgb, 255:red, 74; green, 144; blue, 226 }  ,opacity=1 ]  {$\mathcal{\textcolor[rgb]{0.56,0.07,1}{U}}_{\mathcal{I\textcolor[rgb]{0.82,0.01,0.11}{J}}}$};

\end{tikzpicture}

        \caption{Gauge field as a connection on a non-trivial principal bundle. $\mathbf A$ can only be expressed patch by patch, and on overlapping patches the connections are related by a gauge transformation.}
        \label{fig:magnetic}
    \end{figure}
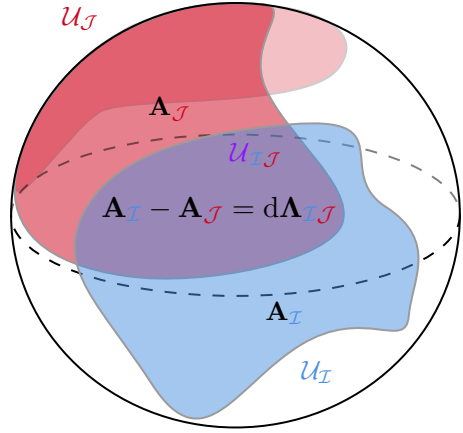

    Let us reconsider the gauge-dependent part of the variation of the improved Noether charge, see  \cref{gdimpnoether},
    \begin{equation}
        \delta_\varphi \tilde{\mathbf Q}^{\text{GD}}_\xi \equiv \delta_\varphi \mathbf{Q}^\text{GD}_\xi - \xi \cdot \mathbf \Theta^{\text{GD}}  = \delta \mathbf{\Upsilon} \wedge (\xi \cdot \mathbf{A}) + (-1)^{D-p} (\xi \cdot \bm \Upsilon) \wedge \delta \mathbf{A}\,.
    \end{equation}
    Previously, in the electrically charged case, we   discarded the term $(\xi \cdot \bm \Upsilon) \wedge \delta \mathbf A$, as it vanishes when pulled back to the horizon  by the boost-weight accounting \cref{thm:boost-weight}. Although this term vanishes on the horizon, in the presence of a non-trivial bundle its patchwise evaluation yields overlap contributions that encode the magnetic charge. To make these contributions explicit, we decompose the horizon into overlapping patches $\mathcal U_\mathcal{I} \subset \tilde{\mathcal C}$ and express the horizon integral
     \begin{equation}
        \int_{\tilde{\mathcal C}} (\xi \cdot \bm \Upsilon) \wedge \delta \mathbf A
    \end{equation}
    in terms of local potentials $\mathbf A_\mathcal{I}$. 
    
   Before doing so, we   first introduce the \emph{magnetic potential form}. By Poincar\'e's lemma and the equation of motion  $\dd{\bm \Upsilon} = 0$, one can locally write $\bm \Upsilon = \dd{\check{\mathbf A}_\mathcal{I}}$  on each patch $\mathcal U_\mathcal{I}$. The $\check{\mathbf A}_\mathcal{I}$'s are the local ``dual'' $(D-p-2)$-form gauge fields that transform as $\check{\mathbf A}_\mathcal{I} \to \check{\mathbf A}_\mathcal{I} + \dd{\check{\bm \Lambda}_\mathcal{I}}$, and   give rise to   magnetostatic potentials as follows. On each horizon patch, we have
    \begin{equation}
        \xi \cdot \bm \Upsilon \overset{\mathcal U_\mathcal{I}}{=} \xi \cdot \dd{\check{\mathbf A}_\mathcal{I}} = \mathcal L_\xi \check{\mathbf A}_\mathcal{I} - \dd{(\xi \cdot \check{\mathbf A}_\mathcal{I})}.
    \end{equation}
    As the  background $\bm \Upsilon$ is invariant under the Killing flow, i.e., $\mathcal L_\xi \bm \Upsilon = 0$, the local 
    gauge field is preserved by the Killing flow up to a   pure gauge transformation, i.e. $\mathcal L_\xi \check{\mathbf A}_\mathcal{I} = \dd{\check{\bm \lambda}_{\xi,\mathcal{I}}}$. Hence, 
    \begin{equation}
        \xi \cdot \bm \Upsilon \overset{\mathcal U_\mathcal{I}}{=} \dd{(\check{\bm \lambda}_{\xi,\mathcal{I}} - \xi \cdot \check{\mathbf A}_\mathcal{I})} \equiv (-1)^{n_\Psi}\dd{\bm \Psi_\mathcal{I}}\,,
    \end{equation}
    where we   defined the magnetic potential form $\bm \Psi_\mathcal{I}$ of degree $D-p-3$ on each patch $\mathcal U_\mathcal{I}$ as 
    \begin{equation}
        (-1)^{n_\Psi}\bm \Psi_\mathcal{I} \equiv - \xi \cdot \check{\mathbf A}_\mathcal{I} + \check{\bm \lambda}_{\xi,\mathcal{I}}\,,
    \end{equation}
    and    introduced a minus sign with power $n_\Psi = (D-p-3)(p+1)$ for later convenience. 
    
    By \cref{thm:boost-weight}, we have $\xi \cdot \bm \Upsilon \fheq 0$, so   the magnetic potential form is closed on the horizon, $\dd{\bm \Psi_\mathcal{I}} \fheq 0$; this is  the magnetic analog of \cref{thm:zeroth-law}. Moreover,   the magnetic potential form is invariant, up to an exact term, under gauge transformations $\check{\mathbf A}_\mathcal{I} \to \check{\mathbf A}_\mathcal{I} + \dd{\check{\bm \Lambda}_\mathcal{I}}$ by the same reasoning as for the electric potential form, see   \cref{gaugetranselectricpot}. On the overlaps $\mathcal{U}_{\mathcal{IJ}} $ between   coordinate patches, the local   potentials are related by a pure gauge transformation,
 $\check{\mathbf A}_\mathcal{I} - \check{\mathbf A}_\mathcal{J} = \dd{\check{\bm \Lambda}_{\mathcal{IJ}}}$. Consequently,  we may write $\xi \cdot \bm \Upsilon \fheq (-1)^{n_\Psi}\dd{\bm \Psi}$ on the entire horizon, because the right-hand side is gauge invariant. Below we choose a  harmonic representative of the magnetic potential, still denoted by~$\bm \Psi$, since that is   free of exact terms, and hence   gauge invariant and patch independent. Then, on each patch $\mathcal U_{\mathcal I}$, the magnetic potential takes the form of $\bm \Psi_{\mathcal I} = \bm \Psi + \dd{\bm \beta_\mathcal{I}}$. 

    Now we decompose the horizon integral on the overlapping patches 
        \begin{align}
            \int_{\tilde{\mathcal C}} (\xi \cdot \bm \Upsilon) \wedge \delta \mathbf A & = \sum_{I} \int_{\mathcal U_\mathcal{I}} (-1)^{n_\Psi}\dd{\bm \Psi} \wedge \delta \mathbf A_\mathcal{I} \nonumber\\
            & = \sum_{I} \left[\int_{\partial \mathcal U_\mathcal{I}} (-1)^{n_\Psi}(\bm \Psi + \dd{\bm \beta_\mathcal{I}}) \wedge \delta \mathbf A_\mathcal{I} - (-1)^{(D-p-3)+n_\Psi} \int_{\mathcal U_{\mathcal I}} (\bm \Psi + \dd{\bm \beta_\mathcal{I}}) \wedge \delta \mathbf F\right] \nonumber\\
             & = \sum_{I} \int_{\partial \mathcal U_\mathcal{I}} (-1)^{n_\Psi}\bm \Psi \wedge \delta \mathbf A_\mathcal{I} - (-1)^{(D-p-3)+n_\Psi} \int_{\tilde{\mathcal C}} \bm \Psi \wedge \delta \mathbf F\,, \label{eq:magnetic-charge-term}
        \end{align}
    where we   integrated by parts and used the fact that the field strength $\mathbf F$ is patch independent because of gauge invariance. We noticed the terms involving the patch-dependent $\dd{\bm \beta_\mathcal{I}}$ cancel out because
    \begin{equation}
    \begin{aligned}
        \int_{\partial \mathcal U_{\mathcal I}} \!\!\!\dd{\bm \beta_{\mathcal I}} \wedge \delta \mathbf A_{\mathcal I} & = (-1)^{D-p-3} \int_{\partial \mathcal U_{\mathcal I}} \!\!\!\bm \beta_{\mathcal I} \wedge \delta \mathbf F\\
        & = (-1)^{(D-p-3)} \int_{\mathcal U_{\mathcal I}} \!\!\! \dd{(\bm \beta_{\mathcal I} \wedge \delta \mathbf F)}\\
        & = (-1)^{(D-p-3)} \int_{\mathcal U_{\mathcal I}} \!\!\! \dd{\bm \beta_{\mathcal I}} \wedge \delta \mathbf F\,.
    \end{aligned}
    \end{equation} 
  There are two types of terms in the final expression of \cref{eq:magnetic-charge-term}: the seam terms localised on the boundary $\partial \mathcal U_\mathcal{I}$ of patches  and a horizon integral that would appear as a magnetic charge contribution to the first law. However, we still have $\xi \cdot \bm \Upsilon \fheq 0$, so the seam integrals exactly cancel the horizon integral,  and   the magnetic charge does not contribute to the first law. 

   From a covariant perspective, the appearance of patch-dependent terms is unsatisfactory, as it indicates that the intermediate quantities are not globally well defined. Moreover, the patches have boundaries, which lead to ambiguities in the improved Noether charge. As we will explain in  \cref{sec:jkm}, the improved Noether charge $\delta_\varphi \tilde{\mathbf Q}_\xi \equiv \delta_\varphi \mathbf{Q}_\xi - \xi \cdot \mathbf{\Theta}(\varphi, \delta \varphi)$ is only defined up to an exact term.  It transforms as 
    \begin{equation}
         \delta_\varphi \tilde{\mathbf Q}_\xi \to \delta_\varphi \tilde{\mathbf Q}_\xi +  \dd \left(  \xi \cdot \mathbf{Y}(\varphi, \delta \varphi) + \delta_\varphi \mathbf Z(\varphi, \xi)\right) 
    \end{equation}
    under the Jacobson-Kang-Myers (JKM) ambiguities of the presymplectic potential and the Noether charge
    \begin{equation}
        \bm \Theta(\varphi, \delta \varphi) \to \bm \Theta(\varphi, \delta \varphi) + \dd{\mathbf{Y}(\varphi, \delta \varphi)}, \quad \mathbf Q_\xi \to \mathbf Q_\xi + \dd{\mathbf Z(\varphi, \xi)}
    \end{equation}
    where $\mathbf Y(\varphi, \delta \varphi)$ is   linear in $\delta \varphi$ and $\mathbf Z(\varphi, \xi)$ is linear in $\xi$. When integrated on compact horizon cross-sections, these
    ambiguities vanish due to Stokes' theorem, if the fields are smooth on the entire cross-section. However, when the gauge field is only defined patchwise, patch-dependent JKM ambiguities could persist at the seams of the patches. In particular, such ambiguity may be linear in $\delta \mathbf A$, and it generally takes the form of $\xi \cdot \mathbf{Y}(\phi, \delta \mathbf A) + \delta_\varphi \mathbf Z(\varphi, \xi) = \mathbf W(\phi,\xi) \wedge \delta \mathbf A$, which, upon integration, produces new seam terms  
    \begin{equation}
        \sum_\mathcal{I} \int_{\partial \mathcal U_\mathcal{I}} \mathbf W(\phi,\xi) \wedge \delta \mathbf A_\mathcal{I}\,.
    \end{equation}
To solve the problems of patch dependence and JKM ambiguities at seams altogether, we notice that there is a unique and physical prescription: we fix the ambiguity and redefine the improved Noether charge $\delta_\varphi \tilde{\mathbf Q}^{\text{GD}}_\xi \to \delta_\varphi \tilde{\mathbf Q}^{\text{BC}}_\xi$ by choosing 
    \begin{equation}
        \mathbf W(\phi,\xi) = -(-1)^{D-p+n_\Psi} \bm \Psi\,,
    \end{equation}
    which cancels the previous patch-dependent term $\bm \Psi \wedge \delta \mathbf A_{\mathcal I}$. We call this the \emph{bundle-covariant} prescription. A similar prescription was also chosen by \cite{Ortin:2022uxa} to reveal the magnetic charge contribution to the first law. 
    
    Together with the contribution from the on-shell gauge-dependent part of the presymplectic current $\bm \omega^{\text{GD}}(\phi, \delta \phi, \mathcal{L}_\xi \varphi) = \dd{(\delta \bm \Upsilon \wedge \bm \lambda_\xi)}$, we have a bundle-covariant quantity
    \begin{equation}
        \delta_\varphi \tilde{\mathbf Q}^{\text{BC}}_\xi - \delta \bm \Upsilon \wedge \bm \lambda_\xi \fheq - \delta \bm \Upsilon \wedge \bm \Phi + (-1)^{n_\Psi}\bm \Psi \wedge \delta \mathbf F = - \delta \bm \Upsilon \wedge \bm \Phi + \delta \mathbf F \wedge \bm \Psi\,.
    \end{equation}
   The second term on the right side gives rise to a magnetic charge  term in the first law,  as follows. By the closedness of $\bm \Psi$ (on the horizon) and $\delta \mathbf F$ we can use the same Poincar\'{e} duality trick as for the electric charge case, i.e., 
    \begin{equation}
        \int_{\tilde{\mathcal C}} \delta \mathbf F \wedge \bm \Psi = \sum_{\nu} \Psi^{\mathcal H^+}_\nu \int_{\nu} \delta \mathbf F = \sum_{\nu} \Psi^{\mathcal H^+}_\nu \delta P_\nu\,. 
    \end{equation}
    Here, we have resolved the harmonic magnetic potential $\bm \Psi$ into an orthonormal   basis~$\{\hat{\bm \zeta}_\nu\}$
    \begin{equation}
        \bm \Psi \fheq \sum_\nu \Psi^{\mathcal H^+}_\nu \hat{\bm \zeta}_\nu\,,
    \end{equation}
    where $\nu$ are the non-trivial $(p+1)$-cycles dual to the $(D-p-3)$-harmonic basis form $\hat{\bm \zeta}_\nu$. Moreover,  $\Psi_\nu^{\mathcal H^+}$ are the associated constant magnetic potentials on the horizon, 
    \begin{equation}
        \Psi_\nu^{\mathcal H^+} = \int_{\tilde \nu} \mathbf \Psi |_{\tilde{\mathcal C}} 
    \end{equation}
   where $\tilde \nu$ are non-trivial $(D-p-3)$-cycles of $\tilde{\mathcal C}$ dual to $\star_{\tilde{\mathcal C}} \hat{\bm \zeta}_\nu$,  and 
    \begin{equation}
        P_\nu = \int_\nu \mathbf F
    \end{equation}
    are the magnetic charges on $\nu \subset \tilde{\mathcal C}$. 
Similarly, at spatial infinity, assuming $\mathcal S_\infty$ has non-trivial cycles that are homologous to the horizon cycles $\nu$, there is a magnetic potential-charge term 
\begin{equation}
\label{magneticatinfinity}
    \int_{{\mathcal S}_\infty} \delta {\mathbf F} \wedge \bm \Psi = \sum_\nu \Psi_\nu^\infty \delta P_\nu\,.
\end{equation}
Note that this term would   vanish, if we   impose the boundary condition $\delta \mathbf A \to \mathcal O(1/r^{D-3})$ at spatial infinity, as we did in the electric case (see below \eqref{gdimpnoether}). Hence, in order to obtain a non-zero magnetic charge contribution from spatial infinity, we do not impose this boundary condition here.   Rather, we require that the field strength variation  does not blow up at infinity, $\delta \mathbf F \to \mathcal O(1)$, such that the magnetic charge variation at infinity remains finite or vanishes. The magnetic charge variation vanishes if the (higher) principal bundle is trivial, since in that case the field strength is globally exact, hence the integral of $\delta \mathbf F$ vanishes on the compact surface $\mathcal S_\infty$.

  The number of different   magnetic charges equals the number of  non-trivial $(p+1)$-cycles of the horizon cross-section, which is given by the rank of the homology group $H_{p+1} (\tilde{\mathcal{C}})$, i.e. by the $(p+1)$-th Betti number $b_{p+1}$ (that equals   $b_{D-p-3}$ by Poincar\'e duality).  Therefore, the number of independent magnetic potential-charge pairs is  
 $b_{p+1}$,  whereas the number of independent electric charges was $b_{p-1}.$

    In the end, the first law becomes
    \begin{equation}
        T \delta S_\text{dyn} = \delta M - \sum_I \Omega_I \delta J_I - \sum_\sigma \bar \Phi_\sigma \delta Q_\sigma - \sum_\nu \bar{\Psi}_\nu  \delta P_\nu
    \end{equation}
    where, $\bar \Phi_\sigma = \Phi_\sigma^{\mathcal H^+} - \Phi_\sigma^{\infty}$ is the electric potential difference and $\bar \Psi_\nu  = \Psi_\nu^{\mathcal H^+} - \Psi_\nu^{\infty}$ is the magnetic potential difference.

    \section{Physical process version of the first law with $p$-form charge}\label{sec:ppfl}

    In contrast to the comparison version, the physical process version of the thermodynamic first law of a black object follows from integrating the fundamental variational identity \eqref{fundidentityorigin} on an interval of the  event horizon. This leads to a relationship  between the state of the black object on the initial and final horizon slices and any matter sources that fall into the black object through the interval. In standard derivations   \cite{Wald:1995yp,Gao:2001ut} of the physical process first law it is assumed that the initial and final state of the black hole are stationary, in which case the black hole entropy is given by the Wald entropy. On the other hand,   for a non-statinary initial and final state at arbitrary affine null times it was shown in \cite{Hollands:2024vbe} (see also \cite{Rignon-Bret:2023fjq,Visser:2024pwz}) that the   black hole entropy receives a dynamical correction term, yielding the same dynamical entropy that appears in the comparison version of the non-stationary first law. In this section we  generalise  this derivation to $p$-form gauge fields that are non-minimally coupled to the metric.  
    
    We consider a  Lagrangian $\mathbf L_\text{source}(g^{ab}, A_{a_1 \cdots a_p}, \{\psi\}, \left\{ \Psi \right\},\nabla_a)$ that sources the metric field via an external stress-energy tensor $T_{ab}$ and   the $p$-form gauge field via an external electric current $\mathbf{j}$. We assume that the external source matter   fields   $\{ \Psi \}$ are minimally coupled to the metric, and that  $\mathbf L_\text{source}$ depends linearly on $\mathbf A$. For the main part of our discussion, we assume the other gauge-independent matter fields $\{\psi\}$ are not externally sourced, however, at the end of this section we will briefly comment on the case of non-zero sources for $\{\psi\}$. For this Lagrangian the    gauge-invariant fields are $\phi\equiv (g, \{\psi\}, \{\Psi\}, \mathbf F)$ and the   collection of gauge-invariant and gauge-dependent fields is $\varphi \equiv (\phi, \mathbf{A})$. 
    
    Since the Lagrangian is now divided into two sectors $L_\text{grav-EM}$ and $L_\text{source}$, all the covariant phase space quantities   split   into two terms accordingly, e.g., 
    \begin{equation}
        \bm \Theta = \bm \Theta_{\text{grav-EM}} + \bm \Theta_{\text{source}}, \qquad \bm \omega = \bm \omega_{\text{grav-EM}} + \bm \omega_{\text{source}},  \qquad \mathbf{Q} = \mathbf{Q}_{\text{grav-EM}} + \mathbf{Q}_{\text{source}}\,,
    \end{equation}
    Hence, there are two separate fundamental variational identities, one for each sector. 
    We emphasise that the observables associated to the black object are solely defined in terms of the covariant phase space quantities in the gravity-electromagnetism sector. We will show below that the  physical process version of the first law follows from the fundamental identity in this sector:
    \begin{align}
            & \quad~\bm \omega_{\text{grav-EM}}^{\text{GI}}(\phi, \delta \varphi, \mathcal{L}_\xi \varphi) +  \bm \omega_{\text{grav-EM}}^{\text{GD}}(\phi, \delta \phi, \mathcal{L}_\xi \varphi) + \delta_\varphi \mathbf{C}_{\text{grav-EM},\xi} - \xi \cdot \left( \mathbf{E}(\phi) \cdot \delta \tilde \phi + \bm{\mathcal{E}} (\phi) \wedge \delta \mathbf{A}\right) \nonumber\\
            & = \dd{\left(\delta_\phi \mathbf{Q}_{\text{grav-EM},\xi}^{\text{GI}} - \xi \cdot \mathbf{\Theta}_{\text{grav-EM}}^{\text{GI}}(\phi, \delta \phi) + \delta_\varphi \mathbf{Q}_{\text{grav-EM},\xi}^{\text{GD}} - \xi \cdot \mathbf{\Theta}_{\text{grav-EM}}^{\text{GD}}(\phi, \delta \mathbf{A})\right)}.
       \label{eq:fund-id-xi-grav-EM}
    \end{align}
    First, we will  evaluate the constraint form on the horizon, then we will see how external matter sources contribute to the first law.

    \subsection{Constraint form on the future horizon}
    In the presence of  non-zero sources, the variation of the gravitational-electromagnetic part of the constraint form,   $\delta_\varphi \mathbf{C}_{\text{grav-EM},\xi}$,  does not vanish. To uncover the physical interpretation of this term, we first rewrite the terms $\mathbf{C}^{(A)}_\xi$ contained in  the unvaried constraint form $ \mathbf{C}_{\text{grav-EM},\xi}$ that are proportional to $\mathbf{A}$, using differential form techniques.

    Let us start with computing
    \begin{align}
            \mathbf{C}^{(A)}_\xi & = C_{(A)}^{ab} \xi_b \bm \epsilon_a \nonumber\\
            & = (E_{(A)}^{a a_2 \cdots a_p} A_{b a_2 \cdots a_p} + E_{(A)}^{a_1 a a_3 \cdots a_p} A_{a_1 b a_3 \cdots a_p} + \cdots + E_{(A)}^{a_1 \cdots a_{p-1} a} A_{a_1 \cdots a_{p-1} b})\xi^b \bm \epsilon_a \nonumber \\
            & = p E_{(A)}^{a a_2 \cdots a_p} A_{b a_2 \cdots a_p} \xi^b \bm \epsilon_a \,,
    \end{align}
    where we recall that $E_{(A)}^{a_1 \cdots a_p}$ is the equation of motion of $A_{a_1 \cdots a_p}$, and we have used the antisymmetry of the indices to simplify the expression. When pulled back to the future horizon, we have (see eq.~(4.47) of our previous paper \cite{Visser:2024pwz})
    \begin{equation}
        \mathbf{C}^{(A)}_\xi \fheq - p E_{(A)}^{a a_2 \cdots a_p} A_{b a_2 \cdots a_p} k_a \xi^b \bm \epsilon_{\mathcal{H^+}} \fheq - p (k \cdot E_{(A)})^{a_2 \cdots a_p} (\xi \cdot A)_{a_2 \cdots a_p} \bm \epsilon_{\mathcal{H^+}} \,,
    \end{equation}
    where $k = \partial_v$ is the generator of affine null translations along $\mathcal{H}^+$, and $\bm \epsilon_\mathcal{H^+}$ is the intrinsic volume form of the horizon $\mathcal{H^+}$. The pullback expression follows from the orientation
    \begin{equation}
        \bm \epsilon = - \bm k \wedge \bm \epsilon_\mathcal{H^+}\,.
    \end{equation}
    The intrinsic volume form of the horizon $\mathcal{H^+}$ can be further decomposed  as 
    \begin{equation}
        \bm \epsilon_\mathcal{H^+} = - \bm l \wedge \bm \epsilon_\mathcal{C}\,,
    \end{equation}
    where $\mathcal{C}$ are   codimension-2 cross-sections of the future horizon with null normal $- \bm l$ within the horizon. 
    
    For the contracted scalar in front of $\bm \epsilon_\mathcal{H^+}$ in the pullback of $\mathbf{C}^{(A)}_\xi$, we carry out a double null decomposition for all of its contracted indices:
        \begin{align}
            (k \cdot E_{(A)})^{a_1 \cdots a \cdots a_{p-1}} (\xi \cdot A)_{a_1 \cdots a \cdots a_{p-1}} & \fheq (k \cdot E_{(A)})^{a_1 \cdots b \cdots a_{p-1}} (-k^a l_b - l^a k_b + \gamma^a_b)(\xi \cdot A)_{a_1 \cdots a \cdots a_{p-1}} \nonumber\\
            & \fheq (k \cdot E_{(A)})^{a_1 \cdots b \cdots a_{p-1}} \gamma^a_b(\xi \cdot A)_{a_1 \cdots a \cdots a_{p-1}} \nonumber\\
            & \fheq (k \cdot E_{(A)})^{b_1  \cdots b_{p-1}} \gamma^{a_1}_{b_1} \cdots \gamma^{a_{p-1}}_{b_{p-1}} (\xi \cdot A)_{a_1 \cdots a_{p-1}} \nonumber\\
            & \fheq (k \cdot E_{(A)})^{i_1  \cdots i_{p-1}} (\xi \cdot A)_{i_1 \cdots i_{p-1}}\,.
        \end{align}
    Here, we used that the null projections are proportional to double (or multiple) contractions, like $k \cdot k \cdot \mathbf{E}_{(A)}$ or $k \cdot \xi \cdot \mathbf{A}$, that  vanish. In the last line, we   converted to intrinsic indices $i_1, i_2, \cdots$ on~$\mathcal{C}$, because the procedure is equivalent to pulling back $k \cdot \mathbf{E}_{(A)}$ and $\xi \cdot \mathbf{A}$ to $\mathcal{C}$ and then contracting them.

   Next,  we employ \cref{cor:subcontraction} in appendix \ref{appaaaa} to rewrite the constraint form that is proportional to  $\mathbf{A}$ as 
    \begin{equation}
        \mathbf C^A_\xi  \fheq p (k \cdot E_{(A)})^{i_1 \cdots i_{p-1}} (\xi \cdot A)_{i_1 \cdots i_{p-1}} \bm l \wedge \bm \epsilon_\mathcal{C} \fheq p!\, \bm l \wedge [\star_\mathcal{C}(k \cdot \mathbf{E}_{(A)})] \wedge (\xi \cdot \mathbf{A})\,, \label{eq:A-constraint-half-way}
    \end{equation}
    where $\star_\mathcal{C}$ is the intrinsic Hodge dual on $\mathcal{C}$.
    To further simplify this expression, we use the results in \cref{lem:dnd-hodge} to carry out a double null decomposition of $\star \mathbf{E}_{(A)}$: 
    \begin{align}
            \star \mathbf{E}_{(A)} & = \bm k \wedge \biggl( \bm l \wedge (\star_\mathcal{C} \mathbf{E}_{(A)}) + (-1)^{D-p-1} {\star_\mathcal{C} (l \cdot \mathbf{E}_{(A)})} \biggr) + (-1)^{D-p} \bm l \wedge [\star_\mathcal{C}(k \cdot \mathbf{E}_{(A)})] + \star_\mathcal{C} (l \cdot k \cdot \mathbf{E}_{(A)}) \nonumber \\
            & \fheq (-1)^{D-p} \bm l \wedge [\star_\mathcal{C}(k \cdot \mathbf{E}_{(A)})] + \star_\mathcal{C} (l \cdot k \cdot \mathbf{E}_{(A)})\,,
    \end{align}
    where in the last line we   pulled back to the future horizon, and the terms involving $\bm k \wedge \cdots$ are projected out. By plugging this back  into \cref{eq:A-constraint-half-way} we obtain
    \begin{equation}
        \begin{split}
            \mathbf{C}_\xi^A & \fheq (-1)^{D-p} p! \,  \bigl( \star \mathbf{E}_{(A)} \wedge (\xi \cdot \mathbf{A}) -  \star_\mathcal{C}(l \cdot k \cdot \mathbf{E}_{(A)}) \wedge (\xi \cdot \mathbf{A}) \bigr)\\
            & \fheq (-1)^{D-p} p! \, {\star \mathbf{E}_{(A)}} \wedge (\xi \cdot \mathbf{A})\,. \label{constraintAfinal}
        \end{split}
    \end{equation}
  Here, the second line follows from the fact  that $\star_\mathcal{C}(l \cdot k \cdot \mathbf{E}_{(A)})$ and $\xi \cdot \mathbf{A}$ both have spatial polarisations only, i.e., there is at least one doubly wedged direction. Hence, $\star_\mathcal{C}(l \cdot k \cdot \mathbf{E}_{(A)}) \wedge (\xi \cdot \mathbf{A}) \fheq 0$. 

    \subsection{Perturbations by matter sources}
Next, we turn on the perturbations by external matter fields $\{\Psi\}$ that source  both the metric and the gauge field.     In the  background geometry the equations of motion  are $E_{(g)ab} = 0$ and $E_{(A)}^{a_1 \cdots a_p} = 0$. The perturbed equations of motions now receive   matter source terms
    \begin{equation}
        \delta E_{(g)ab} = \delta T_{ab} \qquad \text{and} \qquad \delta E_{(A)}^{a_1 \cdots a_p} =  - \frac{(-1)^{D-p}}{p!} \delta j^{a_1 \cdots a_p}\,,
    \end{equation}
  where the factor in front of $\delta j$ is chosen as a convention such that the perturbed equation of motion of the gauge field takes the form
    \begin{equation}
        \dd{(\delta \bm \Upsilon)} = - {\star} \delta \mathbf{j}\,,
    \end{equation}
    which reduces to $\dd{{\star} \mathbf F} = \star \mathbf j$ for the Maxwell   Lagrangian  $ -(1/4) \bm \epsilon F_{ab}F^{ab}= - (1/2) \star \mathbf{F}  \wedge  \mathbf F$. Notice that other matter fields $\{\psi\}$ are not sourced in this setup, hence their perturbations satisfy $\delta E_{(\psi)} = 0$. 
  Thus, the perturbation of the constraint form, when pulled back to the horizon, reads 
    \begin{equation}
        \begin{split}
            \delta_\varphi \mathbf{C}_{\text{grav-EM},\xi} & = (- \delta E^{ab}_{(g)} + \delta C_{(A)}^{ab}) \xi_b \bm \epsilon_a\\
            & \fheq \delta T_{ab} k^a \xi^b \bm \epsilon_\mathcal{H^+} + (-1)^{D-p} p! \, {\star \delta \mathbf{E}_{(A)}} \wedge (\xi \cdot \mathbf{A}) \\
            & \fheq \delta T_{ab} k^a \xi^b \bm \epsilon_\mathcal{H^+} -  {\star \delta \mathbf{j}} \wedge (\xi \cdot \mathbf{A})\,,
        \end{split}
    \end{equation}
    where we   used   $\delta C^{ab}_{(\psi)} \propto \delta E_{(\psi)} = 0$ in the first line, and we inserted \eqref{constraintAfinal} in the second line.  

Now, coming back to the fundamental variational identity \eqref{eq:fund-id-xi-grav-EM}, we compute the other two terms on the left side of the identity, namely the   presymplectic currents. First, the gauge-dependent part of the presymplectic current (\cref{eq:omega-xi-GD}) becomes
    \begin{equation}
        \begin{split}
            \bm \omega^\text{GD}_{\text{grav-EM}}(\phi, \delta \phi, \mathcal{L}_\xi \varphi) & = \dd{\left( \delta \bm \Upsilon \wedge \bm \lambda_\xi \right)} - (-1)^{D-p} p!\, {\star}\delta \mathbf{E}_{(A)} \wedge \bm \lambda_\xi\\
            & = \dd{\left( \delta \bm \Upsilon \wedge \bm \lambda_\xi \right)} + {\star}\delta \mathbf{j} \wedge \bm \lambda_\xi\,,
        \end{split}
    \end{equation}
    where we once again used $\bm{\mathcal{E}} = p!\, {\star} \mathbf{E}_{(A)}$. The gauge-independent part of the presymplectic current again vanishes by the background Killing symmetry, as we saw in \eqref{vanishingcurrent}.

    Finally, pulled back to the horizon, the fundamental identity \eqref{eq:fund-id-xi-grav-EM} becomes 
    \begin{equation} 
        \dd{\left(\delta_\phi \mathbf{Q}_{\text{grav-EM},\xi}^{\text{GI}} - \xi \cdot \mathbf{\Theta}_{\text{grav-EM}}^{\text{GI}}(\phi, \delta \phi) - \delta \bm \Upsilon \wedge \bm \Phi\right)}\fheq \delta T_{ab} k^a \xi^b \bm \epsilon_\mathcal{H^+} +  {\star \delta \mathbf{j}} \wedge \bm \Phi\,,\label{eq:fund-id-sourced}
    \end{equation}
    where we used the definition $\bm \Phi = - \xi \cdot \mathbf{A} + \bm \lambda_\xi$, we employed    \cref{eq:improv-Q-GD}, and we used the fact that the pullback of $\xi \cdot (\cdots)$ vanishes as $\xi$ is tangent to the horizon.

    \subsection{Physical process first law}
    
    We now integrate the fundamental identity \eqref{eq:fund-id-sourced} on a horizon interval $\mathcal{H}_{v_1}^{v_2}$ between two arbitrary cross-sections $\mathcal{C}(v_1)$ and $\mathcal{C}(v_2)$ at affine null parameters $v_1, v_2$ to obtain the physical process version of the first law for a black object. For notational convenience we assume   the horizon cross-sections are compact, i.e. $\mathcal{C} = \tilde{\mathcal{C}}$, but the results below can be straightforwardly generalised to non-compact horizons by the same methods as in \cref{sec:cfl}. It is important in that (non-compact) case to first downgrade the differential forms $\delta \bm \Upsilon$ and $\bm \Phi$ to lower-dimensional forms on the compact horizon cross-section $\tilde{\mathcal C}$, and then perform the cohomology/homology duality trick.   
    
   We want to study the structure of the integral of the electric charge term $-\dd{(\delta \bm \Upsilon \wedge \bm \Phi)}$ and the electric current term ${\star}\delta \mathbf{j} \wedge \bm \Phi$ on the horizon segment $\mathcal H_{v_1}^{v_2}$. We want to use similar arguments as in \cref{sec:cfl} to convert these integrals into those on the non-trivial cycles dual to $\bm \Phi$ in $\mathcal{H}_{v_1}^{v_2}$. However, there are some slight twists in the mathematics:  
    \begin{enumerate}[1)]
        \item The topology of the horizon interval is $\tilde{\mathcal C}\times [0,1]$, so it is technically a compact manifold with boundaries. The algebraic topology statements we use should be generalised to accommodate the boundaries $\tilde{\mathcal C}(v_1)$ and $\tilde{\mathcal C}(v_2)$.
        \item The integral of the exact term $-\dd{(\delta \bm \Upsilon \wedge \bm \Phi)}$  localises on the closed boundary $\tilde{\mathcal C}(v_1) \cup \tilde{\mathcal C}(v_2)$. However, the Poincar\'e duality argument used in \cref{sssec:compact} needs to be modified because $\delta \bm \Upsilon$ is no longer closed in the presence of the source current $\delta \mathbf j$.
        \item Another subtlety arises when evaluating the integral of the electric current term ${\star}\delta \mathbf{j} \wedge \bm \Phi$: The horizon interval $\mathcal{H}_{v_1}^{v_2}$ is a null hypersurface, whose intrinsic Hodge decomposition (technically speaking, the Hodge-Morrey-Friedrichs decomposition as boundaries are present) is ill-defined due to the degenerate metric, but the topological structures such as the homology and cohomology, and the duality between them are still valid.
    \end{enumerate} 
    First, we consider the integral of the electric charge term
    \begin{equation}
        - \int_{\mathcal H_{v_1}^{v_2}} \dd{(\delta \bm \Upsilon \wedge \bm \Phi)} = \left[- \int_{\tilde{\mathcal C}(v)} \delta \bm \Upsilon \wedge \bm \Phi \right]_{v=v_1}^{v=v_2}\,.
    \end{equation}
    It is tempting to use the exact results in \cref{sec:cfl} to evaluate this integral by Hodge decomposing $\bm \Phi$ into a harmonic part and an exact part, and then converting it into integrals of $\delta \bm \Upsilon$ on non-trivial cycles dual to the harmonic part of $\bm \Phi$ using the Poincar\'e duality. However, the previous argument is based on the closedness of $\delta \bm \Upsilon$, which is no longer true. The procedure can be revised as follows. The closedness requirement of $\delta \bm \Upsilon$ can be bypassed as long as we choose an orthonormal \emph{Thom form} basis to resolve the harmonic part of $\bm \Phi$. These Thom forms are special representatives of the harmonic part that are compactly supported in a tubular region around the dual cycles, and such construction is valid for compact and orientable manifolds \cite{Bott:1982xhp,PetersenManifolds} --- our setup satisfies these assumptions. With these Thom forms, we can Hodge decompose the potential form $\bm \Phi$ on a horizon cross-section $\tilde{\mathcal C}(v)$ as 
    \begin{equation}
        \bm \Phi \overset{\tilde{\mathcal{C}}(v)}{=} \bm \eta^\text{Thom} + \dd{\bm \alpha} = \sum_\sigma \Phi_\sigma^{\mathcal H^+} \hat{\bm \eta}_{\sigma}^\text{Thom} + \dd{\bm \alpha}\,,
    \end{equation}
    where the $\hat{\bm \eta}_{\sigma}^\text{Thom}$ are dual to the non-trivial $(D-p-1)$-cycles $\sigma(v)$ in $\tilde{\mathcal C}(v)$, and the coefficients are  defined as $\Phi_\sigma^{\mathcal H^+} = \int_{\tilde \sigma} \bm \Phi   \big |_{\mathcal{C}(v)}$, the same as \cref{electricpotentialfundeq}, which are constants by closedness of $\bm \Phi$.  Poincar\'e duality now implies
    \begin{equation}
        \int_{\mathcal{C}(v)} \mathbf{X} \wedge \hat{\bm \eta}^\text{Thom}_\sigma = \int_{\sigma(v)} \mathbf{X}
    \end{equation}
    for any $(D-p-1)$-form $\mathbf{X}$. Then, we have 
    \begin{equation}
        \begin{split}
            \left[- \int_{\tilde{\mathcal C}(v)} \delta \bm \Upsilon \wedge \bm \Phi \right]_{v=v_1}^{v_2} & = \left[- \sum_\sigma \Phi_\sigma^{\mathcal H^+} \int_{\tilde{\mathcal C}(v)} \delta \bm \Upsilon \wedge \hat{\bm \eta}_\sigma^\text{Thom} - \int_{\tilde{\mathcal C}(v)} \delta \bm \Upsilon \wedge \dd{\bm \alpha} \right]_{v=v_1}^{v=v_2}\\
            & = \left[- \sum_\sigma \Phi_\sigma^{\mathcal H^+} \int_{\sigma(v)} \delta \bm \Upsilon + (-1)^{D-p-1} \int_{\tilde{\mathcal C}(v)} \dd{(\delta \bm \Upsilon)} \wedge \bm \alpha \right]_{v=v_1}^{v=v_2}\\
            & = \left[\sum_\sigma \Phi_\sigma^{\mathcal H^+} \delta Q_\sigma(v) + (-1)^{D-p} \int_{\tilde{\mathcal C}(v)} {\star}\delta \mathbf j \wedge \bm \alpha \right]_{v=v_1}^{v=v_2}\,.
        \end{split}
    \end{equation}
In the second equality, Stokes' theorem and the compactness of $\tilde{\mathcal{C}}(v)$   imply that the integral of $\dd (\delta \bm \Upsilon \wedge \bm \alpha)$ over  $\tilde{\mathcal{C}}(v)$ vanishes. In   the last equality, we   identified the electric charge on a cycle $\sigma$ at affine null time $v$ as $Q_\sigma = - \int_{\sigma(v)} \delta \bm \Upsilon$, and used the equation of motion $\dd{(\delta \bm \Upsilon)} = - {\star} \delta \mathbf{j}$. We noticed that the exact piece $\dd{\bm \alpha}$ would also contribute because $\delta \bm \Upsilon$ is not closed due to the source current.

    Now we turn to the electric current integral. Though the intrinsic Hodge decomposition is ill defined for a null hypersurface, the closed potential form $\bm \Phi$ can still be decomposed on $\mathcal H_{v_1}^{v_2}$ as
    \begin{equation}
        \bm \Phi \hinteq \bm \eta_\text{abs}^\text{Thom} + \dd{\bm \alpha}\,,
    \end{equation}
    by the definition of the de Rham cohomology, where $\bm \eta_\text{abs}^\text{Thom}$ is an absolute Thom form as the cohomology representative. Such representative loses the harmonicity in $\mathcal H_{v_1}^{v_2}$ due to the degenerate metric, yet it does not matter for our purposes. Here, ``absolute'' means that it does not have any $v$-leg, just as for $\bm \Phi$ by \cref{thm:boost-weight}. The Thom form prescription is chosen here for compatibility with the electric charge terms. Further, following the spirit of \cref{harmonicbasis} again, we resolve the cohomology representative part of the potential form in an orthonormal basis $\{\hat{\bm \eta}_{\text{abs},\sigma'}^\text{Thom}\}$ 
    \begin{equation}
        \bm \eta_\text{abs}^\text{Thom} = \sum_{\sigma'} \Phi_{\sigma'}^{\mathcal{H^+}} \hat{\bm \eta}_{\text{abs},\sigma'}^\text{Thom}\,,
    \end{equation}
    where the basis representatives $\hat{\bm \eta}_{\text{abs},\sigma'}^\text{Thom}$ are dual to non-trivial relative $(D-p)$-cycles $\sigma'$ in $\mathcal H_{v_1}^{v_2}$. Here, ``relative'' means that the cycles $\sigma'$ have boundaries purely in $\partial \mathcal H_{v_1}^{v_2} = \tilde{\mathcal C}(v_1) \cup \tilde{\mathcal C}(v_2)$. These relative cycles $\sigma'$ have topology $\sigma \times [0,1]$, where the $\sigma$ slices at null time $v \in [v_1, v_2]$ are the non-trivial cycles dual to the harmonic basis form $\hat{\bm \eta}_\sigma$ within the spatial slices $\tilde{\mathcal C}(v)$ of $\mathcal H_{v_1}^{v_2}$, the same as  in~\cref{harmonicbasis}. The boundary of each $\sigma'$ is $\sigma(v_1) \cup \sigma(v_2) \subset \tilde{\mathcal C}(v_1) \cup \tilde{\mathcal C}(v_2)$. 

    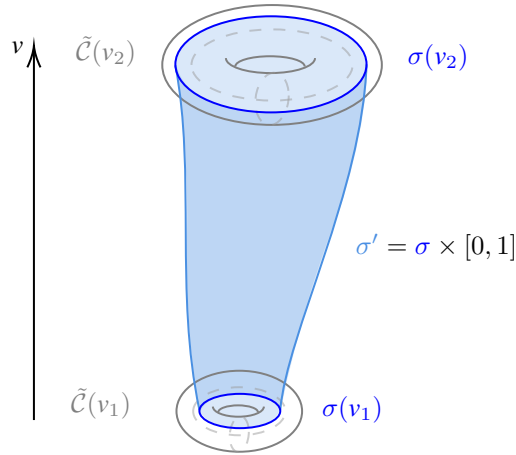
\begin{figure}[H]
        \centering

\tikzset{every picture/.style={line width=0.75pt}} 

\begin{tikzpicture}[x=0.75pt,y=0.75pt,yscale=-1,xscale=1]

\draw  [draw opacity=0][fill={rgb, 255:red, 74; green, 144; blue, 226 }  ,fill opacity=0.2 ] (104.88,80.75) .. controls (105.15,80.59) and (104.34,108.05) .. (152.59,107.74) .. controls (200.85,107.43) and (199.76,83.26) .. (200.3,83.26) .. controls (200.85,83.26) and (200.58,118.73) .. (182.58,167.69) .. controls (164.59,216.66) and (157.37,256.68) .. (157.23,257.15) .. controls (157.09,257.63) and (154.5,266.57) .. (137.6,266.57) .. controls (120.69,266.57) and (117.41,258.4) .. (116.88,258.1) .. controls (116.34,257.79) and (109.24,217.29) .. (110.42,168.63) .. controls (111.6,119.97) and (104.61,80.9) .. (104.88,80.75) -- cycle ;
\draw  [draw opacity=0][fill={rgb, 255:red, 74; green, 144; blue, 226 }  ,fill opacity=0.2 ] (105.15,81.06) .. controls (105.42,80.9) and (107.88,59.71) .. (152.59,59.4) .. controls (197.31,59.09) and (200.03,83.57) .. (200.58,83.57) .. controls (201.12,83.57) and (200.85,119.04) .. (182.86,168.01) .. controls (164.86,216.98) and (157.64,256.99) .. (157.5,257.47) .. controls (157.36,257.95) and (155.86,248.99) .. (137.05,249.31) .. controls (118.24,249.62) and (117.68,258.72) .. (117.15,258.41) .. controls (116.62,258.1) and (109.52,217.6) .. (110.69,168.94) .. controls (111.87,120.28) and (104.88,81.21) .. (105.15,81.06) -- cycle ;
\draw  [color={rgb, 255:red, 128; green, 128; blue, 128 }  ,draw opacity=1 ] (98.34,83.88) .. controls (98.34,67.24) and (122.87,53.75) .. (153.14,53.75) .. controls (183.4,53.75) and (207.94,67.24) .. (207.94,83.88) .. controls (207.94,100.53) and (183.4,114.02) .. (153.14,114.02) .. controls (122.87,114.02) and (98.34,100.53) .. (98.34,83.88) -- cycle ;
\draw  [draw opacity=0] (174.58,80.43) .. controls (174.58,84.59) and (164.63,87.96) .. (152.34,87.96) .. controls (140.06,87.96) and (130.11,84.59) .. (130.11,80.43) -- (152.34,80.43) -- cycle ; \draw  [color={rgb, 255:red, 128; green, 128; blue, 128 }  ,draw opacity=1 ] (174.58,80.43) .. controls (174.58,84.59) and (164.63,87.96) .. (152.34,87.96) .. controls (140.06,87.96) and (130.11,84.59) .. (130.11,80.43) ;  
\draw  [draw opacity=0] (134.78,84.66) .. controls (134.78,84.66) and (134.78,84.66) .. (134.78,84.66) .. controls (134.78,81.81) and (142.54,79.5) .. (152.11,79.5) .. controls (161.67,79.5) and (169.43,81.81) .. (169.43,84.66) -- (152.11,84.66) -- cycle ; \draw  [color={rgb, 255:red, 128; green, 128; blue, 128 }  ,draw opacity=1 ] (134.78,84.66) .. controls (134.78,84.66) and (134.78,84.66) .. (134.78,84.66) .. controls (134.78,81.81) and (142.54,79.5) .. (152.11,79.5) .. controls (161.67,79.5) and (169.43,81.81) .. (169.43,84.66) ;  
\draw  [color={rgb, 255:red, 155; green, 155; blue, 155 }  ,draw opacity=0.5 ][dash pattern={on 4.5pt off 4.5pt}][line width=0.75]  (112.92,83.88) .. controls (112.92,74) and (130.93,65.99) .. (153.14,65.99) .. controls (175.35,65.99) and (193.35,74) .. (193.35,83.88) .. controls (193.35,93.77) and (175.35,101.78) .. (153.14,101.78) .. controls (130.93,101.78) and (112.92,93.77) .. (112.92,83.88) -- cycle ;
\draw  [color={rgb, 255:red, 155; green, 155; blue, 155 }  ,draw opacity=0.5 ][dash pattern={on 4.5pt off 4.5pt}] (161.46,99.26) .. controls (161.46,106.48) and (157.74,113.08) .. (153.14,114.02) .. controls (148.54,114.95) and (144.81,109.86) .. (144.81,102.64) .. controls (144.81,95.42) and (148.54,88.82) .. (153.14,87.88) .. controls (157.74,86.95) and (161.46,92.04) .. (161.46,99.26) -- cycle ;
\draw  [color={rgb, 255:red, 128; green, 128; blue, 128 }  ,draw opacity=1 ] (105.43,258.06) .. controls (105.43,246.81) and (119.46,237.69) .. (136.78,237.69) .. controls (154.1,237.69) and (168.13,246.81) .. (168.13,258.06) .. controls (168.13,269.31) and (154.1,278.43) .. (136.78,278.43) .. controls (119.46,278.43) and (105.43,269.31) .. (105.43,258.06) -- cycle ;
\draw  [draw opacity=0] (149.05,255.72) .. controls (149.05,255.72) and (149.05,255.72) .. (149.05,255.72) .. controls (149.05,258.54) and (143.35,260.81) .. (136.33,260.81) .. controls (129.3,260.81) and (123.6,258.54) .. (123.6,255.72) -- (136.33,255.72) -- cycle ; \draw  [color={rgb, 255:red, 128; green, 128; blue, 128 }  ,draw opacity=1 ] (149.05,255.72) .. controls (149.05,255.72) and (149.05,255.72) .. (149.05,255.72) .. controls (149.05,258.54) and (143.35,260.81) .. (136.33,260.81) .. controls (129.3,260.81) and (123.6,258.54) .. (123.6,255.72) ;  
\draw  [draw opacity=0] (126.28,258.58) .. controls (126.28,256.66) and (130.72,255.1) .. (136.19,255.1) .. controls (141.66,255.1) and (146.1,256.66) .. (146.1,258.58) -- (136.19,258.58) -- cycle ; \draw  [color={rgb, 255:red, 128; green, 128; blue, 128 }  ,draw opacity=1 ] (126.28,258.58) .. controls (126.28,256.66) and (130.72,255.1) .. (136.19,255.1) .. controls (141.66,255.1) and (146.1,256.66) .. (146.1,258.58) ;  
\draw  [color={rgb, 255:red, 155; green, 155; blue, 155 }  ,draw opacity=0.5 ][dash pattern={on 4.5pt off 4.5pt}][line width=0.75]  (113.77,258.06) .. controls (113.77,251.38) and (124.07,245.97) .. (136.78,245.97) .. controls (149.49,245.97) and (159.79,251.38) .. (159.79,258.06) .. controls (159.79,264.74) and (149.49,270.15) .. (136.78,270.15) .. controls (124.07,270.15) and (113.77,264.74) .. (113.77,258.06) -- cycle ;
\draw  [color={rgb, 255:red, 155; green, 155; blue, 155 }  ,draw opacity=0.5 ][dash pattern={on 4.5pt off 4.5pt}] (141.54,268.45) .. controls (141.54,273.33) and (139.41,277.79) .. (136.78,278.43) .. controls (134.15,279.06) and (132.02,275.61) .. (132.02,270.74) .. controls (132.02,265.86) and (134.15,261.39) .. (136.78,260.76) .. controls (139.41,260.13) and (141.54,263.57) .. (141.54,268.45) -- cycle ;
\draw [color={rgb, 255:red, 74; green, 144; blue, 226 }  ,draw opacity=1 ]   (104.88,84.04) .. controls (115.24,138.97) and (104.34,197.51) .. (116.88,257.47) ;
\draw [color={rgb, 255:red, 74; green, 144; blue, 226 }  ,draw opacity=1 ]   (200.3,83.57) .. controls (198.12,147.29) and (160.5,218.86) .. (157.23,257.47) ;
\draw    (34.27,262.49) -- (34.27,76.15) ;
\draw [shift={(34.27,74.15)}, rotate = 90] [color={rgb, 255:red, 0; green, 0; blue, 0 }  ][line width=0.75]    (10.93,-3.29) .. controls (6.95,-1.4) and (3.31,-0.3) .. (0,0) .. controls (3.31,0.3) and (6.95,1.4) .. (10.93,3.29)   ;
\draw  [color={rgb, 255:red, 0; green, 0; blue, 255 }  ,draw opacity=1 ][line width=0.75]  (104.88,83.57) .. controls (104.88,70.22) and (126.24,59.4) .. (152.59,59.4) .. controls (178.94,59.4) and (200.3,70.22) .. (200.3,83.57) .. controls (200.3,96.92) and (178.94,107.74) .. (152.59,107.74) .. controls (126.24,107.74) and (104.88,96.92) .. (104.88,83.57) -- cycle ;
\draw  [color={rgb, 255:red, 0; green, 0; blue, 255 }  ,draw opacity=1 ][line width=0.75]  (116.88,257.94) .. controls (116.88,253.17) and (125.91,249.31) .. (137.05,249.31) .. controls (148.2,249.31) and (157.23,253.17) .. (157.23,257.94) .. controls (157.23,262.71) and (148.2,266.57) .. (137.05,266.57) .. controls (125.91,266.57) and (116.88,262.71) .. (116.88,257.94) -- cycle ;

\draw (26.09,74.78) node    {$v$};
\draw (67.85,254.09) node  [color={rgb, 255:red, 128; green, 128; blue, 128 }  ,opacity=1 ]  {$\tilde{\mathcal{C}}( v_{1})$};
\draw (193.27,257.23) node  [color={rgb, 255:red, 0; green, 0; blue, 255 }  ,opacity=1 ]  {$\sigma ( v_{1})$};
\draw (236.07,79.88) node  [color={rgb, 255:red, 0; green, 0; blue, 255 }  ,opacity=1 ]  {$\sigma ( v_{2})$};
\draw (193.38,167.34) node [anchor=north west][inner sep=0.75pt]    {$\textcolor[rgb]{0.29,0.56,0.89}{\sigma '} =\textcolor[rgb]{0,0,1}{\sigma } \times [ 0,1]$};
\draw (70.31,77.06) node  [color={rgb, 255:red, 128; green, 128; blue, 128 }  ,opacity=1 ]  {$\tilde{\mathcal{C}}( v_{2})$};

\end{tikzpicture}

        \caption{The relative cycle $\sigma'$ of the horizon segment $\mathcal H_{v_1}^{v_2}$ has the topology of $\sigma \times [0,1]$, and it has two boundaries $\sigma(v_1) \subset \tilde{\mathcal C}(v_1)$ and $\sigma(v_2) \subset \tilde{\mathcal C}(v_2)$.}
        \label{fig:relative-cycle}
    \end{figure}
    
    Moreover, the coefficients $\Phi_{\sigma'}^{\mathcal{H^+}}$ can be calculated by
    \begin{equation}
        \Phi_{\sigma'}^{\mathcal H^+} = \int_{\tilde \sigma'} \bm \Phi|_{\mathcal H_{v_1}^{v_2}} = \int_{\tilde \sigma} \bm \Phi|_{\tilde{\mathcal C}} = \Phi_\sigma^{\mathcal H^+}\,,
    \end{equation}
    where $\tilde \sigma'$ are the absolute $(p-1)$-cycles satisfying $\int_{\tilde \sigma_2'} \hat{\bm \eta}_{\text{abs},\sigma_1'}^\text{Thom} = \delta_{\sigma_1'\sigma_2'}$. They are the same as the coefficients $\Phi_\sigma^{\mathcal H^+}$ in \cref{harmonicbasis,electricpotentialfundeq} obtained by the Hodge decomposition of $\bm \Phi$ purely on the spatial slices because the $\tilde \sigma'$ cycles are homologous to $\tilde \sigma$ cycles of the cross-sections within the horizon segment $\mathcal H_{v_1}^{v_2}$. Also, they are constant (independent of the null parameter $v$) over the relative cycles $\sigma'$ as a result of the closedness of $\bm \Phi$ on the horizon.

    Now, on $\mathcal H_{v_1}^{v_2}$, we could use the Poincar\'e-Lefschetz duality (the generalisation of Poincar\'e duality to manifolds with boundaries) to convert the integrals over $\mathcal H_{v_1}^{v_2}$ into integrals along $\sigma'$,
    \begin{equation}
        \int_{\mathcal H_{v_1}^{v_2}} \mathbf X \wedge \hat{\bm \eta}_{\text{abs},\sigma'}^\text{Thom} = \int_{\sigma'} \mathbf X
    \end{equation}
    for any $(D-p)$-form $\mathbf X$. (Here, the Thom form seems to be an overkill as the $\mathbf X$ in consideration is closed, but we need it for compatibility.) 

    Applying this to the electric current integral, we have 
    \begin{equation}
        \begin{split}
            \int_{\mathcal H_{v_1}^{v_2}} {\star}\delta \mathbf j \wedge \bm \Phi & = \sum_{\sigma'} \Phi_\sigma^{\mathcal H^+} \int_{\sigma'} {\star}\delta \mathbf j +  \int_{\mathcal H_{v_1}^{v_2}} {\star}\delta \mathbf j \wedge \dd{\bm \alpha}\\
            & = \sum_{\sigma'} \Phi_\sigma^{\mathcal H^+} \delta q_{\sigma'} + (-1)^{D-p}\left[\int_{\tilde{\mathcal C}(v)} {\star} \delta \mathbf j \wedge \bm \alpha \right]_{v=v_1}^{v=v_2},
        \end{split}
    \end{equation}
    and we notice that the extra contribution from the exact part $\dd{\bm \alpha}$ in $\bm \Phi$ does not vanish because $\mathcal H_{v_1}^{v_2}$ has a boundary. Here, we   used the closedness of ${\star} \delta \mathbf j = - \dd{(\delta \bm \Upsilon)}$, and we   identified the electric charge contribution from the matter source over each relative $(D-p)$-cycle $\sigma'$ as 
    \begin{equation}
        \delta q_{\sigma'} = \int_{\sigma'} {\star}\delta \mathbf j\,.
    \end{equation}
    Combining the results above, we see the $\bm \alpha$ dependent terms cancel each other, and the integrated fundamental identity reads 
    \begin{equation}
        \left[\int_{\mathcal{C}(v)}\!\!\!\!\left( \delta_\phi \mathbf{Q}_{\text{grav-EM},\xi}^{\text{GI}} - \xi \cdot \mathbf{\Theta}_{\text{grav-EM}}^{\text{GI}}(\phi, \delta \phi) \right) + \sum_\sigma \Phi_\sigma^{\mathcal{H}^+} \delta Q_\sigma(v)\right]_{v=v_1}^{v=v_2} = \int_{\mathcal{H}_{v_1}^{v_2}} \!\!\!\!\delta T_{ab} k^a \xi^b \bm \epsilon_\mathcal{H^+} + \sum_{\sigma'} \Phi_\sigma^{\mathcal{H}^+} \delta q_{\sigma'}\,.
    \end{equation} 
    Next, we identify the dynamical entropy of the black object as the integral of the gauge-invariant part of the improved Noether charge of the gravity-electromagnetic sector:
    \begin{equation}
        \begin{aligned}
        \frac{\kappa}{2 \pi} \delta S_\text{dyn}(v) & = \int_{\mathcal{C}(v)}\left( \delta_\phi \mathbf{Q}_{\text{grav-EM},\xi}^{\text{GI}} - \xi \cdot \mathbf{\Theta}_{\text{grav-EM}}^{\text{GI}}(\phi, \delta \phi) \right)\\
        & = \delta_\phi \int_{\mathcal{C}(v)}\left( \mathbf{Q}_{\xi,\text{grav-EM}}^{\text{GI}} - \xi \cdot \mathbf{B}_{\mathcal{H^+},\text{grav-EM}}^{\text{GI}}(\phi) \right).
        \end{aligned}
    \end{equation}
  Thus, we obtain an intermediate version of the  physical process first law for a  black object:
    \begin{equation} \label{physprocessintermediate}
        T \Delta \delta S_\text{dyn}  + \sum_\sigma \Phi_\sigma^{\mathcal{H}^+} \Delta \delta Q_\sigma   = \int_{\mathcal{H}_{v_1}^{v_2}} \delta T_{ab} k^a \xi^b \bm \epsilon_\mathcal{H^+} + \sum_{\sigma'} \Phi_\sigma^{\mathcal{H}^+} \delta q_{\sigma'}\,,
    \end{equation}
    where $T=\kappa/2\pi$ is the background Hawking temperature  and  $\Delta\delta S_\text{dyn} \equiv\delta S_\text{dyn}(v_2) - \delta S_\text{dyn}(v_1) $ is the leading-order entropy difference between the horizon cuts $\tilde{\mathcal C} (v_2)$ and $\tilde{\mathcal C}(v_2)$. Further, the term $\Delta \delta Q_\sigma \equiv \delta Q_\sigma (v_2) - \delta Q_\sigma (v_1)$ on the left-hand side represents the net electric charge variation of the black object during the physical process, whereas $\delta q_{\sigma'}$ on the right-hand side denotes the   flux of electric charge of the matter source flowing into the horizon   interval $\mathcal H_{v_1}^{v_2}$. These first-order charge variations   are, in fact, equal to each other
    \begin{equation} \label{chargesareequal}
        \Delta \delta Q_\sigma = \delta q_{\sigma'}\,,
    \end{equation}
as a consequence of   the first-order variation of the  sourced equation of motion for the $p$-form  gauge field
\begin{equation} 
   \dd{(\delta \bm \Upsilon)} = - {\star} \delta \mathbf{j}\,.
\end{equation}
Integrating this perturbed equation of motion on relative cycle $\sigma'$ yields \eqref{chargesareequal}, because
   \begin{equation}
       \delta Q_\sigma (v_2) - \delta Q_\sigma (v_1) = \left[- \int_{\sigma(v)} \delta \mathbf{\Upsilon}\right]_{v=v_1}^{v=v_2} = - \int_{\partial \sigma'} \delta \mathbf{\Upsilon}  = - \int_{\sigma'} \dd{(\delta \bm{\Upsilon})} = \int_{\sigma'} {\star} \delta \mathbf{j} = \delta q_{\sigma'}\,.
   \end{equation}
 This means that the physical process first law \eqref{physprocessintermediate} simplifies to 
 \begin{equation} \label{physicalprocesssimple}
    T \Delta \delta S_\text{dyn}     = \int_{\mathcal{H}_{v_1}^{v_2}} \delta T_{ab} k^a \xi^b \bm \epsilon_\mathcal{H^+} .
 \end{equation}
The right-hand side is the perturbed matter Killing energy flux, relative to the horizon Killing field~$\xi,$ through the horizon interval $\mathcal{H}_{v_1}^{v_2}. $

We note that   the physical process first law implies a \emph{linearised second law} for the dynamical gravitational entropy. Following~\cite{Hollands:2024vbe},   the local  version of the physical
process  first law  between two infinitesimally close   horizon slices is:
    \begin{equation}
       T \partial_v \delta S_\text{dyn}(v) = \int_{\tilde{\mathcal C}(v)} \delta T_{ab} k^a \xi^b \bm \epsilon_\mathcal{C}  \,.
    \end{equation}
    The right-hand side is non-decreasing if the stress-energy tensor $\delta T_{ab}$ of the external matter source obeys the null energy condition $\delta T_{vv} \geq 0$. The   linearised second law is then $\partial_v \delta S_\text{dyn}   \ge 0$.

 We close the discussion by briefly commenting on the scenario where other gauge-invariant matter fields $\{\psi_{a_1 \cdots a_s}\}$ are also externally sourced. Their perturbed equations of motion are
    \begin{equation}
        \delta E^{a_1 \cdots a_s}_{(\psi)} = - \delta \zeta^{a_1 \cdots a_s}
    \end{equation}
    and the fundamental identity is modified as 
    \begin{equation}
        \dd{\left(\delta_\phi \mathbf{Q}_{\text{grav-EM},\xi}^{\text{GI}} - \xi \cdot \mathbf{\Theta}_{\text{grav-EM}}^{\text{GI}}(\phi, \delta \phi) - \delta \bm \Upsilon \wedge \bm \Phi\right)}\fheq \delta T_{ab} k^a \xi^b \bm \epsilon_\mathcal{H^+} + {\star \delta \mathbf{j}} \wedge \bm \Phi - \delta C_{(\psi)ab} k^a \xi^b \bm \epsilon_{\mathcal H^+}
    \end{equation}
    where 
    \begin{equation}
        -\delta C_{(\psi)ab} = \delta \zeta_{a}{}^{a_2 \cdots a_s} \psi_{ba_2 \cdots a_s} + \delta \zeta^{a_1}{}_a{}^{a_3\cdots a_s} \psi_{a_1ba_3 \cdots a_s} + \cdots + \delta \zeta^{a_1 \cdots a_{s-1}}{}_{a} \psi_{a_1 \cdots a_{s-1}b}\,.
    \end{equation}
    If $\psi$ is a smooth scalar or (non-gauge) vector field, then $\delta C_{(\psi)ab} k^a \xi^b = 0$ on the horizon \cite{Wall:2024lbd}. The non-trivial contributions come from massive gravitons or higher-spin fields with $s \geq 3$ which are externally sourced. If we define 
    \begin{equation}
        \Xi^{(r)}_{a_1 \cdots a_{s-1}} \equiv \xi^b \psi_{a_1 \cdots a_{r-1}b a_{r}\cdots a_{s-1}} \quad \text{and} \quad \delta \tilde \zeta_{(r)}^{a_1 \cdots a_{s-1}} \equiv k^a \delta \zeta^{a_1 \cdots a_{r-1}}{}_a{}^{a_r \cdots a_{s-1}}
    \end{equation}
    for $1\leq r \leq s$ and view $\Xi^{(r)}$ as the ``higher-spin potentials'' and $\delta \tilde \zeta_{(r)}$ as the ``higher-spin currents'', the physical process first law (after cancellations of the electric terms) becomes 
    \begin{equation}
        T \Delta \delta S_\text{dyn}     = \int_{\mathcal{H}_{v_1}^{v_2}} \delta T_{ab} k^a \xi^b \bm \epsilon_\mathcal{H^+} + \sum_{r=1}^s \int_{\mathcal H_{v_1}^{v_2}} \Xi^{(r)}_{a_1 \cdots a_{s-1}} \delta \tilde \zeta_{(r)}^{a_1 \cdots a_{s-1}} \bm \epsilon_\mathcal{H^+}\,.
    \end{equation}

    \section{Differential form ambiguities}
\label{sec:ambi}
    Due to the nilpotency of the exterior derivative on differential forms, quantities defined through their exterior derivative  naturally exhibit differential form ambiguities. For example, if $\bm \beta$ is defined by $\dd{\bm \beta} = \bm \alpha$, then it is only defined up to an exact term: $\bm \beta \sim \bm \beta + \dd{\bm \gamma}$. Below, we discuss three different types of differential form ambiguities: Jacobson-Kang-Myers (JKM) ambiguities of covariant phase space quantities, gauge redundancies of the $p$-form gauge field, and redundancies in the gauge parameter $\bm \lambda_\xi$. We focus, for convenience, on the thermodynamic quantities that appear in the comparison version of the non-stationary first law. We will show for non-compact horizons that,  under appropriate  boundary conditions at the asymptotic ends of the extended horizon directions, the variation of the dynamical gravitational entropy, mass, angular momenta  and the electric charge   in the first law are invariant under these ambiguities (but not the magnetic charge term, see section \ref{ssec:magnetic}).  For compact horizons all the terms in the first law are automatically free of ambiguities irrespective of boundary conditions. 
    
    \subsection{JKM ambiguities}\label{sec:jkm}

    The covariant phase space quantities expressed in differential form language are subject to various ambiguities known as the Jacobson-Kang-Myers (JKM) ambiguities \cite{Jacobson:1993vj} (see also \cite{Iyer:1994ys})
    \begin{equation}
        \begin{aligned}
            \mathbf{L}(\varphi) &\to \mathbf{L}(\varphi) + \dd{\bm \mu(\varphi)}\\
            \bm \Theta(\varphi, \delta \varphi) &\to \bm \Theta(\varphi, \delta \varphi) + \dd{\mathbf{Y}(\varphi, \delta \varphi)}\\
            \mathbf{Q}_\chi & \to \mathbf{Q}_\chi + \dd{\mathbf{Z}(\varphi, \chi)}.
        \end{aligned}
    \end{equation}
  The dynamical field equations are not affected by these transformations.
    In order to match the properties of $\bm \Theta$ and $\mathbf{Q}_\chi$, $\mathbf{Y}(\varphi, \delta \varphi)$ is required to be linear in $\delta \varphi$, and $\mathbf{Z}(\varphi, \chi)$ is required to be linear in the vector field $\chi$ that generates diffeomorphisms.

The variation of the dynamical entropy is defined in terms of the 
 field variation of the gauge-invariant part of the  improved Noether charge   associated to the horizon Killing field $\xi$: \begin{equation}\delta S_{\text{dyn}} = \frac{2\pi}{\kappa} \int_{\mathcal C} \delta_\phi \tilde{\mathbf{Q}}_\xi^{\text{GI}}\,.\end{equation} The latter quantity is given by
\begin{equation}
\delta_\phi \tilde{\mathbf{Q}}_\xi^{\text{GI}}= \delta_\phi \mathbf{Q}_\xi^{\text{GI}} - \xi \cdot \mathbf{\Theta}^{\text{GI}} (\phi, \delta \phi)
\end{equation}
and transforms under the JKM ambiguities as \cite{Hollands:2024vbe,Visser:2024pwz}
\begin{equation}
  \delta_\phi \tilde{\mathbf{Q}}^{\text{GI}}_\xi \to \delta_\phi \tilde{\mathbf{Q}}^{\text{GI}}_\xi +  \dd \left[  \xi \cdot \mathbf{Y}(\phi, \delta \phi) + \delta_\phi \mathbf{Z} (\phi, \xi)\right] \,, 
\end{equation}
which follows from the background stationarity conditions $\mathcal L_\xi \phi = 0$. If the horizon cross-section~$\mathcal C$ is compact, then the exact differential ambiguity vanishes due to Stokes' theorem. This implies the  variation of the dynamical black
hole entropy for first-order perturbations of a stationarity background is JKM-invariant for compact horizons, as shown   in \cite{Hollands:2024vbe,Visser:2024pwz}. For non-compact horizon slices, however, the JKM ambiguities in the dynamical entropy reduce to boundary integrals over the codimension-3 surfaces at the asymptotic ends of the non-compact (flat) horizon directions.  The JKM shift in the variation of the dynamical entropy is 
\begin{equation}
    \Delta_{\text{JKM}} \delta S_{\text{dyn}} =\frac{2\pi}{\kappa} \int_{\partial \mathcal C}    \left [  \xi \cdot \mathbf{Y}(\phi, \delta \phi) + \delta_\phi \mathbf{Z} (\phi, \xi)\right]
    \label{jkmshift}
\end{equation}
We emphasise that, for non-compact horizon cross-sections, this JKM shift can be made to 
vanish by imposing suitable boundary conditions along the flat directions 
(see condition~4 in section~\ref{sec:falloff}). In particular, since we 
restrict to the gauge-invariant sector, the forms $\mathbf Y(\phi,\delta\phi)$ 
and $\delta_\phi \mathbf  Z(\phi,\xi)$ are linear in the perturbations 
$\delta\phi$ and do not depend explicitly on the gauge potential $\mathbf A$. 
It is sufficient to require that both the background fields $\phi$ and 
the perturbations $\delta\phi$ approach the same constant asymptotic values  as 
$y^\aleph \to +\infty$ and $y^\aleph \to -\infty$ along each flat 
direction $y^\aleph$, with $\mathcal{L}_{m_\aleph}\delta\phi \to 0$. Under 
this condition the integrands in~\eqref{jkmshift} take the same constant limit 
on the two opposite faces of~$\partial \mathcal C$. Because these faces enter with opposite 
orientation, their contributions cancel pairwise and hence
\begin{equation}
\Delta_{\rm JKM}\delta S_{\rm dyn} = 0 .
\label{VI.6}
\end{equation}
Thus, under these mild equal-asymptotics conditions, the variation of the 
dynamical entropy is JKM-invariant even for non-compact horizon slices. Similarly, the variation of the canonical mass $M$ and angular momenta $J_I$ in \eqref{eq:int-S-infty} is also JKM-invariant under these asymptotic conditions.

Now, there is a similar JKM ambiguity in the gauge-dependent part of the non-stationary first law:
\begin{equation}  
   \left (  \int_{\mathcal{C}}- \int_{\mathcal{S}_\infty} \right) \left[ \delta_\varphi \mathbf{Q}^{\text{GD}}_\xi - \xi \cdot \mathbf{\Theta}^{\text{GD}}(\phi, \delta \mathbf{A}) - \delta \bm \Upsilon \wedge \bm \lambda_\xi \right]\,.
    \end{equation}
In particular,   the first two terms in the integrand are JKM ambiguous, whereas the last term is not, since  $\delta \bm \Upsilon$ is the perturbed equation of motion for the field strength, and $\bm \lambda_\xi$ is defined via the gauge field. The JKM ambiguity is again an exact term, which vanishes when integrated over a compact horizon, and when integrated over a non-compact horizon it becomes a boundary integral  
\begin{equation}  
   \left (  \int_{\partial \mathcal{C}}- \int_{\partial \mathcal{S}_\infty} \right)  \left [  \xi \cdot \mathbf{Y}(\varphi, \delta \varphi) + \delta_\varphi \mathbf{Z} (\varphi, \xi)\right]  \,.
    \end{equation}
    This ambiguity   vanishes if we require the fourth asymptotic fall-off condition (in section \ref{sec:falloff}), but now applied to all dynamical fields $\varphi$ and their perturbations, including the gauge field~$\mathbf A$.  This ensures that the electric charge variation in the non-stationary first law for non-compact horizons is JKM unambiguous. 

    \subsection{Gauge redundancies}

    Under a gauge transformation  
    \begin{equation}
        \mathbf{A} \to \mathbf{A} + \dd{\bm \Lambda}
    \end{equation}
    the gauge parameter $ \bm{\lambda}_\xi$, defined in \eqref{eq:A-Killing},  transforms as
    \begin{equation}
        \bm{\lambda}_\xi \to \bm{\lambda}_\xi + \mathcal L_\xi \mathbf \Lambda\,,
    \end{equation}
    since $\mathcal L_\xi \mathbf A \to \mathcal L_\xi (\mathbf A + \dd \mathbf \Lambda) = \dd (\bm \lambda_\xi + \mathcal L_\xi \mathbf \Lambda)$.
  Moreover, the  gauge-dependent parts of the presymplectic potential $\bm \Theta^{\text{GD}}(\phi, \delta \mathbf{A})$ and the Noether charge $\mathbf{Q}^{\text{GD}}_\chi$ transform as 
    \begin{equation}
        \begin{split}
            \bm \Theta^{\text{GD}}(\phi, \delta \mathbf{A}) & \to \bm \Theta^{\text{GD}}(\phi, \delta \mathbf{A}) + (-1)^{D-p-1} \bm \Upsilon \wedge \delta (\dd{\bm \Lambda})\\
            \mathbf{Q}^{\text{GD}}_\chi & \to \mathbf{Q}^{\text{GD}}_\chi + \bm \Upsilon \wedge (\chi \cdot \dd{\bm \Lambda})
        \end{split}
    \end{equation}
Hence, the gauge ambiguity in the  combination of these differential forms that appears in \eqref{eq:improv-Q-GD}
   \begin{equation} \label{combinationofformsgauge}
        \delta_\varphi \mathbf{Q}^\text{GD}_\xi - \xi \cdot \mathbf \Theta^{\text{GD}} - \delta \bm \Upsilon \wedge \bm \lambda_\xi = \delta \mathbf{\Upsilon} \wedge (\xi \cdot \mathbf{A}) + (-1)^{D-p} (\xi \cdot \bm \Upsilon) \wedge \delta \mathbf{A}- \delta \bm \Upsilon \wedge \bm \lambda_\xi\,
    \end{equation}
is given by
\begin{equation} \label{twotermsgaugeamb}
       \delta \mathbf{\Upsilon} \wedge \dd(\xi \cdot  \mathbf{\Lambda} )+ (-1)^{D-p} (\xi \cdot \bm \Upsilon) \wedge \delta (\dd\mathbf{\Lambda}) \,.
\end{equation}
In the derivation of  the black hole first law, the combination of forms \eqref{combinationofformsgauge} was integrated over a horizon cross-section and spatial infinity. When pulled back to the Killing  horizon the second term in \eqref{twotermsgaugeamb} vanishes, as shown in \eqref{secondtermvanishesupsilon}. At spatial infinity, we impose the asymptotic condition $\delta \mathbf A \to 0$, which should hold independent of the gauge choice, hence we require $\delta (\dd \mathbf \Lambda) \to 0$ as $r \to \infty. $ After imposing the linearised   gauge field equation $\dd{(\delta \bm \Upsilon)} =0$, the remaining gauge redundancy at $\mathcal C $ and $\mathcal S_\infty$ is thus an exact term
\begin{equation}\label{exacttermgaugetr}
    \dd  (\delta \bm \Upsilon \wedge (\xi \cdot \bm \Lambda))\,.
\end{equation}
    When integrated over a compact horizon this redundancy vanishes, however for non-compact horizons it becomes a codimension-3 boundary integral. 
    By the fourth asymptotic condition (section~\ref{sec:falloff}) the
dynamical fields $\phi$ and their perturbations $\delta\phi$ approach the same
constant limits as $y^\aleph \to \pm \infty$, so that $\delta \bm \Upsilon  \to \delta \bm \Upsilon_\infty. $ If, in addition,
the pullback of  
$\xi \cdot \bm \Lambda$ to the boundary $  \partial \mathcal{C}$  
    goes to the same constant   at both ends of the
non-compact horizon directions, and similarly for the pullback to the boundary $ \partial \mathcal S_\infty$, i.e.,
\begin{equation}
    \xi \cdot \bm \Lambda \big|_{\partial \mathcal{C}} \to \lambda^{\mathcal C}_\infty \,, \qquad \xi \cdot \bm \Lambda \big|_{\partial \mathcal{S}_\infty} \to \lambda^{\mathcal S_\infty}_\infty \qquad \text{as} \qquad y^\aleph \to \pm\infty\,,
\end{equation}
then the codimension-3 contributions from the opposite faces cancel pairwise. 
This restriction on the gauge parameter ensures that residual gauge
transformations act trivially on the physical data at both ends of the
non-compact directions. In particular, it rules out large gauge
transformations that would otherwise shift the relative charge between
the two ends and generate a spurious boundary contribution. Hence, with these boundary conditions, the exact term \eqref{exacttermgaugetr} does not contribute on non-compact horizon cross-sections.

Finally, we discuss another redundancy in the definition of $\bm \lambda_\xi$ in \eqref{eq:A-Killing}, namely that one can freely add an exact term $\bm \lambda_\xi \to \bm \lambda_\xi + \dd  \bm \alpha.$ This redundancy leads to a similar ambiguity in \eqref{combinationofformsgauge} as   the gauge ambiguity, i.e.,
\begin{equation}
   - \dd (\delta \bm \Upsilon \wedge \bm \alpha)\,.
\end{equation}
When integrated over non-compact horizon slices, this ambiguity   vanishes if we require that the pullback of $\bm \alpha$ to $\partial \mathcal C$ goes to the same constant as $y^\aleph \to \pm \infty$ (and similarly for $ \partial \mathcal S_\infty$).
Another  way to remove  this ambiguity is by cancelling it against the ambiguity in \eqref{exacttermgaugetr}, which is realised if $\bm \alpha =   \xi \cdot \bm \Lambda$. Thus, without imposing any boundary condition, and by exploiting the 
interplay between the two ambiguities, it is possible to arrange for their 
cancellation and thereby obtain a fully gauge-invariant definition of the electric potential.

\section{Examples}
\label{sec:exam}

In this section we illustrate our general  framework for electromagnetic potential-charge contributions to
the first law  with three examples: A) four-dimensional dyonic AdS black holes in Einstein-Maxwell gravity with spherical, planar and hyperbolic horizon topology, B) five-dimensional rotating black rings with electric dipole charge, and C) electrically charged $(p-1)$-black branes sourced by a $p$-form gauge field in ten dimensions. We pay special attention to how and where the charges and potentials are defined. Our first law differs from previous work on these examples in the sense that the perturbations are now allowed to be non-stationary, such that the horizon entropy is   given by the dynamical entropy formula.  Moreover, we show how previously derived charges and potentials follow in a systematic fashion from our general formalism.

\subsection{Dyonic AdS black holes}
\label{sec:examplA}
For Einstein--Maxwell theory  with cosmological constant $\Lambda$ 
\begin{equation} \label{einsteinmaxwelltheory}
\mathbf L =  \frac{1}{16\pi G}(R-2 \Lambda)\bm \epsilon - \frac{1}{8\pi}  {\star} \mathbf F \wedge \mathbf F \,,  
\end{equation}
and $\mathbf{F}= \dd \mathbf A$ the two-form field strength,  the relevant covariant phase space quantities   are
\begin{align}
    \mathbf \Upsilon &= - \frac{1}{4\pi}\,{\star} \mathbf F\,, \\
    \mathbf \Theta (\phi, \delta \varphi)&=   \frac{1}{16 \pi G} \bm\epsilon^a \left ( \nabla^b \delta g_{ab} - g^{cd}\nabla_a \delta g_{cd} \right) -  \frac{1}{4\pi}   {\star} \mathbf F \wedge \delta \mathbf A\,,\\
    \mathbf Q_\xi &=   - \frac{1}{16 \pi G}\, {\star} \dd  \xi - \frac{1}{4\pi}(\xi \cdot \mathbf A)\, {\star} \mathbf F \,.
\end{align}
Static, dyonic, asymptotically anti-de Sitter black hole solutions  to \eqref{einsteinmaxwelltheory}  in four spacetime dimensions with constant-curvature horizon cross-sections are described by  \cite{Chamblin:1999tk,Caldarelli:2008ze}
\begin{align}
    ds^2 &= - f(r) \dd{t^2} + \frac{\dd{r^2}}{f(r)} + r^2 (\dd \theta + S_\kappa(\theta) \dd \phi)   \,,\\
    f(r)&= \kappa + \frac{r^2}{L^2} - \frac{m }{  r} + \frac{q^2 + p^2}{r^{2}}\,,\\
    \mathbf F &=  \dd \mathbf A= -\frac{q}{r^2} \dd t \wedge \dd r + p S_\kappa (\theta) \dd \theta \wedge \dd \phi\,,\\
    \star \mathbf F & = - 4 \pi \dd \check{\mathbf A}=    \frac{p}{r^2} \dd t \wedge \dd r +  q S_\kappa (\theta) \dd \theta \wedge \dd \phi\,,
\end{align}
where $S_\kappa (\theta) = \sin \theta (\kappa = 1), 1 (\kappa=0), \sinh \theta (\kappa=-1)$, corresponding to spherical, planar and hyperbolic horizons, respectively. In the blackening factor, $L= \sqrt{-3/  \Lambda}$ is the AdS curvature radius, $m$ is the mass parameter, $q$ is the electric charge parameter and $p$ is the magnetic charge parameter. 
The gauge field and dual gauge field, respectively,   read 
\begin{align}
    \mathbf  A &= \left (- \Phi_\infty  -  \frac{q}{r}\right) \dd t +    p \left ( \int S_\kappa (\theta ) \dd \theta\right) \dd \phi\,,   \\
    \check{\mathbf A} & = \left ( -\Psi_\infty  - \frac{p}{4 \pi r}\right) \dd t     - \frac{q}{4 \pi } \left ( \int S_\kappa (\theta )   \dd \theta \right)  \dd \phi\,,
\end{align}
with   $\Phi_\infty$ and $\Psi_\infty$ being   pure-gauge constants. For the spherical case ($\kappa=1$), the gauge potential must be defined on two overlapping patches, related by a gauge transformation, to avoid the Dirac string singularity. The same   patching applies to the dual potential $\check{\mathbf{A}}$, which is defined on the same principal 
 $U(1)$ bundle, ensuring that both $\mathbf F$ and $\star \mathbf F$ are globally smooth. The standard north and south patch expressions on the two-sphere for the (dual) gauge potentials are: 
 \begin{align}
     A_\phi^\text{N} &= p (1 - \cos \theta) \dd \phi\,, \quad \qquad\,  A_\phi^\text{S} =- p (1 + \cos \theta) \dd \phi\,,\\
      \check{A}_\phi^\text{N} &= - \frac{q}{4\pi} (1 - \cos \theta) \dd \phi\,, \,\,\, \quad  \check{A}_\phi^\text{S} = \frac{q}{4\pi} (1 + \cos \theta) \dd \phi\,,
 \end{align}
 satisfying $  \mathbf{A}^\text{N} - \mathbf{A}^\text{S}=   2 p \dd\phi$ and $   \check{\mathbf{A}}^\text{N}  - \check{\mathbf{A}}^\text{S}= - \frac{q}{2\pi}  \dd\phi$. The gauge transformations relating the north and south patches are the same for the gauge potential and its dual, under the charge transformation 
\((q, p) \to (4\pi p, -q/4\pi)\), which corresponds to the electromagnetic duality 
\(\mathbf{F} \to \check{\mathbf{F}} = -\tfrac{1}{4\pi}{\star}\mathbf{F}\).
 
We now consider dynamical perturbation of these black holes, and study the first law. The dynamical  entropy, ADM mass, and Hawking temperature are, respectively,
 \begin{align}
    S_{\text{dyn}} &= (1- v \partial_v) \frac{\omega_{\kappa}r_h^{ 2}}{4 G}\,, \qquad \qquad 
    M =\frac{ \omega_{\kappa}}{8 \pi G} m\,, \\
    T&= \frac{1}{4\pi} f'(r_h) = \frac{1}{4\pi r_h}\left(\kappa+\frac{3r_h^2}{L^2}-\frac{ q^2 + p^2}{r_h^{4}}\right)\,,
    \end{align}
    where the event horizon radius $r_h$ is the largest root of $f(r_h)=0$. 
    
    Further, the electric and magnetic charge, and the electric and magnetic potential   are   given by
    \begin{align}
      Q &= - \int_{\Sigma_{\kappa}}\!\! \bm \Upsilon = \frac{1}{4 \pi}\int_{\mathcal{C}_{\kappa}}  \!\!{\star}  \mathbf F = \frac{q  }{4\pi }  \omega_{\kappa} \,, \qquad P  = \int_{\mathcal{C}_{\kappa}} \mathbf F = p \omega_\kappa\,,\\
      \bar \Phi &=  \xi \cdot \mathbf A |_{r =\infty} - \xi \cdot \mathbf A |_{r =r_h}=    \frac{q}{r_h}\,, \qquad \,\, \,
      \bar \Psi  = \xi \cdot \check{\mathbf A} |_{r =\infty} - \xi \cdot \check{\mathbf A} |_{r =r_h}=   \frac{p}{4 \pi r_h}\,.
 \end{align} 
  Here, the horizon cross-section $\mathcal{C}_{\kappa}$ is the unit  $2$-dimensional space of constant curvature~$\kappa$, and $\omega_{\kappa}=\int_{\mathcal{C}_\kappa}S_\kappa(\theta)\dd \theta \wedge \dd \phi$ is the volume of $\mathcal{C}_{\kappa}$, which is infinite for $\kappa=0$ and $\kappa=-1$.   The normalisation of the electric and magnetic charges differs by $1/4\pi$, but that is purely conventional. These charges are the same at the horizon and spatial infinity, since $\dd{( {\star} \mathbf F)} = 0 $ and $\dd \mathbf F = 0$. We further observe that the pure-gauge terms $\Phi_\infty$ and $\Psi_\infty$ drop  out in the electric and magnetic potential difference. 
  Moreover, in the spherical case ($\kappa=1$), the   electric potential  
 $\bar \Phi$ (and similarly $\bar \Psi$) is the same on the north and south patches,  as the gauge transformation between the two gauge potentials vanishes identically, since $\xi \cdot \dd \phi=0.$

These thermodynamic quantities satisfy the non-stationary first law for all horizon topologies
    \begin{equation}
        T \delta S_\text{dyn} = \delta M   - \bar \Phi \delta Q - \bar \Psi \delta P\,.
    \end{equation}
    For non-compact horizon slices ($\kappa=0,-1$) one instead often works with finite densities
\begin{equation}
\varepsilon=\frac{M}{\omega_{\kappa}},\qquad 
s_{\text{dyn}}=\frac{S_{\rm dyn}}{\omega_{\kappa}},\qquad 
\rho=\frac{Q}{\omega_{\kappa}}\,, \qquad p= \frac{P}{\omega_k}\,,
\end{equation}
in which case the first law reads
\begin{equation}
   T \delta s_{\text{dyn}} = \delta \varepsilon - \bar \Phi \delta \rho  - \bar \Psi \delta p\,.
\end{equation}
   
\subsection{Black rings with dipole charge}
\label{sec:blackrings}

Next, we  consider five-dimensional Einstein gravity minimally coupled to a dilaton $\chi$ and a two-form gauge potential $\mathbf{B}$,
\begin{equation} \label{EMDfivedim}
\mathbf{L}
=\frac{1}{16\pi G}\left[R\bm \epsilon
-\frac{1}{2}\, {\star} \dd\chi \wedge \dd\chi
-\frac{1}{2}\, e^{-\alpha\chi}\, {\star} \mathbf{H}\wedge  \mathbf{H}\right]\,.
\end{equation}
where $\mathbf{H}=\dd{\mathbf{B}}$ is the three-form field strength, and $\bm \epsilon$ is the five-dimensional volume form. 
The covariant phase space quantities for this Einstein-Maxwell-Dilaton action are  
\begin{align}
\mathbf \Upsilon &= -\frac{1}{16 \pi G}  e^{-\alpha\chi}\,{\star} \mathbf H\,, \\
    \mathbf \Theta (\phi, \delta \varphi)&=   \frac{1}{16 \pi G} \left [\bm\epsilon^a \left ( \nabla^b \delta g_{ab} - g^{cd}\nabla_a \delta g_{cd} \right)-  {\star} \dd \chi \, \delta \chi-    e^{-\alpha \chi} \, {\star} \mathbf H \wedge \delta \mathbf{B} \right]\,,\\
     \mathbf Q_\xi &=   - \frac{1}{16 \pi G}\, {\star} \dd  \xi - \frac{1}{16 \pi G}  e^{-\alpha\chi}    \, {\star} \mathbf{H} \wedge (\xi \cdot \mathbf{B})
    \,.
\end{align}

Emparan \cite{Emparan:2004wy} 
 constructed stationary, asymptotically flat black ring solutions to  \eqref{EMDfivedim} with nonzero electric dipole charge. This is a charge generalisation  of the explicit solution for a neutral rotating black ring found in \cite{Emparan:2001wn}.    The solutions for the metric, dilaton and gauge potential are described in \cite{Emparan:2004wy,Copsey:2005se}, and we do not repeat them here. These black rings  have  horizon topology $S^1\times S^2$ and  are characterised by three independent parameters: the mass $M$,  angular momentum~$J$ along the ring~$S^1$ (with coordinate $\psi$, whose period is $2 \pi$), and a local electric dipole charge~$Q$. Following our definition~\eqref{pcharge}, the   dipole charge is given by integrating the electric flux two-form~$\bm \Upsilon$ over any two-cycle $  S^2 (= \sigma)$ that links the ring
\begin{equation} \label{dipolecharge}
    Q = -  \int_{S^2} \bm \Upsilon = \frac{1}{16 \pi G} \int_{S^2} e^{- \alpha \chi} \,{\star} \mathbf H \,.
\end{equation}
This is independent of the choice of linking two-sphere (inside the horizon cross-section), because the field equation
$\dd(e^{-\alpha\chi}\,{\star} \mathbf H)=0$ ensures the integrand is closed.  As shown in \cite{Copsey:2005se}, the gauge potential $\mathbf{B}$ cannot be defined smoothly on all of spacetime if both  the   dipole charge is nonzero and    the stationarity condition $\mathcal L_\xi \mathbf B =0$ holds. However, this  is a gauge artifact: 
$\mathbf H = \dd \mathbf B$
remains smooth globally, but $\mathbf{B}$
 must be defined patchwise or diverges on either the 
$\psi$-axis or the horizon. In \cite{Copsey:2005se} they set $\mathbf B=0$   on the $\psi$-axis, but allow it to diverge on the bifurcation surface, so that the electrostatic potential is finite on that surface.

Our normalisation for the charge \eqref{dipolecharge} differs   from the one in \cite{Emparan:2004wy,Copsey:2005se}, where the charge is defined as $q_e = 4 G Q $. This is not an issue, since this  is compensated for by our different normalisation for the electrostatic potential $ \Phi$, which differs  from the potential $\phi_e =   \Phi / (4G)$ in~\cite{Emparan:2004wy,Copsey:2005se}. Hence, the product of the   potential and (the variation of) the charge in the first law is the same as in the original works, i.e. $  \Phi \delta Q = \phi_e \delta q_e.$ Equation  \eqref{electricpotentialfundeq} states that the electrostatic potential is defined as the integral of the  potential form over the one-cycle $S^1 (= \tilde \sigma)$
\begin{equation}
  \Phi = - \int_{S^1}  \xi \cdot \mathbf B |_{\mathcal H^+}  =   -  2\pi    B_{t \psi} |_{\mathcal H^+}  \,.
\end{equation}
In the second equality we used $\xi=\partial_t+\Omega\,\partial_\psi$, where $ \Omega$ is the angular velocity along the $S^1$, and $\mathbf{B} = B_{t\psi} \dd t \wedge \dd \psi$. Emparan \cite{Emparan:2004wy} expresses the electrostatic potential as the difference between $B_{t\psi}$ at infinity and the horizon, but $B_{t\psi}|_\infty=0$ because the spacetime is asymptotically flat and the field strength 
  vanishes at infinity, allowing 
 $\mathbf B$ to be gauged to zero there \cite{Copsey:2005se}.
We note the electrostatic potential is   invariant under gauge transformations $\mathbf B \to \mathbf B + \dd{\bm \Lambda}$ that respect the stationarity condition $\mathcal L_\xi \mathbf B = 0$, which implies  $\mathcal L_\xi \bm \Lambda=0$. This is because $\xi \cdot \dd \bm \Lambda = \mathcal L_\xi \bm \Lambda - \dd(\xi \cdot \bm \Lambda)$ and the total derivative term vanishes when integrated on a closed~$S^1.$  This is a special case of \eqref{gaugetranselectricpot}; here we choose a gauge where  $\bm \lambda_\xi =0 $  in accordance with the stationarity condition $\mathcal L_\xi \mathbf B = 0$ used in~\cite{Emparan:2004wy,Copsey:2005se}.   

In this example, the boundary term at spatial infinity \eqref{infintegralcharge} vanishes, since the spacetime is asymptotically flat and the flux variation $\delta \mathbf \Upsilon$ falls off too rapidly to produce a nonzero integral. Moreover, the horizon two-sphere that supports the dipole flux is not homologous to any cycle at infinity, since    $\mathcal{S}_\infty \cong S^3$ has no non-trivial two-cycle, so no asymptotic charge can be defined (see discussion below \eqref{infintegralcharge}). Consequently, the  electric  charge term in the first law arises solely from the horizon contribution, where the potential  $\Phi$ and local charge  $Q$ are evaluated.

Now, for  dipole black rings the   non-starionary comparison first law takes the universal form
\begin{equation}
T\delta S_{\rm dyn}=\delta M-\Omega\delta J- \Phi \delta Q\,,
\end{equation}
where $S_{\rm dyn}=\frac{1}{4G} (1- v \partial_v )  A  $ is the dynamical entropy for general relativity. For   perturbations to nearby stationary black  rings, the first law was derived by Emparan \cite{Emparan:2004wy} and recovered using Hamiltonian methods in \cite{Copsey:2005se} and using covariant phase space methods in \cite{Compere:2007vx}. We have shown that,  for non-stationary perturbations, the black ring first law with the right dipole charge term  follows from the framework of covariant phase space.  We emphasise that, due to Poincar\'{e}  duality between cohomology and homology, the charge precisely localises on the two-sphere and the   potential on the ring of the horizon cross-section.

\subsection{Charged black  branes}

We now turn to the class of electrically charged black ($p-1$)-brane solutions of the ten-dimensional Einstein-Maxwell-Dilaton theory, with $1\leq p \leq 7$ \cite{Gibbons:1987ps,Garfinkle:1990qj,Horowitz:1991cd,Duff:1993ye}.  These backgrounds generalise the neutral branes by supporting a non-trivial $p$-form potential $\mathbf A_{(p)}$ with field strength $\mathbf F_{(p+1)} = \dd \mathbf A_{(p)}.$ For example,  a charged black string corresponds to $p=2$, so it is sourced by a two-form gauge field. The charged black $(p-1)$-branes are solutions to the following theory:
\begin{equation}
    \mathbf L = e^{-\phi} [R \bm \epsilon + 4\, {\star}\dd\phi \wedge \dd\phi] - 2 e^{2 \alpha \phi} {\star} \mathbf F_{(p+1)} \wedge \mathbf F_{(p+1)}.
\end{equation}
This Lagrangian corresponds to the bosonic sector of a low-energy supergravity action truncated to  the metric, dilaton, and a single $p$-form Ramond-Ramond field.
For suitable choices of the coupling~$\alpha$ the theory can be embedded consistently in type II string theory   in ten dimensions, where    the electrically charged solutions correspond to the supergravity limits of   D-branes with Ramond-Ramond charge. The magnetic duals of these electrically charged $p$-branes  are $(7-p)$-brane configurations sourced by the    potential $\check{\mathbf{A}}_{(8-p)}$, whose field strength  is proportional to the Hodge dual of the electric one. 

The covariant phase space quantities for this Lagrangian are
\begin{align}
\mathbf \Upsilon &= - 4 e^{2\alpha \phi} \, {\star} {\mathbf F}\,,\nonumber\\
    \mathbf \Theta (\phi, \delta \varphi)&= e^{-\phi}\bm\epsilon^a \left [ \nabla^b \delta g_{ab} - g^{cd}\nabla_a \delta g_{cd} - (\nabla^b \phi )\delta g_{ab}+  g^{cd}(\nabla_a \phi)\delta g_{cd} + 8 ( \nabla_a \phi) \delta \phi   \right] - 4 e^{2 \alpha \phi}\, {\star} {\mathbf F} \wedge \delta \mathbf A \,, \nonumber\\
     \mathbf Q_\xi &= -  \bm \epsilon_{ab}\left[ e^{-\phi} \nabla^a \xi^b - 2 \xi^a \nabla^b e^{-\phi}  \right] - 4 e^{2\alpha \phi}\,{\star} {\mathbf F} \wedge (\xi \cdot \mathbf A) \, ,
\end{align}
where we have adopted the Killing gauge $\mathcal L_\xi \mathbf A = 0$ for simplicity. These expressions define the electric flux, symplectic potential, and Noether charge, that are needed to derive the first law   below.

The   static, electrically charged black $p$-brane solution, with Euclidean symmetry in $p-1$ dimension and spherical symmetry in $11-p$ dimensions, takes the form  \cite{Horowitz:1991cd}
\begin{align}
    \dd{s^2} &= [f_-(r)]^{\gamma_x}\left(- \frac{f_+(r)}{f_-(r)} \dd{t^2} + \dd{x^i}\dd{x_i}\right) + [f_-(r)]^{\gamma_r + 1} \left(\frac{\dd{r^2}}{f_+(r)f_-(r)} + r^2 \dd{\Omega}_{9-p}^2\right)\,,\\
    \mathbf A & = \frac{Q_0}{(8-p)\gamma_A r_-^{8-p}} [f_-(r)]^{\gamma_A}\bm \epsilon_\text{$p$}\,,\qquad 
    \star \mathbf F   = Q_0 \bm \epsilon_{S^{9-p}}\,,\qquad
    e^{-2 \phi}  = [f_-(r)]^{\gamma_\phi}\,,
\end{align}
where $t,r$ and $x^i$ are the time, radial and $(p-1)$-brane coordinates, respectively, $\bm  \epsilon_{S^{9-p}}$ is the volume element on the unit $(9-p)$-sphere, i.e. $\int_{S^{9-p}}\bm  \epsilon_{S^{9-p}} = \Omega_{9-p}$, and $\bm \epsilon_\text{$p$}=\dd t \wedge \dd x^1 \wedge \cdots \wedge \dd x^{p-1}$. The exponents  
\begin{align}
    \gamma_A & = \frac{p\gamma_x + \gamma_r + 1}{2}\,,\qquad \,\,\,\,\,
    \gamma_r   = \beta (\alpha -1) - \frac{6-p}{8-p}\,,\qquad
    \gamma_x  = \beta (\alpha +1)\,,\\
    \gamma_\phi & = - \beta (4 \alpha + p-4)\,,\qquad 
    \beta = \frac{1}{2 \alpha^2 + (p-4) \alpha + 2}
\end{align}
are fixed by the dilaton coupling $\alpha$   and  ensure that the fields satisfy the coupled Einstein, Maxwell, and scalar equations of motion.

Further, the blackening factors are
\begin{equation}
    f_\pm(r) = 1 - \left(\frac{r_\pm}{r}\right)^{8-p}\,,
\end{equation}
 where $r_+ $ and $r_-$ are integration constants with $r_- \le r_+$. The surface $r=r_+$ is a regular event horizon with cross-sections having topology $\mathcal C \cong S^{9-p} \times \mathbb R^{p-1}$,  and  $r=r_-  $ is the location of a curvature singularity, where the transverse $(9-p)$-spheres shrink to zero size. These two parameters are related to the charge and mass densities of the black brane,  which according to    \cite{Horowitz:1991cd} are given by 
\begin{align}
    Q_0  &= \sqrt{\frac{\beta (8-p)^2 (r_+ r_-)^{8-p}}{2}} \,,\\
     M_0 &= [1 - (8-p) \gamma]r_-^{8-p} + r_{+}^{8-p}\,,
\end{align}
with 
\begin{equation}
    \gamma = \frac{2(10-p)+2\alpha (9-p)(p+2\alpha -4)}{(9-p)(8-p)[2 + \alpha(p+2\alpha -4)]}\,.
\end{equation}
The total electric charge and ADM  mass of the black brane are
\begin{align}
    Q & =- \int_{\mathbb R^{p-1}} \dd[p-1]{x} \int_{S^{9-p}(t)} \mathbf \Upsilon= V_{p-1} \int_{S^{9-p}(t)} 4 e^{2 \alpha \phi} {\star}\mathbf F =4 Q_0 \Omega_{9-p} V_{p-1} e^{2 \alpha \phi} \Big|_{r=r_+} \,,\\
    M&= (9-p) \Omega_{9-p} V_{p-1} M_0\,,
\end{align}
with $ V_{p-1}  = \int_{\mathbb R^{p-1}} \dd{x^1} \cdots \dd{x^{p-1}}$ the    flat, spatial volume of the brane that is formally divergent. Further,    the electric potential difference between the horizon and infinity is  
\begin{equation}
    \bar \Phi  = \partial_{x^{p-1}} \cdot ({\dots} \cdot(\partial_{x^2} \cdot (\partial_{x^1} \cdot [\xi \cdot \mathbf A \big |_{r=\infty} - \xi \cdot \mathbf A \big |_{r=r_+}])))=  \frac{Q_0}{(8-p) \gamma_A r_-^{8-p}}[1-(f_-(r_+))^{\gamma_A}]\,.
    \end{equation}
    Note that the charge is integrated over the non-compact horizon cross-section $\mathcal{C} = \tilde{\mathcal{C}} \times \Sigma_{p-1} \cong S^{9-p} \times \mathbb R^{p-1}$, whereas the potential is a scalar. In the language of section \ref{sssec:non-compact}, this corresponds to the sector with $r=0$ and $k=p-1$, i.e., all of the non-compact legs are on the potential form. The full fibre integration defines the downgraded electric flux density on the compact subspace $\tilde {\mathcal C}$ as $\delta\tilde{\bm \Upsilon} = \int_{\mathbb R^{p-1}} \delta \bm \Upsilon$, which integrates on the fundamental cycle $S^{9-p}$, while the potential downgrades to a scalar potential on $\tilde{\mathcal{C}}$. The same story carries over to the spatial infinity.
    
Finally, the non-stationary comparison version of the first law for this charged black $p$-brane  reads 
\begin{equation}
    T\delta S_\text{dyn} = \delta M - \bar \Phi \delta Q\,,
\end{equation}
where the Hawking temperature and dynamical entropy are given by
\begin{align}
    T&=\frac{\kappa}{2\pi}  = \frac{8-p}{4\pi r_+} [f_-(r_+)]^{(\gamma_x - \gamma_r - 1)/2}\,,\\
    S_\text{dyn} & = (1 - \partial_t) \int_{\mathcal{C}(t)} \bm \epsilon_{\mathcal{C}} e^{-\phi}\,.
\end{align}
The mass, entropy and charge all carry the same divergent $V_{p-1}$ brane volume factor. We could instead work with finite densities 
\begin{equation}
    \varepsilon = \frac{M}{V_{p-1}}, \qquad s_\text{dyn} = \frac{S_\text{dyn}}{V_{p-1}}, \qquad \rho = \frac{Q}{V_{p-1}}\,,
\end{equation}
and write the first law as 
\begin{equation}
    T \delta s_\text{dyn} = \delta \varepsilon - \bar{\Phi} \delta \rho.
\end{equation}

\section{Outlook}

The present analysis provides a universal geometric framework for defining electromagnetic charges and potentials in the covariant phase space formalism and for deriving their contribution to the non-stationary first law of thermodynamics for black objects with non-compact horizons. Several natural extensions of this work remain to be explored. These directions would not only test the robustness of the framework but also deepen our understanding of electromagnetic charges for black objects    and of the entropy of dynamical horizons in higher derivative gravity. We conclude by listing a few promising avenues for future research:

\begin{itemize}
  \item \emph{Non-abelian gauge fields:} 
Generalise the present formalism to non-abelian gauge fields, allowing for black-object solutions carrying Yang-Mills charges and deriving the associated potential-charge terms in the first law. 
The first law for stationary perturbations of Einstein-Yang-Mills black holes has been studied in \cite{Sudarsky:1992ty,Heusler:1993cj,Corichi:2000dm,Ashtekar:2000hw,Gao:2003ys,Prabhu:2015vua}, and it is natural to extend this analysis to non-stationary perturbations. 
The principal-bundle approach developed by Prabhu~\cite{Prabhu:2015vua} is particularly promising, as it provides a fully gauge-covariant formulation of the first law that does not rely on a global choice of gauge. 
This addresses the topological obstructions that arise when gauge potentials cannot be defined   globally as smooth   fields on spacetime, as   in the presence of  magnetic monopoles. Even if such a global gauge choice does exist  for Yang-Mills potentials, it is not clear that such  a choice can be made  consistent with the stationarity condition $\mathcal L_\xi \mathbf A^{\text{YM}}=0$ \cite{Gao:2003ys}.
To avoid this, Prabhu works directly on the principal bundle, and unifies spacetime diffeomorphisms and internal gauge transformations as bundle automorphisms. Our formalism    is closely related in spirit, in that both treat potential-charge pairs as  gauge-covariant quantities rather than as functions of a chosen global gauge, but extends beyond 1-form gauge fields 
(albeit only for the  abelian case) and 
beyond the stationary regime by allowing for non-stationary perturbations. 
A natural next step is to formulate a fully gauge-covariant, non-stationary black hole first law for Yang-Mills fields by combining the dynamical entropy framework with Prabhu's bundle construction.

  \item \emph{Higher-derivative gravities:} Compute dynamical black hole entropy   for diffeomorphism-invariant Lagrangians that depend on the metric field, the Riemann tensor and its covariant derivatives. Even though the dynamical black hole entropy is formally defined by HWZ \cite{Hollands:2024vbe} in terms of the improved Noether charge, computing  the entropy formulae  explicitly for  certain classes of  higher-curvature theories of gravity is not straightforward. In our previous work we evaluated the improved Noether charge     for $f$(Riemann) theories of gravity and showed that it agrees with the general expression for dynamical black hole entropy  $S_{\text{dyn}}= (1- v \partial_v) S_{\text{Wall}}$, where $S_{\text{Wall}}$ was calculated originally by Wall  \cite{Wall:2015raa} for $f$(Riemann) theories. This has not been done for  diffeormphism-invariant Lagrangians  of the form $\mathbf L (g^{ab}, R_{abcd}, \nabla_a)$, nor is the Wall entropy known in those cases. The resulting Wall entropy  should be  compared with the Dong  holographic entanglement entropy for higher-derivative gravity theories  \cite{Dong:2013qoa,Miao:2014nxa}. For $f$(Riemann) theories the Dong and Wall entropy match, but perhaps this is no longer the case for these more general, higher-deriviative theories of gravity.

    \item \emph{Taub-NUT spacetimes:} 
An intriguing direction concerns the inclusion of gravitational magnetic charges, as realised in Taub-NUT spacetimes. The NUT parameter represents a magnetic-type gravitational charge, dual to the ordinary ADM mass in the same sense that magnetic charge is dual to electric charge in electromagnetism. In Lorentzian signature, its presence introduces a Misner string --- the gravitational analogue of the Dirac string of a magnetic monopole --- which reflects the nontrivial $\mathbb{R}$ fibration of the time coordinate over the angular two-sphere and complicates the standard definitions of globally conserved charges and potentials. Recent work~\cite{Hennigar:2019ive,Bordo:2019tyh} on Lorentzian formulations of Taub-NUT thermodynamics has proposed that the NUT charge and its conjugate potential should appear in the first law for the outer Killing horizon. However, it remains unclear what the correct horizon entropy of the Taub-NUT spacetime is --- whether it corresponds to the Bekenstein-Hawking area law or to a more general Noether charge expression~\cite{Fatibene:1999ys,Garfinkle:2000ms}. The covariant phase space formalism may clarify how, if at all, a gravitational magnetic potential-charge pair enters the first law and determine the appropriate entropy for Taub-NUT horizons.

  \item \emph{Extremal horizons:} Investigate the nature of the dynamical entropy for extremal black holes. Extremal horizons have vanishing surface gravity and hence zero temperature but finite entropy, and they lack a bifurcation surface. Consequently, the usual temperature-entropy term is absent from the first law. It is therefore unclear whether the HWZ dynamical black hole entropy formula extends consistently to extremal horizons, even though for general relativity it  applies to arbitrary horizon cross-sections (see \cite{Rignon-Bret:2023fjq} and section 5.1 in \cite{Visser:2024pwz}) and perhaps also in the extremal limit.

  \item \emph{Multi-horizon spacetimes:} Extend the dynamical gravitational entropy to configurations containing multiple disjoint Killing horizons, thereby lifting the assumption that the future horizon is connected. Examples include multi-center Einstein-Maxwell black holes such as the Majumdar-Papapetrou geometries~\cite{Majumdar:1947eu,Papaetrou:1947ib}, higher-dimensional composite objects like black Saturns, black di-rings and  bicycling black rings~\cite{Elvang:2007rd,Iguchi:2007is,Elvang:2007hg,Izumi:2007qx}, and de Sitter black holes with both black hole and cosmological horizons~\cite{Gibbons:1977mu,Kastor:1992nn} (for recent work on dynamical entropy for the de Sitter example, see \cite{Zhao:2025zny}). Each horizon component carries its own entropy, surface gravity  and potential-charge pair --- except for extremal horizons, for which the surface gravities vanishes. A covariant, dynamical formulation encompassing several interacting horizon components could clarify how individual entropies combine and whether a total dynamical entropy can be consistently defined for such multi-horizon systems.

  \item \emph{Comparison with thermodynamics of quasi-local horizons:} 
Clarify how the first law for the dynamical entropy of event horizons relates to the   first laws formulated for quasi-local horizons, such as apparent, trapping, or dynamical horizons~\cite{Collins:1992eca,HaywardPRD,Ashtekar:2002ag,Ashtekar:2004cn}. 
In the HWZ framework, the dynamical entropy of the event horizon agrees to first order with the Bekenstein-Hawking area-entropy of the apparent horizon, suggesting a close correspondence between these notions. 
A precise comparison requires understanding how the Killing vector field entering the event-horizon first law corresponds to the vector field tangent to the quasi-local horizon. It also remains open whether the relation $S_{\text{dyn}} = A_{\text{app}}/4G$ persists beyond first order, or if higher-order corrections distinguish the event horizon and apparent horizon entropies. 
A unified formulation could reveal whether these different horizons encode the same thermodynamics or represent distinct coarse grainings in the  underlying microscopic theory. 

  \item \emph{Fluid/gravity correspondence:} Future work may extend the fluid/gravity correspondence to higher-curvature theories of gravity, with the aim of determining how such corrections modify non-equilibrium entropy production, relaxation dynamics, and transport properties beyond Einstein gravity. This may allow  us to extend the dynamical entropy framework to non-linear   perturbations of equilibrium black hole physics, where the leading order correction to the equilibrium entropy should correspond to the dynamical entropy.  Building on the hydrodynamic reconstruction of Ricci-flat and AdS geometries~\cite{Compere:1103.3022,Skenderis:1404.4824}, one may develop a systematic gradient expansion incorporating higher-derivative interactions~\cite{Eling:2012xa,Eling:2011ct,Chandranathan:2022pfx} and extract the corresponding corrections to the holographic constitutive relations and entropy current. It would also be interesting to compare such results with those obtained in the membrane paradigm and large-$D$ expansions, where analogous higher-curvature corrections to horizon dynamics and entropy currents have been studied~\cite{Saha:2020zho}.

\end{itemize}

\section*{Acknowledgements} We would like to thank Gary Horowitz, Kostas Skenderis, Harvey Reall, Bob Wald and Aron Wall for interesting and helpful conversations. MRV is supported  by the Spinoza Grant of the
Dutch Science Organization (NWO) awarded to Klaas
Landsman. ZY is supported by an Internal Graduate Studentship of Trinity College, Cambridge and the AFOSR grant FA9550-19-1-0260 ``Tensor Networks and Holographic Spacetime''.

\appendix

\section{Differential form identities}
\label{appaaaa}

\begin{lem}\label{lem:contraction}
    For any $p$-forms $\mathbf{A}$ and $\mathbf{B}$ on spacetime $\mathcal{M}$ with volume form $\bm \epsilon$, 
    \begin{equation}
        A^{a_1 \cdots a_p} B_{a_1 \cdots a_p} \bm \epsilon = p!\,{\star \mathbf{A}} \wedge \mathbf{B}.
    \end{equation}
\end{lem}
\begin{proof}
    For $p$-forms $A_{a_1\cdots a_p}$ and $B_{a_1 \cdots a_p}$, we have 
\begin{equation}
    (\star A \wedge B)_{a_1 \cdots a_D} = \frac{D!}{p! (D-p)! p!} \epsilon_{[a_1 \cdots a_{D-p}|b_1 \cdots b_p|} A^{b_1 \cdots b_p} B_{a_{D-p+1} \cdots a_D]}
\end{equation}
so
\begin{equation}
    \begin{split}
        \star (\star A \wedge B) & = \frac{1}{p! (D-p)! p!} \epsilon_{a_1 \cdots a_D} \epsilon^{[a_1 \cdots a_{D-p}|b_1 \cdots b_p|} A_{b_1 \cdots b_p} B^{a_{D-p+1} \cdots a_D]}\\
        & = - \frac{1}{p!} \delta^{b_1}_{[a_{D-p+1}} \cdots \delta^{b_p}_{a_D]} A_{b_1 \cdots b_p} B^{a_{D-p+1} \cdots a_D}\\
        & = - \frac{1}{p!} A_{b_1 \cdots b_p} B^{b_1 \cdots b_p}
    \end{split}
\end{equation}
and we finally have 
\begin{equation}
    (\star A \wedge B)_{a_1 \cdots a_D} = \frac{1}{p!} A^{b_1 \cdots b_p} B_{b_1 \cdots b_p} \epsilon_{a_1 \cdots a_D}
\end{equation}
by using 
\begin{equation}
    \star (\star X_q) = (-1)^{1 + q (D-q)} X_q
\end{equation}
for any $q$-form $X_q$ in Lorentzian signature.
\end{proof}

\begin{cor}\label{cor:subcontraction}
    For any submanifold $\Sigma \subset \mathcal{M}$ with inclusion map $\iota : \Sigma \to \mathcal{M}$ and non-degenerate metric with volume form $\bm \epsilon_\Sigma$, the pullbacks of $p$-forms $\mathbf{A}$ and $\mathbf{B}$ to $\Sigma$ satisfy
    \begin{equation}
        (\iota^* A)^{i_1 \cdots i_p} (\iota^* B)_{i_1 \cdots i_p} \bm \epsilon_\Sigma = p!\, (\star_\Sigma \iota^* \mathbf{A}) \wedge \iota^* \mathbf{B}.
    \end{equation}
\end{cor}

\noindent The spacetime Hodge dual of differential forms can be decomposed in terms of intrinsic Hodge duals of submanifolds as long as the induced metric is non-degenerate. To consider such a decomposition on a null hypersurfaces, one has to carry out a \emph{double null decomposition} to the level of codimension-2 spatial leaves of the null surface because the null metric is degenerate, i.e., the null Hodge dual is undefined.

\begin{lem}[Double null decomposition of Hodge dual]\label{lem:dnd-hodge}
    In a $D$-dimensional Lorentzian spacetime $\mathcal{M}$ with volume form $\bm \epsilon$, consider a codimension-2 spacelike surface $\mathcal{C}$ with two null normal 1-forms $-\bm k$, $-\bm l$, and volume form $\bm \epsilon_\mathcal{C}$ with orientation given by $\bm \epsilon = \bm k \wedge \bm l \wedge \bm \epsilon_\mathcal{C}$. Then for any $p$-form $\mathbf{X}$ on $\mathcal{M}$, its Hodge dual can be decomposed as follows
    \begin{equation}
        \star \mathbf{X} = \bm k \wedge \biggl( \bm l \wedge (\star_\mathcal{C} \mathbf{X}) + (-1)^{D-p-1} {\star_\mathcal{C} (l \cdot \mathbf{X})} \biggr) + (-1)^{D-p} \bm l \wedge [\star_\mathcal{C}(k \cdot \mathbf{X})] + \star_\mathcal{C} (l \cdot k \cdot \mathbf{X})
    \end{equation}
    where the pre-Hodge dual on $\mathcal{C}$ (since we have not yet pulled back to $\mathcal{C}$) is defined as 
    \begin{equation}
        (\star_\mathcal{C} X)_{a_1 \cdots a_{D-p-2}} = \frac{1}{p!} (\epsilon_\mathcal{C})_{a_1 \cdots a_{D-2}} X^{a_{D-p-1} \cdots a_{D-2}}.
    \end{equation}
\end{lem}
\begin{proof}
    For convenience, we first define
    \begin{equation}
        \tilde{\bm \epsilon} = - \bm l \wedge \bm \epsilon_{\mathcal{C}}
    \end{equation}
    which is the volume form on the null hypersurface with normal $\bm k$. Hence, the spacetime volume form can be written as $\bm \epsilon = - \bm k \wedge \tilde{\bm \epsilon}$. 
    
    We can expand
    \begin{equation}
        \epsilon_{a_1 \cdots a_D}  = - (k \wedge \tilde \epsilon)_{a_1 \cdots a_D}= - D k_{[a_1} \tilde \epsilon_{a_2 \cdots a_D]}= - \sum_{m=1}^{D} (-1)^{m-1} k_{a_m} \tilde \epsilon_{a_1 \cdots \hat{a}_m \cdots a_D} \label{eq:vf-decomp-1}
    \end{equation}
    where `` $\hat{}$ '' means omission of an index.

    Similarly, we have
    \begin{equation}
        \tilde \epsilon_{a_1 \cdots a_{D-1}} = - \sum_{n=1}^{D-1} (-1)^{n-1} l_{a_n} (\epsilon_\mathcal{C})_{a_1 \cdots \hat a_n \cdots a_{D-1}}. \label{eq:vf-decomp-2}
    \end{equation}

    We first decompose the Hodge dual of $\mathbf{X}$ as 
    \begin{align*}
        (\star X)_{a_1 \cdots a_{D-p}} & = \frac{1}{p!} \epsilon_{a_1 \cdots a_{D-p} a_{D-p+1} \cdots a_D}X^{a_{D-p+1} \cdots a_D}\\
            & = - \frac{1}{p!} X^{a_{D-p+1} \cdots a_D} \sum_{m=1}^{D} (-1)^{m-1} k_{a_m} \tilde \epsilon_{a_1 \cdots \hat{a}_m \cdots a_D}\\
            & = - \frac{1}{p!} X^{a_{D-p+1} \cdots a_D} \sum_{m=1}^{D-p} (-1)^{m-1} k_{a_m} \tilde \epsilon_{a_1 \cdots \hat{a}_m \cdots a_{D-p} a_{D-p+1} \cdots a_D}\\
            &  \qquad - \frac{1}{p!} X^{a_{D-p+1} \cdots a_D} \sum_{m=D-p+1}^{D} (-1)^{m-1} k_{a_m} \tilde \epsilon_{a_1 \cdots a_{D-p} a_{D-p+1}  \cdots \hat{a}_m \cdots a_D}\\
            & = - \sum_{m=1}^{D-p} (-1)^{m-1} k_{a_m} (\tilde \star X)_{a_1 \cdots \hat a_m \cdots a_{D-p}}\\
            & \qquad - \frac{(-1)^{D-p}}{p!} \sum_{m=1}^{p} (-1)^{m-1} k_{b_m} X^{b_1 \cdots b_m \cdots b_p} \tilde \epsilon_{a_1 \cdots a_{D-p} b_1 \cdots \hat b_m \cdots b_p}\\
            & = - (D-p) k_{[a_1} (\tilde \star X)_{a_2 \cdots a_{D-p}]} - \frac{(-1)^{D-p}}{p!} p \tilde \epsilon_{a_1 \cdots a_{D-p} b_1 \cdots b_{p-1}} (k \cdot X)^{b_1 \cdots b_{p-1}}\\
            & = (- k \wedge \tilde \star X)_{a_1 \cdots a_{D-p}} - (-1)^{D-p} [\tilde \star (k \cdot X)]_{a_1 \cdots a_{D-p}} \,.\numberthis
    \end{align*}
  Expressed differently,
    \begin{equation}
        \star \mathbf{X} = - \bm k \wedge (\tilde \star \mathbf{X}) - (-1)^{D-p} {\tilde \star (k \cdot \mathbf{X})},
    \end{equation}
    where, for convenience, we have defined a ``pseudo-Hodge dual''
    \begin{equation}
        (\tilde \star X)_{a_1 \cdots a_{D-p-1}} = \frac{1}{p!} \tilde \epsilon_{a_1 \cdots a_{D-1}} X^{a_{D-p}\cdots a_{D-1}}
    \end{equation}
    where the spacetime metric is used for the contraction --- it is \emph{not} a well-defined intrinsic Hodge dual for the null surface.
    
    In the same spirit, we can expand 
    \begin{equation}
        \tilde \star \mathbf{X} = - \bm l \wedge (\star_\mathcal{C} \mathbf{X}) - (-1)^{D-p-1} {\star_\mathcal{C}(l \cdot \mathbf{X})}
    \end{equation}
    where we have used the definition of pre-Hodge dual on $\mathcal{C}$.

    We finally have 
    \begin{equation}
        \star \mathbf{X} = \bm k \wedge \biggl( \bm l \wedge (\star_\mathcal{C} \mathbf{X}) + (-1)^{D-p-1} {\star_\mathcal{C} (l \cdot \mathbf{X})} \biggr) + (-1)^{D-p} \bm l \wedge [\star_\mathcal{C}(k \cdot \mathbf{X})] + \star_\mathcal{C} (l \cdot k \cdot \mathbf{X}).
    \end{equation}
\end{proof}

\section{Derivation of constraint forms} \label{app:constraint-form}
\subsubsection{Derivation of \cref{eq:phi-A-constraint} and \cref{eq:constraint-form}}
Here, we review the derivation of the expression of the constraint form (\cref{eq:phi-A-constraint} and \cref{eq:constraint-form}) according to the method used  in \cite{Seifert:2006kv} (see also \cite{Iyer:1994ys}). Under a diffeomorphism generated by an arbitrary  vector field $\chi$, the Lagrangian varies as 
\begin{equation}
    \mathcal L_\chi \mathbf L = \frac{1}{2} \mathbf{E}_{(g)ab}(\phi) \mathcal{L}_\chi g^{ab} + \mathbf{E}_{(\psi)}^{a_1 \cdots a_s}(\phi) \mathcal{L}_\chi \psi_{a_1 \cdots a_s} + \mathbf{E}_{(A)}^{a_1 \cdots a_p}(\phi) \mathcal{L}_\chi A_{a_1 \cdots a_p} + \dd{\bm \Theta(\phi, \mathcal{L}_\chi \varphi)}.
\end{equation}
We use Cartan's magic formula $\mathcal{L}_\chi \mathbf{L} = \dd{(\chi \cdot \mathbf{L})}$ and the expressions of the Lie derivatives 
\begin{align}
    \mathcal{L}_\chi g^{ab} & = - 2 \nabla^{(a} \chi^{b)}\\
    \mathcal{L}_\chi \psi_{a_1 \cdots a_s} & = \chi^b \nabla_b \psi_{a_1 \cdots a_s} + \psi_{b a_2 \cdots a_s} \nabla_{a_1} \chi^b + \psi_{a_1 b a_3 \cdots a_s} \nabla_{a_2} \chi^b + \cdots + \psi_{a_1 \cdots a_{s-1} b} \nabla_{a_s} \chi^b \\
    \mathcal{L}_\chi A_{a_1 \cdots a_p} & = \chi^b \nabla_b A_{a_1 \cdots a_p} + A_{b a_2 \cdots a_p} \nabla_{a_1} \chi^b + A_{a_1 b a_3 \cdots a_p} \nabla_{a_2} \chi^b + \cdots + A_{a_1 \cdots a_{p-1} b} \nabla_{a_p} \chi^b
\end{align}
to obtain 
\begin{equation}
    \dd{\mathbf{J}_\chi}\equiv \dd{(\bm \Theta(\phi, \mathcal{L}_\chi \varphi) - \chi \cdot \mathbf{L})} = \bm \epsilon \left[ \left( E_{(g)}^{ab} - C_{(\psi)}^{ab} - C_{(A)}^{ab} \right) \nabla_a \chi_b - \chi^b P_{(\psi,A)b}  \right] \label{eq:dJ-chi}
\end{equation}
where, for convenience, we have defined
\begin{align}
        C^{ab}_{(\psi)} & \equiv E_{(\psi)}^{a a_2 \cdots a_s} \psi^{b}{}_{a_2 \cdots a_s} + E_{(\psi)}^{a_1 a a_3 \cdots a_s} \psi_{a_1}{}^{b}{}_{a_3 \cdots a_s} + \cdots + E_{(\psi)}^{a_1 \cdots a_{s-1}a} \psi_{a_1 \cdots a_{s-1}}{}^{b},\\
        C^{ab}_{(A)} & \equiv E_{(A)}^{a a_2 \cdots a_p} A^{b}{}_{a_2 \cdots a_p} + E_{(A)}^{a_1 a a_3 \cdots a_p} A_{a_1}{}^{b}{}_{a_3 \cdots a_p} + \cdots + E_{(A)}^{a_1 \cdots a_{p-1}a} A_{a_1 \cdots a_{p-1}}{}^{b}\\
        P_{(\psi,A)b} & \equiv  E^{a_1 \cdots a_s}_{(\psi)} \nabla_b \psi_{a_1 \cdots a_s}+ E^{a_1 \cdots a_p}_{(A)} \nabla_b A_{a_1 \cdots a_p}
    \end{align}
Taking the Hodge dual of \cref{eq:dJ-chi}, we find 
    \begin{align}
        - \nabla^a (\star J_\chi)_a & = \left( E_{(g)}^{ab} - C_{(\psi)}^{ab} - C_{(A)}^{ab} \right) \nabla_a \chi_b - \chi^b P_{(\psi,A)b}\\
     & = \nabla_a\left( \left( E_{(g)}^{ab} - C_{(\psi)}^{ab} - C_{(A)}^{ab} \right)  \chi_b \right) - \chi^b \left( P_{(\psi,A)b} + \nabla^a \left( E_{(g)ab} - C_{(\psi)ab} - C_{(A)ab} \right) \right). \nonumber
    \end{align}
Rearranging this yields
\begin{equation}
    \nabla^a \left( (\star J_\chi)_a +  \left( E_{(g)ab} - C_{(\psi)ab} - C_{(A)ab} \right) \chi^b\right) = \chi^b \left( P_{(\psi,A)b} + \nabla^a \left( E_{(g)ab} - C_{(\psi)ab} - C_{(A)ab} \right) \right).
\end{equation}
As the expression is off shell, we can take $\chi^b$ to be compactly supported and integrate over the entire spacetime to get 
\begin{equation}
    0 = \int \dd[D]{x} \sqrt{-g}\, \nabla^a (\cdots) = \int \dd[D]{x} \sqrt{-g}\, \chi^b \left( P_{(\psi,A)b} + \nabla^a \left( E_{(g)ab} - C_{(\psi)ab} - C_{(A)ab} \right) \right).
\end{equation}
As this should hold for any compactly supported $\chi^b$, we arrive at an identity:
\begin{equation}
    P_{(\psi,A)b} + \nabla^a \left( E_{(g)ab} - C_{(\psi)ab} - C_{(A)ab} \right) = 0. \label{eq:gen-Bianchi-A}
\end{equation}
This is the \emph{generalised Bianchi identity}. 

Furthermore, we have 
\begin{equation}
    \begin{split}
        & \quad~\frac{1}{2} \mathbf{E}_{(g)ab}(\phi) \mathcal{L}_\chi g^{ab} + \mathbf{E}_{(\psi)}^{a_1 \cdots a_s}(\phi) \mathcal{L}_\chi \psi_{a_1 \cdots a_s} + \mathbf{E}_{(A)}^{a_1 \cdots a_p}(\phi) \mathcal{L}_\chi A_{a_1 \cdots a_p} \\
        & = \bm \epsilon \left[ \left( - E_{(g)}^{ab} + C_{(\psi)}^{ab} + C_{(A)}^{ab} \right) \nabla_a \chi_b + \chi^b P_{(\psi,A)b}  \right]\\
        & = \bm \epsilon \nabla^a \left( \left( - E_{(g)ab} + C_{(\psi)ab} + C_{(A)ab} \right) \chi^b \right)\\
        & = \dd{\left( \left(-E_{(g)}^{ab} + C_{(\psi)}^{ab} + C_{(A)}^{ab}\right) \chi_b \bm \epsilon_a \right)}\\
        & = \dd{\mathbf{C}_\chi}
    \end{split}
\end{equation}
where we   used the generalised Bianchi identity \eqref{eq:gen-Bianchi-A} in the second equality, and
\begin{equation}
    \bm \epsilon \nabla^a X_a = (-1)^D \bm \epsilon ({\star} \dd{{\star} \mathbf X}) = (-1)^D {\star}{\star} (\dd{{\star} \mathbf X}) = \dd{\left((-1)^{D+1} {\star} \mathbf X\right)}
\end{equation}
for any 1-form $\mathbf X$ (in this case it is $X_a = \left( - E_{(g)ab} + C_{(\psi)ab} + C_{(A)ab} \right) \chi^b$), and 
\begin{equation}
    (-1)^{D+1} ({\star} X)_{a_1 \cdots a_{D-1}} = (-1)^{D+1} \epsilon_{a_1 \cdots a_{D-1}a} X^a = (-1)^{D+1} (-1)^{D-1} X^a \epsilon_{a a_1 \cdots a_{D-1}} = (X^a \epsilon_a)_{a_1 \cdots a_{D-1}}
\end{equation}
to establish the third equality.

We finally identify the constraint form as 
\begin{equation}
    \mathbf{C}_\chi = \left(-E_{(g)}^{ab} + C_{(\psi)}^{ab} + C_{(A)}^{ab}\right) \chi_b \bm \epsilon_a.
\end{equation}

\subsubsection{Derivation of \cref{eq:phi-F-constraint} and \cref{eq:F-constraint-form}}

Using the same methods as above but treating the field strength $F_{a_1 \cdots a_{p+1}}$ as fundamental, we have 
\begin{equation}
    \mathcal L_\chi \mathbf L = \frac{1}{2} \mathbf{E}_{(g)ab}(\phi) \mathcal{L}_\chi g^{ab} + \mathbf{E}_{(\psi)}^{a_1 \cdots a_s}(\phi) \mathcal{L}_\chi \psi_{a_1 \cdots a_s} + \mathbf Y^{a_1 \cdots a_{p+1}}(\phi) \mathcal{L}_\chi F_{a_1 \cdots a_{p+1}} + \dd{\bm \Theta^\text{GI}(\phi, \mathcal{L}_\chi \phi)}.
\end{equation}
under a diffeomorphism generated by $\chi$. Here, for convenience, we   defined
\begin{equation}
    Y^{a_1 \cdots a_{p+1}} = \frac{1}{\sqrt{-g}}\frac{\delta I}{\delta F_{a_1 \cdots a_{p+1}}} \qquad \text{and} \qquad \mathbf Y^{a_1 \cdots a_{p+1}} = Y^{a_1 \cdots a_{p+1}} \bm \epsilon
\end{equation}
and we notice that $\mathbf Y$ (as a differential form itself) relates to $\bm \Upsilon$ by 
\begin{equation}
    \bm \Upsilon = (p+1)! {\star} \mathbf Y.
\end{equation}
Repeating the same procedure as above but instead using 
\begin{equation}
    \mathcal L_\chi F_{a_1 \cdots a_{p+1}} = \chi^b \nabla_b F_{a_1 \cdots a_{p+1}} + F_{b a_2 \cdots a_{p+1}} \nabla_{a_1} \chi^b + F_{a_1 b a_3 \cdots a_{p+1}} \nabla_{a_2}\chi^b + \cdots + F_{a_1 \cdots a_p b} \nabla_{a_{p+1}} \chi^b 
\end{equation}
we find 
\begin{equation}
    \dd{\mathbf J^\text{GI}_\chi} \equiv \dd{(\bm \Theta^\text{GI}(\phi, \mathcal L_\chi \phi) - \chi \cdot \mathbf L)} = \bm \epsilon \left[\left(E_{(g)}^{ab} - C_{(\psi)}^{ab} - c_{(F)}^{ab}\right)\nabla_a \chi_b - \chi^b P_{(\psi, F)b}\right]
\end{equation}
where $C_{(\psi)}^{ab}$ is the same as before, and 
\begin{align}
    c_{(F)}^{ab} & \equiv Y^{a a_2 \cdots a_{p+1}} F^b{}_{a_2 \cdots a_{p+1}} + Y^{a_1 a a_3 \cdots a_{p+1}} F_{a_1}{}^{b}{}_{a_3 \cdots a_{p+1}} + \cdots + Y^{a_1 \cdots a_p a} F_{a_1 \cdots a_p}{}^b\,,\\
    P_{(\psi, F)b} & \equiv E_{(\psi)}^{a_1 \cdots a_s} \nabla_b \psi_{a_1 \cdots a_s} + Y^{a_1 \cdots a_{p+1}}\nabla_b F_{a_1 \cdots a_{p+1}}.
\end{align}
By the same argument, we have
\begin{equation}
    \nabla^a \left( (\star J^\text{GI}_\chi)_a +  \left( E_{(g)ab} - C_{(\psi)ab} - c_{(F)ab} \right) \chi^b\right) = \chi^b \left( P_{(\psi,F)b} + \nabla^a \left( E_{(g)ab} - C_{(\psi)ab} - c_{(F)ab} \right) \right),
\end{equation}
which integrates to zero for arbitrary compactly supported $\chi^b$. Hence, we obtain another generalised Bianchi identity, treating $\mathbf F$ as  fundamental:
\begin{equation}
    P_{(\psi,F)b} + \nabla^a \left( E_{(g)ab} - C_{(\psi)ab} - c_{(F)ab} \right) = 0.
\end{equation}
Finally, we obtain 
\begin{equation}
    \begin{split}
        & \quad~\frac{1}{2} \mathbf{E}_{(g)ab}(\phi) \mathcal{L}_\chi g^{ab} + \mathbf{E}_{(\psi)}^{a_1 \cdots a_s}(\phi) \mathcal{L}_\chi \psi_{a_1 \cdots a_s} + \mathbf Y^{a_1 \cdots a_{p+1}}(\phi) \mathcal{L}_\chi F_{a_1 \cdots a_{p+1}} \\
        & = \bm \epsilon \left[ \left( - E_{(g)}^{ab} + C_{(\psi)}^{ab} + c_{(F)}^{ab} \right) \nabla_a \chi_b + \chi^b P_{(\psi,F)b}  \right]\\
        & = \bm \epsilon \nabla^a \left( \left( - E_{(g)ab} + C_{(\psi)ab} + c_{(F)ab} \right) \chi^b \right)\\
        & = \dd{\left( \left(-E_{(g)}^{ab} + C_{(\psi)}^{ab} + c_{(F)}^{ab}\right) \chi_b \bm \epsilon_a \right)}\\
        & = \dd{\mathbf{c}^{(F)}_\chi}
    \end{split}
\end{equation}
where we identify the ``would-be constraint form'' for $\mathbf F$ as 
\begin{equation}
    \mathbf{c}^{(F)}_\chi = \left(-E_{(g)}^{ab} + C_{(\psi)}^{ab} + c_{(F)}^{ab}\right) \chi_b \bm \epsilon_a.
\end{equation}

\section{Closedness of downgraded  forms}
\label{appc}
Below we prove the two equations in  \eqref{downgradedoncompactsubspace}.

  \subsubsection{Closedness of downgraded electric flux density perturbation form}
    First, we show that 
    \begin{equation}
        \dd{(\delta \tilde{\bm \Upsilon}^{\aleph_{r+1} \cdots \aleph_k}_{(D-p-r-1)})} \tceq 0.
    \end{equation}
    Consider 
    \begin{equation}
        \begin{split}
            &  \dd{[((\delta \bm \Upsilon \cdot m_{\aleph_r}) \cdot \dots) \cdot m_{\aleph_1}]}\\
            & \overset{\tilde{\mathcal{C}}}{\propto} \dd{[m_{\aleph_r}  \cdot ({\dots} \cdot (m_{\aleph_{1}} \cdot \delta \bm \Upsilon))]}\\
            & \tceq \mathcal{L}_{m_{\aleph_r}}[m_{\aleph_{r-1}}  \cdot ({\dots} \cdot (m_{\aleph_{1}} \cdot \delta \bm \Upsilon))] - m_{\aleph_r}  \cdot \dd{[m_{\aleph_{r-1}}  \cdot ({\dots} \cdot (m_{\aleph_{1}} \cdot \delta \bm \Upsilon))]}\\
            & \tceq \cdots\\
            & \tceq \mathcal{L}_{m_{\aleph_r}}(\cdots) + \mathcal{L}_{m_{\aleph_{r-1}}}(\cdots) + \cdots + \mathcal{L}_{m_{\aleph_1}}(\cdots) + (-1)^r m_{\aleph_r}  \cdot ({\dots} \cdot (m_{\aleph_{1}} \cdot \underbrace{\dd{(\delta \bm \Upsilon)}}_{0}))\\
            & \tceq \mathcal{L}_{m_{\aleph_r}}(\cdots) + \mathcal{L}_{m_{\aleph_{r-1}}}(\cdots) + \cdots + \mathcal{L}_{m_{\aleph_1}}(\cdots)\,,
        \end{split}
    \end{equation}
    where we   repeatedly used   Cartan's magic formula.

    Then, we have 
    \begin{equation}
        \begin{split}
            \dd{(\delta \tilde{\bm \Upsilon}^{\aleph_{r+1} \cdots \aleph_k}_{(D-p-r-1)})} & \tceq \frac{1}{r!}\,\iota^* \int_{\Sigma_k} \left[  \mathcal{L}_{m_{\aleph_r}}(\cdots) + \mathcal{L}_{m_{\aleph_{r-1}}}(\cdots) + \cdots + \mathcal{L}_{m_{\aleph_1}}(\cdots)\right] \dd{y^{\aleph_1}} \cdots \dd{y^{\aleph_k}}\\
            & \tceq \left[ \cdots \right]_{y^{\aleph_1} = -\infty}^\infty + \left[ \cdots \right]_{y^{\aleph_2} = -\infty}^\infty + \cdots + \left[ \cdots \right]_{y^{\aleph_r} = -\infty}^\infty \\
            & \tceq 0\,.
        \end{split}
    \end{equation}
   where we used that  the pullback $\iota^*$ commutes with the exterior derivative. In the last equality we imposed    the asymptotic boundary condition (see number 4 in section \ref{sec:falloff}) that $\delta \phi \to \delta \phi_\infty=\text{const.}$ and $\mathcal L_{m_\aleph} \delta \phi_\infty \to 0$ as $y^{\aleph} \to \pm \infty$.

    \subsubsection{Closedness   of downgraded potential form}

    Next, we show that 
    \begin{equation}
        \dd{\tilde{\bm \Phi}^{(p+r-k-1)}_{\aleph_{r+1}\cdots \aleph_k}} \tceq 0.
    \end{equation}
    To begin with, we use the fact that $\bm \Phi$ does not dependent on the coordinates $y^\aleph$ in the extended directions (see equation \eqref{constpotgaugeoriginal})  to obtain 
    \begin{equation}
        \mathcal{L}_{m_\aleph} \bm \Phi \fheq \dd{(m_\aleph \cdot \bm \Phi)} + m_\aleph \cdot \dd{\bm \Phi} \fheq 0\,.
    \end{equation}
    We also notice the action of Lie derivative commutes with both the exterior derivative and the interior product. So the operation of $\dd$ and $m_\aleph \cdot{}$ on $\bm \Phi$ anticommute on the horizon. Then we have,
    \begin{equation}
        \begin{split}
            \dd{\tilde{\bm \Phi}^{(p+r-k-1)}_{\aleph_{r+1}\cdots \aleph_k}} & \tceq \frac{1}{(k-r)!}\,\dd{\left( \iota^*[m_{\aleph_k}  \cdot (m_{\aleph_{k-1}} \cdot ({\dots} \cdot (m_{\aleph_{r+1}} \cdot \bm \Phi)))] \right)}\\
            & \tceq \frac{1}{(k-r)!}\,\iota^*\dd{[m_{\aleph_k}  \cdot (m_{\aleph_{k-1}} \cdot ({\dots} \cdot (m_{\aleph_{r+1}} \cdot \bm \Phi)))]}\\
            & \tceq - \frac{1}{(k-r)!}\,\iota^*[m_{\aleph_k}  \cdot \dd{(m_{\aleph_{k-1}} \cdot ({\dots} \cdot (m_{\aleph_{r+1}} \cdot \bm \Phi)))}]\\
            & \tceq \frac{(-1)^{k-r}}{(k-r)!}\,\iota^*[m_{\aleph_k}  \cdot (m_{\aleph_{k-1}} \cdot ({\dots} \cdot (m_{\aleph_{r+1}} \cdot \dd{\bm \Phi})))]\\
            & \tceq 0\,,
        \end{split}
    \end{equation}
    where   the first equality follows from the fact that the pullback $\iota^*$ commutes with the exterior derivative $\dd$.  


\bibliography{biblio.bib}

\begin{thebibliography}{91}%
\makeatletter
\providecommand \@ifxundefined [1]{%
 \@ifx{#1\undefined}
}%
\providecommand \@ifnum [1]{%
 \ifnum #1\expandafter \@firstoftwo
 \else \expandafter \@secondoftwo
 \fi
}%
\providecommand \@ifx [1]{%
 \ifx #1\expandafter \@firstoftwo
 \else \expandafter \@secondoftwo
 \fi
}%
\providecommand \natexlab [1]{#1}%
\providecommand \enquote  [1]{``#1''}%
\providecommand \bibnamefont  [1]{#1}%
\providecommand \bibfnamefont [1]{#1}%
\providecommand \citenamefont [1]{#1}%
\providecommand \href@noop [0]{\@secondoftwo}%
\providecommand \href [0]{\begingroup \@sanitize@url \@href}%
\providecommand \@href[1]{\@@startlink{#1}\@@href}%
\providecommand \@@href[1]{\endgroup#1\@@endlink}%
\providecommand \@sanitize@url [0]{\catcode `\\12\catcode `\$12\catcode
  `\&12\catcode `\#12\catcode `\^12\catcode `\_12\catcode `\%12\relax}%
\providecommand \@@startlink[1]{}%
\providecommand \@@endlink[0]{}%
\providecommand \url  [0]{\begingroup\@sanitize@url \@url }%
\providecommand \@url [1]{\endgroup\@href {#1}{\urlprefix }}%
\providecommand \urlprefix  [0]{URL }%
\providecommand \Eprint [0]{\href }%
\providecommand \doibase [0]{https://doi.org/}%
\providecommand \selectlanguage [0]{\@gobble}%
\providecommand \bibinfo  [0]{\@secondoftwo}%
\providecommand \bibfield  [0]{\@secondoftwo}%
\providecommand \translation [1]{[#1]}%
\providecommand \BibitemOpen [0]{}%
\providecommand \bibitemStop [0]{}%
\providecommand \bibitemNoStop [0]{.\EOS\space}%
\providecommand \EOS [0]{\spacefactor3000\relax}%
\providecommand \BibitemShut  [1]{\csname bibitem#1\endcsname}%
\let\auto@bib@innerbib\@empty
\bibitem [{\citenamefont {Hollands}\ \emph {et~al.}(2024)\citenamefont
  {Hollands}, \citenamefont {Wald},\ and\ \citenamefont
  {Zhang}}]{Hollands:2024vbe}%
  \BibitemOpen
  \bibfield  {author} {\bibinfo {author} {\bibfnamefont {S.}~\bibnamefont
  {Hollands}}, \bibinfo {author} {\bibfnamefont {R.~M.}\ \bibnamefont {Wald}},\
  and\ \bibinfo {author} {\bibfnamefont {V.~G.}\ \bibnamefont {Zhang}},\
  }\bibfield  {title} {\bibinfo {title} {{Entropy of dynamical black holes}},\
  }\href {https://doi.org/10.1103/PhysRevD.110.024070} {\bibfield  {journal}
  {\bibinfo  {journal} {Phys. Rev. D}\ }\textbf {\bibinfo {volume} {110}},\
  \bibinfo {pages} {024070} (\bibinfo {year} {2024})},\ \Eprint
  {https://arxiv.org/abs/2402.00818} {arXiv:2402.00818 [hep-th]} \BibitemShut
  {NoStop}%
\bibitem [{\citenamefont {Bekenstein}(1973)}]{Bekenstein:1973ur}%
  \BibitemOpen
  \bibfield  {author} {\bibinfo {author} {\bibfnamefont {J.~D.}\ \bibnamefont
  {Bekenstein}},\ }\bibfield  {title} {\bibinfo {title} {{Black holes and
  entropy}},\ }\href {https://doi.org/10.1103/PhysRevD.7.2333} {\bibfield
  {journal} {\bibinfo  {journal} {Phys. Rev. D}\ }\textbf {\bibinfo {volume}
  {7}},\ \bibinfo {pages} {2333} (\bibinfo {year} {1973})}\BibitemShut
  {NoStop}%
\bibitem [{\citenamefont {Hawking}(1975)}]{Hawking:1975vcx}%
  \BibitemOpen
  \bibfield  {author} {\bibinfo {author} {\bibfnamefont {S.~W.}\ \bibnamefont
  {Hawking}},\ }\bibfield  {title} {\bibinfo {title} {{Particle Creation by
  Black Holes}},\ }\href {https://doi.org/10.1007/BF02345020} {\bibfield
  {journal} {\bibinfo  {journal} {Commun. Math. Phys.}\ }\textbf {\bibinfo
  {volume} {43}},\ \bibinfo {pages} {199} (\bibinfo {year} {1975})},\ \bibinfo
  {note} {[Erratum: Commun.Math.Phys. 46, 206 (1976)]}\BibitemShut {NoStop}%
\bibitem [{\citenamefont {Hawking}(1971)}]{Hawking:1971tu}%
  \BibitemOpen
  \bibfield  {author} {\bibinfo {author} {\bibfnamefont {S.~W.}\ \bibnamefont
  {Hawking}},\ }\bibfield  {title} {\bibinfo {title} {{Gravitational radiation
  from colliding black holes}},\ }\href
  {https://doi.org/10.1103/PhysRevLett.26.1344} {\bibfield  {journal} {\bibinfo
   {journal} {Phys. Rev. Lett.}\ }\textbf {\bibinfo {volume} {26}},\ \bibinfo
  {pages} {1344} (\bibinfo {year} {1971})}\BibitemShut {NoStop}%
\bibitem [{\citenamefont {Bardeen}\ \emph {et~al.}(1973)\citenamefont
  {Bardeen}, \citenamefont {Carter},\ and\ \citenamefont
  {Hawking}}]{Bardeen:1973gs}%
  \BibitemOpen
  \bibfield  {author} {\bibinfo {author} {\bibfnamefont {J.~M.}\ \bibnamefont
  {Bardeen}}, \bibinfo {author} {\bibfnamefont {B.}~\bibnamefont {Carter}},\
  and\ \bibinfo {author} {\bibfnamefont {S.~W.}\ \bibnamefont {Hawking}},\
  }\bibfield  {title} {\bibinfo {title} {{The Four laws of black hole
  mechanics}},\ }\href {https://doi.org/10.1007/BF01645742} {\bibfield
  {journal} {\bibinfo  {journal} {Commun. Math. Phys.}\ }\textbf {\bibinfo
  {volume} {31}},\ \bibinfo {pages} {161} (\bibinfo {year} {1973})}\BibitemShut
  {NoStop}%
\bibitem [{\citenamefont {Hayward}(1994)}]{HaywardPRD}%
  \BibitemOpen
  \bibfield  {author} {\bibinfo {author} {\bibfnamefont {S.~A.}\ \bibnamefont
  {Hayward}},\ }\bibfield  {title} {\bibinfo {title} {General laws of
  black-hole dynamics},\ }\href {https://doi.org/10.1103/PhysRevD.49.6467}
  {\bibfield  {journal} {\bibinfo  {journal} {Phys. Rev. D}\ }\textbf {\bibinfo
  {volume} {49}},\ \bibinfo {pages} {6467} (\bibinfo {year}
  {1994})}\BibitemShut {NoStop}%
\bibitem [{\citenamefont {Hawking}\ and\ \citenamefont
  {Ellis}(2023)}]{Hawking:1973uf}%
  \BibitemOpen
  \bibfield  {author} {\bibinfo {author} {\bibfnamefont {S.~W.}\ \bibnamefont
  {Hawking}}\ and\ \bibinfo {author} {\bibfnamefont {G.~F.~R.}\ \bibnamefont
  {Ellis}},\ }\href {https://doi.org/10.1017/9781009253161} {\emph {\bibinfo
  {title} {{The Large Scale Structure of Space-Time}}}},\ Cambridge Monographs
  on Mathematical Physics\ (\bibinfo  {publisher} {Cambridge University
  Press},\ \bibinfo {year} {2023})\BibitemShut {NoStop}%
\bibitem [{\citenamefont {Ashtekar}\ and\ \citenamefont
  {Krishnan}(2002)}]{Ashtekar:2002ag}%
  \BibitemOpen
  \bibfield  {author} {\bibinfo {author} {\bibfnamefont {A.}~\bibnamefont
  {Ashtekar}}\ and\ \bibinfo {author} {\bibfnamefont {B.}~\bibnamefont
  {Krishnan}},\ }\bibfield  {title} {\bibinfo {title} {{Dynamical horizons:
  Energy, angular momentum, fluxes and balance laws}},\ }\href
  {https://doi.org/10.1103/PhysRevLett.89.261101} {\bibfield  {journal}
  {\bibinfo  {journal} {Phys. Rev. Lett.}\ }\textbf {\bibinfo {volume} {89}},\
  \bibinfo {pages} {261101} (\bibinfo {year} {2002})},\ \Eprint
  {https://arxiv.org/abs/gr-qc/0207080} {arXiv:gr-qc/0207080} \BibitemShut
  {NoStop}%
\bibitem [{\citenamefont {Ashtekar}\ and\ \citenamefont
  {Krishnan}(2004)}]{Ashtekar:2004cn}%
  \BibitemOpen
  \bibfield  {author} {\bibinfo {author} {\bibfnamefont {A.}~\bibnamefont
  {Ashtekar}}\ and\ \bibinfo {author} {\bibfnamefont {B.}~\bibnamefont
  {Krishnan}},\ }\bibfield  {title} {\bibinfo {title} {{Isolated and dynamical
  horizons and their applications}},\ }\href
  {https://doi.org/10.12942/lrr-2004-10} {\bibfield  {journal} {\bibinfo
  {journal} {Living Rev. Rel.}\ }\textbf {\bibinfo {volume} {7}},\ \bibinfo
  {pages} {10} (\bibinfo {year} {2004})},\ \Eprint
  {https://arxiv.org/abs/gr-qc/0407042} {arXiv:gr-qc/0407042} \BibitemShut
  {NoStop}%
\bibitem [{\citenamefont {Bousso}\ and\ \citenamefont
  {Engelhardt}(2015)}]{Bousso:2015mqa}%
  \BibitemOpen
  \bibfield  {author} {\bibinfo {author} {\bibfnamefont {R.}~\bibnamefont
  {Bousso}}\ and\ \bibinfo {author} {\bibfnamefont {N.}~\bibnamefont
  {Engelhardt}},\ }\bibfield  {title} {\bibinfo {title} {{New Area Law in
  General Relativity}},\ }\href
  {https://doi.org/10.1103/PhysRevLett.115.081301} {\bibfield  {journal}
  {\bibinfo  {journal} {Phys. Rev. Lett.}\ }\textbf {\bibinfo {volume} {115}},\
  \bibinfo {pages} {081301} (\bibinfo {year} {2015})},\ \Eprint
  {https://arxiv.org/abs/1504.07627} {arXiv:1504.07627 [hep-th]} \BibitemShut
  {NoStop}%
\bibitem [{\citenamefont {Kelly}\ and\ \citenamefont
  {Wall}(2014)}]{Kelly:2013aja}%
  \BibitemOpen
  \bibfield  {author} {\bibinfo {author} {\bibfnamefont {W.~R.}\ \bibnamefont
  {Kelly}}\ and\ \bibinfo {author} {\bibfnamefont {A.~C.}\ \bibnamefont
  {Wall}},\ }\bibfield  {title} {\bibinfo {title} {{Coarse-grained entropy and
  causal holographic information in AdS/CFT}},\ }\href
  {https://doi.org/10.1007/JHEP03(2014)118} {\bibfield  {journal} {\bibinfo
  {journal} {JHEP}\ }\textbf {\bibinfo {volume} {03}},\ \bibinfo {pages}
  {118}},\ \Eprint {https://arxiv.org/abs/1309.3610} {arXiv:1309.3610 [hep-th]}
  \BibitemShut {NoStop}%
\bibitem [{\citenamefont {Engelhardt}\ and\ \citenamefont
  {Wall}(2018)}]{Engelhardt:2017aux}%
  \BibitemOpen
  \bibfield  {author} {\bibinfo {author} {\bibfnamefont {N.}~\bibnamefont
  {Engelhardt}}\ and\ \bibinfo {author} {\bibfnamefont {A.~C.}\ \bibnamefont
  {Wall}},\ }\bibfield  {title} {\bibinfo {title} {{Decoding the Apparent
  Horizon: Coarse-Grained Holographic Entropy}},\ }\href
  {https://doi.org/10.1103/PhysRevLett.121.211301} {\bibfield  {journal}
  {\bibinfo  {journal} {Phys. Rev. Lett.}\ }\textbf {\bibinfo {volume} {121}},\
  \bibinfo {pages} {211301} (\bibinfo {year} {2018})},\ \Eprint
  {https://arxiv.org/abs/1706.02038} {arXiv:1706.02038 [hep-th]} \BibitemShut
  {NoStop}%
\bibitem [{\citenamefont {Engelhardt}\ and\ \citenamefont
  {Wall}(2019)}]{Engelhardt:2018kcs}%
  \BibitemOpen
  \bibfield  {author} {\bibinfo {author} {\bibfnamefont {N.}~\bibnamefont
  {Engelhardt}}\ and\ \bibinfo {author} {\bibfnamefont {A.~C.}\ \bibnamefont
  {Wall}},\ }\bibfield  {title} {\bibinfo {title} {{Coarse Graining Holographic
  Black Holes}},\ }\href {https://doi.org/10.1007/JHEP05(2019)160} {\bibfield
  {journal} {\bibinfo  {journal} {JHEP}\ }\textbf {\bibinfo {volume} {05}},\
  \bibinfo {pages} {160}},\ \Eprint {https://arxiv.org/abs/1806.01281}
  {arXiv:1806.01281 [hep-th]} \BibitemShut {NoStop}%
\bibitem [{\citenamefont {Sudarsky}\ and\ \citenamefont
  {Wald}(1992)}]{Sudarsky:1992ty}%
  \BibitemOpen
  \bibfield  {author} {\bibinfo {author} {\bibfnamefont {D.}~\bibnamefont
  {Sudarsky}}\ and\ \bibinfo {author} {\bibfnamefont {R.~M.}\ \bibnamefont
  {Wald}},\ }\bibfield  {title} {\bibinfo {title} {{Extrema of mass,
  stationarity, and staticity, and solutions to the Einstein Yang-Mills
  equations}},\ }\href {https://doi.org/10.1103/PhysRevD.46.1453} {\bibfield
  {journal} {\bibinfo  {journal} {Phys. Rev. D}\ }\textbf {\bibinfo {volume}
  {46}},\ \bibinfo {pages} {1453} (\bibinfo {year} {1992})}\BibitemShut
  {NoStop}%
\bibitem [{\citenamefont {Wald}(1993{\natexlab{a}})}]{Wald:1993ki}%
  \BibitemOpen
  \bibfield  {author} {\bibinfo {author} {\bibfnamefont {R.~M.}\ \bibnamefont
  {Wald}},\ }\bibfield  {title} {\bibinfo {title} {{The First law of black hole
  mechanics}},\ }in\ \href@noop {} {\emph {\bibinfo {booktitle} {{Directions in
  General Relativity: An International Symposium in Honor of the 60th Birthdays
  of Dieter Brill and Charles Misner}}}}\ (\bibinfo {year} {1993})\ \Eprint
  {https://arxiv.org/abs/gr-qc/9305022} {arXiv:gr-qc/9305022} \BibitemShut
  {NoStop}%
\bibitem [{\citenamefont {Wald}(1993{\natexlab{b}})}]{Wald:1993nt}%
  \BibitemOpen
  \bibfield  {author} {\bibinfo {author} {\bibfnamefont {R.~M.}\ \bibnamefont
  {Wald}},\ }\bibfield  {title} {\bibinfo {title} {{Black hole entropy is the
  Noether charge}},\ }\href {https://doi.org/10.1103/PhysRevD.48.R3427}
  {\bibfield  {journal} {\bibinfo  {journal} {Phys. Rev. D}\ }\textbf {\bibinfo
  {volume} {48}},\ \bibinfo {pages} {R3427} (\bibinfo {year}
  {1993}{\natexlab{b}})},\ \Eprint {https://arxiv.org/abs/gr-qc/9307038}
  {arXiv:gr-qc/9307038} \BibitemShut {NoStop}%
\bibitem [{\citenamefont {Iyer}\ and\ \citenamefont
  {Wald}(1994)}]{Iyer:1994ys}%
  \BibitemOpen
  \bibfield  {author} {\bibinfo {author} {\bibfnamefont {V.}~\bibnamefont
  {Iyer}}\ and\ \bibinfo {author} {\bibfnamefont {R.~M.}\ \bibnamefont
  {Wald}},\ }\bibfield  {title} {\bibinfo {title} {{Some properties of Noether
  charge and a proposal for dynamical black hole entropy}},\ }\href
  {https://doi.org/10.1103/PhysRevD.50.846} {\bibfield  {journal} {\bibinfo
  {journal} {Phys. Rev. D}\ }\textbf {\bibinfo {volume} {50}},\ \bibinfo
  {pages} {846} (\bibinfo {year} {1994})},\ \Eprint
  {https://arxiv.org/abs/gr-qc/9403028} {arXiv:gr-qc/9403028} \BibitemShut
  {NoStop}%
\bibitem [{\citenamefont {Rignon-Bret}(2023)}]{Rignon-Bret:2023fjq}%
  \BibitemOpen
  \bibfield  {author} {\bibinfo {author} {\bibfnamefont {A.}~\bibnamefont
  {Rignon-Bret}},\ }\bibfield  {title} {\bibinfo {title} {{Second law from the
  Noether current on null hypersurfaces}},\ }\href
  {https://doi.org/10.1103/PhysRevD.108.044069} {\bibfield  {journal} {\bibinfo
   {journal} {Phys. Rev. D}\ }\textbf {\bibinfo {volume} {108}},\ \bibinfo
  {pages} {044069} (\bibinfo {year} {2023})},\ \Eprint
  {https://arxiv.org/abs/2303.07262} {arXiv:2303.07262 [gr-qc]} \BibitemShut
  {NoStop}%
\bibitem [{\citenamefont {Kar}\ \emph {et~al.}(2024)\citenamefont {Kar},
  \citenamefont {Dhivakar}, \citenamefont {Roy}, \citenamefont {Panda},\ and\
  \citenamefont {Shaikh}}]{Kar:2024dqk}%
  \BibitemOpen
  \bibfield  {author} {\bibinfo {author} {\bibfnamefont {A.}~\bibnamefont
  {Kar}}, \bibinfo {author} {\bibfnamefont {P.}~\bibnamefont {Dhivakar}},
  \bibinfo {author} {\bibfnamefont {S.}~\bibnamefont {Roy}}, \bibinfo {author}
  {\bibfnamefont {B.}~\bibnamefont {Panda}},\ and\ \bibinfo {author}
  {\bibfnamefont {A.}~\bibnamefont {Shaikh}},\ }\bibfield  {title} {\bibinfo
  {title} {{Iyer-Wald ambiguities and gauge covariance of Entropy current in
  Higher derivative theories of gravity}},\ }\href
  {https://doi.org/10.1007/JHEP07(2024)016} {\bibfield  {journal} {\bibinfo
  {journal} {JHEP}\ }\textbf {\bibinfo {volume} {07}},\ \bibinfo {pages}
  {016}},\ \Eprint {https://arxiv.org/abs/2403.04749} {arXiv:2403.04749
  [hep-th]} \BibitemShut {NoStop}%
\bibitem [{\citenamefont {Visser}\ and\ \citenamefont
  {Yan}(2024)}]{Visser:2024pwz}%
  \BibitemOpen
  \bibfield  {author} {\bibinfo {author} {\bibfnamefont {M.~R.}\ \bibnamefont
  {Visser}}\ and\ \bibinfo {author} {\bibfnamefont {Z.}~\bibnamefont {Yan}},\
  }\bibfield  {title} {\bibinfo {title} {{Properties of dynamical black hole
  entropy}},\ }\href {https://doi.org/10.1007/JHEP10(2024)029} {\bibfield
  {journal} {\bibinfo  {journal} {JHEP}\ }\textbf {\bibinfo {volume} {10}},\
  \bibinfo {pages} {029}},\ \Eprint {https://arxiv.org/abs/2403.07140}
  {arXiv:2403.07140 [hep-th]} \BibitemShut {NoStop}%
\bibitem [{\citenamefont {Ciambelli}\ \emph {et~al.}(2024)\citenamefont
  {Ciambelli}, \citenamefont {Freidel},\ and\ \citenamefont
  {Leigh}}]{Ciambelli:2023mir}%
  \BibitemOpen
  \bibfield  {author} {\bibinfo {author} {\bibfnamefont {L.}~\bibnamefont
  {Ciambelli}}, \bibinfo {author} {\bibfnamefont {L.}~\bibnamefont {Freidel}},\
  and\ \bibinfo {author} {\bibfnamefont {R.~G.}\ \bibnamefont {Leigh}},\
  }\bibfield  {title} {\bibinfo {title} {{Null Raychaudhuri: canonical
  structure and the dressing time}},\ }\href
  {https://doi.org/10.1007/JHEP01(2024)166} {\bibfield  {journal} {\bibinfo
  {journal} {JHEP}\ }\textbf {\bibinfo {volume} {01}},\ \bibinfo {pages}
  {166}},\ \Eprint {https://arxiv.org/abs/2309.03932} {arXiv:2309.03932
  [hep-th]} \BibitemShut {NoStop}%
\bibitem [{\citenamefont {Wall}(2015)}]{Wall:2015raa}%
  \BibitemOpen
  \bibfield  {author} {\bibinfo {author} {\bibfnamefont {A.~C.}\ \bibnamefont
  {Wall}},\ }\bibfield  {title} {\bibinfo {title} {{A Second Law for Higher
  Curvature Gravity}},\ }\href {https://doi.org/10.1142/S0218271815440149}
  {\bibfield  {journal} {\bibinfo  {journal} {Int. J. Mod. Phys. D}\ }\textbf
  {\bibinfo {volume} {24}},\ \bibinfo {pages} {1544014} (\bibinfo {year}
  {2015})},\ \Eprint {https://arxiv.org/abs/1504.08040} {arXiv:1504.08040
  [gr-qc]} \BibitemShut {NoStop}%
\bibitem [{\citenamefont {Dong}(2014)}]{Dong:2013qoa}%
  \BibitemOpen
  \bibfield  {author} {\bibinfo {author} {\bibfnamefont {X.}~\bibnamefont
  {Dong}},\ }\bibfield  {title} {\bibinfo {title} {{Holographic Entanglement
  Entropy for General Higher Derivative Gravity}},\ }\href
  {https://doi.org/10.1007/JHEP01(2014)044} {\bibfield  {journal} {\bibinfo
  {journal} {JHEP}\ }\textbf {\bibinfo {volume} {01}},\ \bibinfo {pages}
  {044}},\ \Eprint {https://arxiv.org/abs/1310.5713} {arXiv:1310.5713 [hep-th]}
  \BibitemShut {NoStop}%
\bibitem [{\citenamefont {Hollands}\ \emph {et~al.}(2022)\citenamefont
  {Hollands}, \citenamefont {Kov{\'a}cs},\ and\ \citenamefont
  {Reall}}]{Hollands:2022fkn}%
  \BibitemOpen
  \bibfield  {author} {\bibinfo {author} {\bibfnamefont {S.}~\bibnamefont
  {Hollands}}, \bibinfo {author} {\bibfnamefont {{\'A}.~D.}\ \bibnamefont
  {Kov{\'a}cs}},\ and\ \bibinfo {author} {\bibfnamefont {H.~S.}\ \bibnamefont
  {Reall}},\ }\bibfield  {title} {\bibinfo {title} {{The second law of black
  hole mechanics in effective field theory}},\ }\href
  {https://doi.org/10.1007/JHEP08(2022)258} {\bibfield  {journal} {\bibinfo
  {journal} {JHEP}\ }\textbf {\bibinfo {volume} {08}},\ \bibinfo {pages}
  {258}},\ \Eprint {https://arxiv.org/abs/2205.15341} {arXiv:2205.15341
  [hep-th]} \BibitemShut {NoStop}%
\bibitem [{\citenamefont {Davies}\ and\ \citenamefont
  {Reall}(2023)}]{Davies:2022xdq}%
  \BibitemOpen
  \bibfield  {author} {\bibinfo {author} {\bibfnamefont {I.}~\bibnamefont
  {Davies}}\ and\ \bibinfo {author} {\bibfnamefont {H.~S.}\ \bibnamefont
  {Reall}},\ }\bibfield  {title} {\bibinfo {title} {{Dynamical Black Hole
  Entropy in Effective Field Theory}},\ }\href
  {https://doi.org/10.1007/JHEP05(2023)006} {\bibfield  {journal} {\bibinfo
  {journal} {JHEP}\ }\textbf {\bibinfo {volume} {05}},\ \bibinfo {pages}
  {006}},\ \Eprint {https://arxiv.org/abs/2212.09777} {arXiv:2212.09777
  [hep-th]} \BibitemShut {NoStop}%
\bibitem [{\citenamefont {Davies}\ and\ \citenamefont
  {Reall}(2024)}]{Davies:2023qaa}%
  \BibitemOpen
  \bibfield  {author} {\bibinfo {author} {\bibfnamefont {I.}~\bibnamefont
  {Davies}}\ and\ \bibinfo {author} {\bibfnamefont {H.~S.}\ \bibnamefont
  {Reall}},\ }\bibfield  {title} {\bibinfo {title} {{Nonperturbative Second Law
  of Black Hole Mechanics in Effective Field Theory}},\ }\href
  {https://doi.org/10.1103/PhysRevLett.132.171402} {\bibfield  {journal}
  {\bibinfo  {journal} {Phys. Rev. Lett.}\ }\textbf {\bibinfo {volume} {132}},\
  \bibinfo {pages} {171402} (\bibinfo {year} {2024})},\ \Eprint
  {https://arxiv.org/abs/2312.07659} {arXiv:2312.07659 [hep-th]} \BibitemShut
  {NoStop}%
\bibitem [{\citenamefont {Biswas}\ \emph {et~al.}(2022)\citenamefont {Biswas},
  \citenamefont {Dhivakar},\ and\ \citenamefont {Kundu}}]{Biswas:2022grc}%
  \BibitemOpen
  \bibfield  {author} {\bibinfo {author} {\bibfnamefont {P.}~\bibnamefont
  {Biswas}}, \bibinfo {author} {\bibfnamefont {P.}~\bibnamefont {Dhivakar}},\
  and\ \bibinfo {author} {\bibfnamefont {N.}~\bibnamefont {Kundu}},\ }\bibfield
   {title} {\bibinfo {title} {{Non-minimal coupling of scalar and gauge fields
  with gravity: an entropy current and linearized second law}},\ }\href
  {https://doi.org/10.1007/JHEP12(2022)036} {\bibfield  {journal} {\bibinfo
  {journal} {JHEP}\ }\textbf {\bibinfo {volume} {12}},\ \bibinfo {pages}
  {036}},\ \Eprint {https://arxiv.org/abs/2206.04538} {arXiv:2206.04538
  [hep-th]} \BibitemShut {NoStop}%
\bibitem [{\citenamefont {Wall}\ and\ \citenamefont
  {Yan}(2024)}]{Wall:2024lbd}%
  \BibitemOpen
  \bibfield  {author} {\bibinfo {author} {\bibfnamefont {A.~C.}\ \bibnamefont
  {Wall}}\ and\ \bibinfo {author} {\bibfnamefont {Z.}~\bibnamefont {Yan}},\
  }\bibfield  {title} {\bibinfo {title} {{Linearized second law for higher
  curvature gravity and nonminimally coupled vector fields}},\ }\href
  {https://doi.org/10.1103/PhysRevD.110.084005} {\bibfield  {journal} {\bibinfo
   {journal} {Phys. Rev. D}\ }\textbf {\bibinfo {volume} {110}},\ \bibinfo
  {pages} {084005} (\bibinfo {year} {2024})},\ \Eprint
  {https://arxiv.org/abs/2402.05411} {arXiv:2402.05411 [gr-qc]} \BibitemShut
  {NoStop}%
\bibitem [{\citenamefont {Yan}(2024)}]{Yan:2024gbz}%
  \BibitemOpen
  \bibfield  {author} {\bibinfo {author} {\bibfnamefont {Z.}~\bibnamefont
  {Yan}},\ }\href@noop {} {\bibinfo {title} {{Gravitational focusing and
  horizon entropy for higher-spin fields}}} (\bibinfo {year} {2024}),\ \Eprint
  {https://arxiv.org/abs/2412.07107} {arXiv:2412.07107 [gr-qc]} \BibitemShut
  {NoStop}%
\bibitem [{\citenamefont {Bhattacharyya}\ \emph {et~al.}(2021)\citenamefont
  {Bhattacharyya}, \citenamefont {Dhivakar}, \citenamefont {Dinda},
  \citenamefont {Kundu}, \citenamefont {Patra},\ and\ \citenamefont
  {Roy}}]{Bhattacharyya:2021jhr}%
  \BibitemOpen
  \bibfield  {author} {\bibinfo {author} {\bibfnamefont {S.}~\bibnamefont
  {Bhattacharyya}}, \bibinfo {author} {\bibfnamefont {P.}~\bibnamefont
  {Dhivakar}}, \bibinfo {author} {\bibfnamefont {A.}~\bibnamefont {Dinda}},
  \bibinfo {author} {\bibfnamefont {N.}~\bibnamefont {Kundu}}, \bibinfo
  {author} {\bibfnamefont {M.}~\bibnamefont {Patra}},\ and\ \bibinfo {author}
  {\bibfnamefont {S.}~\bibnamefont {Roy}},\ }\bibfield  {title} {\bibinfo
  {title} {{An entropy current and the second law in higher derivative theories
  of gravity}},\ }\href {https://doi.org/10.1007/JHEP09(2021)169} {\bibfield
  {journal} {\bibinfo  {journal} {JHEP}\ }\textbf {\bibinfo {volume} {09}},\
  \bibinfo {pages} {169}},\ \Eprint {https://arxiv.org/abs/2105.06455}
  {arXiv:2105.06455 [hep-th]} \BibitemShut {NoStop}%
\bibitem [{\citenamefont {Carter}(2009)}]{Carter:2009nex}%
  \BibitemOpen
  \bibfield  {author} {\bibinfo {author} {\bibfnamefont {B.}~\bibnamefont
  {Carter}},\ }\bibfield  {title} {\bibinfo {title} {{Republication of: Black
  hole equilibrium states}},\ }\href
  {https://doi.org/10.1007/s10714-009-0888-5} {\bibfield  {journal} {\bibinfo
  {journal} {Gen. Rel. Grav.}\ }\textbf {\bibinfo {volume} {41}},\ \bibinfo
  {pages} {2873} (\bibinfo {year} {2009})}\BibitemShut {NoStop}%
\bibitem [{\citenamefont {Gauntlett}\ \emph {et~al.}(1999)\citenamefont
  {Gauntlett}, \citenamefont {Myers},\ and\ \citenamefont
  {Townsend}}]{Gauntlett:1998fz}%
  \BibitemOpen
  \bibfield  {author} {\bibinfo {author} {\bibfnamefont {J.~P.}\ \bibnamefont
  {Gauntlett}}, \bibinfo {author} {\bibfnamefont {R.~C.}\ \bibnamefont
  {Myers}},\ and\ \bibinfo {author} {\bibfnamefont {P.~K.}\ \bibnamefont
  {Townsend}},\ }\bibfield  {title} {\bibinfo {title} {{Black holes of D = 5
  supergravity}},\ }\href {https://doi.org/10.1088/0264-9381/16/1/001}
  {\bibfield  {journal} {\bibinfo  {journal} {Class. Quant. Grav.}\ }\textbf
  {\bibinfo {volume} {16}},\ \bibinfo {pages} {1} (\bibinfo {year} {1999})},\
  \Eprint {https://arxiv.org/abs/hep-th/9810204} {arXiv:hep-th/9810204}
  \BibitemShut {NoStop}%
\bibitem [{\citenamefont {Gao}\ and\ \citenamefont {Wald}(2001)}]{Gao:2001ut}%
  \BibitemOpen
  \bibfield  {author} {\bibinfo {author} {\bibfnamefont {S.}~\bibnamefont
  {Gao}}\ and\ \bibinfo {author} {\bibfnamefont {R.~M.}\ \bibnamefont {Wald}},\
  }\bibfield  {title} {\bibinfo {title} {{The 'Physical process' version of the
  first law and the generalized second law for charged and rotating black
  holes}},\ }\href {https://doi.org/10.1103/PhysRevD.64.084020} {\bibfield
  {journal} {\bibinfo  {journal} {Phys. Rev. D}\ }\textbf {\bibinfo {volume}
  {64}},\ \bibinfo {pages} {084020} (\bibinfo {year} {2001})},\ \Eprint
  {https://arxiv.org/abs/gr-qc/0106071} {arXiv:gr-qc/0106071} \BibitemShut
  {NoStop}%
\bibitem [{\citenamefont {Smoli{\'c}}(2014)}]{Smolic:2014swa}%
  \BibitemOpen
  \bibfield  {author} {\bibinfo {author} {\bibfnamefont {I.}~\bibnamefont
  {Smoli{\'c}}},\ }\bibfield  {title} {\bibinfo {title} {{On the various
  aspects of electromagnetic potentials in spacetimes with symmetries}},\
  }\href {https://doi.org/10.1088/0264-9381/31/23/235002} {\bibfield  {journal}
  {\bibinfo  {journal} {Class. Quant. Grav.}\ }\textbf {\bibinfo {volume}
  {31}},\ \bibinfo {pages} {235002} (\bibinfo {year} {2014})},\ \Eprint
  {https://arxiv.org/abs/1404.1936} {arXiv:1404.1936 [gr-qc]} \BibitemShut
  {NoStop}%
\bibitem [{\citenamefont {Gao}(2003)}]{Gao:2003ys}%
  \BibitemOpen
  \bibfield  {author} {\bibinfo {author} {\bibfnamefont {S.}~\bibnamefont
  {Gao}},\ }\bibfield  {title} {\bibinfo {title} {{The First law of black hole
  mechanics in Einstein-Maxwell and Einstein-Yang-Mills theories}},\ }\href
  {https://doi.org/10.1103/PhysRevD.68.044016} {\bibfield  {journal} {\bibinfo
  {journal} {Phys. Rev. D}\ }\textbf {\bibinfo {volume} {68}},\ \bibinfo
  {pages} {044016} (\bibinfo {year} {2003})},\ \Eprint
  {https://arxiv.org/abs/gr-qc/0304094} {arXiv:gr-qc/0304094} \BibitemShut
  {NoStop}%
\bibitem [{\citenamefont {Emparan}\ and\ \citenamefont
  {Reall}(2008)}]{Emparan:2008eg}%
  \BibitemOpen
  \bibfield  {author} {\bibinfo {author} {\bibfnamefont {R.}~\bibnamefont
  {Emparan}}\ and\ \bibinfo {author} {\bibfnamefont {H.~S.}\ \bibnamefont
  {Reall}},\ }\bibfield  {title} {\bibinfo {title} {{Black Holes in Higher
  Dimensions}},\ }\href {https://doi.org/10.12942/lrr-2008-6} {\bibfield
  {journal} {\bibinfo  {journal} {Living Rev. Rel.}\ }\textbf {\bibinfo
  {volume} {11}},\ \bibinfo {pages} {6} (\bibinfo {year} {2008})},\ \Eprint
  {https://arxiv.org/abs/0801.3471} {arXiv:0801.3471 [hep-th]} \BibitemShut
  {NoStop}%
\bibitem [{\citenamefont {Prabhu}(2017)}]{Prabhu:2015vua}%
  \BibitemOpen
  \bibfield  {author} {\bibinfo {author} {\bibfnamefont {K.}~\bibnamefont
  {Prabhu}},\ }\bibfield  {title} {\bibinfo {title} {{The First Law of Black
  Hole Mechanics for Fields with Internal Gauge Freedom}},\ }\href
  {https://doi.org/10.1088/1361-6382/aa536b} {\bibfield  {journal} {\bibinfo
  {journal} {Class. Quant. Grav.}\ }\textbf {\bibinfo {volume} {34}},\ \bibinfo
  {pages} {035011} (\bibinfo {year} {2017})},\ \Eprint
  {https://arxiv.org/abs/1511.00388} {arXiv:1511.00388 [gr-qc]} \BibitemShut
  {NoStop}%
\bibitem [{\citenamefont {Jacobson}\ \emph {et~al.}(1994)\citenamefont
  {Jacobson}, \citenamefont {Kang},\ and\ \citenamefont
  {Myers}}]{Jacobson:1993vj}%
  \BibitemOpen
  \bibfield  {author} {\bibinfo {author} {\bibfnamefont {T.}~\bibnamefont
  {Jacobson}}, \bibinfo {author} {\bibfnamefont {G.}~\bibnamefont {Kang}},\
  and\ \bibinfo {author} {\bibfnamefont {R.~C.}\ \bibnamefont {Myers}},\
  }\bibfield  {title} {\bibinfo {title} {{On black hole entropy}},\ }\href
  {https://doi.org/10.1103/PhysRevD.49.6587} {\bibfield  {journal} {\bibinfo
  {journal} {Phys. Rev. D}\ }\textbf {\bibinfo {volume} {49}},\ \bibinfo
  {pages} {6587} (\bibinfo {year} {1994})},\ \Eprint
  {https://arxiv.org/abs/gr-qc/9312023} {arXiv:gr-qc/9312023} \BibitemShut
  {NoStop}%
\bibitem [{\citenamefont {Emparan}(2004)}]{Emparan:2004wy}%
  \BibitemOpen
  \bibfield  {author} {\bibinfo {author} {\bibfnamefont {R.}~\bibnamefont
  {Emparan}},\ }\bibfield  {title} {\bibinfo {title} {{Rotating circular
  strings, and infinite nonuniqueness of black rings}},\ }\href
  {https://doi.org/10.1088/1126-6708/2004/03/064} {\bibfield  {journal}
  {\bibinfo  {journal} {JHEP}\ }\textbf {\bibinfo {volume} {03}},\ \bibinfo
  {pages} {064}},\ \Eprint {https://arxiv.org/abs/hep-th/0402149}
  {arXiv:hep-th/0402149} \BibitemShut {NoStop}%
\bibitem [{\citenamefont {Emparan}\ and\ \citenamefont
  {Reall}(2002)}]{Emparan:2001wn}%
  \BibitemOpen
  \bibfield  {author} {\bibinfo {author} {\bibfnamefont {R.}~\bibnamefont
  {Emparan}}\ and\ \bibinfo {author} {\bibfnamefont {H.~S.}\ \bibnamefont
  {Reall}},\ }\bibfield  {title} {\bibinfo {title} {{A Rotating black ring
  solution in five-dimensions}},\ }\href
  {https://doi.org/10.1103/PhysRevLett.88.101101} {\bibfield  {journal}
  {\bibinfo  {journal} {Phys. Rev. Lett.}\ }\textbf {\bibinfo {volume} {88}},\
  \bibinfo {pages} {101101} (\bibinfo {year} {2002})},\ \Eprint
  {https://arxiv.org/abs/hep-th/0110260} {arXiv:hep-th/0110260} \BibitemShut
  {NoStop}%
\bibitem [{\citenamefont {Emparan}\ \emph
  {et~al.}(2010{\natexlab{a}})\citenamefont {Emparan}, \citenamefont {Harmark},
  \citenamefont {Niarchos},\ and\ \citenamefont {Obers}}]{Emparan:2009at}%
  \BibitemOpen
  \bibfield  {author} {\bibinfo {author} {\bibfnamefont {R.}~\bibnamefont
  {Emparan}}, \bibinfo {author} {\bibfnamefont {T.}~\bibnamefont {Harmark}},
  \bibinfo {author} {\bibfnamefont {V.}~\bibnamefont {Niarchos}},\ and\
  \bibinfo {author} {\bibfnamefont {N.~A.}\ \bibnamefont {Obers}},\ }\bibfield
  {title} {\bibinfo {title} {{Essentials of Blackfold Dynamics}},\ }\href
  {https://doi.org/10.1007/JHEP03(2010)063} {\bibfield  {journal} {\bibinfo
  {journal} {JHEP}\ }\textbf {\bibinfo {volume} {03}},\ \bibinfo {pages}
  {063}},\ \Eprint {https://arxiv.org/abs/0910.1601} {arXiv:0910.1601 [hep-th]}
  \BibitemShut {NoStop}%
\bibitem [{\citenamefont {Emparan}\ \emph
  {et~al.}(2010{\natexlab{b}})\citenamefont {Emparan}, \citenamefont {Harmark},
  \citenamefont {Niarchos},\ and\ \citenamefont {Obers}}]{Emparan:2009vd}%
  \BibitemOpen
  \bibfield  {author} {\bibinfo {author} {\bibfnamefont {R.}~\bibnamefont
  {Emparan}}, \bibinfo {author} {\bibfnamefont {T.}~\bibnamefont {Harmark}},
  \bibinfo {author} {\bibfnamefont {V.}~\bibnamefont {Niarchos}},\ and\
  \bibinfo {author} {\bibfnamefont {N.~A.}\ \bibnamefont {Obers}},\ }\bibfield
  {title} {\bibinfo {title} {{New Horizons for Black Holes and Branes}},\
  }\href {https://doi.org/10.1007/JHEP04(2010)046} {\bibfield  {journal}
  {\bibinfo  {journal} {JHEP}\ }\textbf {\bibinfo {volume} {04}},\ \bibinfo
  {pages} {046}},\ \Eprint {https://arxiv.org/abs/0912.2352} {arXiv:0912.2352
  [hep-th]} \BibitemShut {NoStop}%
\bibitem [{\citenamefont {Compere}(2006)}]{Compere:2006my}%
  \BibitemOpen
  \bibfield  {author} {\bibinfo {author} {\bibfnamefont {G.}~\bibnamefont
  {Compere}},\ }\bibfield  {title} {\bibinfo {title} {{An introduction to the
  mechanics of black holes}},\ }in\ \href@noop {} {\emph {\bibinfo {booktitle}
  {{2nd Modave Summer School in Theoretical Physics}}}}\ (\bibinfo {year}
  {2006})\ \Eprint {https://arxiv.org/abs/gr-qc/0611129} {arXiv:gr-qc/0611129}
  \BibitemShut {NoStop}%
\bibitem [{\citenamefont {Wald}(1984)}]{Wald:1984rg}%
  \BibitemOpen
  \bibfield  {author} {\bibinfo {author} {\bibfnamefont {R.~M.}\ \bibnamefont
  {Wald}},\ }\href {https://doi.org/10.7208/chicago/9780226870373.001.0001}
  {\emph {\bibinfo {title} {{General Relativity}}}}\ (\bibinfo  {publisher}
  {Chicago Univ. Pr.},\ \bibinfo {address} {Chicago, USA},\ \bibinfo {year}
  {1984})\BibitemShut {NoStop}%
\bibitem [{\citenamefont {Elgood}\ \emph {et~al.}(2020)\citenamefont {Elgood},
  \citenamefont {Meessen},\ and\ \citenamefont {Ort{\'\i}n}}]{Elgood:2020svt}%
  \BibitemOpen
  \bibfield  {author} {\bibinfo {author} {\bibfnamefont {Z.}~\bibnamefont
  {Elgood}}, \bibinfo {author} {\bibfnamefont {P.}~\bibnamefont {Meessen}},\
  and\ \bibinfo {author} {\bibfnamefont {T.}~\bibnamefont {Ort{\'\i}n}},\
  }\bibfield  {title} {\bibinfo {title} {{The first law of black hole mechanics
  in the Einstein-Maxwell theory revisited}},\ }\href
  {https://doi.org/10.1007/JHEP09(2020)026} {\bibfield  {journal} {\bibinfo
  {journal} {JHEP}\ }\textbf {\bibinfo {volume} {09}},\ \bibinfo {pages}
  {026}},\ \Eprint {https://arxiv.org/abs/2006.02792} {arXiv:2006.02792
  [hep-th]} \BibitemShut {NoStop}%
\bibitem [{\citenamefont {Copsey}\ and\ \citenamefont
  {Horowitz}(2006)}]{Copsey:2005se}%
  \BibitemOpen
  \bibfield  {author} {\bibinfo {author} {\bibfnamefont {K.}~\bibnamefont
  {Copsey}}\ and\ \bibinfo {author} {\bibfnamefont {G.~T.}\ \bibnamefont
  {Horowitz}},\ }\bibfield  {title} {\bibinfo {title} {{The Role of dipole
  charges in black hole thermodynamics}},\ }\href
  {https://doi.org/10.1103/PhysRevD.73.024015} {\bibfield  {journal} {\bibinfo
  {journal} {Phys. Rev. D}\ }\textbf {\bibinfo {volume} {73}},\ \bibinfo
  {pages} {024015} (\bibinfo {year} {2006})},\ \Eprint
  {https://arxiv.org/abs/hep-th/0505278} {arXiv:hep-th/0505278} \BibitemShut
  {NoStop}%
\bibitem [{\citenamefont {Compere}(2007)}]{Compere:2007vx}%
  \BibitemOpen
  \bibfield  {author} {\bibinfo {author} {\bibfnamefont {G.}~\bibnamefont
  {Compere}},\ }\bibfield  {title} {\bibinfo {title} {{Note on the First Law
  with p-form potentials}},\ }\href
  {https://doi.org/10.1103/PhysRevD.75.124020} {\bibfield  {journal} {\bibinfo
  {journal} {Phys. Rev. D}\ }\textbf {\bibinfo {volume} {75}},\ \bibinfo
  {pages} {124020} (\bibinfo {year} {2007})},\ \Eprint
  {https://arxiv.org/abs/hep-th/0703004} {arXiv:hep-th/0703004} \BibitemShut
  {NoStop}%
\bibitem [{\citenamefont {Fischetti}\ \emph {et~al.}(2013)\citenamefont
  {Fischetti}, \citenamefont {Marolf},\ and\ \citenamefont
  {Santos}}]{Fischetti:2012vt}%
  \BibitemOpen
  \bibfield  {author} {\bibinfo {author} {\bibfnamefont {S.}~\bibnamefont
  {Fischetti}}, \bibinfo {author} {\bibfnamefont {D.}~\bibnamefont {Marolf}},\
  and\ \bibinfo {author} {\bibfnamefont {J.~E.}\ \bibnamefont {Santos}},\
  }\bibfield  {title} {\bibinfo {title} {{AdS flowing black funnels: Stationary
  AdS black holes with non-Killing horizons and heat transport in the dual
  CFT}},\ }\href {https://doi.org/10.1088/0264-9381/30/7/075001} {\bibfield
  {journal} {\bibinfo  {journal} {Class. Quant. Grav.}\ }\textbf {\bibinfo
  {volume} {30}},\ \bibinfo {pages} {075001} (\bibinfo {year} {2013})},\
  \Eprint {https://arxiv.org/abs/1212.4820} {arXiv:1212.4820 [hep-th]}
  \BibitemShut {NoStop}%
\bibitem [{\citenamefont {Figueras}\ and\ \citenamefont
  {Wiseman}(2013)}]{Figueras:2012rb}%
  \BibitemOpen
  \bibfield  {author} {\bibinfo {author} {\bibfnamefont {P.}~\bibnamefont
  {Figueras}}\ and\ \bibinfo {author} {\bibfnamefont {T.}~\bibnamefont
  {Wiseman}},\ }\bibfield  {title} {\bibinfo {title} {{Stationary holographic
  plasma quenches and numerical methods for non-Killing horizons}},\ }\href
  {https://doi.org/10.1103/PhysRevLett.110.171602} {\bibfield  {journal}
  {\bibinfo  {journal} {Phys. Rev. Lett.}\ }\textbf {\bibinfo {volume} {110}},\
  \bibinfo {pages} {171602} (\bibinfo {year} {2013})},\ \Eprint
  {https://arxiv.org/abs/1212.4498} {arXiv:1212.4498 [hep-th]} \BibitemShut
  {NoStop}%
\bibitem [{\citenamefont {Papadimitriou}\ and\ \citenamefont
  {Skenderis}(2005)}]{Papadimitriou:2005ii}%
  \BibitemOpen
  \bibfield  {author} {\bibinfo {author} {\bibfnamefont {I.}~\bibnamefont
  {Papadimitriou}}\ and\ \bibinfo {author} {\bibfnamefont {K.}~\bibnamefont
  {Skenderis}},\ }\bibfield  {title} {\bibinfo {title} {{Thermodynamics of
  asymptotically locally AdS spacetimes}},\ }\href
  {https://doi.org/10.1088/1126-6708/2005/08/004} {\bibfield  {journal}
  {\bibinfo  {journal} {JHEP}\ }\textbf {\bibinfo {volume} {08}},\ \bibinfo
  {pages} {004}},\ \Eprint {https://arxiv.org/abs/hep-th/0505190}
  {arXiv:hep-th/0505190} \BibitemShut {NoStop}%
\bibitem [{\citenamefont {Rogatko}(2009)}]{Rogatko:2009th}%
  \BibitemOpen
  \bibfield  {author} {\bibinfo {author} {\bibfnamefont {M.}~\bibnamefont
  {Rogatko}},\ }\bibfield  {title} {\bibinfo {title} {{First Law of p-brane
  Thermodynamics}},\ }\href {https://doi.org/10.1103/PhysRevD.80.044035}
  {\bibfield  {journal} {\bibinfo  {journal} {Phys. Rev. D}\ }\textbf {\bibinfo
  {volume} {80}},\ \bibinfo {pages} {044035} (\bibinfo {year} {2009})},\
  \Eprint {https://arxiv.org/abs/0909.0323} {arXiv:0909.0323 [hep-th]}
  \BibitemShut {NoStop}%
\bibitem [{\citenamefont {Keir}(2014)}]{Keir:2013jga}%
  \BibitemOpen
  \bibfield  {author} {\bibinfo {author} {\bibfnamefont {J.}~\bibnamefont
  {Keir}},\ }\bibfield  {title} {\bibinfo {title} {{Stability, Instability,
  Canonical Energy and Charged Black Holes}},\ }\href
  {https://doi.org/10.1088/0264-9381/31/3/035014} {\bibfield  {journal}
  {\bibinfo  {journal} {Class. Quant. Grav.}\ }\textbf {\bibinfo {volume}
  {31}},\ \bibinfo {pages} {035014} (\bibinfo {year} {2014})},\ \Eprint
  {https://arxiv.org/abs/1306.6087} {arXiv:1306.6087 [gr-qc]} \BibitemShut
  {NoStop}%
\bibitem [{\citenamefont {Harlow}\ and\ \citenamefont
  {Wu}(2020)}]{Harlow:2019yfa}%
  \BibitemOpen
  \bibfield  {author} {\bibinfo {author} {\bibfnamefont {D.}~\bibnamefont
  {Harlow}}\ and\ \bibinfo {author} {\bibfnamefont {J.-Q.}\ \bibnamefont
  {Wu}},\ }\bibfield  {title} {\bibinfo {title} {{Covariant phase space with
  boundaries}},\ }\href {https://doi.org/10.1007/JHEP10(2020)146} {\bibfield
  {journal} {\bibinfo  {journal} {JHEP}\ }\textbf {\bibinfo {volume} {10}},\
  \bibinfo {pages} {146}},\ \Eprint {https://arxiv.org/abs/1906.08616}
  {arXiv:1906.08616 [hep-th]} \BibitemShut {NoStop}%
\bibitem [{\citenamefont {Hajian}\ \emph {et~al.}(2022)\citenamefont {Hajian},
  \citenamefont {Sheikh-Jabbari},\ and\ \citenamefont
  {Tekin}}]{Hajian:2022lgy}%
  \BibitemOpen
  \bibfield  {author} {\bibinfo {author} {\bibfnamefont {K.}~\bibnamefont
  {Hajian}}, \bibinfo {author} {\bibfnamefont {M.~M.}\ \bibnamefont
  {Sheikh-Jabbari}},\ and\ \bibinfo {author} {\bibfnamefont {B.}~\bibnamefont
  {Tekin}},\ }\bibfield  {title} {\bibinfo {title} {{Gauge invariant derivation
  of zeroth and first laws of black hole thermodynamics}},\ }\href
  {https://doi.org/10.1103/PhysRevD.106.104030} {\bibfield  {journal} {\bibinfo
   {journal} {Phys. Rev. D}\ }\textbf {\bibinfo {volume} {106}},\ \bibinfo
  {pages} {104030} (\bibinfo {year} {2022})},\ \Eprint
  {https://arxiv.org/abs/2209.00563} {arXiv:2209.00563 [hep-th]} \BibitemShut
  {NoStop}%
\bibitem [{\citenamefont {Seifert}\ and\ \citenamefont
  {Wald}(2007)}]{Seifert:2006kv}%
  \BibitemOpen
  \bibfield  {author} {\bibinfo {author} {\bibfnamefont {M.~D.}\ \bibnamefont
  {Seifert}}\ and\ \bibinfo {author} {\bibfnamefont {R.~M.}\ \bibnamefont
  {Wald}},\ }\bibfield  {title} {\bibinfo {title} {{A General variational
  principle for spherically symmetric perturbations in diffeomorphism covariant
  theories}},\ }\href {https://doi.org/10.1103/PhysRevD.75.084029} {\bibfield
  {journal} {\bibinfo  {journal} {Phys. Rev. D}\ }\textbf {\bibinfo {volume}
  {75}},\ \bibinfo {pages} {084029} (\bibinfo {year} {2007})},\ \Eprint
  {https://arxiv.org/abs/gr-qc/0612121} {arXiv:gr-qc/0612121} \BibitemShut
  {NoStop}%
\bibitem [{\citenamefont {Wald}(1990)}]{Wald:1990mme}%
  \BibitemOpen
  \bibfield  {author} {\bibinfo {author} {\bibfnamefont {R.~M.}\ \bibnamefont
  {Wald}},\ }\bibfield  {title} {\bibinfo {title} {{On identically closed forms
  locally constructed from a field}},\ }\href
  {https://doi.org/10.1063/1.528839} {\bibfield  {journal} {\bibinfo  {journal}
  {J. Math. Phys.}\ }\textbf {\bibinfo {volume} {31}},\ \bibinfo {pages} {2378}
  (\bibinfo {year} {1990})}\BibitemShut {NoStop}%
\bibitem [{\citenamefont {Hollands}\ and\ \citenamefont
  {Wald}(2013)}]{Hollands:2012sf}%
  \BibitemOpen
  \bibfield  {author} {\bibinfo {author} {\bibfnamefont {S.}~\bibnamefont
  {Hollands}}\ and\ \bibinfo {author} {\bibfnamefont {R.~M.}\ \bibnamefont
  {Wald}},\ }\bibfield  {title} {\bibinfo {title} {{Stability of Black Holes
  and Black Branes}},\ }\href {https://doi.org/10.1007/s00220-012-1638-1}
  {\bibfield  {journal} {\bibinfo  {journal} {Commun. Math. Phys.}\ }\textbf
  {\bibinfo {volume} {321}},\ \bibinfo {pages} {629} (\bibinfo {year}
  {2013})},\ \Eprint {https://arxiv.org/abs/1201.0463} {arXiv:1201.0463
  [gr-qc]} \BibitemShut {NoStop}%
\bibitem [{\citenamefont {Bott}\ and\ \citenamefont {Tu}(1982)}]{Bott:1982xhp}%
  \BibitemOpen
  \bibfield  {author} {\bibinfo {author} {\bibfnamefont {R.}~\bibnamefont
  {Bott}}\ and\ \bibinfo {author} {\bibfnamefont {L.~W.}\ \bibnamefont {Tu}},\
  }\href {https://doi.org/10.1007/978-1-4757-3951-0} {\emph {\bibinfo {title}
  {{Differential Forms in Algebraic Topology}}}}\ (\bibinfo  {publisher}
  {Springer},\ \bibinfo {year} {1982})\BibitemShut {NoStop}%
\bibitem [{\citenamefont {Petersen}()}]{PetersenManifolds}%
  \BibitemOpen
  \bibfield  {author} {\bibinfo {author} {\bibfnamefont {P.}~\bibnamefont
  {Petersen}},\ }\href {https://www.math.ucla.edu/~petersen/manifolds.pdf}
  {\bibinfo {title} {{Manifold Theory}}},\ \bibinfo {note} {lecture notes,
  UCLA}\BibitemShut {NoStop}%
\bibitem [{\citenamefont {Ortin}\ and\ \citenamefont
  {Pere{\~n}iguez}(2022)}]{Ortin:2022uxa}%
  \BibitemOpen
  \bibfield  {author} {\bibinfo {author} {\bibfnamefont {T.}~\bibnamefont
  {Ortin}}\ and\ \bibinfo {author} {\bibfnamefont {D.}~\bibnamefont
  {Pere{\~n}iguez}},\ }\bibfield  {title} {\bibinfo {title} {{Magnetic charges
  and Wald entropy}},\ }\href {https://doi.org/10.1007/JHEP11(2022)081}
  {\bibfield  {journal} {\bibinfo  {journal} {JHEP}\ }\textbf {\bibinfo
  {volume} {11}},\ \bibinfo {pages} {081}},\ \Eprint
  {https://arxiv.org/abs/2207.12008} {arXiv:2207.12008 [hep-th]} \BibitemShut
  {NoStop}%
\bibitem [{\citenamefont {Wald}(1995)}]{Wald:1995yp}%
  \BibitemOpen
  \bibfield  {author} {\bibinfo {author} {\bibfnamefont {R.~M.}\ \bibnamefont
  {Wald}},\ }\href@noop {} {\emph {\bibinfo {title} {{Quantum Field Theory in
  Curved Space-Time and Black Hole Thermodynamics}}}},\ Chicago Lectures in
  Physics\ (\bibinfo  {publisher} {University of Chicago Press},\ \bibinfo
  {address} {Chicago, IL},\ \bibinfo {year} {1995})\BibitemShut {NoStop}%
\bibitem [{\citenamefont {Chamblin}\ \emph {et~al.}(1999)\citenamefont
  {Chamblin}, \citenamefont {Emparan}, \citenamefont {Johnson},\ and\
  \citenamefont {Myers}}]{Chamblin:1999tk}%
  \BibitemOpen
  \bibfield  {author} {\bibinfo {author} {\bibfnamefont {A.}~\bibnamefont
  {Chamblin}}, \bibinfo {author} {\bibfnamefont {R.}~\bibnamefont {Emparan}},
  \bibinfo {author} {\bibfnamefont {C.~V.}\ \bibnamefont {Johnson}},\ and\
  \bibinfo {author} {\bibfnamefont {R.~C.}\ \bibnamefont {Myers}},\ }\bibfield
  {title} {\bibinfo {title} {{Charged AdS black holes and catastrophic
  holography}},\ }\href {https://doi.org/10.1103/PhysRevD.60.064018} {\bibfield
   {journal} {\bibinfo  {journal} {Phys. Rev. D}\ }\textbf {\bibinfo {volume}
  {60}},\ \bibinfo {pages} {064018} (\bibinfo {year} {1999})},\ \Eprint
  {https://arxiv.org/abs/hep-th/9902170} {arXiv:hep-th/9902170} \BibitemShut
  {NoStop}%
\bibitem [{\citenamefont {Caldarelli}\ \emph {et~al.}(2009)\citenamefont
  {Caldarelli}, \citenamefont {Dias},\ and\ \citenamefont
  {Klemm}}]{Caldarelli:2008ze}%
  \BibitemOpen
  \bibfield  {author} {\bibinfo {author} {\bibfnamefont {M.~M.}\ \bibnamefont
  {Caldarelli}}, \bibinfo {author} {\bibfnamefont {O.~J.~C.}\ \bibnamefont
  {Dias}},\ and\ \bibinfo {author} {\bibfnamefont {D.}~\bibnamefont {Klemm}},\
  }\bibfield  {title} {\bibinfo {title} {{Dyonic AdS black holes from
  magnetohydrodynamics}},\ }\href
  {https://doi.org/10.1088/1126-6708/2009/03/025} {\bibfield  {journal}
  {\bibinfo  {journal} {JHEP}\ }\textbf {\bibinfo {volume} {03}},\ \bibinfo
  {pages} {025}},\ \Eprint {https://arxiv.org/abs/0812.0801} {arXiv:0812.0801
  [hep-th]} \BibitemShut {NoStop}%
\bibitem [{\citenamefont {Gibbons}\ and\ \citenamefont
  {Maeda}(1988)}]{Gibbons:1987ps}%
  \BibitemOpen
  \bibfield  {author} {\bibinfo {author} {\bibfnamefont {G.~W.}\ \bibnamefont
  {Gibbons}}\ and\ \bibinfo {author} {\bibfnamefont {K.-i.}\ \bibnamefont
  {Maeda}},\ }\bibfield  {title} {\bibinfo {title} {{Black Holes and Membranes
  in Higher Dimensional Theories with Dilaton Fields}},\ }\href
  {https://doi.org/10.1016/0550-3213(88)90006-5} {\bibfield  {journal}
  {\bibinfo  {journal} {Nucl. Phys. B}\ }\textbf {\bibinfo {volume} {298}},\
  \bibinfo {pages} {741} (\bibinfo {year} {1988})}\BibitemShut {NoStop}%
\bibitem [{\citenamefont {Garfinkle}\ \emph {et~al.}(1991)\citenamefont
  {Garfinkle}, \citenamefont {Horowitz},\ and\ \citenamefont
  {Strominger}}]{Garfinkle:1990qj}%
  \BibitemOpen
  \bibfield  {author} {\bibinfo {author} {\bibfnamefont {D.}~\bibnamefont
  {Garfinkle}}, \bibinfo {author} {\bibfnamefont {G.~T.}\ \bibnamefont
  {Horowitz}},\ and\ \bibinfo {author} {\bibfnamefont {A.}~\bibnamefont
  {Strominger}},\ }\bibfield  {title} {\bibinfo {title} {{Charged black holes
  in string theory}},\ }\href {https://doi.org/10.1103/PhysRevD.43.3140}
  {\bibfield  {journal} {\bibinfo  {journal} {Phys. Rev. D}\ }\textbf {\bibinfo
  {volume} {43}},\ \bibinfo {pages} {3140} (\bibinfo {year} {1991})},\ \bibinfo
  {note} {[Erratum: Phys.Rev.D 45, 3888 (1992)]}\BibitemShut {NoStop}%
\bibitem [{\citenamefont {Horowitz}\ and\ \citenamefont
  {Strominger}(1991)}]{Horowitz:1991cd}%
  \BibitemOpen
  \bibfield  {author} {\bibinfo {author} {\bibfnamefont {G.~T.}\ \bibnamefont
  {Horowitz}}\ and\ \bibinfo {author} {\bibfnamefont {A.}~\bibnamefont
  {Strominger}},\ }\bibfield  {title} {\bibinfo {title} {{Black strings and
  P-branes}},\ }\href {https://doi.org/10.1016/0550-3213(91)90440-9} {\bibfield
   {journal} {\bibinfo  {journal} {Nucl. Phys. B}\ }\textbf {\bibinfo {volume}
  {360}},\ \bibinfo {pages} {197} (\bibinfo {year} {1991})}\BibitemShut
  {NoStop}%
\bibitem [{\citenamefont {Duff}\ and\ \citenamefont {Lu}(1994)}]{Duff:1993ye}%
  \BibitemOpen
  \bibfield  {author} {\bibinfo {author} {\bibfnamefont {M.~J.}\ \bibnamefont
  {Duff}}\ and\ \bibinfo {author} {\bibfnamefont {J.~X.}\ \bibnamefont {Lu}},\
  }\bibfield  {title} {\bibinfo {title} {{Black and super p-branes in diverse
  dimensions}},\ }\href {https://doi.org/10.1016/0550-3213(94)90586-X}
  {\bibfield  {journal} {\bibinfo  {journal} {Nucl. Phys. B}\ }\textbf
  {\bibinfo {volume} {416}},\ \bibinfo {pages} {301} (\bibinfo {year}
  {1994})},\ \Eprint {https://arxiv.org/abs/hep-th/9306052}
  {arXiv:hep-th/9306052} \BibitemShut {NoStop}%
\bibitem [{\citenamefont {Heusler}\ and\ \citenamefont
  {Straumann}(1993)}]{Heusler:1993cj}%
  \BibitemOpen
  \bibfield  {author} {\bibinfo {author} {\bibfnamefont {M.}~\bibnamefont
  {Heusler}}\ and\ \bibinfo {author} {\bibfnamefont {N.}~\bibnamefont
  {Straumann}},\ }\bibfield  {title} {\bibinfo {title} {{The First law of black
  hole physics for a class of nonlinear matter models}},\ }\href
  {https://doi.org/10.1088/0264-9381/10/7/008} {\bibfield  {journal} {\bibinfo
  {journal} {Class. Quant. Grav.}\ }\textbf {\bibinfo {volume} {10}},\ \bibinfo
  {pages} {1299} (\bibinfo {year} {1993})}\BibitemShut {NoStop}%
\bibitem [{\citenamefont {Corichi}\ \emph {et~al.}(2000)\citenamefont
  {Corichi}, \citenamefont {Nucamendi},\ and\ \citenamefont
  {Sudarsky}}]{Corichi:2000dm}%
  \BibitemOpen
  \bibfield  {author} {\bibinfo {author} {\bibfnamefont {A.}~\bibnamefont
  {Corichi}}, \bibinfo {author} {\bibfnamefont {U.}~\bibnamefont {Nucamendi}},\
  and\ \bibinfo {author} {\bibfnamefont {D.}~\bibnamefont {Sudarsky}},\
  }\bibfield  {title} {\bibinfo {title} {{Einstein-Yang-Mills isolated
  horizons: Phase space, mechanics, hair and conjectures}},\ }\href
  {https://doi.org/10.1103/PhysRevD.62.044046} {\bibfield  {journal} {\bibinfo
  {journal} {Phys. Rev. D}\ }\textbf {\bibinfo {volume} {62}},\ \bibinfo
  {pages} {044046} (\bibinfo {year} {2000})},\ \Eprint
  {https://arxiv.org/abs/gr-qc/0002078} {arXiv:gr-qc/0002078} \BibitemShut
  {NoStop}%
\bibitem [{\citenamefont {Ashtekar}\ \emph {et~al.}(2000)\citenamefont
  {Ashtekar}, \citenamefont {Fairhurst},\ and\ \citenamefont
  {Krishnan}}]{Ashtekar:2000hw}%
  \BibitemOpen
  \bibfield  {author} {\bibinfo {author} {\bibfnamefont {A.}~\bibnamefont
  {Ashtekar}}, \bibinfo {author} {\bibfnamefont {S.}~\bibnamefont
  {Fairhurst}},\ and\ \bibinfo {author} {\bibfnamefont {B.}~\bibnamefont
  {Krishnan}},\ }\bibfield  {title} {\bibinfo {title} {{Isolated horizons:
  Hamiltonian evolution and the first law}},\ }\href
  {https://doi.org/10.1103/PhysRevD.62.104025} {\bibfield  {journal} {\bibinfo
  {journal} {Phys. Rev. D}\ }\textbf {\bibinfo {volume} {62}},\ \bibinfo
  {pages} {104025} (\bibinfo {year} {2000})},\ \Eprint
  {https://arxiv.org/abs/gr-qc/0005083} {arXiv:gr-qc/0005083} \BibitemShut
  {NoStop}%
\bibitem [{\citenamefont {Miao}\ and\ \citenamefont
  {Guo}(2015)}]{Miao:2014nxa}%
  \BibitemOpen
  \bibfield  {author} {\bibinfo {author} {\bibfnamefont {R.-X.}\ \bibnamefont
  {Miao}}\ and\ \bibinfo {author} {\bibfnamefont {W.-z.}\ \bibnamefont {Guo}},\
  }\bibfield  {title} {\bibinfo {title} {{Holographic Entanglement Entropy for
  the Most General Higher Derivative Gravity}},\ }\href
  {https://doi.org/10.1007/JHEP08(2015)031} {\bibfield  {journal} {\bibinfo
  {journal} {JHEP}\ }\textbf {\bibinfo {volume} {08}},\ \bibinfo {pages}
  {031}},\ \Eprint {https://arxiv.org/abs/1411.5579} {arXiv:1411.5579 [hep-th]}
  \BibitemShut {NoStop}%
\bibitem [{\citenamefont {Hennigar}\ \emph {et~al.}(2019)\citenamefont
  {Hennigar}, \citenamefont {Kubiz{\v{n}}{\'a}k},\ and\ \citenamefont
  {Mann}}]{Hennigar:2019ive}%
  \BibitemOpen
  \bibfield  {author} {\bibinfo {author} {\bibfnamefont {R.~A.}\ \bibnamefont
  {Hennigar}}, \bibinfo {author} {\bibfnamefont {D.}~\bibnamefont
  {Kubiz{\v{n}}{\'a}k}},\ and\ \bibinfo {author} {\bibfnamefont {R.~B.}\
  \bibnamefont {Mann}},\ }\bibfield  {title} {\bibinfo {title} {{Thermodynamics
  of Lorentzian Taub-NUT spacetimes}},\ }\href
  {https://doi.org/10.1103/PhysRevD.100.064055} {\bibfield  {journal} {\bibinfo
   {journal} {Phys. Rev. D}\ }\textbf {\bibinfo {volume} {100}},\ \bibinfo
  {pages} {064055} (\bibinfo {year} {2019})},\ \Eprint
  {https://arxiv.org/abs/1903.08668} {arXiv:1903.08668 [hep-th]} \BibitemShut
  {NoStop}%
\bibitem [{\citenamefont {Bordo}\ \emph {et~al.}(2019)\citenamefont {Bordo},
  \citenamefont {Gray}, \citenamefont {Hennigar},\ and\ \citenamefont
  {Kubiz{\v{n}}{\'a}k}}]{Bordo:2019tyh}%
  \BibitemOpen
  \bibfield  {author} {\bibinfo {author} {\bibfnamefont {A.~B.}\ \bibnamefont
  {Bordo}}, \bibinfo {author} {\bibfnamefont {F.}~\bibnamefont {Gray}},
  \bibinfo {author} {\bibfnamefont {R.~A.}\ \bibnamefont {Hennigar}},\ and\
  \bibinfo {author} {\bibfnamefont {D.}~\bibnamefont {Kubiz{\v{n}}{\'a}k}},\
  }\bibfield  {title} {\bibinfo {title} {{Misner Gravitational Charges and
  Variable String Strengths}},\ }\href
  {https://doi.org/10.1088/1361-6382/ab3d4d} {\bibfield  {journal} {\bibinfo
  {journal} {Class. Quant. Grav.}\ }\textbf {\bibinfo {volume} {36}},\ \bibinfo
  {pages} {194001} (\bibinfo {year} {2019})},\ \Eprint
  {https://arxiv.org/abs/1905.03785} {arXiv:1905.03785 [hep-th]} \BibitemShut
  {NoStop}%
\bibitem [{\citenamefont {Fatibene}\ \emph {et~al.}(2000)\citenamefont
  {Fatibene}, \citenamefont {Ferraris}, \citenamefont {Francaviglia},\ and\
  \citenamefont {Raiteri}}]{Fatibene:1999ys}%
  \BibitemOpen
  \bibfield  {author} {\bibinfo {author} {\bibfnamefont {L.}~\bibnamefont
  {Fatibene}}, \bibinfo {author} {\bibfnamefont {M.}~\bibnamefont {Ferraris}},
  \bibinfo {author} {\bibfnamefont {M.}~\bibnamefont {Francaviglia}},\ and\
  \bibinfo {author} {\bibfnamefont {M.}~\bibnamefont {Raiteri}},\ }\bibfield
  {title} {\bibinfo {title} {{The Entropy of Taub-Bolt solution}},\ }\href
  {https://doi.org/10.1006/aphy.2000.6062} {\bibfield  {journal} {\bibinfo
  {journal} {Annals Phys.}\ }\textbf {\bibinfo {volume} {284}},\ \bibinfo
  {pages} {197} (\bibinfo {year} {2000})},\ \Eprint
  {https://arxiv.org/abs/gr-qc/9906114} {arXiv:gr-qc/9906114} \BibitemShut
  {NoStop}%
\bibitem [{\citenamefont {Garfinkle}\ and\ \citenamefont
  {Mann}(2000)}]{Garfinkle:2000ms}%
  \BibitemOpen
  \bibfield  {author} {\bibinfo {author} {\bibfnamefont {D.}~\bibnamefont
  {Garfinkle}}\ and\ \bibinfo {author} {\bibfnamefont {R.~B.}\ \bibnamefont
  {Mann}},\ }\bibfield  {title} {\bibinfo {title} {{Generalized entropy and
  Noether charge}},\ }\href {https://doi.org/10.1088/0264-9381/17/16/314}
  {\bibfield  {journal} {\bibinfo  {journal} {Class. Quant. Grav.}\ }\textbf
  {\bibinfo {volume} {17}},\ \bibinfo {pages} {3317} (\bibinfo {year}
  {2000})},\ \Eprint {https://arxiv.org/abs/gr-qc/0004056}
  {arXiv:gr-qc/0004056} \BibitemShut {NoStop}%
\bibitem [{\citenamefont {Majumdar}(1947)}]{Majumdar:1947eu}%
  \BibitemOpen
  \bibfield  {author} {\bibinfo {author} {\bibfnamefont {S.~D.}\ \bibnamefont
  {Majumdar}},\ }\bibfield  {title} {\bibinfo {title} {A class of exact
  solutions of einstein's field equations},\ }\href
  {https://doi.org/10.1103/PhysRev.72.390} {\bibfield  {journal} {\bibinfo
  {journal} {Phys. Rev.}\ }\textbf {\bibinfo {volume} {72}},\ \bibinfo {pages}
  {390} (\bibinfo {year} {1947})}\BibitemShut {NoStop}%
\bibitem [{\citenamefont {Papaetrou}(1947)}]{Papaetrou:1947ib}%
  \BibitemOpen
  \bibfield  {author} {\bibinfo {author} {\bibfnamefont {A.}~\bibnamefont
  {Papaetrou}},\ }\bibfield  {title} {\bibinfo {title} {{A Static solution of
  the equations of the gravitational field for an arbitrary charge
  distribution}},\ }\href@noop {} {\bibfield  {journal} {\bibinfo  {journal}
  {Proc. Roy. Irish Acad. A}\ }\textbf {\bibinfo {volume} {51}},\ \bibinfo
  {pages} {191} (\bibinfo {year} {1947})}\BibitemShut {NoStop}%
\bibitem [{\citenamefont {Elvang}\ and\ \citenamefont
  {Figueras}(2007)}]{Elvang:2007rd}%
  \BibitemOpen
  \bibfield  {author} {\bibinfo {author} {\bibfnamefont {H.}~\bibnamefont
  {Elvang}}\ and\ \bibinfo {author} {\bibfnamefont {P.}~\bibnamefont
  {Figueras}},\ }\bibfield  {title} {\bibinfo {title} {Black saturn},\ }\href
  {https://doi.org/10.1088/1126-6708/2007/05/050} {\bibfield  {journal}
  {\bibinfo  {journal} {JHEP}\ }\textbf {\bibinfo {volume} {05}},\ \bibinfo
  {pages} {050}},\ \Eprint {https://arxiv.org/abs/hep-th/0701035}
  {hep-th/0701035} \BibitemShut {NoStop}%
\bibitem [{\citenamefont {Iguchi}\ and\ \citenamefont
  {Mishima}(2007)}]{Iguchi:2007is}%
  \BibitemOpen
  \bibfield  {author} {\bibinfo {author} {\bibfnamefont {H.}~\bibnamefont
  {Iguchi}}\ and\ \bibinfo {author} {\bibfnamefont {T.}~\bibnamefont
  {Mishima}},\ }\bibfield  {title} {\bibinfo {title} {{Black di-ring and
  infinite nonuniqueness}},\ }\href
  {https://doi.org/10.1103/PhysRevD.78.069903} {\bibfield  {journal} {\bibinfo
  {journal} {Phys. Rev. D}\ }\textbf {\bibinfo {volume} {75}},\ \bibinfo
  {pages} {064018} (\bibinfo {year} {2007})},\ \bibinfo {note} {[Erratum:
  Phys.Rev.D 78, 069903 (2008)]},\ \Eprint
  {https://arxiv.org/abs/hep-th/0701043} {arXiv:hep-th/0701043} \BibitemShut
  {NoStop}%
\bibitem [{\citenamefont {Elvang}\ and\ \citenamefont
  {Rodriguez}(2008)}]{Elvang:2007hg}%
  \BibitemOpen
  \bibfield  {author} {\bibinfo {author} {\bibfnamefont {H.}~\bibnamefont
  {Elvang}}\ and\ \bibinfo {author} {\bibfnamefont {M.~J.}\ \bibnamefont
  {Rodriguez}},\ }\bibfield  {title} {\bibinfo {title} {Bicycling black
  rings},\ }\href {https://doi.org/10.1088/1126-6708/2008/04/045} {\bibfield
  {journal} {\bibinfo  {journal} {JHEP}\ }\textbf {\bibinfo {volume} {04}},\
  \bibinfo {pages} {045}},\ \Eprint {https://arxiv.org/abs/0712.2425}
  {arXiv:0712.2425 [hep-th]} \BibitemShut {NoStop}%
\bibitem [{\citenamefont {Izumi}(2008)}]{Izumi:2007qx}%
  \BibitemOpen
  \bibfield  {author} {\bibinfo {author} {\bibfnamefont {K.}~\bibnamefont
  {Izumi}},\ }\bibfield  {title} {\bibinfo {title} {{Orthogonal black di-ring
  solution}},\ }\href {https://doi.org/10.1143/PTP.119.757} {\bibfield
  {journal} {\bibinfo  {journal} {Prog. Theor. Phys.}\ }\textbf {\bibinfo
  {volume} {119}},\ \bibinfo {pages} {757} (\bibinfo {year} {2008})},\ \Eprint
  {https://arxiv.org/abs/0712.0902} {arXiv:0712.0902 [hep-th]} \BibitemShut
  {NoStop}%
\bibitem [{\citenamefont {Gibbons}\ and\ \citenamefont
  {Hawking}(1977)}]{Gibbons:1977mu}%
  \BibitemOpen
  \bibfield  {author} {\bibinfo {author} {\bibfnamefont {G.~W.}\ \bibnamefont
  {Gibbons}}\ and\ \bibinfo {author} {\bibfnamefont {S.~W.}\ \bibnamefont
  {Hawking}},\ }\bibfield  {title} {\bibinfo {title} {Cosmological event
  horizons, thermodynamics, and particle creation},\ }\href
  {https://doi.org/10.1103/PhysRevD.15.2738} {\bibfield  {journal} {\bibinfo
  {journal} {Phys. Rev. D}\ }\textbf {\bibinfo {volume} {15}},\ \bibinfo
  {pages} {2738} (\bibinfo {year} {1977})}\BibitemShut {NoStop}%
\bibitem [{\citenamefont {Kastor}\ and\ \citenamefont
  {Traschen}(1993)}]{Kastor:1992nn}%
  \BibitemOpen
  \bibfield  {author} {\bibinfo {author} {\bibfnamefont {D.}~\bibnamefont
  {Kastor}}\ and\ \bibinfo {author} {\bibfnamefont {J.}~\bibnamefont
  {Traschen}},\ }\bibfield  {title} {\bibinfo {title} {Cosmological
  multi-black-hole solutions},\ }\href
  {https://doi.org/10.1103/PhysRevD.47.5370} {\bibfield  {journal} {\bibinfo
  {journal} {Phys. Rev. D}\ }\textbf {\bibinfo {volume} {47}},\ \bibinfo
  {pages} {5370} (\bibinfo {year} {1993})},\ \Eprint
  {https://arxiv.org/abs/hep-th/9212035} {hep-th/9212035} \BibitemShut
  {NoStop}%
\bibitem [{\citenamefont {Zhao}(2025)}]{Zhao:2025zny}%
  \BibitemOpen
  \bibfield  {author} {\bibinfo {author} {\bibfnamefont {J.}~\bibnamefont
  {Zhao}},\ }\bibfield  {title} {\bibinfo {title} {{The entropy of dynamical de
  Sitter horizons}},\ }\href {https://doi.org/10.1140/epjc/s10052-025-14471-9}
  {\bibfield  {journal} {\bibinfo  {journal} {Eur. Phys. J. C}\ }\textbf
  {\bibinfo {volume} {85}},\ \bibinfo {pages} {750} (\bibinfo {year} {2025})},\
  \Eprint {https://arxiv.org/abs/2503.16138} {arXiv:2503.16138 [gr-qc]}
  \BibitemShut {NoStop}%
\bibitem [{\citenamefont {Collins}(1992)}]{Collins:1992eca}%
  \BibitemOpen
  \bibfield  {author} {\bibinfo {author} {\bibfnamefont {W.}~\bibnamefont
  {Collins}},\ }\bibfield  {title} {\bibinfo {title} {{Mechanics of apparent
  horizons}},\ }\href {https://doi.org/10.1103/PhysRevD.45.495} {\bibfield
  {journal} {\bibinfo  {journal} {Phys. Rev. D}\ }\textbf {\bibinfo {volume}
  {45}},\ \bibinfo {pages} {495} (\bibinfo {year} {1992})}\BibitemShut
  {NoStop}%
\bibitem [{\citenamefont {Compère}\ \emph {et~al.}(2011)\citenamefont
  {Compère}, \citenamefont {McFadden}, \citenamefont {Skenderis},\ and\
  \citenamefont {Taylor}}]{Compere:1103.3022}%
  \BibitemOpen
  \bibfield  {author} {\bibinfo {author} {\bibfnamefont {G.}~\bibnamefont
  {Compère}}, \bibinfo {author} {\bibfnamefont {P.}~\bibnamefont {McFadden}},
  \bibinfo {author} {\bibfnamefont {K.}~\bibnamefont {Skenderis}},\ and\
  \bibinfo {author} {\bibfnamefont {M.}~\bibnamefont {Taylor}},\ }\bibfield
  {title} {\bibinfo {title} {{The Holographic fluid dual to vacuum Einstein
  gravity}},\ }\href {https://doi.org/10.1007/JHEP07(2011)050} {\bibfield
  {journal} {\bibinfo  {journal} {JHEP}\ }\textbf {\bibinfo {volume} {07}},\
  \bibinfo {pages} {050}},\ \Eprint {https://arxiv.org/abs/1103.3022}
  {arXiv:1103.3022 [hep-th]} \BibitemShut {NoStop}%
\bibitem [{\citenamefont {Bakas}\ and\ \citenamefont
  {Skenderis}(2014)}]{Skenderis:1404.4824}%
  \BibitemOpen
  \bibfield  {author} {\bibinfo {author} {\bibfnamefont {I.}~\bibnamefont
  {Bakas}}\ and\ \bibinfo {author} {\bibfnamefont {K.}~\bibnamefont
  {Skenderis}},\ }\bibfield  {title} {\bibinfo {title} {{Non-equilibrium
  dynamics and AdS$_4$ Robinson–Trautman}},\ }\href
  {https://doi.org/10.1007/JHEP08(2014)056} {\bibfield  {journal} {\bibinfo
  {journal} {JHEP}\ }\textbf {\bibinfo {volume} {08}},\ \bibinfo {pages}
  {056}},\ \Eprint {https://arxiv.org/abs/1404.4824} {arXiv:1404.4824 [hep-th]}
  \BibitemShut {NoStop}%
\bibitem [{\citenamefont {Eling}\ \emph {et~al.}(2012)\citenamefont {Eling},
  \citenamefont {Meyer},\ and\ \citenamefont {Oz}}]{Eling:2012xa}%
  \BibitemOpen
  \bibfield  {author} {\bibinfo {author} {\bibfnamefont {C.}~\bibnamefont
  {Eling}}, \bibinfo {author} {\bibfnamefont {A.}~\bibnamefont {Meyer}},\ and\
  \bibinfo {author} {\bibfnamefont {Y.}~\bibnamefont {Oz}},\ }\bibfield
  {title} {\bibinfo {title} {{Local Entropy Current in Higher Curvature Gravity
  and Rindler Hydrodynamics}},\ }\href
  {https://doi.org/10.1007/JHEP08(2012)088} {\bibfield  {journal} {\bibinfo
  {journal} {JHEP}\ }\textbf {\bibinfo {volume} {08}},\ \bibinfo {pages}
  {088}},\ \Eprint {https://arxiv.org/abs/1205.4249} {arXiv:1205.4249 [hep-th]}
  \BibitemShut {NoStop}%
\bibitem [{\citenamefont {Eling}\ and\ \citenamefont
  {Oz}(2011)}]{Eling:2011ct}%
  \BibitemOpen
  \bibfield  {author} {\bibinfo {author} {\bibfnamefont {C.}~\bibnamefont
  {Eling}}\ and\ \bibinfo {author} {\bibfnamefont {Y.}~\bibnamefont {Oz}},\
  }\bibfield  {title} {\bibinfo {title} {{Holographic Screens and Transport
  Coefficients in the Fluid/Gravity Correspondence}},\ }\href
  {https://doi.org/10.1103/PhysRevLett.107.201602} {\bibfield  {journal}
  {\bibinfo  {journal} {Phys. Rev. Lett.}\ }\textbf {\bibinfo {volume} {107}},\
  \bibinfo {pages} {201602} (\bibinfo {year} {2011})},\ \Eprint
  {https://arxiv.org/abs/1107.2134} {arXiv:1107.2134 [hep-th]} \BibitemShut
  {NoStop}%
\bibitem [{\citenamefont {Chandranathan}\ \emph {et~al.}(2023)\citenamefont
  {Chandranathan}, \citenamefont {Bhattacharyya}, \citenamefont {Patra},\ and\
  \citenamefont {Roy}}]{Chandranathan:2022pfx}%
  \BibitemOpen
  \bibfield  {author} {\bibinfo {author} {\bibfnamefont {A.}~\bibnamefont
  {Chandranathan}}, \bibinfo {author} {\bibfnamefont {S.}~\bibnamefont
  {Bhattacharyya}}, \bibinfo {author} {\bibfnamefont {M.}~\bibnamefont
  {Patra}},\ and\ \bibinfo {author} {\bibfnamefont {S.}~\bibnamefont {Roy}},\
  }\bibfield  {title} {\bibinfo {title} {{Entropy current and fluid-gravity
  duality in Gauss-Bonnet theory}},\ }\href
  {https://doi.org/10.1007/JHEP09(2023)070} {\bibfield  {journal} {\bibinfo
  {journal} {JHEP}\ }\textbf {\bibinfo {volume} {09}},\ \bibinfo {pages}
  {070}},\ \Eprint {https://arxiv.org/abs/2208.07856} {arXiv:2208.07856
  [hep-th]} \BibitemShut {NoStop}%
\bibitem [{\citenamefont {Saha}(2021)}]{Saha:2020zho}%
  \BibitemOpen
  \bibfield  {author} {\bibinfo {author} {\bibfnamefont {A.}~\bibnamefont
  {Saha}},\ }\bibfield  {title} {\bibinfo {title} {{General Theory of Large D
  Membranes Consistent with Second Law of Thermodynamics}},\ }\href
  {https://doi.org/10.1007/JHEP04(2021)152} {\bibfield  {journal} {\bibinfo
  {journal} {JHEP}\ }\textbf {\bibinfo {volume} {04}},\ \bibinfo {pages}
  {152}},\ \Eprint {https://arxiv.org/abs/2012.12834} {arXiv:2012.12834
  [hep-th]} \BibitemShut {NoStop}%
\end{thebibliography}%

\end{document}